\documentclass[11pt]{article}

\usepackage{float}
\usepackage[normalem]{ulem}
\usepackage[a4paper,margin=1in]{geometry}
\usepackage{amsmath,amssymb,amsthm,mathtools,mathrsfs}
\usepackage{xcolor}
\usepackage{enumitem}
\usepackage{booktabs}
\usepackage{tikz}
\usepackage{hyperref}
\usepackage{aliascnt}
\usepackage[nameinlink,capitalize]{cleveref}

\usetikzlibrary{arrows.meta,positioning,calc,fit}

\hypersetup{
  colorlinks=true,
  linkcolor=blue,
  citecolor=blue,
  urlcolor=blue
}

\newtheorem{theorem}{Theorem}
\newtheorem*{informaltheorem}{Theorem}
\newaliascnt{lemma}{theorem}
\newtheorem{lemma}[lemma]{Lemma}
\aliascntresetthe{lemma}
\newaliascnt{proposition}{theorem}
\newtheorem{proposition}[proposition]{Proposition}
\aliascntresetthe{proposition}
\newaliascnt{corollary}{theorem}
\newtheorem{corollary}[corollary]{Corollary}
\aliascntresetthe{corollary}
\newaliascnt{assumption}{theorem}

\aliascntresetthe{assumption}

\crefname{proofinput}{model condition}{model conditions}
\Crefname{proofinput}{Model Condition}{Model Conditions}
\theoremstyle{definition}
\newaliascnt{definition}{theorem}
\newtheorem{definition}[definition]{Definition}
\aliascntresetthe{definition}
\newaliascnt{remark}{theorem}
\newtheorem{remark}[remark]{Remark}
\aliascntresetthe{remark}

\newcommand{\C}{\mathbb C}

\newcommand{\ran}{\operatorname{ran}}

\newcommand{\gap}{\operatorname{gap}}

\newcommand{\cB}{\mathcal B}

\newcommand{\cG}{\mathcal G}
\newcommand{\cH}{\mathcal H}

\definecolor{figureInterface}{HTML}{008779}
\definecolor{figureComplement}{HTML}{7955A4}
\definecolor{figureCoupling}{HTML}{C45D43}
\definecolor{figureGeometry}{HTML}{287DAD}
\definecolor{figureEstimate}{HTML}{347C9D}
\definecolor{figureMechanism}{HTML}{70548C}
\definecolor{figureAmber}{HTML}{D08B2D}
\definecolor{figurePink}{HTML}{B45182}

\usepackage{needspace}
\usepackage{tabularx}
\usepackage{placeins}
\usepackage{authblk}

\usepackage[most]{tcolorbox}
\renewenvironment{quote}
  {\list{}{\leftmargin=1em\rightmargin=1em}\item\relax}
  {\endlist}
\newcolumntype{Y}{>{\raggedright\arraybackslash}X}

\newcommand{\PathFam}{\ensuremath{\mathscr P}}

\newif\ifcoauthorchecklist
\coauthorchecklisttrue
\theoremstyle{plain}
\newaliascnt{claim}{theorem} 
\newtheorem{claim}[claim]{Claim}
\aliascntresetthe{claim}
\crefname{claim}{claim}{claims}
\Crefname{claim}{Claim}{Claims}
\theoremstyle{definition}
\newtcolorbox{examplebox}[1]{
  enhanced,
  breakable,
  colback=blue!2,
  colbacktitle=blue!8,
  colframe=blue!35!black,
  coltitle=black,
  boxrule=0.45pt,
  arc=1pt,
  left=6pt,
  right=6pt,
  top=5pt,
  bottom=5pt,
  before skip=8pt,
  after skip=8pt,
  fonttitle=\bfseries,
  title={Example: #1}
}

\title{Double Localization for Quantum Gibbs Sampler Gaps:\\
From an Abstract Framework to Finite-Group Models}

\author[1]{Ryu Hayakawa
\thanks{Email: \texttt{ryu.hayakawa@yukawa.kyoto-u.ac.jp}}}
\author[2]{Angus Southwell}
\author[2,4]{Caesnan M. G. Leditto}
\author[3]{\authorcr Kuo-Chin Chen}
\author[3]{Min-Hsiu Hsieh
\thanks{Email: \texttt{min-hsiu.hsieh@foxconn.com}}}
\affil[1]{Yukawa Institute for Theoretical Physics \& The Hakubi Center, Kyoto University, Japan}
\affil[2]{School of Physics and Astronomy, Monash University, Clayton, VIC 3168, Australia}
\affil[3]{Hon-Hai Research Institute, Taipei, Taiwan}
\affil[4]{Faculty of Science and Computer, Universitas Kristen Immanuel, Yogyakarta, Indonesia}

\date{}

\begin{document}

\noindent
\hspace{\fill} YITP-26-129
\begingroup
\let\newpage\relax
\maketitle
\endgroup

\vspace{-1em}
\begin{abstract}
We introduce \emph{double localization}, a framework for proving spectral gaps of quantum Gibbs samplers through two complementary operations: \emph{geometric localization}, which selects updates supported in small spatial regions, and \emph{interface localization}, which focuses on a model-defined subspace of observables while remaining geometrically global. Our abstract gap theorem combines a global bound on this subspace, called the \emph{interface}, with local quantum gap estimates under quantitative coupling and geometric assembly conditions. The resulting explicit lower bound controls the gap of the full dynamics without requiring the interface to be invariant under the generator. This separation allows global estimates, including classical comparison, to be combined with local control of the remaining quantum directions.

We apply the framework to Hamiltonians built from the vertex terms of Kitaev's finite-group quantum double construction, without plaquette terms, on finite simple three-regular graphs of girth at least six. For the specified local Davies dynamics, we prove an unconditional spectral-gap lower bound by a positive constant independent of graph size for every fixed nontrivial finite group and each fixed inverse temperature $\beta$ with $0\le\beta J<\log(5/3)$, where $J>0$ is the coupling strength. For the smallest non-Abelian group $S_3$, double localization yields a positive lower bound uniform in both graph size and the entire interval $0\le\beta J\le\log 3$, providing a guarantee that does not follow directly from existing results.
This demonstrates how model-specific finite calculations can extend spectral-gap guarantees for the full quantum dynamics.
\end{abstract}

\section{Introduction}
\label{sec:new-introduction}

{
How quickly does a quantum many-body system approach thermal equilibrium?
Answering this question is crucial to understanding thermalization in many-body physics and to designing Gibbs-state preparation algorithms for quantum computers
\cite{kastoryano2016quantum,chen2023exact,rouze2024optimal}.
Beyond its natural role in many-body physics \cite{temme2011metropolisintro}, quantum Gibbs-state preparation is a useful ingredient in algorithms for optimisation \cite{brandao2017sdpintro} and in quantum machine-learning models \cite{amin2018boltzmannintro}.
The computational study of thermal equilibrium also has close connections to quantum complexity theory \cite{bravyi2022thermalintro}.
These connections motivate the search for efficient methods of preparing Gibbs states for different classes of Hamiltonians.

A central approach to Gibbs-state preparation is to simulate dynamics that relaxes toward the target Gibbs state.
Such dynamics can be described by a Lindbladian, a linear map that generates the evolution of an open quantum system, analogous to the generator of a continuous-time Markov chain.
The Davies construction derives such generators from the weak-coupling limit of a system interacting with a thermal bath \cite{davies1974markovian}.
Alternatively, samplers can be designed for implementation on a quantum computer, as in the Chen--Kastoryano--Gily\'en (CKG) construction \cite{chen2023exact}.
These approaches motivate studying both the relaxation of specified physical dynamics and the design of efficiently implementable Gibbs samplers.

To quantify the rate of convergence to equilibrium, one studies the generator of the evolution, whose spectral gap provides a bound on that rate
\cite{temme2013lower,kastoryano2016quantum}.
Bounding this gap is thus key to quantifying thermal relaxation and establishing guarantees for Gibbs-state preparation.
Inverse-polynomial gap bounds can support polynomial-time preparation guarantees, while a size-independent lower bound gives stronger control of relaxation as the system grows.

Determining the spectral gap and relaxation rate of a quantum many-body system remains a major challenge, even when its dynamics consists of local updates.
Even for commuting Hamiltonians, commutativity alone does not provide a useful bound on the relaxation rate.
A variety of methods have advanced our understanding of relaxation in such locally updated systems, including approaches based on clustering, tensor networks, and quantum-to-classical comparison \cite{kastoryano2016quantum,lucia2023thermalization,basso2025quantum}.
However, verifying the conditions required by these methods can itself be difficult for a given many-body system.
Overcoming this difficulty and developing methods that extend quantitative relaxation bounds to a broader range of systems are therefore urgent tasks.
\par}
Our starting point is two observations suggested by existing approaches to spectral-gap bounds: one concerns spatial locality, and the other concerns model-specific structure.
First, spatial locality can make a global relaxation problem accessible through estimates on smaller regions. Clustering and tensorization methods control dependence between regions and use this control to combine regional estimates into global spectral-gap or mixing bounds \cite{kastoryano2016quantum,bardet2022approximate}. This perspective is illustrated in the left panel of \Cref{fig:intro-prior-perspectives}.
Second, model-specific structure can make classical methods useful for quantum problems. For certain commuting Hamiltonians, Gibbs-state preparation can be reduced to classical sampling \cite{hwang2024gibbs,paez-velasco2025efficient}. 
For specified Davies dynamics, quantum-to-classical comparison can instead use classical gap estimates to bound the gap of the full quantum dynamics \cite{basso2025quantum}.
The right panel of \Cref{fig:intro-prior-perspectives} distinguishes these two routes: one concerns state preparation, while the other concerns the gap of a specified dynamical generator.

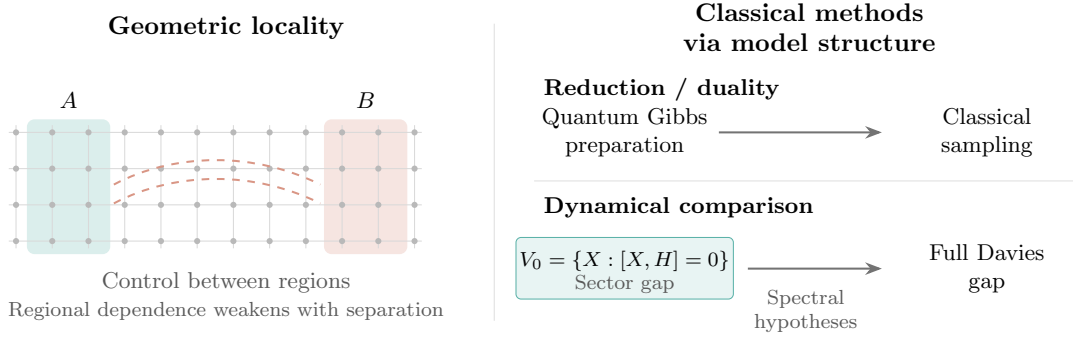
\begin{figure}[!t]
  \centering
  \resizebox{0.95\linewidth}{!}{
  \begingroup
\begin{tikzpicture}[font=\footnotesize,>=Stealth,
  flow/.style={->,black!55,line width=.8pt}]
  \path[use as bounding box] (0,0) rectangle (15.4,4.4);
  \node[font=\small\bfseries] at (3.35,4.3) {Geometric locality};
  \node[font=\small\bfseries,align=center] at (11.4,4.3)
    {Classical methods\\via model structure};
  \draw[gray!25] (7.05,.25)--(7.05,4.35);
  \begin{scope}[shift={(.45,1.35)}]
    \fill[figureInterface!14,rounded corners=3pt] (.15,-.18) rectangle (1.3,1.68);
    \fill[figureCoupling!16,rounded corners=3pt] (4.25,-.18) rectangle (5.4,1.68);
    \foreach \i in {0,...,11} \draw[gray!30] ({.5*\i},-.1)--({.5*\i},1.6);
    \foreach \j in {0,...,3} \draw[gray!30] (-.1,{.5*\j})--(5.6,{.5*\j});
    \foreach \i in {0,...,11} \foreach \j in {0,...,3}
      \fill[gray!55] ({.5*\i},{.5*\j}) circle (1.2pt);
    \node at (.725,1.95) {$A$};
    \node at (4.825,1.95) {$B$};
    \foreach \y in {.52,.78}
      \draw[figureCoupling!65,dashed,line width=.8pt] (1.35,\y)
        .. controls (2.2,{\y+.45}) and (3.3,{\y+.45}) .. (4.2,\y);
  \end{scope}
  \node[text=black!65] at (3.35,.78) {Control between regions};
  \node[text=black!65,font=\scriptsize] at (3.35,.38) {Regional dependence weakens with separation};
  \node[anchor=west,font=\footnotesize\bfseries] at (7.6,3.45) {Reduction / duality};
  \node[align=center] at (8.85,2.85) {Quantum Gibbs\\preparation};
  \draw[flow] (10.15,2.85)--(12.4,2.85);
  \node[align=center] at (13.85,2.85) {Classical\\sampling};
  \draw[gray!25] (7.6,2.18)--(15.2,2.18);
  \node[anchor=west,font=\footnotesize\bfseries] at (7.6,1.82) {Dynamical comparison};
  \filldraw[fill=figureInterface!9,draw=figureInterface!65,rounded corners=2pt]
    (7.35,.55) rectangle (10.35,1.4);
  \node[font=\scriptsize] at (8.85,1.08) {$V_0=\{X:[X,H]=0\}$};
  \node[font=\scriptsize,text=black!65] at (8.85,.76) {Sector gap};
  \draw[flow] (10.55,.97)--(12.4,.97);
  \node[font=\scriptsize,text=black!65,align=center] at (11.35,.39)
    {Spectral\\hypotheses};
  \node[align=center] at (13.85,.97) {Full Davies\\gap};
\end{tikzpicture}
\endgroup
  }
  \caption{Two perspectives motivating our approach. Left: clustering controls
  quantum dependence between regions \cite{kastoryano2016quantum,bardet2022approximate}.
  Right: reductions connect Gibbs preparation to classical sampling
  \cite{hwang2024gibbs,paez-velasco2025efficient}, whereas dynamical comparison
  lifts a sector gap under spectral hypotheses \cite{basso2025quantum}.
  }
  \label{fig:intro-prior-perspectives}
\end{figure}

These observations suggest combining the two perspectives within a unified framework. We therefore ask: 
\vspace{-0.5em}
\begin{quote}
\textbf{Can spatial locality and model-specific structure be combined to establish a global spectral gap under tractable conditions?}
\end{quote}

\noindent
Here, by ``tractable'' we mean conditions that are easier to verify than a direct estimate of the global spectral gap.
We address this question by developing a framework that combines these two perspectives and demonstrating its effectiveness in a concrete family of quantum many-body systems.

\subsection{Our Contributions}
\label{sec:intro-contributions}

Our main contributions are twofold. First, we prove an abstract spectral-gap
theorem based on a method we call \emph{double localization}: a combination
of (1) \emph{geometric localization} and (2) \emph{interface localization}, specified below.
This technique offers considerable flexibility in choosing both a subset of updates and a subspace of observables.
Consequently, deriving bounds on the spectral gap becomes more tractable.
Second, we apply this theorem to Hamiltonians built from the vertex terms of Kitaev's quantum double construction
\cite{kitaev2003faulttolerant}.
For the associated Davies dynamics,\footnote{Davies generators describe
the weak-coupling limit of a quantum system interacting with a thermal bath
and have the Gibbs state as a stationary state
\cite{davies1974markovian}.}
we establish a spectral gap bounded below by a positive constant
independent of system size over the temperature ranges specified below.
For the smallest non-Abelian group $S_3$, double localization extends
this guarantee into a temperature range beyond the reach of known
techniques. Although the vertex-term Hamiltonians provide a simplified
setting, their connection to quantum double models and fault-tolerant
quantum computation \cite{kitaev2003faulttolerant,lucia2023thermalization}
makes this result relevant to understanding thermal relaxation in
such systems.

\begin{figure}[!t]
  \centering
  \color{black}
  \resizebox{0.9\linewidth}{!}{%
    \begingroup
\colorlet{axisblue}{figureGeometry}
\colorlet{axisamber}{figureAmber}
\colorlet{axispink}{figurePink}
\colorlet{axisteal}{figureInterface}
\colorlet{axispurple}{figureComplement}
\colorlet{axiscoral}{figureCoupling}
\begin{tikzpicture}[font=\small,>=Stealth,
  flow/.style={->,black!55,line width=.8pt},
  split/.pic={
    \fill[axisteal!12,rounded corners=3pt] (0,0) rectangle (2,1.85);
    \fill[axispurple!13,rounded corners=3pt] (2.12,0) rectangle (4.24,1.85);
    \node[text=axisteal,font=\scriptsize\bfseries] at (1,1.48) {Interface};
    \node[text=axispurple,font=\scriptsize\bfseries] at (3.18,1.48) {Complement};
    \node[text=axisteal] at (1,.97) {$\mathcal M$};
    \node[text=axispurple] at (3.18,.97) {$\mathcal M^\perp$};
    \node[text=axisteal,font=\footnotesize] at (1,.48) {$A\blockindex$};
    \draw[<->,axiscoral,line width=.8pt] (1.52,.5) .. controls (1.82,.85) and (2.42,.85) .. (2.72,.5);
    \node[text=axiscoral,font=\footnotesize] at (2.12,.2) {Coupling};
  }]
  \path[use as bounding box] (0,-.1) rectangle (15.4,7.5);
  \begin{scope}[shift={(.3,4.1)}]
    \foreach \i in {0,...,10} \draw[gray!35] ({.64*\i},0)--({.64*\i},3.2);
    \foreach \j in {0,...,5} \draw[gray!35] (0,{.64*\j})--(6.4,{.64*\j});
    \foreach \i in {0,...,10} \foreach \j in {0,...,5}
      \fill[gray!55] ({.64*\i},{.64*\j}) circle (1.4pt);
    \draw[axisblue,line width=2.2pt,line join=round] (1.28,1.92)--(1.92,1.92)--(2.56,1.92)--(2.56,1.28);
    \draw[axisamber,line width=2.2pt,line join=round] (2.61,1.92)--(2.61,1.28)--(3.25,1.28)--(3.89,1.28);
    \draw[axispink,line width=2.2pt,line join=round] (3.2,1.23)--(3.84,1.23)--(3.84,.59)--(4.48,.59);
    \node[text=axisblue] at (1.6,2.25) {$p_1$};
    \node[text=axisamber] at (3.2,1.64) {$p_2$};
    \node[text=axispink] at (4.25,.24) {$p_3$};
  \end{scope}
  \node[font=\small\bfseries] at (12.82,7.25) {Global operator $K$};
  \def\blockindex{}
  \pic at (10.7,4.75) {split};
  \draw[flow] (7.25,5.75)--(10.25,5.75);
  \node[align=center,text=black!65,font=\footnotesize] at (8.75,6.35) {Interface\\localization};
  \draw[flow] (3.5,3.75)--(3.5,2.65);
  \node[anchor=west,align=left,text=black!65,font=\footnotesize] at (3.8,3.2) {Geometric\\localization};
  \draw[flow] (12.82,4.3)--(12.82,2.65);
  \node[anchor=west,align=left,text=black!65,font=\footnotesize] at (13,3.35)
    {Same geometric\\localization};
  \node[font=\small\bfseries] at (3.5,2.2) {Overlapping update collections};
  \foreach \x/\col/\idx in {.7/axisblue/1,3/axisamber/2,5.3/axispink/3} {
    \begin{scope}[shift={(\x,1.25)}]
      \ifnum\idx=1
        \draw[\col,line width=1.7pt] (0,0)--(.5,0)--(1,0)--(1,-.5);
        \foreach \a/\b in {0/0,.5/0,1/0,1/-.5} \filldraw[fill=white,draw=\col] (\a,\b) circle (1.8pt);
      \else\ifnum\idx=2
        \draw[\col,line width=1.7pt] (0,0)--(0,-.5)--(.5,-.5)--(1,-.5);
        \foreach \a/\b in {0/0,0/-.5,.5/-.5,1/-.5} \filldraw[fill=white,draw=\col] (\a,\b) circle (1.8pt);
      \else
        \draw[\col,line width=1.7pt] (0,0)--(.5,0)--(.5,-.5)--(1,-.5);
        \foreach \a/\b in {0/0,.5/0,.5/-.5,1/-.5} \filldraw[fill=white,draw=\col] (\a,\b) circle (1.8pt);
      \fi\fi
      \node[text=\col] at (.5,-.95) {$p_{\idx}$};
    \end{scope}
  }
  \node[font=\small\bfseries] at (12.82,2.2) {Local operator $K_p$};
  \def\blockindex{_p}
  \pic at (10.7,0) {split};
  \draw[flow] (7.25,1)--(10.25,1);
  \node[align=center,text=black!65,font=\footnotesize] at (8.75,1.65) {Apply the\\same split};
\end{tikzpicture}
\endgroup%
  }
  \caption{Double localization. Geometric localization selects local updates,
  while interface localization focuses on a chosen subspace of observables.
  For each fixed $p$, the two operations commute.
  Here $A$ and $A_p$ denote the interface restrictions of $K$ and $K_p$,
  respectively, using the same projection onto $\mathcal M$.
  The lattice is schematic, not the graph of the finite-group application.}
  \label{fig:intro-two-axis}
\end{figure}
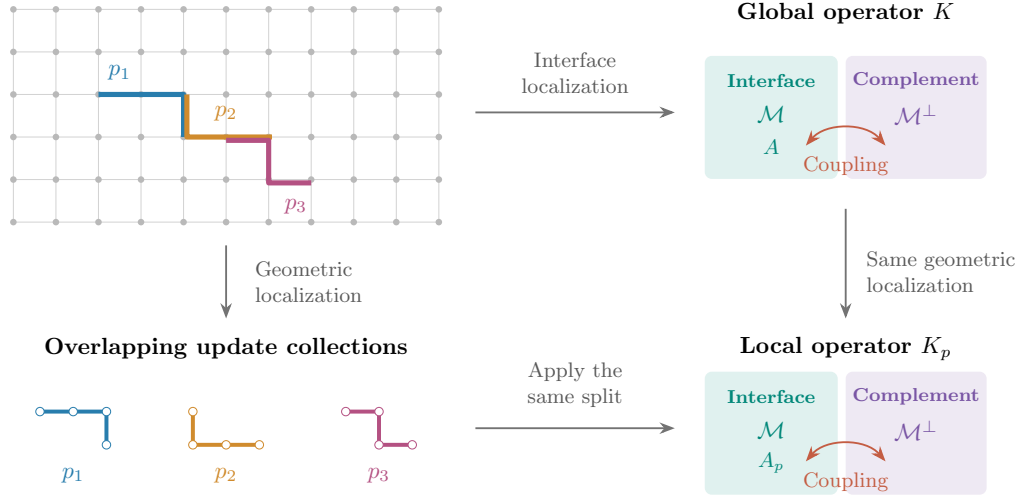

\subsubsection{Double localization and an abstract gap theorem}
\label{sec:intro-contribution-abstract}

We first introduce our double localization technique. Let $\mathcal H$ be the system Hilbert space and equip its observable
space $\mathcal B(\mathcal H)$ with an inner product. We consider a
positive operator $K$ with a specified decomposition
$K=\sum_{i\in\mathcal I}K_i,$
where each $K_i\succeq0$ acts only on a bounded set of sites\footnote{For the thermal dynamics studied
here, $K=-\mathcal L$, where $\mathcal L$ is the Lindbladian governing
the evolution of observables, $e^{t\mathcal L}=e^{-tK}$.
The inner product is chosen so that $K$ is self-adjoint
and positive.}.

(1) \emph{Geometric localization.}
Fix a finite family $\mathcal P$ of update collections
$p\subseteq\mathcal I$, each supported within a small spatial region,
and write $K_p:=\sum_{i\in p}K_i$ for $p\in\mathcal P$.
Geometric localization replaces $K$ by $K_p$, restricting attention
to updates in the corresponding region.
The collections in $\mathcal P$ may overlap, allowing the chosen regions to collectively capture global behavior while keeping local gap estimation tractable.

(2) \emph{Interface localization.}
Interface localization selects a model-dependent subspace $\mathcal M\subseteq\mathcal B(\mathcal H)$, called the \emph{interface}, with orthogonal projection $\Pi$ (in the chosen inner product), so that $\mathcal B(\mathcal H)=\mathcal M\oplus\mathcal M^\perp$.
This allows us to focus on a smaller class of observables in $\mathcal M$, while remaining geometrically global. For a suitable choice of interface $\mathcal M$, its observables may depend only on the outcomes in a preferred basis, rather than on coherences between basis states. In such cases, comparison with classical dynamics may yield a useful lower bound on the interface gap.

For each fixed $p$, the two localization operations commute, as illustrated in \Cref{fig:intro-two-axis}.

\begin{informaltheorem}[Abstract gap theorem, informal]
Suppose that the smallest positive eigenvalues of $A:=\left.\Pi K\Pi\right|_{\mathcal M}$ and of each $K_p$, $p\in\mathcal P$, have known positive lower bounds.
Under additional quantitative conditions on the coupling between $\mathcal M$ and $\mathcal M^\perp$ and on the assembly of local bounds,\footnote{
These conditions control the loss of dissipation due to coupling, ensure that the local operators cover complementary observables orthogonal to $\ker K$, bound their overlap, and ensure that the retained local interface bounds together control $A$.}
$K$ has a positive global spectral gap with an explicit lower bound in terms of these input values.
\end{informaltheorem}

The precise assumptions and quantitative bound are given in \cref{thm:abstract-two-axis-gap-assembly}.
The proof combines the interface
gap bound for $A$ with the local gap bounds for $K_p$, using the coupling
and assembly conditions to obtain a gap for $K$, as illustrated in
\Cref{fig:intro-assembly}.

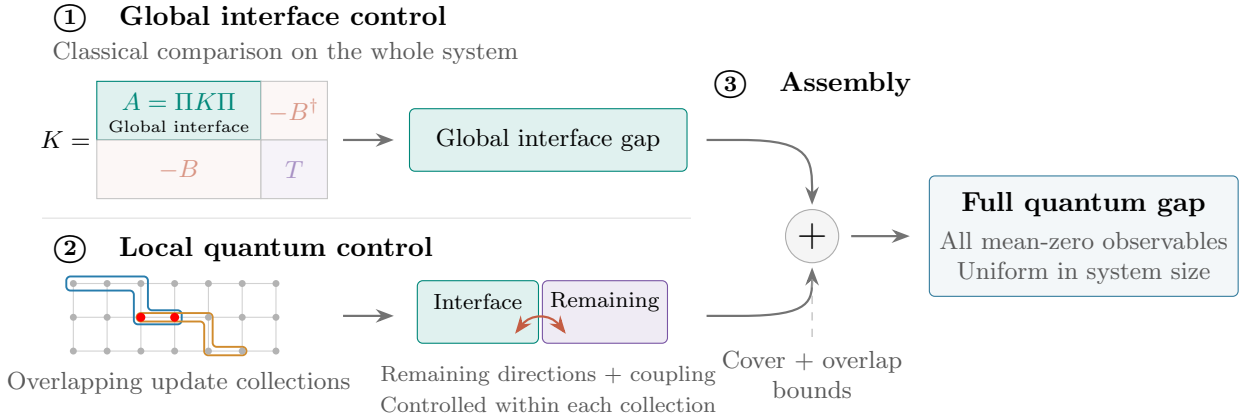
\begin{figure}[htbp]
  \centering
  \begingroup
\resizebox{\linewidth}{!}{%
\begin{tikzpicture}[font=\footnotesize,>=Stealth,
  flow/.style={->,black!55,line width=.8pt}]
  \path[use as bounding box] (0,0) rectangle (15.4,5.5);
  \node[anchor=west,font=\small\bfseries] at (0,5.2)
    {\textcircled{\scriptsize 1}\quad Global interface control};
  \node[anchor=west,text=black!65] at (0,4.77)
    {Classical comparison on the whole system};
  \node at (.3,3.65) {$K=$};
  \filldraw[fill=figureInterface!12,draw=figureInterface] (.7,3.65) rectangle (2.8,4.4);
  \filldraw[fill=figureCoupling!5,draw=gray!45] (2.8,3.65) rectangle (3.65,4.4);
  \filldraw[fill=figureCoupling!5,draw=gray!45] (.7,2.9) rectangle (2.8,3.65);
  \filldraw[fill=figureComplement!7,draw=gray!45] (2.8,2.9) rectangle (3.65,3.65);
  \node[text=figureInterface] at (1.75,4.14) {$A=\Pi K\Pi$};
  \node[font=\tiny] at (1.75,3.83) {Global interface};
  \node[text=figureCoupling!65] at (3.225,4.025) {$-B^\dagger$};
  \node[text=figureCoupling!65] at (1.75,3.275) {$-B$};
  \node[text=figureComplement!65] at (3.225,3.275) {$T$};
  \draw[flow] (3.85,3.65)--(4.5,3.65);
  \filldraw[fill=figureInterface!12,draw=figureInterface,rounded corners=2pt]
    (4.7,3.25) rectangle (8.25,4.05);
  \node at (6.475,3.65) {Global interface gap};
  \draw[flow] (8.5,3.65)--(8.95,3.65)
    .. controls (9.85,3.65) and (9.85,3.4) .. (9.85,2.82);
  \draw[gray!25] (0,2.65)--(8.3,2.65);
  \node[anchor=west,font=\small\bfseries] at (0,2.25)
    {\textcircled{\scriptsize 2}\quad Local quantum control};
  \begin{scope}[shift={(.4,.95)},scale=.9]
    \foreach \i in {0,...,6} \draw[gray!35] ({.48*\i},0)--({.48*\i},.96);
    \foreach \j in {0,...,2} \draw[gray!35] (0,{.48*\j})--(2.88,{.48*\j});
    \draw[figureGeometry,line width=.7pt,rounded corners=1.5pt,fill=none]
      (-.10,.86)--(.86,.86)--(.86,.38)--(1.54,.38)
      --(1.54,.58)--(1.06,.58)--(1.06,1.06)--(-.10,1.06)--cycle;
    \draw[figureAmber,line width=.7pt,rounded corners=1.5pt,fill=none]
      (.90,.42)--(1.86,.42)--(1.86,-.06)--(2.46,-.06)
      --(2.46,.06)--(1.98,.06)--(1.98,.54)--(.90,.54)--cycle;
    \foreach \i in {0,...,6} \foreach \j in {0,...,2}
      \fill[gray!60] ({.48*\i},{.48*\j}) circle (1.4pt);
    \foreach \x in {.96,1.44}
      \fill[red] (\x,.48) circle (1.8pt);
  \end{scope}
  \node[text=black!65] at (1.75,.55) {Overlapping update collections};
  \draw[flow] (3.85,1.4)--(4.5,1.4);
  \filldraw[fill=figureInterface!12,draw=figureInterface,rounded corners=2pt]
    (4.8,1.05) rectangle (6.35,1.85);
  \filldraw[fill=figureComplement!12,draw=figureComplement,rounded corners=2pt]
    (6.4,1.05) rectangle (8,1.85);
  \node[font=\scriptsize] at (5.575,1.58) {Interface};
  \node[font=\scriptsize] at (7.2,1.58) {Remaining};
  \draw[<->,figureCoupling,line width=.8pt] (6.05,1.17)
    .. controls (6.25,1.43) and (6.5,1.43) .. (6.7,1.17);
  \node[text=black!65,font=\scriptsize] at (6.475,.64) {Remaining directions + coupling};
  \node[text=black!65,font=\scriptsize] at (6.475,.27) {Controlled within each collection};
  \draw[flow] (8.5,1.4)--(8.95,1.4)
    .. controls (9.85,1.4) and (9.85,1.6) .. (9.85,2.02);
  \node[font=\small\bfseries] at (9.85,4.35)
    {\textcircled{\scriptsize 3}\quad Assembly};
  \filldraw[fill=gray!8,draw=gray!65] (9.85,2.42) circle (.34);
  \node[font=\Large] at (9.85,2.42) {$+$};
  \node[text=black!65,align=center] at (9.85,.65) {Cover + overlap\\bounds};
  \draw[gray!60,dashed] (9.85,1.13)--(9.85,1.9);
  \draw[flow] (10.35,2.42)--(11.1,2.42);
  \filldraw[fill=figureEstimate!6,draw=figureEstimate,rounded corners=2pt]
    (11.35,1.65) rectangle (15.3,3.19);
  \node[font=\small\bfseries] at (13.325,2.83) {Full quantum gap};
  \node[text=black!65] at (13.325,2.36) {All mean-zero observables};
  \node[text=black!65] at (13.325,1.95) {Uniform in system size};
\end{tikzpicture}%
}
\endgroup
  \caption{Assembling the two estimates.  The block $A=\Pi K\Pi$ is the
  global interface problem, estimated here by classical comparison.
  The remaining quantum directions and their coupling are controlled on
  bounded collections of updates.  Fixed-space coverage and quantitative overlap and
  assembly bounds join these estimates to give a gap on the full mean-zero
  observable space.
  The resulting bound
  is volume-uniform when the required estimates are uniform.}
  \label{fig:intro-assembly}
\end{figure}

\subsubsection{An unconditional uniform gap theorem for finite-group Davies dynamics}
\label{sec:intro-contribution-model}

Let the underlying graph be finite, simple, and three-regular, with girth at least six, and assign to each edge the Hilbert space $\mathbb C[\Gamma]$ for a finite group $\Gamma$. We consider the Hamiltonian $H=-J\sum_v A_v$, where $J>0$ and $A_v$ is the projector obtained by averaging the group action at vertex $v$.
This Hamiltonian consists of the vertex terms of Kitaev's quantum double construction, with no plaquette terms. The system is coupled to the specified unit-rate Davies bath through
left and right regular translations and group-basis projectors on every edge. Our second result is as follows.

\begin{informaltheorem}[Finite-group gap theorem, informal]
For the Hamiltonian and Davies bath specified above, the full quantum
dynamics has a spectral gap bounded below by a positive constant independent of graph size
for every fixed nontrivial finite group $\Gamma$ and fixed inverse
temperature $\beta$ with $0\le\beta J<\log(5/3)$.
For $\Gamma=S_3$, the lower bound is uniform over all graphs in the
stated family and the entire interval $0\le\beta J\le\log 3$.
\end{informaltheorem}

The precise statement is given in \cref{thm:finite-group-davies-gaps}.

A key feature of our approach is that the spectral-gap bound follows from estimates proved for the model, with no local or global gap assumptions left to verify. We choose an interface consisting of observables that depend only on the irreducible-representation labels of the edge variables and bound its gap by comparison with classical dynamics on those labels.
{For this comparison, we prove a theorem
(\cref{thm:finite-block-davies-label-comparison}) that turns classical
heat-bath gap estimates into lower bounds on the interface gap.}
We also establish the local quantum gap and coupling bounds, along with the required geometric coverage and
overlap conditions.
Our exact vertex-star analysis reaches $\beta J=\log3$ for $S_3$, where the one-site Dobrushin criterion \cite{dobrushin1968description} does not certify a gap. \Cref{fig:intro-certificate-reuse} illustrates the resulting extension of the quantum gap guarantee.

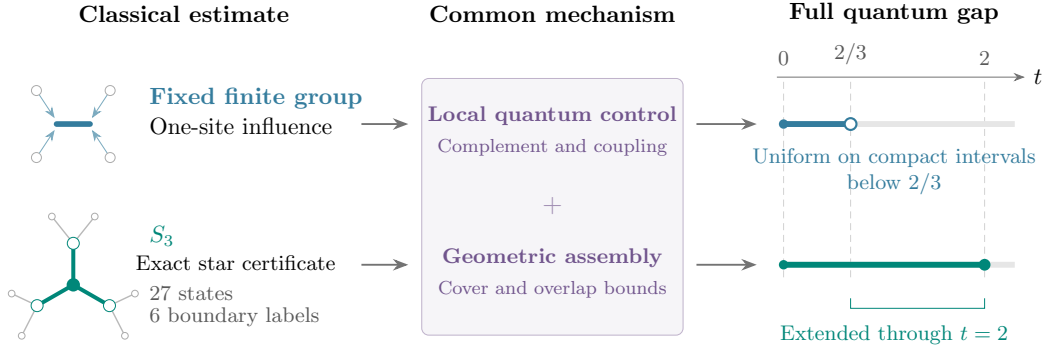
\begin{figure}[!t]
  \centering
  \resizebox{0.9\linewidth}{!}{
  \begingroup
\colorlet{certblue}{figureEstimate}
\colorlet{certteal}{figureInterface}
\colorlet{certpurple}{figureMechanism}
\begin{tikzpicture}[font=\small, >=Stealth,
  boundary/.style={draw=gray!55, line width=.6pt},
  point/.style={circle, inner sep=1.5pt, fill=white, draw=gray!65},
  flow/.style={->, draw=black!55, line width=.8pt}]
  \path[use as bounding box] (0,-.2) rectangle (15.4,5.2);
  \node[font=\small\bfseries] at (2.6,5) {Classical estimate};
  \node[font=\small\bfseries] at (8.1,5) {Common mechanism};
  \node[font=\small\bfseries] at (13.2,5) {Full quantum gap};

  \begin{scope}[yshift=-3mm]
  \begin{scope}[shift={(.95,3.65)}]
    \foreach \x/\y in {-.55/.5,.55/.5,-.55/-.5,.55/-.5} {
      \draw[->,draw=certblue!65,shorten >=3pt] (\x,\y)--({.4*\x},0);
      \node[point] at (\x,\y) {};
    }
    \draw[certblue,line width=2.4pt,line cap=round] (-.25,0)--(.25,0);
  \end{scope}
  \node[anchor=west,text=certblue,font=\small\bfseries] at (1.95,4.03) {Fixed finite group};
  \node[anchor=west] at (1.95,3.63) {One-site influence};

  \begin{scope}[shift={(.95,1.25)}]
    \coordinate (v) at (0,0);
    \coordinate (u1) at (0,.62);
    \coordinate (u2) at (-.54,-.31);
    \coordinate (u3) at (.54,-.31);
    \foreach \u/\x/\y in {u1/-.32/1.02,u1/.32/1.02,
      u2/-.92/-.13,u2/-.77/-.76,u3/.92/-.13,u3/.77/-.76} {
      \draw[boundary] (\u)--(\x,\y);
      \node[point,inner sep=1pt] at (\x,\y) {};
    }
    \foreach \u in {u1,u2,u3} {
      \draw[certteal,line width=1.7pt] (v)--(\u);
      \node[circle,inner sep=1.8pt,fill=white,draw=certteal] at (\u) {};
    }
    \node[circle,inner sep=2pt,fill=certteal] at (v) {};
  \end{scope}
  \node[anchor=west,text=certteal,font=\small\bfseries] at (1.95,1.98) {$S_3$};
  \node[anchor=west,font=\footnotesize] at (1.75,1.58) {Exact star certificate};
  \node[anchor=west,font=\footnotesize,text=black!65] at (1.95,1.13) {$27$ states};
  \node[anchor=west,font=\footnotesize,text=black!65] at (1.95,.78) {$6$ boundary labels};

  \filldraw[fill=certpurple!6,draw=certpurple!45,rounded corners=3pt]
    (6.15,.5) rectangle (10.05,4.33);
  \node[text=certpurple,font=\footnotesize\bfseries] at (8.1,3.77) {Local quantum control};
  \node[text=certpurple,font=\scriptsize] at (8.1,3.31) {Complement and coupling};
  \node[text=certpurple!65,font=\normalsize] at (8.1,2.43) {$+$};
  \node[text=certpurple,font=\footnotesize\bfseries] at (8.1,1.65) {Geometric assembly};
  \node[text=certpurple,font=\scriptsize] at (8.1,1.19) {Cover and overlap bounds};
  \foreach \y in {3.65,1.55} {
    \draw[flow] (5.25,\y)--(5.95,\y);
    \draw[flow] (10.25,\y)--(11.15,\y);
  }
  \end{scope}
  \begin{scope}[yshift=-3mm]
  \draw[->,black!55] (11.45,4.35)--(15.15,4.35)
    node[right,text=black] {$t$};
  \foreach \x/\lab in {11.55/0,12.55/{2/3},14.55/2} {
    \node[above,font=\footnotesize,text=black!65] at (\x,4.4) {$\lab$};
    \draw[gray!35,densely dashed] (\x,4.28)--(\x,1.3);
  }
  \foreach \y in {3.65,1.55}
    \draw[gray!20,line width=2.5pt] (11.55,\y)--(15,\y);
  \draw[certblue,line width=2.5pt] (11.55,3.65)--(12.55,3.65);
  \fill[certblue] (11.55,3.65) circle (2pt);
  \filldraw[fill=white,draw=certblue,line width=.9pt]
    (12.55,3.65) circle (2.5pt);
  \node[text=certblue,font=\footnotesize,align=center] at (13.2,2.95)
    {Uniform on compact intervals\\below $2/3$};
  \draw[certteal,line width=2.5pt] (11.55,1.55)--(14.55,1.55);
  \fill[certteal] (11.55,1.55) circle (2pt);
  \fill[certteal] (14.55,1.55) circle (2.5pt);
  \draw[certteal] (12.55,1.05)--(12.55,.9)--(14.55,.9)--(14.55,1.05);
  \node[text=certteal,font=\footnotesize] at (13.2,.5)
    {Extended through $t=2$};
  \end{scope}
\end{tikzpicture}
\endgroup
  }
  \caption{Stronger classical estimates, a wider quantum gap guarantee.
  Here $t=e^{\beta J}-1$.
  The same framework turns a one-site estimate into a full quantum gap
  for $t<2/3$, and an exact $S_3$ star certificate into a guarantee through
  $t=2$.  The marked ranges describe what the proofs certify, not where
  the gap disappears.}
  \label{fig:intro-certificate-reuse}
\end{figure}

\subsection{Significance}
\label{sec:new-intro-significance}

The contribution of this work is both conceptual and technical, with a concrete
application demonstrating the strength of the framework.  We highlight
three aspects below.

\paragraph{Flexibility of the framework.}
Double localization allows both the local update collections and the
interface subspace to be chosen to suit the model, subject to the conditions
of the abstract theorem. The interface retains observables whose
dissipation admits a useful global estimate, while local estimates control
the remaining directions and their coupling. Irreducible-representation
labels provide one choice of interface, but the framework is not limited
to this choice or to classical comparison. In our finite-group application,
it combines classical comparison on the label interface with local quantum
gap estimates and geometric assembly. This illustrates how different tools
can address different parts of the global gap problem within a single
framework.
This flexibility also suggests several directions for further applications and extensions, discussed in the outlook below.

\paragraph{Qualitative comparison with other methods.}
\label{sec:new-intro-proof-architectures}
\Cref{tab:proof-interface-comparison} compares the structure of several approaches to proving spectral-gap and mixing bounds.
Clustering and tensorization methods assemble regional estimates using decay or factorization properties \cite{kastoryano2016quantum,bardet2022approximate}. Our approach instead retains a global bound on the interface while localizing the additional quantum estimates. For standard quantum double models, factorization of PEPS boundary states can establish a parent-Hamiltonian gap that transfers to the Davies dynamics \cite{lucia2023thermalization}. Those models include plaquette terms, whereas our vertex-term model does not; we estimate the specified Davies operator directly through its local components.

Quantum-to-classical Davies comparison \cite{basso2025quantum} uses
a gap on the full space $V_0(H)=\{X:[X,H]=0\}$ together with spectral assumptions.
Our approach uses a smaller interface defined by representation
labels, supplemented by local quantum estimates and geometric bounds
verified for the model. General high-temperature mixing results
\cite{bakshi2025dobrushin,bergamaschi2026fast} concern different sampler constructions and do not directly establish a gap for the Davies dynamics studied here. We discuss a quantitative comparison with \cite{basso2025quantum} and \cite{bergamaschi2026fast} below.

\begin{table}[!htbp]
  \centering
  \small
  \renewcommand{\arraystretch}{1.16}
  \def\qipComparisonScale{0.95}
  \resizebox{\qipComparisonScale\textwidth}{!}{%
  \begin{tabularx}{\textwidth}{@{}
    >{\hsize=.82\hsize}Y
    >{\hsize=1.05\hsize}Y
    >{\hsize=1.13\hsize}Y@{}}
    \toprule
    Approach & Main input & Route to a global gap \\
    \midrule
    Clustering-based methods \cite{kastoryano2016quantum}
      & Regional conditional expectations and correlation bounds
      & Combine regional estimates. \\
    Quantum-to-classical comparison \cite{basso2025quantum}
      & A gap on the full space $V_0(H)$ of observables commuting with $H$
      & Transfer the bound under spectral assumptions. \\
    \addlinespace[.45em]
    \textbf{Double localization (this work)}
      & \textbf{A global interface bound and local gap bounds}
      & \textbf{Combine them under geometric assembly conditions.} \\
    \bottomrule
  \end{tabularx}%
  }
  \caption{Selected approaches to global gap bounds. Here $V_0(H)=\{X:[X,H]=0\}$.}
  \label{tab:proof-interface-comparison}
\end{table}

\paragraph{Quantitative comparison for $S_3$.}
\label{sec:new-intro-quantitative-calibration}
For the specified $S_3$ Davies dynamics, we prove a positive spectral-gap
lower bound uniform in graph size throughout $0\le\beta J\le\log 3$. This guarantee does not follow directly from the quantum-to-classical
comparison of \cite{basso2025quantum} or the all-to-all
high-temperature result of \cite{bergamaschi2026fast}. Applied to our model, the arithmetic-progression bound of
\cite{basso2025quantum} loses size independence, while its alternative
bound at fixed $\beta>0$ requires control of all observables commuting
with $H$, beyond what our label-interface estimate provides. Likewise,
when each edge is treated as a $|\Gamma|$-dimensional qudit and our
interaction parameters are inserted into the sufficient condition of
\cite{bergamaschi2026fast}, that condition does not extend to
$\beta J=\log 3$.  The detailed comparisons are given in
\cref{app:baseline-calibrations}.

\subsection{Outlook}
\label{sec:new-intro-outlook}

The framework suggests extensions through new choices of interface and
local regions, as well as combinations with other local-to-global
methods.\footnote{These directions concern spectral gaps.  The present
quadratic argument does not establish a modified logarithmic Sobolev
inequality (MLSI).  Uniform MLSIs are known
for the standard two-dimensional toric code and Abelian quantum doubles
at every positive temperature
\cite{stengele2026css,stengele2026modified}.}

\paragraph{Interfaces for other models.}
One future direction is to consider whether there are other problems where this interface subspace is natural and helpful. In this paper, we consider a model defined by a graph and a group, but neither of these ingredients is required by \Cref{thm:abstract-two-axis-gap-assembly}. One promising candidate in qubit and qudit models is the space defined by observables that are diagonal in the computational basis. Another class of models that combine qubit and group dynamics could arise from Hamiltonians related to stabilizer codes, which form a bridge between many-body physics and quantum error correction. In this case, a candidate interface could be the observables generated by the stabilizer group, which in an experiment would record the syndrome. 

A concrete extension would be to extend to the standard quantum-double Hamiltonian, including the vertex terms we restrict to and the plaquette terms that we neglected in this paper. The challenging part here would be to identify a subspace which is compatible with both of these different terms, and to determine how the Hamiltonian couples this interface subspace with the remaining space.

\paragraph{More flexible geometric assembly.}
Another direction is to explore more flexible choices of local regions and their geometric assembly for $K$. In our case, \Cref{thm:abstract-two-axis-gap-assembly} is designed around the intuition of picking a set of local update sites. However, there are many other ways of breaking a system up into smaller, overlapping pieces. 
For systems with higher-order interactions, local regions adapted to the interaction structure may provide more useful coverage and overlap bounds than regions chosen from an underlying graph alone. 
For example, neighbourhoods based on links in a simplicial complex (a higher-dimensional generalisation of a graph) are natural candidates.
If qubits or qudits are defined on the $d$-simplices of a complex and the Hamiltonian acts in accordance with the up- and down-neighbourhood relationships defined by the complex, such neighbourhoods may help organise the relevant local updates, depending on how the interactions are defined. The abstract theorem already permits overlapping sets of updates, so such choices may fit within its existing assembly conditions. They would require new coverage and overlap estimates, together with local spectral bounds. This could allow these methods to be applied to a wider range of Hamiltonians.

\paragraph{Application for noncommuting Hamiltonians.}
A further question is whether the framework can yield useful bounds for noncommuting Hamiltonians. The abstract theorem does not assume commutativity, but our model-specific verification uses it to obtain exact local reductions. Extending these arguments would require control of the errors arising when such reductions are only approximate.

\vspace{1em}
We believe that this framework will help establish spectral-gap bounds
for a broader range of quantum many-body systems, contributing to our
understanding of thermal relaxation and quantum Gibbs sampling.

\subsection{Organization of the Paper}

\Cref{sec:main-results} formally states the main results.
\Cref{sec:preliminaries} collects the operator tools for the abstract
gap theorem.
\Cref{sec:three-local-inputs} proves the abstract gap theorem by
assembling the local bounds and applying a two-block estimate.
\Cref{sec:model-preliminaries} provides the KMS, GNS, Davies, finite-group,
and classical heat-bath background for the application.
\Cref{sec:finite-group-specialization} identifies the abstract conditions
and assembly constants for the finite-group model.
\Cref{sec:finite-block-davies-label-comparison} establishes the general
Davies-to-label comparison, and
\cref{sec:common-analytic-certificates} verifies the three conditions for
the finite-group star model.  \Cref{sec:s3-strengthening} gives the
$S_3$ bound using fusion rules.  The appendices contain the full comparison
proof, the exact centered-path kernel calculation, the reproducible finite
$S_3$ certificates, and the quantitative baseline calibrations.

\section{Main Results}
\label{sec:main-results}

In this section, we formally state the results of the abstract gap criterion and its
finite-group application. Only the
notation needed for these statements is introduced here; the operator
preliminaries and detailed model construction follow in
\cref{sec:prelim-shorted-operators,sec:prelim-davies,sec:prelim-quantum-double}.

\subsection{Abstract Gap Criterion for Double Localization}
\label{sec:main-abstract-criterion}

All operators in the abstract statement act on a finite-dimensional
Hilbert space.  The criterion concerns positive operators $K$ and does
not require a Markov or Davies construction.  In applications to
KMS-self-adjoint quantum Markov dynamics, $K=-\mathcal L\succeq0$,
where $\mathcal L$ generates the observable evolution
$e^{t\mathcal L}=e^{-tK}$.  We call $K$ the positive operator, reserving
the dynamical generator for $\mathcal L$.
In the Davies application the Hilbert space is the mean-zero observable
space
\[
 L^2_0(\rho_\beta)=\{X:\operatorname{Tr}(\rho_\beta X)=0\},
 \qquad
 \langle X,Y\rangle_{\rho_\beta}
 =\operatorname{Tr}(\rho_\beta^{1/2}X^\dagger\rho_\beta^{1/2}Y).
\]
Choose an orthogonal projection $\Pi$ from the model structure, rather
than from an unknown low eigenspace, and put $\Pi^\perp=I-\Pi$.
Its range is the \emph{interface subspace}. In the Davies application,
this is the mean-zero part of the interface $\mathcal M$ introduced in
\cref{sec:intro-contribution-abstract}:
\[
  \mathcal M_0:=\mathcal M\cap L^2_0(\rho_\beta),
  \qquad \operatorname{ran}\Pi=\mathcal M_0.
\]
Thus we retain the notation $\Pi$ for the interface projection, now acting
on the mean-zero observable space. Neither the abstract criterion nor its
application requires this subspace to be invariant under $K$.
The corresponding block decomposition is
\[
 K=\begin{pmatrix}A&-B^\dagger\\-B&T\end{pmatrix}_{\Pi\oplus\Pi^\perp},
 \qquad
 A=\Pi K\Pi,\quad B=-\Pi^\perp K\Pi,\quad T=\Pi^\perp K\Pi^\perp.
\]
We use the same notation for a compressed operator on its subspace and
its extension by zero to the orthogonal complement.\footnote{In
\cref{sec:intro-contribution-abstract}, the restriction
$A=\left.\Pi K\Pi\right|_{\mathcal M}$ explicitly indicates the interface
as its domain. In the Davies application here, we work on the mean-zero
observable space, so the corresponding compression acts on $\mathcal M_0$:
$\left.\Pi K\Pi\right|_{\mathcal M_0}$. We also denote its extension by
zero on the orthogonal complement within $L^2_0(\rho_\beta)$ by $A$.}
In particular, $A$
acts on $\operatorname{ran}\Pi$ as the interface compression and vanishes
on $\operatorname{ran}\Pi^\perp$.
For $G\succeq0$, the \emph{shorted operator}, denoted by $\mathcal S_\Pi(G)$,
is the positive operator on
$\operatorname{ran}\Pi$, extended by zero on $\operatorname{ran}\Pi^\perp$,
whose quadratic form is
\[
 \langle x,\mathcal S_{\Pi}(G)x\rangle
 =\min_{y\in\operatorname{ran}\Pi^\perp}
   \langle x+y,G(x+y)\rangle,
 \qquad x\in\operatorname{ran}\Pi.
\]
It measures the dissipation that survives when the interface component is fixed and the complementary component is varied to minimize the quadratic form.
The definition and its Schur-complement formula are discussed in
\cref{sec:prelim-shorted-operators}.

Geometrically, $K_i$ denotes an elementary update, whereas $K_p$ combines
a bounded collection $I(p)$ of updates, with overlaps allowed.
Here $p$ indexes the collection $I(p)$, which is denoted simply by $p$ in
\cref{sec:intro-contribution-abstract}.
For $K_p=\sum_{i\in I(p)}K_i$, write
\[
 K_p=\begin{pmatrix}A_p&-B_p^\dagger\\-B_p&T_p\end{pmatrix}_{\Pi\oplus\Pi^\perp},
 \qquad S_p=\mathcal S_{\Pi}(K_p).
\]

For each fixed $p$, update selection and interface compression commute,
as summarized in \Cref{fig:abstract-localization-square}.

\begin{figure}[!htbp]
  \centering
  \def\abstractLocalizationScale{1}
  \scalebox{\abstractLocalizationScale}{%
  \begin{tikzpicture}[>=Stealth, every node/.style={font=\small}]
    \node (K) at (0,0) {$K=\sum_{i\in\mathcal I}K_i$};
    \node (A) at (7.6,0) {$A=\Pi K\Pi$};
    \node (Kp) at (0,-1.8) {$K_p=\sum_{i\in I(p)}K_i$};
    \node (Ap) at (7.6,-1.8) {$A_p=\Pi K_p\Pi$};
    \draw[->] (K) -- node[above,font=\footnotesize]
      {Interface compression} (A);
    \draw[->] (Kp) -- node[below,font=\footnotesize]
      {Same interface compression} (Ap);
    \draw[->] (K) -- node[right,align=left,font=\footnotesize]
      {Select updates\\$i\in I(p)$} (Kp);
    \draw[->] (A) -- node[left,align=right,font=\footnotesize]
      {Select the same updates\\$i\in I(p)$} (Ap);
  \end{tikzpicture}%
  }
  \caption{The two localizations for a fixed $p$. Writing $A_i=\Pi K_i\Pi$,
  the right vertical arrow selects $A_p=\sum_{i\in I(p)}A_i$ from
  $A=\sum_{i\in\mathcal I}A_i$. The same projection $\Pi$ is used throughout.
  This diagram records the operator identities, not the assembly of gap bounds.}
  \label{fig:abstract-localization-square}
\end{figure}
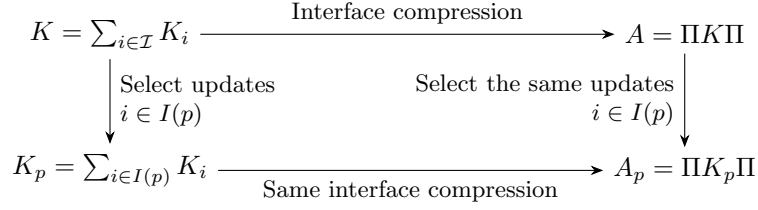

\begin{theorem}[Abstract gap assembly under double localization]
\label{thm:abstract-two-axis-gap-assembly}
Let $\mathcal I$ and $\mathcal P$ be finite, let $K_i\succeq0$ for
$i\in\mathcal I$, and let $\Pi$ be an orthogonal projection with
$\Pi^\perp=I-\Pi$.  For each $p\in\mathcal P$, choose
$I(p)\subseteq\mathcal I$; overlaps are allowed.  Set
\[
  K:=\sum_{i\in\mathcal I}K_i,
  \qquad
  K_p:=\sum_{i\in I(p)}K_i,
  \qquad
  A:=\Pi K\Pi,
  \qquad
  T:=\Pi^\perp K\Pi^\perp.
\]
For each $p$, let $E_p$ be the orthogonal projection onto $\ker K_p$, put
$E_p^\perp=I-E_p$, and choose a positive target form
$C_p=\Pi C_p\Pi$.  Let
$a,g_{\mathrm{loc}},\kappa_{\mathrm{cov}},c_{\mathrm{Sch}},r_K,r_A>0$.
Suppose the following spectral conditions hold:
{\small
\[
  \begin{array}{@{}c@{\qquad}c@{\qquad}c@{}}
    \textnormal{(i)}\; A\succeq a\Pi,
    & \textnormal{(ii)}\; K_p\succeq g_{\mathrm{loc}}E_p^\perp
      \quad(p\in\mathcal P),
    & \textnormal{(iii)}\;
      \mathcal S_{\Pi}(K_p)\succeq c_{\mathrm{Sch}}C_p
      \quad(p\in\mathcal P).
  \end{array}
\]
}
Suppose also that the following assembly conditions hold:
{\small
\[
  \begin{array}{@{}c@{\qquad}c@{\qquad}c@{}}
    \textnormal{(G1)}\;
      \Pi^\perp\!\Big(\textstyle\sum_{p\in\mathcal P}E_p^\perp\Big)\!\Pi^\perp
      \succeq\kappa_{\mathrm{cov}}\Pi^\perp,
    & \textnormal{(G2)}\;
      \textstyle\sum_{p\in\mathcal P}K_p\preceq r_KK,
    & \textnormal{(G3)}\;
      \textstyle\sum_{p\in\mathcal P}C_p\succeq r_AA.
  \end{array}
\]
}
With
$
  \tau:=\frac{g_{\mathrm{loc}}\kappa_{\mathrm{cov}}}{r_K}$ and
 $  \eta:=\min\left\{\frac{c_{\mathrm{Sch}}r_A}{r_K},1\right\},
$
one has
\[
  T\succeq\tau\Pi^\perp,
  \qquad
  \mathcal S_{\Pi}(K)\succeq\eta A.
\]
Consequently,
\begin{equation}
  \label{eq:abstract-two-axis-assembled-gap}
  K\succeq
  \left(1-\sqrt{1-\eta}\right)\min\{a,\tau\}\,I.
\end{equation}
\end{theorem}

The proof is given in \cref{sec:abstract-two-axis-assembly}, where
\cref{prop:global-bounds-from-local-inputs,lem:two-block-gap-assembly}
together establish the conclusion.

\Cref{fig:two-axis-proof-architecture} illustrates the spectral and assembly conditions associated with double localization.
In the figure,
$A_p=\Pi K_p\Pi$, $T_p=\Pi^\perp K_p\Pi^\perp$,
$B_p=-\Pi^\perp K_p\Pi$, and $S_p=\mathcal S_\Pi(K_p)$;
the global cross block is $B=-\Pi^\perp K\Pi$.

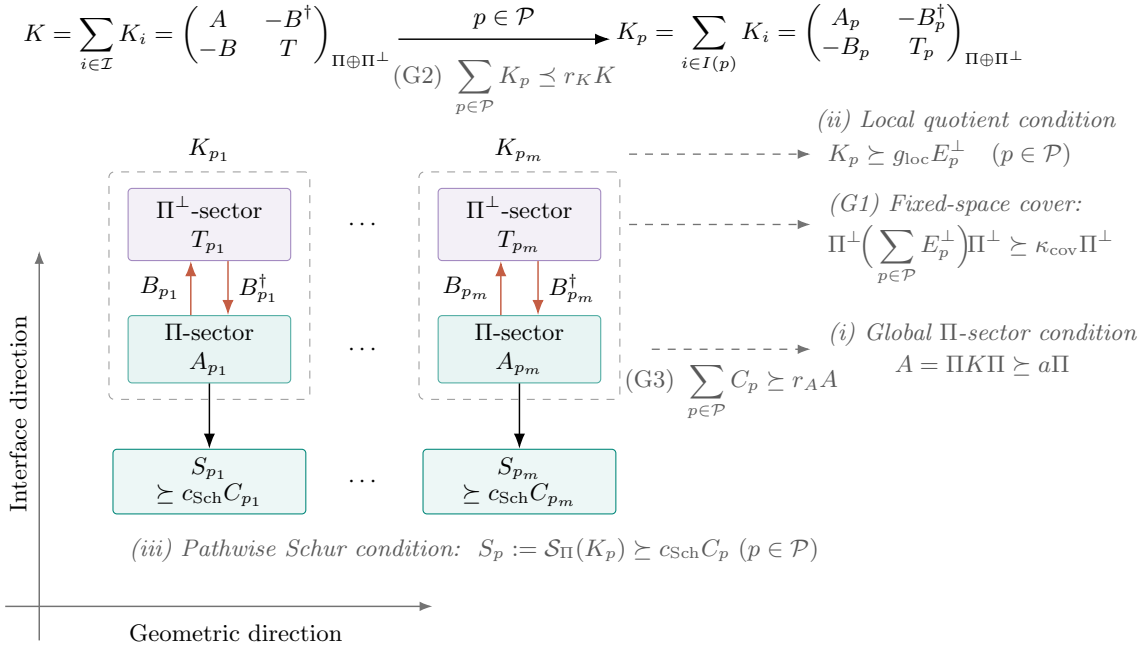
\begin{figure}[htb]
  \centering
\begin{tikzpicture}[
    block/.style={draw=black!55, rounded corners=1.5pt, align=center,
      minimum width=2.15cm, minimum height=0.62cm, inner sep=3pt},
    transverse sector/.style={block, draw=figureComplement!70, fill=figureComplement!9,
      text=black},
    interface sector/.style={block, draw=figureInterface!70, fill=figureInterface!9,
      text=black},
    certificate/.style={block, fill=figureInterface!7, draw=figureInterface,
      text=black,
      minimum width=2.55cm, minimum height=0.72cm},
    window/.style={draw=black!38, dashed, rounded corners=2pt,
      inner xsep=7pt, inner ysep=6pt},
    decomposition/.style={align=center, inner sep=2pt},
    condition text/.style={text=black!65},
    flow/.style={-{Latex[length=2mm]}, line width=0.55pt},
    relation/.style={-{Latex[length=1.7mm]},
      dashed, draw=black!55, line width=0.5pt},
    conceptual axis/.style={-{Latex[length=1.7mm]},
      draw=black!55, line width=0.55pt},
    every node/.style={font=\footnotesize}
  ]
    \node[decomposition, anchor=east] (globaldec) at (-1.40,3.40)
      {$K=\displaystyle\sum_{i\in\mathcal I}K_i
       =\begin{pmatrix}
         A&-B^\dagger\\
         -B&T
         \end{pmatrix}_{\Pi\oplus\Pi^\perp}$};
    \node[decomposition, anchor=west] (localdec) at (1.40,3.40)
      {$K_p=\displaystyle\sum_{i\in I(p)}K_i
       =\begin{pmatrix}
         A_p&-B_p^\dagger\\
         -B_p&T_p
         \end{pmatrix}_{\Pi\oplus\Pi^\perp}$};
    \draw[flow] (globaldec.east) --
      node[above, font=\footnotesize] {$p\in\mathcal P$}
      node[below, font=\footnotesize, yshift=-3pt, condition text]
        {$\textnormal{(G2)}\
         \displaystyle\sum_{p\in\mathcal P}K_p\preceq r_KK$}
      (localdec.west);

    \foreach \x/\tag in {-3.90/1,0.20/m} {
      \node[transverse sector] (t\tag) at (\x,0.95)
        {$\Pi^\perp$-sector\\$T_{p_\tag}$};
      \node[interface sector] (a\tag) at (\x,-0.70)
        {$\Pi$-sector\\$A_{p_\tag}$};
      \node[certificate] (s\tag) at (\x,-2.45)
        {$S_{p_\tag}$\\[-1pt]
          $\succeq c_{\mathrm{Sch}}C_{p_\tag}$};

      \draw[flow, draw=figureCoupling, text=black] ([xshift=-7pt]a\tag.north) --
        node[left, font=\footnotesize] {$B_{p_\tag}$}
        ([xshift=-7pt]t\tag.south);
      \draw[flow, draw=figureCoupling, text=black] ([xshift=7pt]t\tag.south) --
        node[right, font=\footnotesize] {$B_{p_\tag}^{\dagger}$}
        ([xshift=7pt]a\tag.north);
      \draw[flow] (a\tag.south) -- (s\tag.north);

      \node[window, fit=(t\tag)(a\tag),
        label={[font=\footnotesize]above:$K_{p_\tag}$}] (k\tag) {};
    }

    \node at (-1.85,0.95) {$\cdots$};
    \node at (-1.85,-0.70) {$\cdots$};
    \node at (-1.85,-2.45) {$\cdots$};

    \node[align=center, font=\footnotesize, condition text] at (-0.40,-3.35)
      {{\footnotesize\itshape(iii) Pathwise Schur condition:}\;
       $S_p:=\mathcal S_{\Pi}(K_p)\succeq c_{\mathrm{Sch}}C_p$
       $(p\in\mathcal P)$};

    \node[align=center, inner sep=1pt, condition text] at (6.10,2.32)
      {{\footnotesize\itshape(ii) Local quotient condition}};
    \node[align=left, anchor=west, inner sep=2pt, font=\footnotesize, condition text]
      (inputtwolocal) at (4.20,1.88)
      {$K_p\succeq g_{\mathrm{loc}}E_p^\perp\quad(p\in\mathcal P)$};
    \node[align=left, anchor=west, inner sep=2pt, font=\footnotesize, condition text]
      (inputtwocover) at (4.20,0.75)
      {{\itshape(G1) Fixed-space cover:}\\[2pt]
       $\Pi^\perp\!\Big(\displaystyle\sum_{p\in\mathcal P}E_p^\perp\Big)\!\Pi^\perp
       \succeq\kappa_{\mathrm{cov}}\Pi^\perp$};
    \node[align=center, anchor=west, inner sep=2pt, condition text]
      (inputone) at (4.20,-0.70)
      {{\footnotesize\itshape(i) Global $\Pi$-sector condition}\\[1pt]
       $A=\Pi K\Pi\succeq a\Pi$};

    \draw[relation] (1.65,1.88) -- (4.05,1.88);
    \draw[relation] (1.65,0.95) -- (4.05,0.95);
    \draw[relation] (1.95,-0.70) --
      node[below, font=\footnotesize, yshift=-1pt, condition text]
        {$\textnormal{(G3)}\
         \displaystyle\sum_{p\in\mathcal P}C_p\succeq r_AA$}
      (4.05,-0.70);

    \coordinate (axisorigin) at (-6.15,-4.10);
    \draw[draw=black!55, line width=0.55pt]
      (-6.60,-4.10) -- (axisorigin);
    \draw[conceptual axis] (axisorigin) -- (-0.95,-4.10)
      node[midway, below=3pt] {Geometric direction};
    \draw[draw=black!55, line width=0.55pt]
      (-6.15,-4.55) -- (axisorigin);
    \draw[conceptual axis] (axisorigin) -- (-6.15,0.60);
    \node[rotate=90] at (-6.43,-1.55) {Interface direction};
  \end{tikzpicture}
  \caption{
    Conditions of the abstract gap theorem under double localization: spectral and Schur bounds (i)–(iii), and geometric assembly conditions (G1)–(G3). The same interface projection \(\Pi\) is used throughout.
  }
  \label{fig:two-axis-proof-architecture}
\end{figure}

Intuitively, the conditions split into three categories. The first category is ``interface-local'': Condition~(i) says that, when compressed to a particular subspace, the positive operator $K$ has a spectral lower bound of at least $a$. The second category is ``geometrically local'': Condition~(ii) says that each piece $K_p$ of $K$ has a decent spectral gap outside of its respective kernel, while Condition (G2) says that these pieces can overlap, but their combined contribution is at most $r_K$ times that of $K$. Finally, the third category of conditions enforces that the two notions of interface and geometric locality have to play nicely with each other. Condition (G1) implies that the images of the geometrically-local operators $K_p$ collectively cover $\operatorname{ran}\Pi^\perp$ sufficiently well, where $\kappa_{\mathrm{cov}}$ measures how well the space is covered. Condition (iii), at an intuitive level, imposes that each geometrically-local operator $K_p$ retains enough control of the interface component even when an arbitrary component in the ``uncontrolled'' space $\operatorname{ran}\Pi^\perp$ is allowed to reduce the energy (the shorted operator $\mathcal S_\Pi(K_p)$ records the minimum over that component). The operators $C_p$, which act only on $\operatorname{ran}\Pi$, specify the control that must remain, up to the common factor $c_{\mathrm{Sch}}$. Finally, Condition (G3) says that these $C_p$ operators, when added together, must provide at least $r_A$ times the control supplied by $A$.
The assembly of these conditions is illustrated in
\Cref{fig:abstract-gap-dependency-map} in \cref{sec:abstract-two-axis-assembly}.

\subsection{Finite-Group Davies Gaps}
\label{sec:main-finite-group-gaps}

Fix a nontrivial finite group $\Gamma$ and $J>0$.  Let
$\mathfrak G^{(3)}_{\geq6}$ denote the finite simple cubic graphs of girth
at least six.  For $\cG\in\mathfrak G^{(3)}_{\geq6}$, orient the edges and set
\[
 \cH_{\cG}=\bigotimes_{e\in E(\cG)}\C[\Gamma],
 \qquad
 A_v=\frac1{|\Gamma|}\sum_{g\in\Gamma}U_v(g),
 \qquad
 H_{\cG}=-J\sum_{v\in V(\cG)}A_v.
\]
Here $\C[\Gamma]$ has orthonormal basis $\{\lvert g\rangle:g\in\Gamma\}$,
and $U_v(g)$ applies the left or right regular action on the three
incident edges according to their orientations.  The $A_v$ are commuting
projections.  We use the specified unit-rate complete star Davies bath,
built from left and right regular translations and group-basis projectors
on each edge, with the construction and normalization of
\cref{sec:prelim-davies,sec:prelim-quantum-double}.  Write
\[
 \rho_\beta=\frac{e^{-\beta H_{\cG}}}{\operatorname{Tr}(e^{-\beta H_{\cG}})},
 \qquad K_{\cG,\beta}=\sum_{e\in E(\cG)}K_e,
 \qquad t=e^{\beta J}-1.
\]
Figure~\ref{fig:finite-group-star-model} shows an admissible graph, a
vertex star, and a centered three-edge path
$p=(e_-,e_0,e_+)$.  These paths implement the geometric localization,
while the irrep-label observables serve as the interface sector.

\begin{figure}[t]
  \centering
  \begin{tikzpicture}[
    >=Stealth,
    graph edge/.style={draw=black!42, line width=0.55pt, ->},
    marked edge/.style={draw=figureGeometry, line width=1.35pt, ->},
    ambient edge/.style={draw=black!22, line width=0.5pt},
    path edge/.style={draw=figureGeometry, line width=1.2pt},
    middle edge/.style={draw=figureGeometry, line width=1.7pt},
    graph vertex/.style={circle, draw=black!65, fill=white,
      inner sep=0pt, minimum size=3.6pt},
    ambient vertex/.style={circle, draw=black!38, fill=white,
      inner sep=0pt, minimum size=3.5pt},
    marked vertex/.style={circle, draw=figureGeometry,
      fill=figureGeometry!16, line width=0.8pt, inner sep=0pt, minimum size=5.5pt},
    path vertex/.style={circle, draw=figureGeometry, fill=figureGeometry!10,
      line width=0.8pt, inner sep=0pt, minimum size=5.8pt},
    local edge/.style={draw=black!75, line width=0.9pt, ->},
    every node/.style={font=\small}
  ]
    \foreach \i in {0,...,13} {
      \coordinate (h\i) at ({90-360*\i/14}:1.9);
    }
    \foreach \i [evaluate=\i as \j using {int(mod(\i+1,14))}]
      in {0,...,13} {
      \draw[graph edge] (h\i) -- (h\j);
    }
    \foreach \i/\j in {0/5,2/7,4/9,6/11,8/13,10/1,12/3} {
      \draw[graph edge] (h\i) -- (h\j);
    }

    \draw[marked edge] (h13) -- (h0);
    \draw[marked edge] (h0) -- (h1);
    \draw[marked edge] (h0) -- (h5);
    \foreach \i in {0,...,13} {
      \node[graph vertex] at (h\i) {};
    }
    \node[marked vertex] (hv) at (h0) {};
    \node[above=2pt of hv, text=figureGeometry] {$v$};
    \node[below=2.16cm, text width=4cm, align=center] at (0,0)
      {Heawood graph (cubic, girth $6$)};

    \coordinate (sv) at (5.20,0.25);
    \coordinate (s1) at (4.00,1.27);
    \coordinate (s2) at (6.53,0.56);
    \coordinate (s3) at (4.65,-1.18);
    \draw[local edge] (sv) -- node[pos=.62, above=5pt] {$L_g$} (s1);
    \draw[local edge] (s2) -- node[above=1pt] {$R_g$} (sv);
    \draw[local edge] (sv) -- node[right=1pt] {$L_g$} (s3);
    \node[marked vertex, minimum size=7pt] at (sv) {};
    \node[left=4pt of sv, text=figureGeometry] {$v$};
    \foreach \s in {s1,s2,s3} {
      \node[graph vertex, minimum size=4.5pt] at (\s) {};
    }
    \node[below=2.16cm, text width=3.5cm, align=center] at (5.20,0)
      {local star term $A_v$};

    \begin{scope}[xshift=10.45cm]
      \foreach \i in {0,...,13} {
        \coordinate (p\i) at ({90-360*\i/14}:1.9);
      }
      \foreach \i [evaluate=\i as \j using {int(mod(\i+1,14))}]
        in {0,...,13} {
        \draw[ambient edge] (p\i) -- (p\j);
      }
      \foreach \i/\j in {0/5,2/7,4/9,6/11,8/13,10/1,12/3} {
        \draw[ambient edge] (p\i) -- (p\j);
      }
      \draw[path edge] (p13) -- node[midway, above left=1pt] {$e_-$} (p0);
      \draw[middle edge] (p0) -- node[midway, above=3pt] {$e_0$} (p1);
      \draw[path edge] (p1) -- node[midway, right=3pt] {$e_+$} (p2);
      \foreach \i in {0,...,13} {
        \node[ambient vertex] at (p\i) {};
      }
      \foreach \i in {13,0,1,2} {
        \node[path vertex] at (p\i) {};
      }
      \node[below=5pt of p0, text=figureGeometry, fill=white, inner sep=1pt,
        font=\footnotesize] {$u$};
      \node[below left=6pt of p1, text=figureGeometry, fill=white,
        inner sep=1pt, font=\footnotesize] {$v$};
      \node[below=2.16cm, text width=4cm, align=center] at (0,0)
        {centered path $p=(e_-,e_0,e_+)$};
    \end{scope}
  \end{tikzpicture}
  \caption{An admissible graph and the two local structures used in the
  proof.  Left: the Heawood graph, a cubic graph of girth six, with one
  star highlighted.  Center: the enlarged star defines
  $A_v=|\Gamma|^{-1}\sum_{g\in\Gamma}U_v(g)$.  Outgoing and incoming half-edges
  carry the left and right regular actions.  Right: a centered three-edge
  path with target middle edge $e_0$ and flank edges $e_-$ and $e_+$.  The
  rest of the graph is shown lightly.}
  \label{fig:finite-group-star-model}
\end{figure}
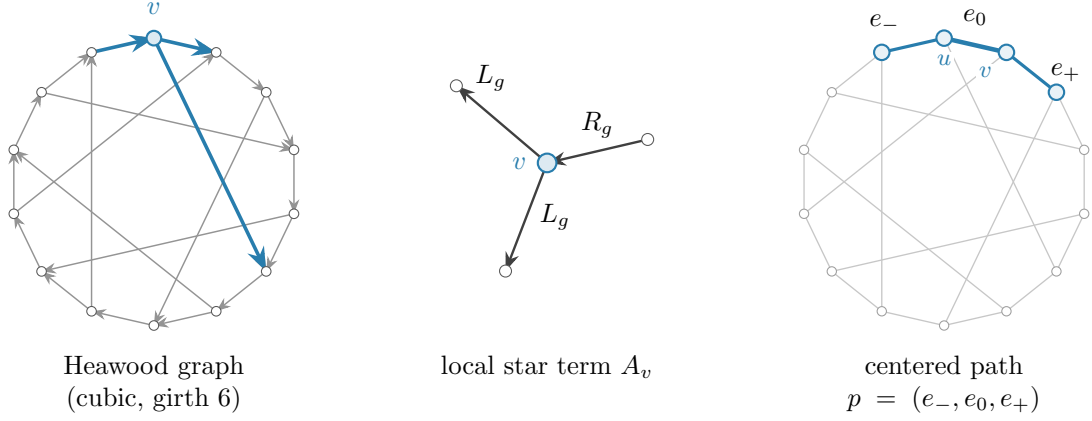

The girth assumption prevents short-cycle identifications among the
centered paths and neighboring stars used in the proof.
The parameter $t$ increases with inverse temperature, not physical
temperature.  The group, coupling, and bath normalization are fixed as the
graph grows, while $\dim\cH_{\cG}=|\Gamma|^{|E(\cG)|}$ increases.
The full-GNS gap is the gap on all mean-zero observables in the GNS inner
product $\langle X,Y\rangle_{\mathrm{GNS},\rho_\beta}
=\operatorname{Tr}(\rho_\beta X^\dagger Y)$.
For these Davies generators it agrees with the KMS gap by
\cref{sec:prelim-davies}; the results below also establish
$\ker K_{\cG,\beta}=\C I$.

\begin{theorem}[Finite-group Davies gaps]
\label{thm:finite-group-davies-gaps}
Fix a nontrivial finite group $\Gamma$.
\begin{enumerate}[label=\textnormal{(\roman*)}]
\item For every $0\le t_*<2/3$, set
$\beta_*:=J^{-1}\log(1+t_*)$, so that
$0\le\beta_*J<\log(5/3)$.  There is a constant
$c_{\Gamma,t_*}>0$ independent of $|V(\cG)|$ such that, for every $\cG\in \mathfrak G^{(3)}_{\geq 6}$ and every
$0\le t\le t_*$, or equivalently $0\le\beta\le\beta_*$,
the positive Davies operator $K_{\cG,\beta}$ satisfies
\[
  \gap_{\mathrm{GNS}}(K_{\cG,\beta})\ge c_{\Gamma,t_*}.
\]

\item For $\Gamma=S_3$, the same volume-uniform full-GNS conclusion holds throughout
$0\le t\le2$, equivalently $0\le\beta J\le\log 3$.
\end{enumerate}
\end{theorem}

\begin{proof}
For part~(i), \cref{sec:common-population-certificate,sec:common-fibre-certificate,sec:common-schur-certificate}
verify Conditions~(i)--(iii) of \cref{thm:abstract-two-axis-gap-assembly}
with positive constants uniform in graph size and $0\le t\le t_*$.  The
fixed-space cover in \cref{lem:cubic-defect-frame} and the incidence identities
in \cref{lem:cubic-path-ownership} give
$\kappa_{\mathrm{cov}}=6$, $r_K=12$, and $r_A=4$.
The correspondence between the abstract and model constants is summarized
in \cref{sec:model-constant-summary}.
The abstract theorem therefore yields a uniform positive lower bound on
$K_{\cG,\beta}$ on $L^2_0(\rho_\beta)$ in the KMS inner product.
Any fixed observable $X$ has a fixed mean-zero part
$X-\operatorname{Tr}(\rho_\beta X)I$, which this bound forces to vanish;
hence $\ker K_{\cG,\beta}=\C I$.
By \cref{prop:davies-kms-gns-gap-bridge}, the same bound gives the stated
GNS gap.  Part~(ii) is proved in \cref{sec:s3-strengthening}.
Uniformity is over the graph size and the stated temperature interval,
not over $|\Gamma|$.
\end{proof}

The temperature restriction in part~(i) comes only from the classical
estimate for the label interface.  The local quantum estimates remain
valid on any bounded interval of inverse temperatures, with constants
that may depend on the upper endpoint, while the geometric conditions
are independent of temperature.  Thus, improving the classical interface
estimate can extend the temperature range without changing the remaining
parts of the argument.

\FloatBarrier

The $S_3$ improvement follows by verifying a finite classical certificate
on conditioned vertex stars and combining it with the local quantum and
geometric estimates. The certificate-to-gap principle and its verification
are given in \cref{sec:s3-star-certificate,sec:s3-exact-certificate-application}.

The model-specific projection, local update collections, and assembly conditions are
specified in \cref{sec:formal-model-split-path-cover,sec:three-certificate-interfaces,sec:finite-group-incidence-assembly}
and verified in \cref{sec:common-analytic-certificates}.

\section{Preliminaries for the Abstract Gap Theorem}
\label{sec:preliminaries}

This section collects the finite-dimensional operator tools used in the
abstract gap argument. \Cref{sec:prelim-general-conventions} fixes the
operator conventions and records the kernel identity for positive sums.
\Cref{sec:prelim-shorted-operators} introduces shorted operators and their
Schur-complement formulas, while \cref{sec:prelim-above-kernel-poincare}
records gap bounds above a nontrivial kernel and a compactness criterion
for their uniformity. The background for the finite-group application
is collected separately in \cref{sec:model-preliminaries}.

\subsection{Operator Conventions and Kernels}
\label{sec:prelim-general-conventions}

All Hilbert spaces are finite dimensional. Adjoints and orthogonal
projections are defined using the chosen inner product.
For self-adjoint operators $A$ and $B$, we write $A\succeq B$ when
$\langle x,(A-B)x\rangle\ge0$ for every $x$.
The notation $\lVert\cdot\rVert$ denotes the vector norm or the corresponding
operator norm, depending on its argument.
An operator defined on a subspace is also viewed as an operator on the
whole space by setting it to zero on the orthogonal complement.

A trace subscript specifies the space over which the full trace is taken.
Partial traces are explicitly identified, with the subscript indicating
the factors being traced out.
The Hilbert--Schmidt norm is
\[
  \lVert X\rVert_{\mathrm{HS}}
  :=\bigl(\operatorname{Tr}(X^\dagger X)\bigr)^{1/2}.
\]

\begin{lemma}[Kernel of a sum of positive operators]
\label{lem:positive-sum-kernel}
Let $\{K_i\}_{i\in\mathcal I}$ be a finite family of positive operators
on a common Hilbert space. For every nonempty $\mathcal J\subseteq\mathcal I$,
\[
  \ker\Bigl(\sum_{i\in\mathcal J}K_i\Bigr)
  =\bigcap_{i\in\mathcal J}\ker K_i.
\]
\end{lemma}

\begin{proof}
If $x$ belongs to the kernel of the sum, positivity gives
\[
  0=\Bigl\langle x,\Bigl(\sum_{i\in\mathcal J}K_i\Bigr)x\Bigr\rangle
   =\sum_{i\in\mathcal J}\lVert K_i^{1/2}x\rVert^2.
\]
Thus $K_ix=0$ for every $i\in\mathcal J$. The reverse inclusion is immediate.
\end{proof}

\subsection{Shorted Operators and Schur Complements}
\label{sec:prelim-shorted-operators}

Let $\mathcal K$ be a finite-dimensional Hilbert space, let $\Pi$ be an
orthogonal projection on $\mathcal K$, and set $\Pi^\perp=I-\Pi$.
Relative to
$\mathcal K=\operatorname{ran}\Pi\oplus\operatorname{ran}\Pi^\perp$, write
\[
  G=
  \begin{pmatrix}
    A&-B^\dagger\\
    -B&T
  \end{pmatrix}
  \succeq0.
\]
In this finite-dimensional setting, we use the following block formula as
the definition.

\begin{definition}[Shorted operator]
\label{def:shorted-operator}
The shorted operator of $G$ to $\operatorname{ran}\Pi$ is
\begin{equation}
  \label{eq:short-pseudoinverse-formula}
  \mathcal S_{\Pi}(G)
  :=
  \begin{pmatrix}
    A-B^\dagger T^+B&0\\
    0&0
  \end{pmatrix}_{\operatorname{ran}\Pi\oplus\operatorname{ran}\Pi^\perp},
\end{equation}
where $T^+$ is the Moore--Penrose pseudoinverse.
\end{definition}

The retained block $A-B^\dagger T^+B$ is the generalized Schur complement of
$T$ in $G$; when $T$ is invertible, it is the ordinary Schur complement
$A-B^\dagger T^{-1}B$.  This is the finite-dimensional block form of the
classical shorted-operator construction
\cite{anderson1975shorted,friedrich2018generalized}.  It also makes the
support convention explicit: $\mathcal S_{\Pi}(G)$ maps
$\operatorname{ran}\Pi$ into itself and vanishes on $\operatorname{ran}\Pi^\perp$.

The first lemma gives two intrinsic characterizations of the block
definition.

\begin{lemma}[Variational and maximal characterizations]
\label{lem:shorted-operator-characterizations}
The operator $\mathcal S_{\Pi}(G)$ is positive and is the greatest
element, in the Loewner order, of
\begin{equation}
  \label{eq:short-maximal-definition}
  \bigl\{X\succeq0:X\preceq G,\ X=\Pi X\Pi\bigr\}.
\end{equation}
Equivalently, for every $x\in\operatorname{ran}\Pi$,
\begin{equation}
  \label{eq:short-variational-definition}
  \langle x,\mathcal S_{\Pi}(G)x\rangle
  =\inf_{y\in\operatorname{ran}\Pi^\perp}
    \langle x+y,G(x+y)\rangle.
\end{equation}
\end{lemma}

\begin{proof}
Positivity of $G$ implies
$\ker T\subseteq\ker B^\dagger$ and hence, in finite dimension,
$\operatorname{ran}B\subseteq\operatorname{ran}T$.  Completing the square
therefore gives
\[
  \langle x+y,G(x+y)\rangle
  =
  \langle x,(A-B^\dagger T^+B)x\rangle
  +\langle y-T^+Bx,T(y-T^+Bx)\rangle.
\]
This proves \eqref{eq:short-variational-definition} and shows that the
operator in \eqref{eq:short-pseudoinverse-formula} is positive and bounded
above by $G$.  If $0\preceq X\preceq G$ and $X=\Pi X\Pi$, then for every
$x\in\operatorname{ran}\Pi$ and $y\in\operatorname{ran}\Pi^\perp$,
\[
  \langle x,Xx\rangle
  =\langle x+y,X(x+y)\rangle
  \leq\langle x+y,G(x+y)\rangle.
\]
Taking the infimum over $y$ proves that the operator in
\eqref{eq:short-pseudoinverse-formula} is the greatest element in
\eqref{eq:short-maximal-definition}.
\end{proof}

The variational formula has a simple interpretation: we fix the
interface component $x$ and choose the complementary component $y$
to minimize the quadratic form.
The compression $A$ corresponds to setting $y=0$, whereas
$\mathcal S_\Pi(G)$ measures what remains after this minimization.
Thus a lower bound on $\mathcal S_\Pi(G)$ controls the loss caused
by coupling to the complementary subspace.

The second lemma isolates the order properties used in the assembly
argument.

\begin{lemma}[Order properties of shorting]
\label{lem:shorted-operator-order-properties}
On a fixed split, shorting is positively homogeneous and monotone:
\[
  \mathcal S_{\Pi}(cG)
  =c\mathcal S_{\Pi}(G)
  \quad(c\geq0),
  \qquad
  G\preceq H
  \ \Longrightarrow\
  \mathcal S_{\Pi}(G)
  \preceq\mathcal S_{\Pi}(H).
\]
It is also superadditive: for any finite family $G_j\succeq0$,
\begin{equation}
  \label{eq:short-superadditivity}
  \mathcal S_{\Pi}\!\left(\sum_jG_j\right)
  \succeq\sum_j\mathcal S_{\Pi}(G_j).
\end{equation}
In particular,
\[
  \mathcal S_{\Pi}(G)\preceq\Pi G\Pi.
\]
\end{lemma}

\begin{proof}
Positive homogeneity and monotonicity follow immediately from
\eqref{eq:short-variational-definition}.  For superadditivity, the same
formula gives
\[
  \inf_y\sum_j\langle x+y,G_j(x+y)\rangle
  \geq
  \sum_j\inf_{y_j}\langle x+y_j,G_j(x+y_j)\rangle.
\]
Comparison of quadratic forms proves \eqref{eq:short-superadditivity}.
Finally, taking $y=0$ in \eqref{eq:short-variational-definition} gives the
last inequality.
\end{proof}

\subsection{Above-Kernel Poincar\'e Inequalities}
\label{sec:prelim-above-kernel-poincare}

Let $D\succeq0$ be self-adjoint on a finite-dimensional Hilbert space, let
$F:=\ker D$, and let $E_F$ be the orthogonal projection onto $F$.  A lower
bound above the kernel may be written either as the operator inequality
\[
  D\succeq g(I-E_F)
\]
or, equivalently, as the Poincar\'e inequality
\begin{equation}
  \label{eq:above-kernel-poincare}
  \langle X,DX\rangle
  \ge g\,\operatorname{dist}(X,F)^2,
  \qquad
  \operatorname{dist}(X,F):=\inf_{Y\in F}\lVert X-Y\rVert.
\end{equation}
Indeed, $\operatorname{dist}(X,F)=\lVert(I-E_F)X\rVert$, and both statements
say that the restriction of $D$ to $F^\perp$ has spectral bottom at least
$g$.  We use the two forms interchangeably below.

The following standard finite-dimensional perturbation argument makes this
bound uniform in a compact parameter once the kernel is fixed; see, for
example, \cite[Chapter~II]{kato1995perturbation} and
\cite[Section~1.3]{tao2012topics}.

\begin{lemma}[Finite-dimensional compactness above a fixed kernel]
\label{lem:finite-compactness-fixed-kernel}
Let $V$ be finite dimensional, let $I\subset\mathbb R$ be compact, and let
$\langle\cdot,\cdot\rangle_t$ be a continuous family of positive-definite
inner products on $V$.  Suppose $D_t$ is a continuous family of positive
operators, self-adjoint in $\langle\cdot,\cdot\rangle_t$, and that
$\ker D_t=F$ is independent of $t\in I$.  Write
\[
  \|X\|_t^2:=\langle X,X\rangle_t,
  \qquad
  \operatorname{dist}_t(X,F):=\inf_{Y\in F}\|X-Y\|_t.
\]
Then there is $g>0$ such that
\begin{equation}
  \label{eq:finite-compactness-fixed-kernel}
  \langle X,D_tX\rangle_t
  \ge g\,\operatorname{dist}_t(X,F)^2
  \qquad(X\in V,\ t\in I).
\end{equation}
\end{lemma}

\begin{proof}
If $F=V$, then the right-hand side of
\eqref{eq:finite-compactness-fixed-kernel} vanishes and the conclusion holds
for any $g>0$.  Assume henceforth that $F\ne V$.
Choose a fixed linear complement $V=F\oplus W$.  On $W$, the numerator in
\eqref{eq:finite-compactness-fixed-kernel} and the squared distance to $F$
are represented by continuous families of positive-definite matrices.  The
best constant is their smallest generalized eigenvalue.  It is positive for
each $t$ because the kernel is exactly $F$, and it depends continuously on
$t$.  Its minimum on the compact set $I$ is therefore strictly positive.
\end{proof}

\section{Proof of the Abstract Gap Theorem}
\label{sec:three-local-inputs}

We prove \cref{thm:abstract-two-axis-gap-assembly} by assembling the local
bounds and then applying a two-block estimate. The argument uses no group,
graph, or representation-theoretic structure.

\subsection{Abstract Local-to-Global Assembly}
\label{sec:abstract-two-axis-assembly}

\Cref{fig:abstract-gap-dependency-map} summarizes how the local and geometric
conditions give the two global bounds used in the final block estimate.

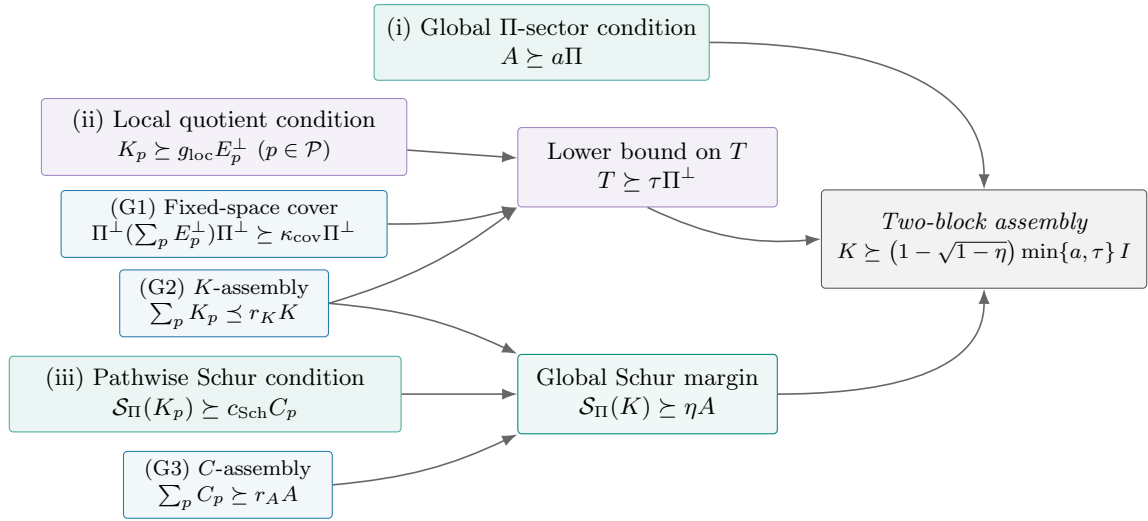
\begin{figure}[htbp]
  \centering
  \begin{tikzpicture}[
    assumption/.style={draw=black!45, rounded corners=1.5pt, fill=black!3,
      align=center, text width=4.90cm, minimum height=1.00cm, inner sep=4pt},
    assembly/.style={draw=figureGeometry, rounded corners=1.5pt, fill=figureGeometry!6,
      align=center, text width=2.55cm, minimum height=0.78cm, inner sep=3pt,
      font=\scriptsize},
    consequence/.style={draw=figureInterface, rounded corners=1.5pt,
      fill=figureInterface!6, align=center, text width=3.15cm, minimum height=1.05cm,
      inner sep=4pt},
    conclusion/.style={draw=black!65, rounded corners=1.5pt, fill=black!5,
      align=center, text width=3.95cm, minimum height=1.30cm, inner sep=5pt},
    dependency/.style={-{Latex[length=2mm]}, draw=black!60,
      line width=0.6pt},
    every node/.style={font=\footnotesize,text=black}
  ]
    \node[assumption, draw=figureInterface!70, fill=figureInterface!8, text width=4.15cm] (selectedinput) at (-0.2,2.60)
      {\mbox{\textnormal{(i)} Global $\Pi$-sector condition}\par
       \mbox{$A\succeq a\Pi$}};
    \node[assumption, draw=figureComplement!70, fill=figureComplement!8, text width=4.55cm] (transverseinput) at (-4.40,1.35)
      {\mbox{\textnormal{(ii)} Local quotient condition}\par
       {\scriptsize\mbox{$K_p\succeq g_{\mathrm{loc}}E_p^\perp\ (p\in\mathcal P)$}}};
    \node[assumption, draw=figureInterface!70, fill=figureInterface!8] (schurinput) at (-4.65,-2.05)
      {\mbox{\textnormal{(iii)} Pathwise Schur condition}\par
       \mbox{$\mathcal S_{\Pi}(K_p)\succeq c_{\mathrm{Sch}}C_p$}};

    \node[assembly, text width=4.10cm] (coverassembly) at (-4.40,0.20)
      {\mbox{\textnormal{(G1)} Fixed-space cover}\par
       \mbox{$\Pi^\perp(\sum_pE_p^\perp)\Pi^\perp
       \succeq\kappa_{\mathrm{cov}}\Pi^\perp$}};
    \node[assembly] (windowassembly) at (-4.40,-0.85)
      {\mbox{\textnormal{(G2)} $K$-assembly}\par
       \mbox{$\textstyle\sum_pK_p\preceq r_KK$}};
    \node[assembly] (targetassembly) at (-4.35,-3.25)
      {\mbox{\textnormal{(G3)} $C$-assembly}\par
       \mbox{$\textstyle\sum_pC_p\succeq r_AA$}};

    \node[consequence, draw=figureComplement!70, fill=figureComplement!9,
      text=black] (transversefloor) at (1.20,0.95)
      {\mbox{Lower bound on $T$}\par
       \mbox{$T\succeq\tau\Pi^\perp$}};
    \node[consequence] (schurmargin) at (1.20,-2.05)
      {\mbox{Global Schur margin}\par
       \mbox{$\mathcal S_{\Pi}(K)\succeq\eta A$}};

    \node[conclusion] (fullgap) at (5.65,0)
      {\mbox{\textit{Two-block assembly}}\par
       {\scriptsize\mbox{$K\succeq
       \bigl(1-\sqrt{1-\eta}\bigr)\min\{a,\tau\}\,I$}}};

    \draw[dependency] (transverseinput) -- (transversefloor);
    \draw[dependency] (coverassembly.east)
      to[bend right=8] (transversefloor.south west);
    \draw[dependency] (schurinput) -- (schurmargin);
    \draw[dependency] (windowassembly.east)
      to[bend right=10] (transversefloor.south west);
    \draw[dependency] (windowassembly.east)
      to[bend left=10] (schurmargin.north west);
    \draw[dependency] (targetassembly.east)
      to[bend right=10] (schurmargin.south west);
    \draw[dependency] (selectedinput.east) to[out=0,in=90] (fullgap.north);
    \draw[dependency] (transversefloor.south)
      to[bend right=15] (fullgap.west);
    \draw[dependency] (schurmargin.east) to[out=0,in=-90] (fullgap.south);
  \end{tikzpicture}
  \caption{Dependency map for the abstract gap estimate.  The blue
  fixed-space cover and $K$-assembly blocks combine with Condition~(ii) to
  produce the lower bound on $T$.  The $K$- and $C$-assembly blocks combine
  with Condition~(iii) to produce the global Schur margin.  Condition~(i)
  enters the final two-block estimate directly.}
  \label{fig:abstract-gap-dependency-map}
\end{figure}

We use the notation of \cref{thm:abstract-two-axis-gap-assembly}.  Thus
$K\succeq0$ is split by an orthogonal projection $\Pi$ as
\[
  K=
  \begin{pmatrix}
    A&-B^\dagger\\
    -B&T
  \end{pmatrix}_{\Pi\oplus\Pi^\perp},
\]
and $\{K_p:p\in\mathcal P\}$ is a family of positive local operators.
For each $p$, $E_p$ projects onto $\ker K_p$, and
$C_p=\Pi C_p\Pi\succeq0$ is its target form.  Conditions~(ii)--(iii) are the
local quotient and shorted-form bounds, while Geometric
Conditions~(G1)--(G3) are the fixed-space cover and the two assembly
comparisons.  No structure of
the index set $\mathcal P$ is used in the next two steps.

\begin{proposition}[Global bounds from local conditions]
\label{prop:global-bounds-from-local-inputs}
Assume Conditions~\textnormal{(ii)}--\textnormal{(iii)} and Geometric
Conditions~\textnormal{(G1)}--\textnormal{(G3)} of
\cref{thm:abstract-two-axis-gap-assembly}, and define
$\tau = \frac{g_\mathrm{loc}\kappa_{\mathrm{cov}}}{r_K}$ and
$\eta = \min \left\{\frac{c_{\mathrm{Sch}} r_A}{r_K}, 1 \right\}$.  Then
$T\succeq\tau\Pi^\perp$ and
$\mathcal S_{\Pi}(K)\succeq\eta A$.
\end{proposition}

\begin{proof}
By Geometric Condition~\textnormal{(G2)},
$\sum_{p}K_p\preceq r_KK$, and since
$\Pi^\perp K\Pi^\perp=T$, it follows that
\[ 
    r_KT \succeq\Pi^\perp\!\Big(\sum_{p\in\mathcal P}K_p\Big)\!\Pi^\perp.
\]
Recall that $\Pi^{\perp}$ is an orthogonal projector and thus $\left(\Pi^{\perp}\right)^2 = \left(\Pi^{\perp}\right)^\dagger = \Pi^{\perp}$. Applying Condition~\textnormal{(ii)} and then Geometric Condition~\textnormal{(G1)} of \cref{thm:abstract-two-axis-gap-assembly}, it follows that
\[
    \Pi^\perp\!\Big(\sum_{p\in\mathcal P}K_p\Big)\!\Pi^\perp
    \succeq
  g_{\mathrm{loc}}
  \Pi^\perp\!\Big(\sum_{p\in\mathcal P}E_p^\perp\Big)\!\Pi^\perp
  \succeq
  g_{\mathrm{loc}}\kappa_{\mathrm{cov}}\Pi^\perp.
\]
Thus, it immediately follows that
\begin{equation}
  \label{eq:assembled-transverse-floor}
  T\succeq \tau\Pi^\perp.
\end{equation}
The claim about $\mathcal S_{\Pi}(K)$ follows from Geometric
Conditions~\textnormal{(G2)}--\textnormal{(G3)}, the superadditivity of the
shorting operator in \cref{lem:shorted-operator-order-properties}, and
Condition~\textnormal{(iii)} of \cref{thm:abstract-two-axis-gap-assembly}:
\begin{align*}
  r_K\mathcal S_{\Pi}(K)
  &=\mathcal S_{\Pi}(r_KK)
  \succeq
    \mathcal S_{\Pi}\!\left(\sum_{p\in\mathcal P}K_p\right)
  \succeq
    \sum_{p\in\mathcal P}\mathcal S_{\Pi}(K_p)\\
  &\succeq
    c_{\mathrm{Sch}}\sum_{p\in\mathcal P}C_p
  \succeq c_{\mathrm{Sch}}r_AA.
\end{align*}
Therefore
$
  \mathcal S_{\Pi}(K)
  \succeq \frac{c_{\mathrm{Sch}}r_A}{r_K}A
  \succeq \eta A.
$
\end{proof}

It remains to combine these two global bounds with Condition~\textnormal{(i)}.
The following two-block estimate completes the assembly.

\begin{lemma}[Two-block gap assembly]
\label{lem:two-block-gap-assembly}
Let
\[
  K=
  \begin{pmatrix}
    A&-B^\dagger\\
    -B&T
  \end{pmatrix}_{\Pi\oplus\Pi^\perp}
  \succeq0.
\]
If, for some $a,\tau>0$ and $0<\eta\le1$,
\[
  A\succeq a\Pi,
  \qquad
  T\succeq \tau\Pi^\perp,
  \qquad
  \mathcal S_{\Pi}(K)\succeq\eta A,
\]
then
\begin{equation}
  \label{eq:two-block-assembled-gap}
  K\succeq
  \left(1-\sqrt{1-\eta}\right)\min\{a,\tau\}\,I.
\end{equation}
\end{lemma}

\begin{proof}
The first two bounds make $A$ and $T$ invertible on $\Pi$ and $\Pi^\perp$,
respectively. Hence \cref{def:shorted-operator} becomes
\[
  \mathcal S_{\Pi}(K)=A-B^\dagger T^{-1}B.
\]
We define
\[
  C:=T^{-1/2}BA^{-1/2}:\operatorname{ran}\Pi
  \longrightarrow\operatorname{ran}\Pi^\perp.
\]
Since $A, T$ are PSD matrices, so are $A^{-1/2}$ and $T^{-1/2}$. Thus,
\[
    A^{-1/2} \mathcal S_{\Pi}(K) A^{-1/2} = I - C^\dagger C.
\]
Since $\mathcal S_{\Pi}(K) \succeq \eta A$ and $A$ is invertible on $\Pi$, it follows that $A^{-1/2} \mathcal S_{\Pi}(K) A^{-1/2} \succeq \eta \Pi$. Thus, 
\[
    C^\dagger C \preceq (1 - \eta)\Pi,
\]
and therefore $\lVert C\rVert\le\sqrt{1-\eta}$. We factorize $K$ into block matrices
\[
  K=
  \begin{pmatrix}A^{1/2}&0\\0&T^{1/2}\end{pmatrix}
  \begin{pmatrix}\Pi&-C^\dagger\\-C&\Pi^\perp\end{pmatrix}
  \begin{pmatrix}A^{1/2}&0\\0&T^{1/2}\end{pmatrix}.
\]
To bound the effect of the middle term, let $z = x + y$ be an arbitrary vector with $x \in \operatorname{ran}\Pi$ and $y \in \operatorname{ran}\Pi^\perp$. Then
\[
  \begin{aligned}
    z^{\dagger} \begin{pmatrix} 0 &-C^\dagger\\-C& 0 \end{pmatrix} z
    &= -2\operatorname{Re}(y^\dagger Cx)\\
    &\geq -2\|C\|\|x\|\|y\|\\
    &\geq -\|C\|(\|x\|^2+\|y\|^2)
     = -\|C\|\|z\|^2.
  \end{aligned}
\]
Thus, 
\[
  \begin{pmatrix}\Pi&-C^\dagger\\-C&\Pi^\perp\end{pmatrix}
  \succeq (1-\lVert C\rVert)I.
\]
Combining this with the assumptions that $A \succeq a \Pi$ and $T \succeq \tau \Pi^\perp$, we obtain
\[
  K
  \succeq
  \left(1-\sqrt{1-\eta}\right)
  \begin{pmatrix}A&0\\0&T\end{pmatrix}
  \succeq
  \left(1-\sqrt{1-\eta}\right)\min\{a,\tau\}\,I.
\]
This is \eqref{eq:two-block-assembled-gap}.
\end{proof}

\begin{proof}[Proof of \cref{thm:abstract-two-axis-gap-assembly}]
Conditions~(ii)--(iii), together with Geometric
Conditions~(G1)--(G3), give the stated bounds on $T$ and
$\mathcal S_{\Pi}(K)$ by
\cref{prop:global-bounds-from-local-inputs}.  Combining these bounds with
Condition~(i), \cref{lem:two-block-gap-assembly} gives
\eqref{eq:abstract-two-axis-assembled-gap}.
\end{proof}

\section{Preliminaries for the Finite-Group Model}
\label{sec:model-preliminaries}

This section collects the background for the finite-group theorem.
\Cref{sec:prelim-gns} fixes the KMS and GNS inner products and the mean-zero
observable spaces. \Cref{sec:prelim-davies} recalls the Davies construction,
its fixed points, and the equality of its KMS and GNS gaps.
\Cref{sec:prelim-quantum-double} describes the quantum-double vertex terms
retained by our model, and \cref{sec:prelim-finite-groups} develops the
representation theory used to describe observable blocks and irrep labels.
Finally, \cref{sec:prelim-dobrushin} recalls classical heat-bath chains and
Dobrushin comparison for the label-gap estimates.

\subsection{KMS and GNS Geometry}
\label{sec:prelim-gns}

Let $\cH$ be a finite-dimensional Hilbert space, let $H$ be a
Hamiltonian on $\cH$, and fix $\beta\ge0$.  Its Gibbs state is
\[
  \rho_\beta:=Z_\beta^{-1}e^{-\beta H},
  \qquad
  Z_\beta:=\operatorname{Tr}(e^{-\beta H}).
\]
For the fixed Gibbs state $\rho_\beta$, we use the inner products
\[
  \langle X,Y\rangle_{\mathrm{KMS}}
  :=\operatorname{Tr}\!\left(
    \rho_\beta^{1/2}X^\dagger
    \rho_\beta^{1/2}Y
  \right),
  \qquad
  \langle X,Y\rangle_{\mathrm{GNS}}
  :=\operatorname{Tr}(\rho_\beta X^\dagger Y).
\]
The notation $L^2(\rho_\beta)$ and the alternative notation
$\langle\cdot,\cdot\rangle_{\rho_\beta}$ refer to the KMS geometry unless
``GNS'' is written explicitly.  The Davies generators used below admit both
realizations, and their gaps are identified in
\cref{prop:davies-kms-gns-gap-bridge}.
Since
\[
  \langle I,X\rangle_{\mathrm{KMS}}
  =\langle I,X\rangle_{\mathrm{GNS}}
  =\operatorname{Tr}(\rho_\beta X),
\]
the mean-zero observable space is 
\[
  L^2_0(\rho_\beta)
  :=
  \{X\in\cB(\cH):\operatorname{Tr}(\rho_\beta X)=0\}
  =\{I\}^{\perp_{\rho_\beta}}.
\]
Indeed, every observable has the decomposition
\[
  X=\operatorname{Tr}(\rho_\beta X)I+X_0,
  \qquad
  X_0:=X-\operatorname{Tr}(\rho_\beta X)I
  \in L^2_0(\rho_\beta).
\]
Thus restricting to $L^2_0(\rho_\beta)$ removes the stationary identity
direction and retains the fluctuations about equilibrium.

Let $\mathcal L$ be a quantum Markov generator (Lindbladian)
acting on observables in the Heisenberg picture.
Assume that $\mathcal L$ is self-adjoint in the KMS inner product,
and write
\[
  K:=-\mathcal L\succeq0
\]
for its positive realization.  Its Dirichlet form is
\[
  \mathcal E_K(X,Y):=\langle X,KY\rangle_{\rho_\beta}.
\]
The Heisenberg dynamics preserves the identity observable, so
$\mathcal L I=0$ and hence $KI=0$.
KMS self-adjointness then gives
$\langle I,KX\rangle_{\rho_\beta}=\langle KI,X\rangle_{\rho_\beta}=0$.
Hence $L^2_0(\rho_\beta)$ is invariant under $K$, and
$KX=KX_0$ for the centered observable above.
When $\ker K=\C I$, its KMS spectral gap is defined through the Rayleigh quotient as
\[
  \gap_{\mathrm{KMS}}(K)
  :=
  \inf_{0\ne X\in L^2_0(\rho_\beta)}
  \frac{\langle X,KX\rangle_{\rho_\beta}}
       {\langle X,X\rangle_{\rho_\beta}}.
\]
In finite dimension, this is the smallest eigenvalue of the restriction of
$K$ to $L^2_0(\rho_\beta)$.  Equivalently,
$K\succeq\gap_{\mathrm{KMS}}(K)I$ on
$L^2_0(\rho_\beta)$.
If $K$ is also GNS self-adjoint, $\gap_{\mathrm{GNS}}(K)$ is defined by the
same Rayleigh quotient with
$\langle\cdot,\cdot\rangle_{\mathrm{GNS}}$ in the numerator and
denominator.
 The physical bath
used below will be shown to satisfy $\ker K=\C I$.

The KMS norm associated with the above inner product is
\[
  \lVert X\rVert_{\mathrm{KMS}}
  :=\langle X,X\rangle_{\rho_\beta}^{1/2}.
\]
Unlike this norm, the Hilbert--Schmidt norm defined in
\cref{sec:prelim-general-conventions} is independent of the Gibbs weight.

\subsection{Davies Generators}
\label{sec:prelim-davies}

We recall the finite-dimensional Davies construction in the convention used
below; see, for example, \cite{davies1974markovian,kastoryano2016quantum}.
Davies generators describe the weak-coupling limit of a quantum
system interacting with a thermal bath.
In the Heisenberg picture, the full generator has the form
\[
  \mathcal L_{\mathrm{full}}(X)
  =i[H_{\mathrm{eff}},X]+\mathcal L(X),
\]
where $H_{\mathrm{eff}}$ generates the coherent evolution, including any
bath-induced Hamiltonian correction, and $\mathcal L$ describes dissipation.
Here we study the dissipative part $\mathcal L$ and write $K=-\mathcal L$
for the corresponding positive Davies operator.

Let
\[
  H=\sum_EE\,P_E
\]
be the spectral resolution of $H$.  Let
$\mathscr S=\{S_\alpha\}_{\alpha\in\mathsf A}\subseteq\cB(\cH)$ be a finite
family of system--bath coupling operators.  At this abstract stage,
$\mathsf A$ is only an index set and $\alpha$ merely labels the operator
$S_\alpha$.  In
lattice models the couplings are typically chosen with local support, which
is specified as part of the model.
The individual $S_\alpha$ need not be self-adjoint.  For
$S_\alpha\in\mathscr S$, define its Bohr component at frequency $\omega$ by
the convention
\[
  S_{\alpha,\omega}
  :=
  \sum_{E-E'=\omega}P_ES_\alpha P_{E'},
  \qquad
  [H,S_{\alpha,\omega}]=\omega S_{\alpha,\omega}.
\]
Thus $\omega$ is the energy gained by the system under this component.
The observables with zero Bohr frequency form the space
\begin{equation}
  \label{eq:zero-bohr-frequency-space}
  V_0(H):=\{X:[X,H]=0\}
  =\Bigl\{X:X=\sum_E P_E X P_E\Bigr\}.
\end{equation}
Here ``zero'' refers to the energy difference $E-E'=0$.  Thus $V_0(H)$ consists of the observables that
preserve every energy eigenspace.  
Let the nonnegative rates satisfy the KMS relation
\[
  \gamma_\alpha(\omega)
  =e^{-\beta\omega}\gamma_\alpha(-\omega).
\]
As in the coupling families used below, we assume that the list is closed
under adjoints.  Whenever $S_{\bar\alpha}=S_\alpha^\dagger$, the two partners
carry the same rate function and weight,
\[
  \gamma_{\bar\alpha}=\gamma_\alpha,
  \qquad
  w_{\bar\alpha}=w_\alpha>0.
\]
With this convention, the corresponding Heisenberg-picture Davies generator is
\[
  \mathcal L_{\beta,H}(X)
  =
  \sum_{\alpha,\omega}w_\alpha\gamma_\alpha(\omega)
  \left(
    S_{\alpha,\omega}^\dagger XS_{\alpha,\omega}
    -\frac12
     \{S_{\alpha,\omega}^\dagger S_{\alpha,\omega},X\}
  \right).
\]
The Gibbs state $\rho_\beta$ is stationary, and the Davies generator is
self-adjoint in the KMS inner product.  Hence
$K_{\beta,H}:=-\mathcal L_{\beta,H}\succeq0$ on $L^2(\rho_\beta)$; see, for
example, \cite[Sec.~II.A]{kastoryano2013quantum}.

The following standard consequence of detailed balance connects the two
gap conventions used here; see \cite[Lemma~5 and Sec.~2.1]{ding2025efficient}.
We include a short proof for completeness.

\begin{proposition}[Equality of the KMS and GNS gaps for Davies generators]
\label{prop:davies-kms-gns-gap-bridge}
The Davies operator $K_{\beta,H}$ is self-adjoint in both the KMS and GNS
inner products. If
$\ker K_{\beta,H}=\C I$, then
\[
  \gap_{\mathrm{KMS}}(K_{\beta,H})
  =\gap_{\mathrm{GNS}}(K_{\beta,H}).
\]
\end{proposition}

\begin{proof}
Define the modular operator by
\[
  \Delta_{\rho_\beta}(X):=\rho_\beta X\rho_\beta^{-1}.
\]
The Bohr relation gives
$\Delta_{\rho_\beta}(S_{\alpha,\omega})
=e^{-\beta\omega}S_{\alpha,\omega}$.  The two modular factors cancel in
$S_{\alpha,\omega}^\dagger X S_{\alpha,\omega}$, while
$S_{\alpha,\omega}^\dagger S_{\alpha,\omega}$ has zero Bohr frequency.
Thus every Davies summand, and hence $K_{\beta,H}$, commutes with
$\Delta_{\rho_\beta}$.

Decompose $\cB(\cH)$ into the eigenspaces of the
positive operator $\Delta_{\rho_\beta}$.  These spaces are mutually
orthogonal in both geometries and are preserved by $K_{\beta,H}$.
For $X,Y$ in the same eigenspace with eigenvalue $\lambda>0$,
\[
  \langle X,Y\rangle_{\mathrm{KMS},\rho_\beta}
  =\sqrt\lambda\,
    \langle X,Y\rangle_{\mathrm{GNS},\rho_\beta}.
\]
Hence KMS self-adjointness implies GNS self-adjointness on each
eigenspace, and therefore on their orthogonal direct sum.
When $\ker K_{\beta,H}=\C I$, both gaps are the smallest eigenvalue of
the same operator restricted to the same mean-zero space
$L^2_0(\rho_\beta)$, so they coincide.
\end{proof}

For a single coupling $S$, write $K_{\beta,H}[S]$ for its
contribution to $-\mathcal L_{\beta,H}$.  Explicitly, after decomposing
$S=\sum_\omega S_\omega$ relative to $H$, we set
\[
  K_{\beta,H}[S](X)
  :=-\sum_\omega\gamma(\omega)
  \left(
    S_\omega^\dagger XS_\omega
    -\frac12\{S_\omega^\dagger S_\omega,X\}
  \right).
\]
Thus this notation fixes one coupling in the preceding Davies formula, sums
over all of its Bohr components, and reverses the generator sign.  When $S$
is not self-adjoint, this individual contribution need not itself be
KMS-self-adjoint or positive.  An adjoint-closed family with equal weights on
adjoint partners does give a KMS-self-adjoint positive operator.  This is the
convention used in the complete star edge update operator
\eqref{eq:complete-star-edge-generator}.

Suppose that the coupling index set is partitioned as
\[
  \mathsf A=\bigsqcup_{i\in\mathcal I}\mathsf A_i,
  \qquad
  \mathscr S_i:=\{S_\alpha:\alpha\in\mathsf A_i\},
\]
where every subfamily $\mathscr S_i$ is closed under adjoints, and the weights
remain equal on adjoint partners.  Thus $i$
labels one grouped update, and we denote its positive Davies operator by
\[
  K_i:=\sum_{\alpha\in\mathsf A_i}
    w_\alpha K_{\beta,H}[S_\alpha],
  \qquad w_\alpha>0,
\]
where the weights specify the normalization within the grouped update.  Once
each $K_i$ has the internal normalization specified by the model, the
\emph{unit-rate convention} means that these grouped updates are summed with
coefficient one:
\[
  K:=K_{\beta,H}=\sum_{i\in\mathcal I}K_i.
\]
Here ``unit rate'' refers to this outer sum over the already normalized
$K_i$; the couplings within a single $K_i$ may carry nontrivial weights.

\begin{lemma}[Positive local Davies operators and frustration freeness]
\label{lem:prelim-davies-generator-frustration-free}
Let $\{K_i\}_{i\in\mathcal I}$ be the positive local Davies operators associated with
adjoint-closed coupling families whose adjoint partners have equal weights and
which satisfy the rate assumptions above.
Then $K_i\succeq0$ and $K_iI=0$.  For every
$\mathcal J\subseteq\mathcal I$,
\begin{equation}
  \label{eq:prelim-davies-generator-kernel-intersection}
  K_{\mathcal J}:=\sum_{i\in\mathcal J}K_i\succeq0,
  \qquad
  K_{\mathcal J}I=0,
  \qquad
  \ker K_{\mathcal J}=\bigcap_{i\in\mathcal J}\ker K_i.
\end{equation}
\end{lemma}

\begin{proof}
Each coupling subfamily is closed under adjoints and inherits the rate
and weight assumptions of the Davies construction above.  The resulting
generator $\mathcal L_i$ therefore has the same KMS self-adjointness and
nonpositivity properties, so $K_i=-\mathcal L_i\succeq0$.

The identity is fixed by direct substitution.  Indeed, each Lindblad term
cancels separately, and hence
\[
  K_i(I)
  =-\sum_{\alpha\in\mathsf A_i,\omega}
    w_\alpha\gamma_\alpha(\omega)
    \left(
      S_{\alpha,\omega}^\dagger S_{\alpha,\omega}
      -\frac12
       \{S_{\alpha,\omega}^\dagger S_{\alpha,\omega},I\}
    \right)
  =0.
\]
Summing over $i\in\mathcal J$ gives
$K_{\mathcal J}\succeq0$ and $K_{\mathcal J}I=0$.

The kernel identity follows from \cref{lem:positive-sum-kernel} for
nonempty $\mathcal J$; for the empty family, both sides are the full
observable space, with the usual empty-intersection convention.
\end{proof}

Here \emph{frustration free} refers only to the kernel identity in
\eqref{eq:prelim-davies-generator-kernel-intersection}.

\begin{lemma}[Fixed points of an adjoint-closed Davies generator]
\label{lem:davies-fixed-point-commutant}
Let $K_S=-\mathcal L_S$ be a sum of positive local Davies operators for a faithful
(full-rank) Gibbs state.
Assume that every included Bohr rate is strictly positive and that the
included jump family is closed under adjoints.  Then
\begin{equation}
  \label{eq:davies-fixed-point-commutant}
  X\in\ker K_S
  \quad\Longleftrightarrow\quad
  [X,S_{\alpha,\omega}]=0
  \quad\text{for every included }(\alpha,\omega).
\end{equation}
\end{lemma}

\begin{proof}
For one Lindblad jump $V$, the completely positive
carr\'e-du-champ identity \cite{lindblad1976generators,junge2022stability}
reads
\begin{equation}
  \label{eq:lindblad-carre-du-champ}
  \mathcal L_V(X^\dagger X)
  -\mathcal L_V(X^\dagger)X
  -X^\dagger\mathcal L_V(X)
  =[V,X]^\dagger[V,X].
\end{equation}
Suppose $X\in\ker K_S$.  The Davies generator is $*$-preserving, so
$X^\dagger$ is fixed as well.  Sum the identity over the included Bohr
jumps with their positive rates and take expectation in the stationary Gibbs
state.  The expectation of the left side vanishes, whereas the right side is
a sum of nonnegative expectations.  Faithfulness of the Gibbs state forces
every commutator $[S_{\alpha,\omega},X]$ to vanish.  Conversely, if $X$
commutes with every included jump and, by adjoint closure, with its adjoint,
then each Lindblad summand annihilates $X$.  Hence $X\in\ker K_S$.
\end{proof}

\subsection{Quantum-Double Star Structure}
\label{sec:prelim-quantum-double}

We introduce the part of Kitaev's finite-group quantum-double construction
that is needed to locate our model within the standard framework
\cite{kitaev2003faulttolerant}.  Detailed treatments of Davies thermalization
for the standard two-dimensional models may be found in
\cite{alicki2009thermalization,lucia2023thermalization}.  An explicit
Peter--Weyl formulation that assigns the two matrix indices of an edge
representation block to its endpoints appears in
\cite{buerschaper2009mapping}.  For a recent representation-theoretic
treatment of protected sectors in the full finite-group model, see
\cite{hu2026protected}. 

Consider a two-dimensional lattice with an
orientation assigned to each edge and one copy of the complex group algebra
$\C[\Gamma]$, viewed as a Hilbert space, on each edge.  The full Hilbert
space is therefore
\[
  \cH_{\cG}=\bigotimes_{e\in E(\cG)}\C[\Gamma].
\]
An operator $S_e$ is \emph{supported on the edge $e$} if it acts as the
identity on every other edge: after placing the $e$ factor first, it has
the form
\[
  S_e=s_e\otimes\bigotimes_{f\in E(\cG)\setminus\{e\}}I_f,
  \qquad s_e\in\mathcal B(\C[\Gamma]),
\]
where $I_f$ is the identity on the edge $f$.
The orthonormal basis vectors $\lvert h\rangle$ of each edge space are
indexed by $h\in\Gamma$, and a general edge vector has the form
$\sum_{h\in\Gamma}c_h\lvert h\rangle$ with $c_h\in\C$.  The commuting left
and right regular actions may be written
\[
  L_g\lvert h\rangle=\lvert gh\rangle,
  \qquad
  R_g\lvert h\rangle=\lvert hg^{-1}\rangle.
\]
At a vertex $v$, define $U_v(g)$ by applying $L_g$ to each outgoing edge
and $R_g$ to each incoming edge.  This is the standard vertex gauge
transformation in the quantum-double construction; below we call it the
\emph{vertex action}.  More explicitly, if
$\operatorname{Out}(v)$ and $\operatorname{In}(v)$ denote the incident edges
oriented away from and toward $v$, respectively, then on the incident-edge
Hilbert space
\begin{equation}
  \label{eq:explicit-local-gauge-action}
  U_v(g)
  =
  \left(\bigotimes_{e\in\operatorname{Out}(v)}L_g^{(e)}\right)
  \otimes
  \left(\bigotimes_{e\in\operatorname{In}(v)}R_g^{(e)}\right),
\end{equation}
extended by the identity on all edges not incident to $v$.  The two index
sets partition the three incident edges of the cubic graph.  Changing an
edge orientation exchanges its $L_g$ and $R_g$ assignments at the two
endpoints.  Averaging this action gives
the star, or electric, projector
\[
  A_v
  =\frac1{|\Gamma|}\sum_{g\in\Gamma}U_v(g).
\]
Here the vertex constraint $A_v\psi=\psi$ requires the state to be fixed
by every $U_v(g)$.  A violation of this constraint is called an electric
charge excitation in the quantum-double construction.
The full quantum-double Hamiltonian also contains plaquette terms associated with faces, in contrast to the vertex terms \(A_v\). We consider only the vertex terms in this paper.

On an oriented graph $\cG$, the vertex-term Hamiltonian is
\[
  H_{\cG}=-J\sum_{v\in V(\cG)}A_v.
\]
The following basic properties will be used below.
\begin{lemma}[Commutativity and locality of the vertex terms]
\label{lem:vertex-commutativity-locality}
The operators $A_v$ are mutually commuting orthogonal projections onto
the subspaces fixed by the respective vertex actions.  In particular,
\begin{equation}
  \label{eq:vertex-projectors-commute}
  [A_u,A_v]=0 \qquad\text{for all }u,v.
\end{equation}
If $e:u\to v$ and $S_e$ is supported on $e$, then
\begin{equation}
  \label{eq:edge-coupling-commutator-locality}
  [H_{\cG},S_e]=-J[A_u+A_v,S_e].
\end{equation}
\end{lemma}
\begin{proof}
Since $U_v$ is a unitary representation, its group average satisfies
\begin{equation}
  \label{eq:vertex-group-average-projection}
  A_v^\dagger=A_v,
  \qquad
  A_v^2
  =\frac1{|\Gamma|^2}\sum_{g,h\in\Gamma}U_v(gh)
  =\frac1{|\Gamma|}\sum_{k\in\Gamma}U_v(k)
  =A_v.
\end{equation}
For the middle equality, each $k\in\Gamma$ has exactly $|\Gamma|$ factorizations
$k=gh$, one for each choice of $g$.  The range of this average is precisely
the subspace fixed by the vertex action at $v$.

On an edge directed from $u$ to $v$, transformations by
$g\in\Gamma$ at $u$ and by $h\in\Gamma$ at $v$ act as $x\mapsto gx$ and
$x\mapsto xh^{-1}$, respectively.  These actions commute because
\[
  g(xh^{-1})=(gx)h^{-1}.
\]
On edges that are not shared, the transformations act on separate tensor
factors.  Thus, actions at different vertices commute, even
when $\Gamma$ is non-Abelian.  Averaging gives
\eqref{eq:vertex-projectors-commute}.
Finally, $S_e$ commutes with every star projector except possibly $A_u$
and $A_v$, since the other projectors act on edges disjoint from $e$.
This gives \eqref{eq:edge-coupling-commutator-locality}.
\end{proof}

Together, these properties imply that the exact Bohr components of an edge coupling are determined by the
fixed-size system on that edge and its two endpoints, rather than by
diagonalizing the full graph Hamiltonian.

We now specify the Davies dynamics used in this paper by choosing the system–bath coupling operators. 
The operators specified at this stage
are the \emph{bare couplings}, meaning the couplings before their decomposition
into Bohr-frequency components.  Writing $q:=|\Gamma|$ and
$Q_h^{(e)}:=|h\rangle\!\langle h|^{(e)}$, we choose on each edge the
three indexed coupling families
\begin{equation}
  \label{eq:active-edge-bare-coupling-family}
  \mathscr S_e^{\mathsf L}:=\{L_g^{(e)}:g\in\Gamma\},\qquad
  \mathscr S_e^{\mathsf R}:=\{R_g^{(e)}:g\in\Gamma\},\qquad
  \mathscr S_e^{\mathsf Q}:=\{Q_h^{(e)}:h\in\Gamma\}.
\end{equation}
Set
\[
  \mathscr S_e:=\mathscr S_e^{\mathsf L}\sqcup
  \mathscr S_e^{\mathsf R}\sqcup\mathscr S_e^{\mathsf Q}.
\]
In the Davies sum, the three families are counted separately even when two
of their entries define the same operator.
For $S\in\mathscr S_e$, the Hamiltonian determines its Bohr
components, and hence its single-coupling Davies contribution
$K_{\beta,H_{\cG}}[S]$.  Their normalized sum defines the complete
star edge update.
Every
coupling uses the common logistic rate
\begin{equation}
  \label{eq:complete-star-logistic-rate}
  \gamma_\beta(\omega):=\frac1{1+e^{\beta\omega}},
\end{equation}
which obeys the KMS relation and is shared by adjoint partners.
Thus
\begin{align}
  \label{eq:complete-star-edge-generator}
  K_e
  :&=\frac1{3q}\sum_{S\in\mathscr S_e}
       K_{\beta,H_{\cG}}[S]\nonumber\\ 
   &=\frac1{3q}\sum_{g\in\Gamma}K_{\beta,H_{\cG}}[L_g^{(e)}]
    +\frac1{3q}\sum_{g\in\Gamma}K_{\beta,H_{\cG}}[R_g^{(e)}]
    +\frac1{3q}\sum_{h\in\Gamma}
       K_{\beta,H_{\cG}}[Q_h^{(e)}].
\end{align}
The factor $1/(3q)$ gives the $3q$ coupling labels within $K_e$ total edge
clock one.  The global unit-rate positive operator then sums these normalized edge
updates with coefficient one as  $$K=\sum_{e\in E}K_e.$$  

\subsection{Finite-Group Harmonic Analysis}
\label{sec:prelim-finite-groups}

We introduce the representation-theoretic decomposition used to define the label interface in \cref{sec:finite-group-specialization}. It also organizes the edge degrees of freedom by their endpoints, which simplifies the local analysis in the model proofs.

\subsubsection{Peter--Weyl Blocks and Endpoint Factors}

Let $\Gamma$ be a finite group.  Irreducible representations are the basic building blocks of
finite-dimensional unitary representations: they cannot be decomposed into
smaller invariant subspaces.  Two representations are considered equivalent
if they differ only by a change of basis.  We use $\widehat\Gamma$ to label
these equivalence classes of irreducible unitary representations.  Its
elements therefore specify representation types, rather than elements of
the group $\Gamma$.

For each $\lambda\in\widehat\Gamma$, choose a representative
$\rho_\lambda:\Gamma\to\mathcal U(V_\lambda)$ acting on the irreducible
representation space $V_\lambda$.  Denote its dual space by $V_\lambda^*$
and write $d_\lambda:=\dim V_\lambda$.

Every edge in the present model carries the same fixed
regular representation $\C[\Gamma]$.  Its orthonormal group-element basis is
$\{\lvert g\rangle \mid g\in\Gamma\}$.  The Peter--Weyl decomposition is a
unitary regrouping of this fixed edge space into its irreducible blocks
\cite{serre1977linear,etingof2011introduction}:
\[
  \C[\Gamma]
  \cong
  \bigoplus_{\lambda\in\widehat\Gamma}
  V_\lambda\otimes V_\lambda^*,
  \qquad
  |\Gamma|=\sum_{\lambda\in\widehat\Gamma}d_\lambda^2.
\]
Choose an orthonormal basis $\{e_{\lambda,i}\}_{i=1}^{d_\lambda}$ of each
$V_\lambda$ and its dual basis $\{e_\lambda^j\}_{j=1}^{d_\lambda}$.  Under
this identification, write
\[
  \lvert\lambda,i,j\rangle
  :=e_{\lambda,i}\otimes e_\lambda^j,
  \qquad 1\le i,j\le d_\lambda,
\]
for the corresponding Peter--Weyl basis of the $\lambda$ summand.
With the standard commuting left and right regular actions, the two tensor
factors carry the two representation actions.  In a compatible convention,
\[
  L_g
  \cong
  \bigoplus_\lambda\rho_\lambda(g)\otimes I,
  \qquad
  R_g
  \cong
  \bigoplus_\lambda I\otimes\rho_\lambda^*(g).
\]
Set $\cH_\lambda:=V_\lambda\otimes V_\lambda^*$.  If an oriented edge
$e:u\to v$ carries the label $\lambda_e$, we assign its two Peter--Weyl
factors to its endpoints by
\[
  V_{\lambda_e}^{(u)}:=V_{\lambda_e},
  \qquad
  V_{\lambda_e}^{(v)}:=V_{\lambda_e}^*,
  \qquad
  \cH_{\lambda_e}
  =V_{\lambda_e}^{(u)}\otimes V_{\lambda_e}^{(v)}.
\]
Thus the factor owned by $u$ carries the left action and the factor owned by
$v$ carries the right action.  Reversing the edge orientation exchanges the
two assignments.

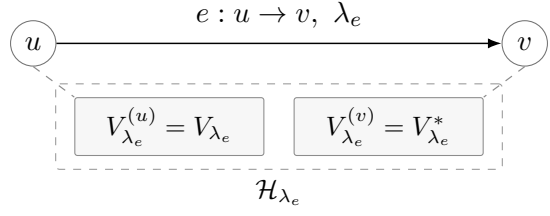
\begin{figure}[H]
  \centering
  \begin{tikzpicture}[
    endpoint/.style={circle, draw=black!55, minimum size=6mm, inner sep=0pt},
    factor/.style={draw=black!45, rounded corners=1pt, minimum width=2.5cm,
      minimum height=0.62cm, fill=black!3, font=\small},
    ownership/.style={dashed, draw=black!45, line width=0.45pt}
  ]
    \node[endpoint] (u) at (-3.25,0.25) {$u$};
    \node[endpoint] (v) at (3.25,0.25) {$v$};
    \draw[-{Latex[length=1.8mm]}, line width=0.55pt]
      (u) -- node[above=2pt] {$e:u\to v,\ \lambda_e$} (v);
    \node[factor] (fu) at (-1.45,-0.85)
      {$V_{\lambda_e}^{(u)}=V_{\lambda_e}$};
    \node[factor] (fv) at (1.45,-0.85)
      {$V_{\lambda_e}^{(v)}=V_{\lambda_e}^*$};
    \draw[ownership] (u.south) -- (fu.north west);
    \draw[ownership] (v.south) -- (fv.north east);
    \node[draw=black!40, dashed, rounded corners=1.5pt,
      fit=(fu)(fv), inner xsep=7pt, inner ysep=5pt,
      label={[font=\small]below:$\cH_{\lambda_e}$}] {};
  \end{tikzpicture}
  \caption{Endpoint assignment of the two factors in one Peter--Weyl edge
  block.  The dashed lines indicate which endpoint vertex action acts on each
  factor.  They do not introduce additional tensor factors.}
  \label{fig:peter-weyl-endpoint-assignment}
\end{figure}

After a global label configuration $\boldsymbol\lambda$ is fixed, these
endpoint assignments give the vertex-wise regrouping
\begin{equation}
  \label{eq:edge-to-vertex-regrouping}
  \cH_{\boldsymbol\lambda}
  :=\bigotimes_{e\in E(\cG)}\cH_{\lambda_e}
  \cong
  \bigotimes_{v\in V(\cG)}\mathcal V_v(\boldsymbol\lambda),
  \qquad
  \mathcal V_v(\boldsymbol\lambda)
  :=\bigotimes_{e\ni v}V_{\lambda_e}^{(v)}.
\end{equation}
The isomorphism in \eqref{eq:edge-to-vertex-regrouping} only permutes and
rebrackets the existing tensor factors: no degree of freedom is removed or
added.

The orthogonal projection onto the $V_\lambda\otimes V_\lambda^*$ summand is
\[
  P_\lambda
  =\sum_{i,j=1}^{d_\lambda}
    \lvert\lambda,i,j\rangle\!\langle\lambda,i,j\rvert.
\]
This projection commutes with both the left and right regular actions:
\[
  [P_\lambda,L_g]=[P_\lambda,R_g]=0
  \qquad\text{for all }g\in\Gamma.
\]
Observables generated by these projections will form the irrep-label
subspace in the model below.
Although a concrete representative and basis may be chosen for each class,
the isotypic projector $P_\lambda$ depends only on the equivalence class
$\lambda$.

\begin{examplebox}{one \texorpdfstring{$S_3$}{S3} edge: from group elements to irrep blocks}
We use one-line notation for permutations; for example,
\[
  \begin{pmatrix}
    1&2&3\\
    2&3&1
  \end{pmatrix}
  \quad\text{is written as }231.
\]
Thus, the entries list the images of $1,2,3$, in that order.
The group-element basis of one
physical edge is
\[
  \C[S_3]
  =\operatorname{span}\{
    \lvert123\rangle,\lvert132\rangle,\lvert213\rangle,
    \lvert231\rangle,\lvert312\rangle,\lvert321\rangle
  \}.
\]
These six kets are indexed by group elements, not by irreducible
representations.  The three irrep labels are
$\widehat{S_3}=\{\mathbf1,\mathrm{sgn},\tau\}$.
Here $\mathbf1$ is the trivial representation, $\mathrm{sgn}$ is the sign
representation, and $\tau$ is the two-dimensional standard irreducible
representation; see \cite[Sec.~4.3, item~2]{etingof2011introduction}
for explicit descriptions.  The Peter--Weyl decomposition regroups the same
six-dimensional space as
\[
  \C[S_3]
  \cong \cH_{\mathbf1}\oplus\cH_{\mathrm{sgn}}\oplus\cH_\tau,
  \qquad
  \dim\cH_{\mathbf1}=1,
  \quad \dim\cH_{\mathrm{sgn}}=1,
  \quad \dim\cH_\tau=4.
\]
Normalized vectors spanning the trivial and sign blocks are
\[
  \lvert u_+\rangle
  :=\frac1{\sqrt6}\sum_{g\in S_3}\lvert g\rangle,
  \qquad
  \lvert u_-\rangle
  :=\frac1{\sqrt6}\sum_{g\in S_3}\mathrm{sgn}(g)\lvert g\rangle.
\]
Under the left group action $L_h\lvert g\rangle=\lvert hg\rangle$, they satisfy
\[
  L_h\lvert u_+\rangle=\lvert u_+\rangle,
  \qquad
  L_h\lvert u_-\rangle=\mathrm{sgn}(h)\lvert u_-\rangle.
\]
Thus $u_+$ is unchanged by every group element, while $u_-$ is unchanged
by even permutations and changes sign under odd permutations.
These requirements determine each normalized vector up to an overall phase;
they are not arbitrary orthogonal vectors.
In Peter--Weyl notation we may take
$\lvert\mathbf1,1,1\rangle:=\lvert u_+\rangle$ and
$\lvert\mathrm{sgn},1,1\rangle:=\lvert u_-\rangle$, and choose an
orthonormal basis of the remaining block so that
\[
  \begin{aligned}
  \cH_{\mathbf1}
    &=\operatorname{span}\{\lvert\mathbf1,1,1\rangle\},
  &
  \cH_{\mathrm{sgn}}
    &=\operatorname{span}\{\lvert\mathrm{sgn},1,1\rangle\},
  \\
  \cH_\tau
    &=\operatorname{span}\{
       \lvert\tau,1,1\rangle,\lvert\tau,1,2\rangle,
       \lvert\tau,2,1\rangle,\lvert\tau,2,2\rangle\}.
  \end{aligned}
\]
Here the indices $i,j\in\{1,2\}$ in
$\lvert\tau,i,j\rangle$ label basis vectors in the two factors
$V_\tau$ and $V_\tau^*$, respectively.  Thus the same edge Hilbert space has
the Peter--Weyl basis
\[
  \{\lvert\mathbf1,1,1\rangle,
    \lvert\mathrm{sgn},1,1\rangle,
    \lvert\tau,1,1\rangle,\lvert\tau,1,2\rangle,
    \lvert\tau,2,1\rangle,\lvert\tau,2,2\rangle\}.
\]
The three label projectors are then
\[
  P_{\mathbf1}
    =\lvert\mathbf1,1,1\rangle\!\langle\mathbf1,1,1\rvert
    =\lvert u_+\rangle\!\langle u_+\rvert,
  \qquad
  P_{\mathrm{sgn}}
    =\lvert\mathrm{sgn},1,1\rangle\!\langle\mathrm{sgn},1,1\rvert
    =\lvert u_-\rangle\!\langle u_-\rvert,
\]
and
\[
  P_\tau
  =\sum_{i,j=1}^{2}
    \lvert\tau,i,j\rangle\!\langle\tau,i,j\rvert
  =I-P_{\mathbf1}-P_{\mathrm{sgn}}.
\]
Thus $P_\tau$ has rank four: observing the label $\tau$ selects an entire
four-dimensional block, not a single vector.  At $\beta=0$, the maximally
mixed edge state consequently gives label probabilities
$1/6$, $1/6$, and $4/6$, respectively.
\end{examplebox}

\subsubsection{Observable Blocks and the Classical Label Interface}

We will also use the Peter--Weyl decomposition on the \emph{observable}
space.
Each Peter--Weyl summand of $\C[\Gamma]$ is indexed by one irrep label
$\lambda$; an arbitrary vector may have components in several summands.
An operator $X\in\cB(\C[\Gamma])$ has two such labels through its blocks
\begin{equation}
  \label{eq:peter-weyl-ket-bra-block}
  X_{\lambda^{\mathrm{ket}},\lambda^{\mathrm{bra}}}
  :=P_{\lambda^{\mathrm{ket}}}XP_{\lambda^{\mathrm{bra}}}.
\end{equation}
The superscript ``ket'' records the target sector of the operator, and
``bra'' records its source sector, as in a matrix unit
$|\mathrm{ket}\rangle\langle\mathrm{bra}|$.  These are labels on the two
sides of an observable, not two physical copies of the edge.  On a graph the
same decomposition is made edge by edge, giving ket and bra label configurations
$\boldsymbol{\lambda}^{\mathrm{ket}}
=(\lambda_e^{\mathrm{ket}})_{e\in E(\cG)}$ and
$\boldsymbol{\lambda}^{\mathrm{bra}}
=(\lambda_e^{\mathrm{bra}})_{e\in E(\cG)}$.  For either label configuration
$\boldsymbol\lambda=(\lambda_e)_{e\in E(\cG)}$, set
\begin{equation}
  \label{eq:global-peter-weyl-projector}
  \begin{aligned}
  \lvert\boldsymbol\lambda,\mathbf i,\mathbf j\rangle
  &:=\bigotimes_{e\in E(\cG)}
       \lvert\lambda_e,i_e,j_e\rangle,\\
  P_{\boldsymbol\lambda}
  &:=\bigotimes_{e\in E(\cG)}P_{\lambda_e}
    =\sum_{\mathbf i,\mathbf j}
       \lvert\boldsymbol\lambda,\mathbf i,\mathbf j\rangle
       \!\langle\boldsymbol\lambda,\mathbf i,\mathbf j\rvert.
  \end{aligned}
\end{equation}
Here $\mathbf i=(i_e)_{e\in E(\cG)}$ and
$\mathbf j=(j_e)_{e\in E(\cG)}$, with
$1\le i_e,j_e\le d_{\lambda_e}$ for each edge $e$.
Thus $P_{\boldsymbol\lambda}$ selects the fixed-label Hilbert block:
\[
  P_{\boldsymbol\lambda}\mathcal H_{\cG}
  \cong \cH_{\boldsymbol\lambda}.
\]
The corresponding projectors on the two sides of an observable are
$P_{\boldsymbol{\lambda}^{\mathrm{ket}}}$ and
$P_{\boldsymbol{\lambda}^{\mathrm{bra}}}$.

\paragraph{Classical observables on irrep labels.}
A function $f$ of the irrep-label configuration defines an observable
\[
  X_f:=\sum_{\boldsymbol\lambda}
    f(\boldsymbol\lambda)P_{\boldsymbol\lambda}.
\]
This observable acts as the scalar $f(\boldsymbol\lambda)$ within each
fixed-label block, so it depends only on the labels and not on the
internal state of that block.  Such observables commute with one another
and can be identified with classical functions of the label configuration.
They form the label interface used in
\cref{sec:finite-group-specialization}, providing the setting for the
classical comparison developed later.

The full observable space contains more information.
Even within a fixed-label block, an observable may act nontrivially
on the internal matrix degrees of freedom.
This distinction is already visible on one edge.  If $d_\lambda>1$, then
$\dim\cH_\lambda=d_\lambda^2$, so the fixed-label block
$\cB(\cH_\lambda)$ contains many non-scalar operators, such as traceless
diagonal matrices and coherences between internal matrix indices.  Measuring
the irrep label reveals only that the state lies in $\cH_\lambda$ and cannot
distinguish these internal states.

\begin{examplebox}{the label projector and internal directions for
\texorpdfstring{$S_3$}{S3}}
On one $S_3$ edge, choose a function of the irrep label by
\[
  f(\mathbf1)=0,\qquad f(\mathrm{sgn})=1,\qquad f(\tau)=2.
\]
The corresponding observable is
\[
  X_f=P_{\mathrm{sgn}}+2P_\tau
  \cong 0\oplus1\oplus2I_4
  =\operatorname{diag}(0,1,2,2,2,2),
\]
where the matrix is written in a Peter--Weyl basis ordered by
$\mathbf1,\mathrm{sgn},\tau$.  Thus $X_f$ distinguishes the three labels,
but assigns the same value $2$ to every state in the four-dimensional
$\tau$ block.  It cannot distinguish states within that block.

To see the additional observables available within this block, recall that
\[
  P_\tau
  =\sum_{i,j=1}^{2}
    \lvert\tau,i,j\rangle\!\langle\tau,i,j\rvert,
  \qquad
  \cH_\tau=\ran P_\tau\cong\C^4.
\]
The label sees only the scalar projector direction inside the full
fixed-label block:
\[
  \underbrace{\C P_\tau}_{\text{label-visible}}
  \ \subsetneq\
  \underbrace{P_\tau\cB(\C[S_3])P_\tau
    =\cB(\cH_\tau)\cong\cB(\C^4)}_{\text{fixed label $\tau$, with internal directions}}.
\]
For example, two observables outside $\C P_\tau$ are
\[
  \begin{aligned}
  X_{\mathrm{diag}}&:=
  \lvert\tau,1,1\rangle\!\langle\tau,1,1\rvert
  -\lvert\tau,2,1\rangle\!\langle\tau,2,1\rvert,\\
  X_{\mathrm{coh}}&:=
  \lvert\tau,1,1\rangle\!\langle\tau,2,1\rvert
  +\lvert\tau,2,1\rangle\!\langle\tau,1,1\rvert.
  \end{aligned}
\]
For $X=P_\tau XP_\tau$, removing the label component gives
\[
  X_\perp
  :=X-\mathcal E_{\mathrm{lab}}X
  =X-
  \frac{\operatorname{Tr}(\rho_\beta P_\tau X)}
       {\operatorname{Tr}(\rho_\beta P_\tau)}P_\tau,
  \qquad
  \mathcal E_{\mathrm{lab}}X_\perp=0.
\]
Thus the label projection retains the $P_\tau$ direction and misses the
centered internal directions.  The classical label gap controls the former;
the estimate on the orthogonal complement of the label interface controls the latter.
\end{examplebox}

\subsubsection{Invariant Spaces and Intertwiners}
\label{sec:prelim-invariant-spaces-intertwiners}

This section introduces further tools used in the three model-verification
subsections.  In \cref{sec:common-population-certificate}, dimensions of
invariant spaces determine the classical vertex weights.  In
\cref{sec:common-fibre-certificate}, intertwiners describe the exact path
fixed spaces and support the fixed-space cover argument
(\cref{sec:exact-centered-path-fixed-space,lem:cubic-defect-frame}).
In \cref{sec:common-schur-certificate}, these tools enter the proof of
middle-edge decoupling (\cref{prop:middle-edge-shielding}).
Readers interested primarily in the main results may skip this section
on a first reading.

\paragraph{Invariant vectors and vertex constraints.}
For a unitary representation $W:\Gamma\to\mathcal U(\mathcal H_W)$, write
\begin{equation}
  \label{eq:invariant-vector-space}
  \operatorname{Inv}(W)
  :=\{\xi\in\mathcal H_W:W(g)\xi=\xi
       \text{ for every }g\in\Gamma\}.
\end{equation}
For representations $W_1,\ldots,W_m$, the notation
$\operatorname{Inv}(W_1\otimes\cdots\otimes W_m)$ refers to the fixed space
of the diagonal action
$g\mapsto W_1(g)\otimes\cdots\otimes W_m(g)$.

At a vertex, this invariant space describes the states that satisfy
the vertex constraint for the chosen incident-edge labels:
\[
  A_v\xi=\xi
  \quad\Longleftrightarrow\quad
  U_v(g)\xi=\xi
  \quad\text{for every }g\in\Gamma.
\]
Here the operators act on the fixed-label vertex space
$\mathcal V_v(\boldsymbol\lambda)$.

\begin{lemma}[Group averaging onto invariant vectors]
\label{lem:group-averaging-invariants}
For any positive integer $m$, let
$W_i:\Gamma\to\mathcal U(\mathcal V_i)$, $i=1,\ldots,m$, be unitary
representations.  The orthogonal projection onto
$\operatorname{Inv}(W_1\otimes\cdots\otimes W_m)$ is
\[
  \Pi_{\mathrm{inv}}
  =
  \frac1{|\Gamma|}\sum_{g\in\Gamma}
  W_1(g)\otimes\cdots\otimes W_m(g).
\]
In particular,
\[
  \operatorname{rank}\Pi_{\mathrm{inv}}
  =
  \dim\operatorname{Inv}
  (W_1\otimes\cdots\otimes W_m).
\]
\end{lemma}
\begin{proof}
The tensor product action is a unitary representation.  As in the proof of
\cref{lem:vertex-commutativity-locality}, its group average is self-adjoint
and idempotent.  Its image is fixed by every group element, and it acts as
the identity on every fixed vector.  This identifies its range and proves
the rank formula.  See also
\cite[proof of Theorem~4.5.1]{etingof2011introduction}.
\end{proof}
Here $m$ counts the representation factors on which the same group element
acts.  In the principal application below, these are the three legs
incident to a degree-three vertex, so $m=3$; the formula also applies at any fixed
degree.

\paragraph{Fusion channels.}
For fixed incident-edge labels, the dimension of the invariant space at a
vertex counts the independent ways these labels can be coupled
to form an invariant state.  These possibilities are called
\emph{fusion channels}, and their number is the \emph{fusion multiplicity}.
Equivalently, this dimension is the number of copies of the trivial
representation in the tensor product.

At a degree-three vertex, dimension zero means that the three labels cannot
satisfy the vertex constraint, dimension one means that their invariant
coupling is unique up to scale, and a larger dimension means that several
linearly independent couplings are available.

The same interpretation applies to arbitrary edge orientations: each
$\epsilon_i$ records the corresponding choice between $V_{\lambda_i}$ and
its dual.  Thus
\[
  \dim\operatorname{Inv}
  (V_{\lambda_1}^{\epsilon_1}\otimes
   V_{\lambda_2}^{\epsilon_2}\otimes
   V_{\lambda_3}^{\epsilon_3})>1,
\]
means that the same three edge labels admit several linearly independent
invariant tensors.  An operator can mix these possibilities without changing
any irrep label, so it need not be a label-only observable.  Together with the
internal matrix directions in the preceding example, such operators belong
to the complement of the irrep-label subspace below.

\begin{examplebox}{\texorpdfstring{$S_3$}{S3} fusion rules}
\phantomsection\label{ex:s3-fusion-vertex-weights}
The three irreducible representations are $\mathbf 1$, $\mathrm{sgn}$,
and $\tau$, of dimensions $1,1,2$, respectively.  All are self-dual.
Tensoring with $\mathbf 1$ leaves a representation unchanged, and the
remaining fusion rules are
\[
  \mathrm{sgn}\otimes\mathrm{sgn}\cong\mathbf 1,
  \qquad
  \mathrm{sgn}\otimes\tau\cong\tau,
  \qquad
  \tau\otimes\tau
  \cong\mathbf 1\oplus\mathrm{sgn}\oplus\tau.
\]
See \cite[Example~4.9.1]{etingof2011introduction}.
Each summand occurs at most once.  Since the irreps are self-dual,
reversing an edge replaces its representation by an equivalent dual, so
the following dimension count is independent of the edge orientations.
The three-edge invariant space has dimension one precisely when
$V_{\lambda_3}$ occurs in $V_{\lambda_1}\otimes V_{\lambda_2}$, and
dimension zero otherwise.  For instance,
\[
  \dim\operatorname{Inv}(V_\tau\otimes V_\tau\otimes V_\tau)=1,
  \qquad
  \dim\operatorname{Inv}(V_{\mathbf 1}\otimes V_{\mathbf 1}\otimes V_\tau)=0.
\]
Thus a three-edge vertex has either one invariant channel or none and 
there are no multiple invariant channels for a fixed triple of labels.

These fusion rules are used in the finite calculations in
\cref{app:s3-finite-certificates}.
\end{examplebox}

Invariant spaces of dimension greater than one can occur for larger groups
and are allowed in the general finite-group theorem.

\paragraph{Intertwiners and multiplicity spaces.}
For finite-dimensional Hilbert spaces $\mathcal H_U$ and $\mathcal H_V$,
$\operatorname{Hom}(\mathcal H_U,\mathcal H_V)$ denotes the vector space of
all linear maps from $\mathcal H_U$ to $\mathcal H_V$.  In this notation, the
Peter--Weyl block in \eqref{eq:peter-weyl-ket-bra-block} satisfies
\[
  P_{\lambda^{\mathrm{ket}}}X P_{\lambda^{\mathrm{bra}}}
  \in
  \operatorname{Hom}
  (\cH_{\lambda^{\mathrm{bra}}},
   \cH_{\lambda^{\mathrm{ket}}}).
\]
For representations $U$ on $\mathcal H_U$ and $V$ on $\mathcal H_V$, write
\begin{equation}
  \label{eq:intertwiner-space}
  \operatorname{Hom}_\Gamma(U,V)
  :=
  \left\{
    F:\mathcal H_U\to\mathcal H_V:
    V(g)F=FU(g)\ \text{for every }g\in\Gamma
  \right\}.
\end{equation}
This definition also fixes the order of the intertwiner spaces used
later.  A local condition
$U^{\mathrm{ket}}(g)X=XU^{\mathrm{bra}}(g)$ says precisely that $X$ maps the
bra representation equivariantly into the ket representation, so its
solution space is
$\operatorname{Hom}_\Gamma(U^{\mathrm{bra}},U^{\mathrm{ket}})$.
Its elements are \emph{intertwiners}: linear maps compatible with the two
group actions.  Thus the first argument is the source representation and the
second is the target representation.

At a degree-three vertex with one incoming edge labelled $\lambda_1$ and two
outgoing edges labelled $\lambda_2,\lambda_3$, the vertex space is
$V_{\lambda_1}^*\otimes V_{\lambda_2}\otimes V_{\lambda_3}$.
A tensor in this space can also be viewed as a linear map from $V_{\lambda_1}$
to $V_{\lambda_2}\otimes V_{\lambda_3}$, by treating the dual factor as an input:
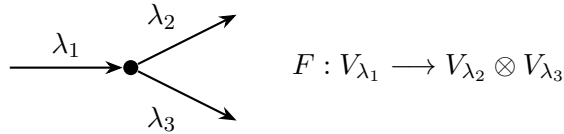
\begin{center}
  \begin{tikzpicture}[baseline=(v.base), >=Stealth]
    \node[circle, fill=black, inner sep=2pt] (v) at (0,0) {};
    \draw[->, thick] (-1.6,0) -- node[above] {$\lambda_1$} (v);
    \draw[->, thick] (v) -- node[above left] {$\lambda_2$} (1.4,0.7);
    \draw[->, thick] (v) -- node[below left] {$\lambda_3$} (1.4,-0.7);
    \node[anchor=west] at (2.0,0)
      {$F:V_{\lambda_1}\longrightarrow V_{\lambda_2}\otimes V_{\lambda_3}$};
  \end{tikzpicture}
\end{center}
Under this identification, invariance of the tensor is exactly
compatibility of the map with the group actions.  Thus,
\[
  \operatorname{Inv}
  (V_{\lambda_1}^*\otimes V_{\lambda_2}\otimes V_{\lambda_3})
  \cong
  \operatorname{Hom}_\Gamma
  (V_{\lambda_1},V_{\lambda_2}\otimes V_{\lambda_3})
\]
as in
\cite[proof of Theorem~4.5.1]{etingof2011introduction}.
Its dimension is the multiplicity of $V_{\lambda_1}$ in
$V_{\lambda_2}\otimes V_{\lambda_3}$.

For irreducibles, Schur's lemma
\cite[Proposition~2.3.9 and Corollary~2.3.10]{etingof2011introduction} gives
\begin{equation}
  \label{eq:schur-intertwiner-dimension}
  \dim\operatorname{Hom}_\Gamma(V_\lambda,V_\mu)
  =\mathbf1_{\lambda=\mu}.
\end{equation}
For a general representation, the following decomposition separates each
irreducible type from the space recording its multiplicity.

\begin{lemma}[Finite-group isotypic decomposition]
\label{lem:finite-group-isotypic-decomposition}
Every finite-dimensional unitary representation $U$ of $\Gamma$ on
$\mathcal W$ admits an orthogonal equivariant decomposition
\begin{equation}
  \label{eq:finite-group-isotypic-decomposition}
  \mathcal W
  \cong
  \bigoplus_{\alpha\in\widehat\Gamma}V_\alpha\otimes M_\alpha,
  \qquad
  M_\alpha\cong\operatorname{Hom}_\Gamma(V_\alpha,\mathcal W),
\end{equation}
where zero multiplicity spaces may be omitted.  Under this identification,
\[
  U(g)\cong\bigoplus_{\alpha\in\widehat\Gamma}
    \rho_\alpha(g)\otimes I_{M_\alpha}.
\]
Thus $\dim M_\alpha$ is the number of copies of the irreducible type
$V_\alpha$ in $\mathcal W$.
\end{lemma}

\begin{proof}
If a subspace of $\mathcal W$ is invariant, then so is its orthogonal
complement, because $U(g)$ is unitary.  Induction on $\dim\mathcal W$
therefore decomposes $\mathcal W$ into irreducible summands.  Grouping
equivalent summands gives \eqref{eq:finite-group-isotypic-decomposition}, and
the evaluation map identifies their multiplicity space with
$\operatorname{Hom}_\Gamma(V_\alpha,\mathcal W)$.  This is Maschke's theorem
in the present unitary setting; see
\cite[Remark~3.1.3 and Theorem~4.6.3]{etingof2011introduction}
and \cite{serre1977linear}.
\end{proof}

The group need not be abelian,
and tensor products of irreducible representations need not be multiplicity-free.

\subsection{Classical Heat-Bath Chains and Dobrushin Comparison}
\label{sec:prelim-dobrushin}

This subsection defines the classical heat-bath dynamics used in the
comparison in \cref{sec:common-population-certificate} and introduces
the Dobrushin criterion used there to bound its spectral gap.

Let $\mu$ be a strictly positive probability measure on a finite product
space $\Xi^E$.  For $e\in E$, let $\mathbb E_e$ denote conditional
expectation onto the coordinates in $E\setminus\{e\}$: it keeps those
coordinates fixed and averages only over the conditional distribution at $e$.
Explicitly, for $x\in\Xi^E$,
\[
  (\mathbb E_e f)(x)
  =\sum_{a\in\Xi}
    \mu\!\left(
      \xi_e=a\,\middle|\,
      \xi_{E\setminus\{e\}}=x_{E\setminus\{e\}}
    \right)
    f(x^{e\leftarrow a}),
\]
where $x^{e\leftarrow a}$ is obtained from $x$ by replacing its $e$th
coordinate by $a$.  Thus $\mathbb E_e$ is the orthogonal projection in
$L^2(\mu)$ onto the functions independent of the $e$th coordinate.  The unit-rate
continuous-time heat-bath operator is
\[
  \mathsf K_{\mathrm{hb}}
  :=\sum_{e\in E}(I-\mathbb E_e)
\]
on $L^2(\mu)$.  We also call this a positive heat-bath generator;
the actual Markov generator is $-\mathsf K_{\mathrm{hb}}$, with
semigroup $e^{-t\mathsf K_{\mathrm{hb}}}$.  All classical heat-bath
operators below use this positive-sign convention.

Analogously to the mean-zero observable space in the quantum setting, define
\[
  L^2_0(\mu):=\{f\in L^2(\mu):\mathbb E_\mu f=0\}.
\]
The operator $\mathsf K_{\mathrm{hb}}$ is self-adjoint in $L^2(\mu)$ and
annihilates constant functions, so it preserves $L^2_0(\mu)$.
Its spectral gap is defined by the Rayleigh quotient
\begin{equation}
  \label{eq:classical-heat-bath-gap}
  \gap(\mathsf K_{\mathrm{hb}})
  :=\inf_{0\ne f\in L^2_0(\mu)}
  \frac{\langle f,\mathsf K_{\mathrm{hb}}f\rangle_{L^2(\mu)}}
       {\langle f,f\rangle_{L^2(\mu)}}.
\end{equation}
Thus it is the smallest eigenvalue on the mean-zero subspace, just as in
the quantum setting.  Strict positivity of $\mu$ on the product space
ensures that the kernel consists only of constant functions.

For later use, consider also a single heat-bath update on a finite state
space $\mathsf I$ with probability vector $\pi$.  Writing
$\mathbb E_\pi f:=\sum_{i\in\mathsf I}\pi_i f_i$, the update replaces the
current state by an independent sample from $\pi$.  Its positive generator
and quadratic form are
\begin{equation}
  \label{eq:abstract-block-heat-bath-form}
  \mathsf K_{\mathrm{hb},\pi}:=I-\mathbb E_\pi,
  \qquad
  \langle f,\mathsf K_{\mathrm{hb},\pi}f\rangle_\pi
  =\operatorname{Var}_\pi(f)
  =\frac12\sum_{i,j\in\mathsf I}
    \pi_i\pi_j|f_i-f_j|^2.
\end{equation}
After the exterior coordinates of a product-space chain are fixed, the
$e$th summand $I-\mathbb E_e$ has exactly this form with $\pi$ equal to the
conditional distribution at $e$.

For probability vectors $p,q$ on a finite set, their total-variation distance
is
\begin{equation}
  \label{eq:total-variation-distance}
  \lVert p-q\rVert_{\mathrm{TV}}
  :=\frac12\sum_i|p_i-q_i|.
\end{equation}
Write $\mu_e(\,\cdot\mid x_{E\setminus\{e\}})$ for the conditional
distribution at $e$.  For distinct $e,f\in E$, define the corresponding
total-variation influence by
\begin{align*}
  c_{ef}
  &:=
  \sup_{x\sim_fy}
  \left\|
    \mu_e(\,\cdot\mid x_{E\setminus\{e\}})
    -
    \mu_e(\,\cdot\mid y_{E\setminus\{e\}})
  \right\|_{\mathrm{TV}}
  \\
  &=\frac12\sup_{x\sim_fy}\sum_{a\in\Xi}
  \left|
    \mu_e(a\mid x_{E\setminus\{e\}})
    -\mu_e(a\mid y_{E\setminus\{e\}})
  \right|,
\end{align*}
where $x\sim_fy$ means that the two boundary configurations differ only at
coordinate $f$.  Set
\[
  \alpha:=\max_{e\in E}\sum_{f\ne e}c_{ef}.
\]
Under the Dobrushin condition $\alpha<1$, the classical heat-bath operator
satisfies \cite{dobrushin1968description,wu2006poincare}
\[
  \gap(\mathsf K_{\mathrm{hb}})\ge1-\alpha.
\]
We will apply the unit-rate statement after deriving the
exact conditional distributions of the finite-group irrep-label model.  The stronger
block certificate for $S_3$, which uses its fusion rules, is model-specific
and is therefore developed separately rather than included in these
preliminaries.

\section{Gap Conditions for the Finite-Group Model}
\label{sec:finite-group-specialization}

We now explain how to apply the abstract gap theorem to the finite-group model. We specify the ingredients of double localization and state the conditions to be verified in the following sections.

\subsection{The Finite-Group Split and Centered-Path Cover}
\label{sec:formal-model-split-path-cover}

Fix a graph $\cG$ in the theorem family, write $E=E(\cG)$, and let
$K=\sum_{e\in E}K_e$ be its positive Davies operator on the mean-zero
observable space.  A label configuration
$\boldsymbol\lambda=(\lambda_e)_{e\in E}\in\widehat\Gamma^E$ determines the
joint Peter--Weyl projector $P_{\boldsymbol\lambda}$ in
\eqref{eq:global-peter-weyl-projector}.  These projectors span the commutative
label-observable space
\[
  \mathcal M_{\mathrm{lab}}
  :=\operatorname{span}
    \{P_{\boldsymbol\lambda}:\boldsymbol\lambda\in\widehat\Gamma^E\}.
\]
Thus a label configuration chooses one irreducible \emph{block} on each
edge, and $P_{\boldsymbol\lambda}$ projects onto their tensor product.  The family
$\{P_{\boldsymbol\lambda}\}_{\boldsymbol\lambda\in\widehat\Gamma^E}$ is the
joint projective measurement of the edge irrep labels.
The Gibbs state induces the classical distribution
\[
  \mu_{\mathrm{lab}}(\boldsymbol\lambda)
  :=\operatorname{Tr}(\rho_\beta P_{\boldsymbol\lambda}).
\]
Accordingly, $\mu_{\mathrm{lab}}(\boldsymbol\lambda)$ is precisely the
probability of obtaining the joint label outcome $\boldsymbol\lambda$ when
this measurement is performed in the Gibbs state.
We identify classical functions with label observables by
\[
  L^2(\mu_{\mathrm{lab}})\ni f
  \longmapsto\sum_{\boldsymbol\lambda}
    f(\boldsymbol\lambda)P_{\boldsymbol\lambda}
  \in\mathcal M_{\mathrm{lab}}.
\]
This identification preserves the norm:
\[
  \left\lVert\sum_{\boldsymbol\lambda}
    f(\boldsymbol\lambda)P_{\boldsymbol\lambda}\right\rVert_{\mathrm{KMS}}^2
  =\sum_{\boldsymbol\lambda}
    \mu_{\mathrm{lab}}(\boldsymbol\lambda)
    \lvert f(\boldsymbol\lambda)\rvert^2
  =\lVert f\rVert_{L^2(\mu_{\mathrm{lab}})}^2.
\]
It therefore gives the natural space in which
the compressed label dynamics is compared with a classical heat-bath chain.
On the mean-zero spaces this becomes
\begin{equation}
  \label{eq:mean-zero-label-identification}
  \mathcal M_{\mathrm{lab}}\cap L^2_0(\rho_\beta)
  \cong L^2_0(\mu_{\mathrm{lab}}).
\end{equation}
Since every $P_{\boldsymbol\lambda}$ commutes with the vertex Hamiltonian
and hence with $\rho_\beta$, the KMS-orthogonal conditional expectation
onto the label algebra is
\begin{equation}
  \label{eq:label-coarse-graining-formula}
  \mathcal E_{\mathrm{lab}}X
  :=\sum_{\boldsymbol{\lambda}\in\widehat\Gamma^E}
    \frac{
      \operatorname{Tr}
      (\rho_\beta P_{\boldsymbol{\lambda}}X)
    }{
      \operatorname{Tr}
      (\rho_\beta P_{\boldsymbol{\lambda}})
    }
    P_{\boldsymbol{\lambda}},
  \qquad X\in\cal B(\mathcal H_{\cG}).
\end{equation}
Under the identification of label observables with functions on label
configurations, we write
$(\mathcal E_{\mathrm{lab}}X)(\boldsymbol\lambda)$ for the scalar coefficient
of $P_{\boldsymbol\lambda}$ in this expansion.

The conditional expectation preserves the Gibbs expectation.
Explicitly,
$\operatorname{Tr}(\rho_\beta\mathcal E_{\mathrm{lab}}X)
=\operatorname{Tr}(\rho_\beta X)$.  It deletes coherences between
distinct label blocks and replaces the remaining internal matrix data in
each fixed-label block by its Gibbs-weighted scalar average.
We therefore define the mean-zero label projection by
\[
  \Pi_{\mathrm{lab}}
  :=\mathcal E_{\mathrm{lab}}\big|_{L^2_0(\rho_\beta)}.
\]
Throughout the model analysis below, we use the abstract notation
\[
  \mathcal M:=\mathcal M_{\mathrm{lab}},
  \qquad
  \mathcal M_0:=\mathcal M\cap L^2_0(\rho_\beta),
  \qquad
  \Pi:=\Pi_{\mathrm{lab}},
  \qquad
  \Pi^\perp:=I-\Pi.
\]
In particular, $\operatorname{ran}\Pi=\mathcal M_0$.
Thus $\Pi$ is the projection onto the left-hand side of
\eqref{eq:mean-zero-label-identification}.  Here
$P_{\boldsymbol\lambda}$ acts on the physical many-edge Hilbert space,
whereas $\mathcal E_{\mathrm{lab}}$, $\Pi$, and $\Pi^\perp$ act on observables.
Each
vertex projector $A_v$ preserves the joint irrep blocks. Using the
zero-Bohr-frequency space defined in \eqref{eq:zero-bohr-frequency-space},
we therefore have
\begin{equation}
  \label{eq:label-algebra-zero-bohr-inclusion}
  \mathcal M
  \subseteq
  V_0(H_{\cG})=\{X:[X,H_{\cG}]=0\}.
\end{equation}
Degeneracies can make this inclusion strict.  Thus $\Pi^\perp$ contains both
$V_0(H_{\cG})\ominus\mathcal M$ and the nonzero-frequency subspaces.
In terms of label blocks, $\Pi^\perp$ contains both
blocks between different label configurations and non-scalar matrix or
fusion-channel directions within fixed-label blocks.

The label block decomposes over the edge update operators as
\[
  A=\Pi K\Pi=\sum_{e\in E}A_e,
  \qquad
  A_e:=\Pi K_e\Pi.
\]
We also use the family $\PathFam$ of unoriented induced centered
three-edge paths.  Each path is counted once and, after choosing either
endpoint order for its middle edge, is written as $p=(e_-,e_0,e_+)$.  Put
\[
  K_p:=K_{e_-}+K_{e_0}+K_{e_+},
  \qquad
  C_p:=A_{\operatorname{mid}(p)}:=A_{e_0}.
\]
Write the path vertices in order as $v_L-v_1-v_2-v_R$.
We call $v_1,v_2$ the \emph{internal vertices}, $v_L,v_R$ the
\emph{path endpoints}, and the three selected edges the \emph{active edges}.
Edges incident to a path vertex but not belonging to the path are its
\emph{outer edges}; each endpoint has two and each internal vertex has one.
The middle edge is the update to be certified, while the two flank edges
retain its interaction environment at the endpoints; see
\Cref{fig:finite-group-star-model}. We use the endpoint-factor assignment
in \cref{eq:edge-to-vertex-regrouping}.

The two flanks make the fixed space of $K_p$ tractable and give the
middle-edge shielding property proved in \cref{prop:middle-edge-shielding}:
\[
  \Pi\bigl(\ker K_p\cap L^2_0(\rho_\beta)\bigr)
  \subseteq\ker A_{\operatorname{mid}(p)}.
\]
Together with the fixed incidence multiplicities established below, this
allows the same paths to support local coupling control and global
assembly.

By \cref{lem:prelim-davies-generator-frustration-free}, this path sum is positive
and frustration free in the precise sense that
\[
  \ker K_p
  =\ker K_{e_-}\cap\ker K_{e_0}\cap\ker K_{e_+}.
\]
Identifying this intersection and bounding $K_p$ above its kernel are the
model-specific tasks carried out later.
Let $E_p$ be the KMS-orthogonal projection onto
$\ker K_p\cap L^2_0(\rho_\beta)$ within the mean-zero space and denote its
orthogonal-complement projection by $E_p^\perp:=I-E_p$.  The local form
$\mathcal S_{\Pi}(K_p)$ is the specialization of
\eqref{eq:short-variational-definition} to this split.

\subsection{The Three Model Conditions}
\label{sec:three-certificate-interfaces}

We verify Conditions~(i)--(iii) of
\cref{thm:abstract-two-axis-gap-assembly} as follows. The constants
$a,g_{\mathrm{loc}},c_{\mathrm{Sch}}>0$ are uniform in graph size, but may
depend on the fixed group and the chosen closed, bounded temperature interval.
\begin{enumerate}
\item[(i)] The interface bound $A\succeq a\Pi$ is proved in
\cref{sec:common-population-certificate} by classical comparison.
\item[(ii)] The local bound $K_p\succeq g_{\mathrm{loc}}E_p^\perp$
for every $p\in\PathFam$ is proved in \cref{sec:common-fibre-certificate}
after identifying the exact path kernel.
\item[(iii)] The Schur bound
$\mathcal S_{\Pi}(K_p)\succeq c_{\mathrm{Sch}}C_p$ for every $p\in\PathFam$
is proved in \cref{sec:common-schur-certificate} using middle-edge
decoupling and a finite path quotient bound.
\end{enumerate}

\subsection{Geometric Conditions and the Fixed-Group Gap}
\label{sec:finite-group-incidence-assembly}

The remaining hypotheses of \cref{thm:abstract-two-axis-gap-assembly} are
geometric. Condition~(G1) holds with $\kappa_{\mathrm{cov}}=6$:
\[
  \Pi^\perp\Big(\sum_{p\in\PathFam}E_p^\perp\Big)\Pi^\perp
  \succeq 6\Pi^\perp,
\]
as proved in \cref{lem:cubic-defect-frame}.
Conditions~(G2)--(G3) hold with $r_K=12$ and $r_A=4$,
respectively, through the exact incidence identities in the following lemma.

\begin{lemma}[Centered-path incidence counting]
\label{lem:cubic-path-ownership}
For every finite simple cubic graph of girth at least six,
\begin{equation}
  \label{eq:cubic-path-ownership}
  \sum_{p\in\PathFam}K_p=12K,
  \qquad
  \sum_{p\in\PathFam}C_p=4A.
\end{equation}
\end{lemma}

\begin{proof}
Fix an edge $e$.  To form a centered path having $e$ as its middle edge,
choose one of the other two incident edges at each endpoint of $e$.  Hence
$e$ is the middle edge of exactly $2\cdot2=4$ paths.  To make $e$ a flank,
choose one of its two endpoints, one of the other two edges there as the
middle edge, and one of the two continuations at the opposite endpoint of
that middle edge.  Thus $e$ is a flank of exactly $2\cdot2\cdot2=8$ paths.
The girth assumption ensures that these are induced three-edge paths and
that no choices collapse.  Consequently every $K_e$ occurs $4+8=12$ times
in the first sum, while every $A_e$ occurs four times as a middle-edge form
in the second.
\end{proof}

\subsection{Summary of the Constants}
\label{sec:model-constant-summary}

For a fixed group $\Gamma$ and $0\le t\le t_*<2/3$, the constants in
\cref{thm:abstract-two-axis-gap-assembly} are identified as follows:
\[
\begin{aligned}
  a&=a_{\Gamma,*}
    &&\text{(\cref{sec:common-population-certificate})},\\
  g_{\mathrm{loc}}&=g_{\Gamma,\mathrm{loc},*}
    &&\text{(\cref{sec:common-fibre-certificate})},\\
  c_{\mathrm{Sch}}&=c_{\Gamma,3,*}
    &&\text{(\cref{sec:common-schur-certificate})},\\
  \kappa_{\mathrm{cov}}&=6
    &&\text{(\cref{lem:cubic-defect-frame})},\\
  r_K&=12,\qquad r_A=4
    &&\text{(\cref{lem:cubic-path-ownership})}.
\end{aligned}
\]
The first three constants are positive and uniform in graph size and
$0\le t\le t_*$; the geometric constants are independent of $\Gamma$ and
temperature.

\section{Classical Comparison for the Label Interface Gap}
\label{sec:finite-block-davies-label-comparison}

In this section, we prove a model-independent comparison principle for
establishing the interface condition.  It relates the Davies form on
observables that are constant on prescribed operator blocks to the
classical heat-bath form on the corresponding labels.
We apply this principle to the finite-group model in
\cref{sec:common-population-certificate}, where its hypotheses are verified
using the Peter--Weyl decomposition.

Let $\sigma=Z^{-1}e^{-\beta H_0}$ be a strictly positive state on a finite-dimensional
Hilbert space, and let $\{P_i\}_{i\in\mathsf I}$ be nonzero mutually
orthogonal projectors summing to the identity and commuting with $H_0$.  Set
\[
  \pi_i:=\operatorname{Tr}(\sigma P_i),
  \qquad
  r_i:=\operatorname{rank}P_i,
\]
and identify the block-scalar algebra with functions
$f=\sum_i f_iP_i$ on $\mathsf I$.  We use the one-step classical
heat-bath generator $\mathsf K_{\mathrm{hb},\pi}$ from
\eqref{eq:abstract-block-heat-bath-form}, with the block index $i$ as its
classical state.

Let $K$ be a positive KMS-symmetric Davies operator for
$(H_0,\sigma)$, specified by a finite weighted family $\mathscr S$ of bare
couplings, with clock weights $w_S>0$.

\begin{theorem}[Finite-block Davies-to-label comparison]
\label{thm:finite-block-davies-label-comparison}
Fix an adjoint-closed subfamily
$\mathscr S^{\mathrm{sel}}\subseteq\mathscr S$.
Suppose that there are $b,\gamma_{\min},\kappa>0$ such that
\begin{align}
  \label{eq:abstract-block-density-floor}
  &\sigma P_i
  \succeq b\frac{\pi_i}{r_i}P_i,
  \tag{C1}
  \\
  \label{eq:abstract-block-rate-floor}
  &\gamma_\beta(\omega)
  \ge \gamma_{\min}
  \quad\text{for every Bohr component of every }
     S\in\mathscr S^{\mathrm{sel}},
  \tag{C2}
  \\
  \label{eq:abstract-block-overlap-floor}
  &\sum_{S\in\mathscr S^{\mathrm{sel}}}w_S
    \lVert P_jSP_i\rVert_{\mathrm{HS}}^2
  \ge \kappa r_i\pi_j
  \quad (i,j\in\mathsf I).
  \tag{C3}
\end{align}
Then every block-scalar observable satisfies
\begin{equation}
  \label{eq:abstract-davies-label-comparison}
  \langle f,Kf\rangle_\sigma
  \ge b\gamma_{\min}\kappa\,
    \langle f,\mathsf K_{\mathrm{hb},\pi}f\rangle_\pi.
\end{equation}
Equivalently, the compression of $K$ to the block-scalar
subspace dominates $b\gamma_{\min}\kappa$ times the label heat-bath operator.
\end{theorem}

The three conditions ensure that the selected couplings provide enough
mixing between labels.  Condition (C1) gives a lower bound on the Gibbs
weight of every state within each block.  Condition (C2) ensures that the
relevant transitions occur at a rate bounded away from zero, while (C3)
ensures that the selected couplings connect every pair of blocks strongly
enough.  Together, they allow us to bound the Davies form on label
observables from below by the classical heat-bath form in
\eqref{eq:abstract-block-heat-bath-form}.  It suffices to use the selected
couplings because all other contributions to the Davies form are
nonnegative.  The full proof is given in
\cref{app:finite-block-davies-label-comparison-proof}.

\begin{remark}[Relation to zero-Bohr-frequency comparisons]
\label{rem:comparison-with-zero-bohr-results}
The comparison of \cite{basso2025quantum} lifts a gap on the whole
zero-Bohr-frequency space $V_0(H_{\cG})$ to the full Davies generator
under spectral arithmetic hypotheses.  Our comparison instead bounds
the compression to the label interface $\mathcal M$ through its classical
label chain.  Although $\mathcal M\subseteq V_0(H_{\cG})$
by \eqref{eq:label-algebra-zero-bohr-inclusion}, a bound on $\mathcal M$
alone does not control $V_0(H_{\cG})\ominus\mathcal M$.
We use double localization to control the remaining observables and their
coupling to the interface.  Further differences in spectral assumptions
and temperature ranges are discussed in \cref{app:baseline-calibrations}.
\end{remark}

\section{Verification for Finite-Group Davies Dynamics}
\label{sec:common-analytic-certificates}

This section verifies the model conditions listed in
\cref{sec:three-certificate-interfaces,sec:finite-group-incidence-assembly}.
\Cref{sec:common-population-certificate} proves the interface bound (i)
by classical comparison.
\Cref{sec:common-fibre-certificate} proves the local gap bound (ii) and
the fixed-space cover (G1).
\Cref{sec:common-schur-certificate} proves the Schur bound (iii).
The geometric conditions (G2)--(G3) were established in
\cref{lem:cubic-path-ownership}.

We work under the hypotheses and unit-rate bath convention of
\cref{thm:finite-group-davies-gaps}(i).  Constants constructed below may
depend on $\Gamma$ and $t_*$, but not on the graph.  The restriction
$t_*<2/3$ enters only the label-subspace condition of \cref{sec:common-population-certificate}.  The local quotient and Schur conditions remain valid on any bounded temperature-parameter interval, with constants that may depend on its upper endpoint. The geometric conditions are independent of temperature.

\subsection{The Label-Subspace Condition: Classical Label Dynamics}
\label{sec:common-population-certificate}

In this section, we first establish a uniform gap for the classical irrep-label
chain in \cref{sec:auxiliary-irrep-label-chain}.  Then we apply
\cref{thm:finite-block-davies-label-comparison} in
\cref{sec:physical-davies-label-comparison} to lift it to the label-subspace
compression of the Davies form in the full observable KMS space.

\subsubsection{The Auxiliary Irrep-Label Chain}
\label{sec:auxiliary-irrep-label-chain}

We compare the Davies operator compressed to the label interface with
the positive generator of a classical chain on the irrep labels:
\[
  K
  \quad\xrightarrow{\;\Pi(\,\cdot\,)\Pi\;}\quad
  A:=\Pi K\Pi,
  \qquad
  A\quad\xleftrightarrow{\;\text{form comparison}\;}\quad
  \mathsf K_{\mathrm{lab}}.
\]
Here $K$ is the positive Davies operator on the mean-zero observable space,
$A=\Pi K\Pi$ is its compression to the label interface, and
$\mathsf K_{\mathrm{lab}}$ is the positive generator of a classical chain
that updates the irrep labels.

Recall from \cref{sec:formal-model-split-path-cover} that
$\Pi L^2_0(\rho_\beta)$ is identified with
$L^2_0(\mu_{\mathrm{lab}})$.  Here $\mu_{\mathrm{lab}}$ is the classical
distribution of the diagonal irrep-label blocks,
$$\mu_{\mathrm{lab}}(\boldsymbol\lambda)
=\operatorname{Tr}(\rho_\beta P_{\boldsymbol\lambda})$$.  Its explicit product
formula is derived below.  Whenever $\mathsf K_{\mathrm{lab}}$ is compared
with $A$, its quadratic form is transported through this identification to
$\operatorname{ran}\Pi$.
For each edge $e$, let
$\mathbb E_e^{\mathrm{lab}}$ be conditional expectation that resamples
$\lambda_e$ from this distribution while holding all other labels fixed.  The
unit-rate single-edge heat-bath generator used in this subsection is
\begin{equation}
  \label{eq:label-heat-bath-generator}
  \mathsf K_{\mathrm{lab}}
  :=\sum_{e\in E}\bigl(I-\mathbb E_e^{\mathrm{lab}}\bigr)
  \quad\text{on }L^2(\mu_{\mathrm{lab}}).
\end{equation}

\begin{proposition}[Uniform gap for the irrep-label chain]
\label{lem:uniform-label-influence}
The unit-rate irrep-label heat-bath generator $\mathsf K_{\mathrm{lab}}$ satisfies
\begin{equation}
  \label{eq:common-label-gap}
  \gap(\mathsf K_{\mathrm{lab}})
  \ge\delta_*:=1-\frac{4t_*}{t_*+2}
  =\frac{2-3t_*}{t_*+2}>0.
\end{equation}
Equivalently, $\mathsf K_{\mathrm{lab}}\succeq\delta_* I$
on $L^2_0(\mu_{\mathrm{lab}})$.
\end{proposition}

\begin{proof}
We first compute $\mu_{\mathrm{lab}}$ and then verify the one-site
Dobrushin criterion.  

\paragraph{Vertex weights and the label distribution.}
Fix a global label configuration
$\boldsymbol\lambda=(\lambda_e)_{e\in E(\cG)}$.
At a degree-three vertex $v$, let $e_1,e_2,e_3$ be its three incident
edges and write $\lambda_i:=\lambda_{e_i}\in\widehat\Gamma$.
Thus $\lambda_1,\lambda_2,\lambda_3$ are the three components of
$\boldsymbol\lambda$ on the edges meeting at $v$.
We replace $V_{\lambda_i}$ by its dual when required by the edge
orientation.  By \eqref{eq:vertex-group-average-projection}, $A_v$ is a
projection, and hence
\[
  e^{\beta J A_v}
  =(I-A_v)+e^{\beta J}A_v
  =I+tA_v,
  \qquad t=e^{\beta J}-1.
\]
By \eqref{eq:vertex-projectors-commute}, the vertex projectors commute, so
\[
  e^{-\beta H_{\cG}}
  =e^{\beta J\sum_{v\in V(\cG)}A_v}
  =\prod_{v\in V(\cG)}e^{\beta J A_v}
  =\prod_{v\in V(\cG)}(I+tA_v).
\]

We first trace one vertex factor.  By the vertex-wise regrouping
\eqref{eq:edge-to-vertex-regrouping}, each edge incident to $v$ contributes
exactly one of its two Peter--Weyl factors to $v$.  Since $v$ is degree-three, its
factors form
\[
  \mathcal V_v(\lambda_1,\lambda_2,\lambda_3)
  :=V_{\lambda_1}^{(v)}\otimes
    V_{\lambda_2}^{(v)}\otimes
    V_{\lambda_3}^{(v)},
\]
where each factor is $V_{\lambda_i}$ or its dual according to the edge
orientation.  Dualization does not change dimension, so
$\dim\mathcal V_v=\prod_{i=1}^3d_{\lambda_i}$.
Taking the trace gives
\begin{equation}
  \label{eq:vertex-label-weight}
  \operatorname{Tr}_{\mathcal V_v}(I+tA_v)
  =\prod_{i=1}^3d_{\lambda_i}
   +t\,\operatorname{Tr}_{\mathcal V_v}A_v
  =:z_t(\lambda_1,\lambda_2,\lambda_3).
\end{equation}
For a fixed global label configuration $\boldsymbol\lambda$,
$z_t(\lambda_1,\lambda_2,\lambda_3)$ is the unnormalized vertex weight
at $v$, depending only on its three incident labels.

We now evaluate the full label weight
$\mu_{\mathrm{lab}}(\boldsymbol\lambda)$.
The projector $P_{\boldsymbol\lambda}$ selects the subspace with this
fixed label configuration.  By \eqref{eq:edge-to-vertex-regrouping},
this subspace is a tensor product of the vertex spaces $\mathcal V_v$:
the two Peter--Weyl factors on each edge belong to its two endpoints.
On this subspace, each $I+tA_v$ acts only on the factor $\mathcal V_v$,
so the trace of their product factorizes over vertices.  Thus
\begin{equation}
  \label{eq:global-label-distribution}
  \begin{aligned}
  \mu_{\mathrm{lab}}(\boldsymbol{\lambda})
  &=\operatorname{Tr}
    (\rho_\beta P_{\boldsymbol\lambda})
  =\frac1{Z_\beta}\operatorname{Tr}\!\left[
    P_{\boldsymbol\lambda}\prod_{v\in V(\cG)}(I+tA_v)
    \right]\\
  &=\frac1{Z_\beta}\prod_{v\in V(\cG)}
    \operatorname{Tr}_{\mathcal V_v}(I+tA_v)
  =\frac1{Z_\beta}
   \prod_{v\in V(\cG)}
   z_t(\lambda_{v,1},\lambda_{v,2},\lambda_{v,3}).
  \end{aligned}
\end{equation}
Here the partition function is
\[
  Z_\beta:=\operatorname{Tr}(e^{-\beta H_{\cG}})
  =\sum_{\boldsymbol\lambda\in\widehat\Gamma^{E(\cG)}}
    \prod_{v\in V(\cG)}
    z_t(\lambda_{v,1},\lambda_{v,2},\lambda_{v,3}),
\]
so that $\mu_{\mathrm{lab}}$ sums to one over all label configurations.

\paragraph{Single-edge conditional distributions.}
Fix an oriented edge $e:u\to v$.  Besides $e$, the degree-three vertex $u$ is
incident to edges $a,b$, and $v$ is incident to edges $c,d$.  Their irrep
labels are $\lambda_a,\lambda_b,\lambda_c,\lambda_d$, respectively.  See
\Cref{fig:edge-label-dobrushin-neighborhood}.

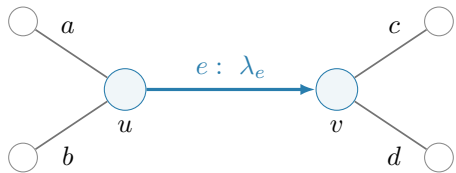
\begin{figure}[H]
  \centering
  \begin{tikzpicture}[
    endpoint/.style={circle, draw=figureGeometry, fill=figureGeometry!7,
      minimum size=5.5mm, inner sep=0pt},
    outer vertex/.style={circle, draw=black!45, fill=white,
      minimum size=3.8mm, inner sep=0pt},
    updated edge/.style={-{Latex[length=2mm]}, figureGeometry,
      line width=1.2pt},
    boundary edge/.style={black!55, line width=0.65pt},
    every node/.style={font=\small}
  ]
    \node[endpoint, label={[font=\small]below:$u$}] (u) at (-1.4,0) {};
    \node[endpoint, label={[font=\small]below:$v$}] (v) at (1.4,0) {};
    \draw[updated edge] (u) --
      node[above, text=figureGeometry] {$e:\ \lambda_e$} (v);

    \node[outer vertex] (a) at (-2.75,0.9) {};
    \node[outer vertex] (b) at (-2.75,-0.9) {};
    \draw[boundary edge] (u) --
      node[pos=0.55, above=3pt, text=black] {$a$} (a);
    \draw[boundary edge] (u) --
      node[pos=0.55, below=3pt, text=black] {$b$} (b);

    \node[outer vertex] (c) at (2.75,0.9) {};
    \node[outer vertex] (d) at (2.75,-0.9) {};
    \draw[boundary edge] (v) --
      node[pos=0.55, above=3pt, text=black] {$c$} (c);
    \draw[boundary edge] (v) --
      node[pos=0.55, below=3pt, text=black] {$d$} (d);
  \end{tikzpicture}
  \caption{The neighborhood of the updated edge $e$.  The conditional
  distribution of $\lambda_e$ depends only on the four adjacent labels
  $\lambda_a,\lambda_b,\lambda_c,\lambda_d$.}
  \label{fig:edge-label-dobrushin-neighborhood}
\end{figure}

The conditioning configuration $\boldsymbol{\lambda}_{\setminus e}$ fixes
every edge label except $\lambda_e$.  In particular,
$(\boldsymbol{\lambda}_{\setminus e})_a=\lambda_a$, and similarly for
$b,c,d$; only these four components occur in the two vertex factors at $u$
and $v$ that depend on the updated label.  The Peter--Weyl factors owned by
the two endpoints are $V_{\lambda_e}^{(u)}=V_{\lambda_e}$ and
$V_{\lambda_e}^{(v)}=V_{\lambda_e}^*$ by the convention of
\Cref{fig:peter-weyl-endpoint-assignment}.
In the conditional distribution and the estimates below, we place the
label on the updated edge in the first argument of each vertex weight:
\[
  z_t(\bullet,\lambda_a,\lambda_b),
  \qquad
  z_t((\bullet)^*,\lambda_c,\lambda_d).
\]
Here $\bullet$ denotes the variable label on $e$, with its dual used at
the opposite endpoint; the other two arguments are fixed by the
conditioning configuration.  In the denominator below, $\nu$ runs over
all possible labels on $e$.
Consequently the normalized
one-edge conditional distribution used by the heat-bath update is
\begin{equation}
  \label{eq:edge-label-conditional-distribution}
  \mu_{\mathrm{lab}}
  \bigl(\lambda_e\mid\boldsymbol{\lambda}_{\setminus e}\bigr)
  =
  \frac{z_t(\lambda_e,\lambda_a,\lambda_b)
        z_t(\lambda_e^*,\lambda_c,\lambda_d)}
       {\displaystyle\sum_{\nu\in\widehat\Gamma}
        z_t(\nu,\lambda_a,\lambda_b)
        z_t(\nu^*,\lambda_c,\lambda_d)}.
\end{equation}
All vertex factors away from $u$ and $v$ are independent of $\lambda_e$ and
cancel between the numerator and denominator.  Thus the product in the
numerator of \eqref{eq:edge-label-conditional-distribution} follows directly
from the vertex factorization in \eqref{eq:global-label-distribution}.  It is
not an additional assumption on the label chain.

For clarity, in the present label model the Dobrushin influence coefficient
used below is
\[
\begin{aligned}
  c_{ef}
  &:=
  \sup_{\boldsymbol x\sim_f\boldsymbol y}
  \left\|
    \mu_{\mathrm{lab}}
      \bigl(\,\cdot\mid
             \boldsymbol{\lambda}_{\setminus e}=\boldsymbol x\bigr)
    -
    \mu_{\mathrm{lab}}
      \bigl(\,\cdot\mid
             \boldsymbol{\lambda}_{\setminus e}=\boldsymbol y\bigr)
  \right\|_{\mathrm{TV}}\\
  &=\frac12\sup_{\boldsymbol x\sim_f\boldsymbol y}
    \sum_{\lambda_e\in\widehat\Gamma}
    \left|
      \mu_{\mathrm{lab}}
        (\lambda_e\mid\boldsymbol\lambda_{\setminus e}=\boldsymbol x)
      -\mu_{\mathrm{lab}}
        (\lambda_e\mid\boldsymbol\lambda_{\setminus e}=\boldsymbol y)
    \right|,
\end{aligned}
\]
where $\boldsymbol x$ and $\boldsymbol y$ are assignments of all edge labels other than $\lambda_e$
and agree except at edge $f$.  Although the conditioning formally includes
all of $E\setminus\{e\}$,
\eqref{eq:edge-label-conditional-distribution} shows that the conditional
distribution depends only on the four neighboring labels
$\lambda_a,\lambda_b,\lambda_c,\lambda_d$.  Thus $c_{ef}=0$ if $f$ does not
share an endpoint with $e$.  We treat the representative neighboring edge
$f=a$; the other three cases are identical.

Fix $\lambda_b,\lambda_c,\lambda_d$.  Since the conditional distribution
depends only on the four neighboring labels, we abbreviate it by
\begin{equation}
  \label{eq:neighbor-conditional-q}
  q_{\lambda_e}^{(\lambda_a)}
  :=\mu_{\mathrm{lab}}
    (\lambda_e\mid\lambda_a,\lambda_b,\lambda_c,\lambda_d)
  =\frac{z_t(\lambda_e,\lambda_a,\lambda_b)
          z_t(\lambda_e^*,\lambda_c,\lambda_d)}
        {\displaystyle\sum_{\nu\in\widehat\Gamma}
          z_t(\nu,\lambda_a,\lambda_b)
          z_t(\nu^*,\lambda_c,\lambda_d)},
  \qquad \lambda_e\in\widehat\Gamma.
\end{equation}
The superscript records the neighboring label that we vary, while the
subscript is the label on the updated edge $e$.
For $f=a$, the two conditioning configurations differ only in
$\lambda_a$, with $\lambda_b,\lambda_c,\lambda_d$ unchanged.
Since no other labels affect the conditional distribution, the supremum
above reduces to a maximum over these neighboring labels:
\begin{equation}
  \label{eq:neighbor-label-influence-tv}
  \begin{aligned}
  c_{ea}
  &=\frac12
   \max_{\lambda_a,\lambda_a',\lambda_b,\lambda_c,\lambda_d}
   \sum_{\lambda_e\in\widehat\Gamma}
   \Bigl|
     \mu_{\mathrm{lab}}(\lambda_e\mid\lambda_a',\lambda_b,\lambda_c,\lambda_d)
     -\mu_{\mathrm{lab}}(\lambda_e\mid\lambda_a,\lambda_b,\lambda_c,\lambda_d)
   \Bigr|\\
  &=\frac12
   \max_{\lambda_a,\lambda_a',\lambda_b,\lambda_c,\lambda_d}
   \sum_{\lambda_e\in\widehat\Gamma}
   \left|q_{\lambda_e}^{(\lambda_a')}-q_{\lambda_e}^{(\lambda_a)}\right|.
  \end{aligned}
\end{equation}

\paragraph{Comparing two neighboring configurations.}
To bound the right-hand side, for three irrep labels define
\begin{equation}
  \label{eq:normalized-vertex-weight}
  r_t(\lambda_1,\lambda_2,\lambda_3)
  :=\frac{z_t(\lambda_1,\lambda_2,\lambda_3)}
          {d_{\lambda_1}d_{\lambda_2}d_{\lambda_3}}.
\end{equation}
Orientation-dependent duals do not affect the denominator because
$d_{\lambda^*}=d_\lambda$.
Since $A_v$ is an orthogonal projection on $\mathcal V_v$,
$0\le\operatorname{Tr}_{\mathcal V_v}A_v\le\dim\mathcal V_v$.
Substituting this into \eqref{eq:vertex-label-weight} gives
\begin{equation}
  \label{eq:vertex-weight-ratio}
  \qquad 1\le r_t\le1+t.
\end{equation}
For two values $\lambda_a,\lambda_a'$, let $Z_{\lambda_a}$ denote the
normalizing denominator in \eqref{eq:edge-label-conditional-distribution}
with this value of $\lambda_a$, written explicitly as
\begin{equation}
  \label{eq:neighbor-label-normalizer}
  Z_{\lambda_a}
  :=\sum_{\nu\in\widehat\Gamma}
    z_t(\nu,\lambda_a,\lambda_b)
    z_t(\nu^*,\lambda_c,\lambda_d).
\end{equation}
Set
\begin{equation}
  \label{eq:neighbor-weight-change}
  h_{\lambda_e}
  :=\frac{r_t(\lambda_e,\lambda_a',\lambda_b)}
          {r_t(\lambda_e,\lambda_a,\lambda_b)},
  \qquad
  \overline h
  :=\sum_{\lambda_e} q_{\lambda_e}^{(\lambda_a)}h_{\lambda_e}.
\end{equation}
We use $h_{\lambda_e}$ to express the change in the vertex weight and
its average $\overline h$ to express the change in normalization.
Changing $\lambda_a$ to $\lambda_a'$ affects only the vertex weight
at $u$; the factor $z_t(\lambda_e^*,\lambda_c,\lambda_d)$ at $v$
is unchanged.  Using the definitions of $r_t$ and $h_{\lambda_e}$ in
\eqref{eq:normalized-vertex-weight} and \eqref{eq:neighbor-weight-change},
we obtain
\begin{align*}
  z_t(\lambda_e,\lambda_a',\lambda_b)
  &=d_{\lambda_e}d_{\lambda_a'}d_{\lambda_b}
    r_t(\lambda_e,\lambda_a',\lambda_b)\\
  &=d_{\lambda_e}d_{\lambda_a'}d_{\lambda_b}
    h_{\lambda_e}r_t(\lambda_e,\lambda_a,\lambda_b)\\
  &=\frac{d_{\lambda_a'}}{d_{\lambda_a}}\,
    h_{\lambda_e} z_t(\lambda_e,\lambda_a,\lambda_b).
\end{align*}
Substituting this into \eqref{eq:neighbor-label-normalizer} gives
\begin{align*}
  Z_{\lambda_a'}
  &=\frac{d_{\lambda_a'}}{d_{\lambda_a}}\,
    \sum_{\nu\in\widehat\Gamma}h_\nu
    z_t(\nu,\lambda_a,\lambda_b)z_t(\nu^*,\lambda_c,\lambda_d)\\
  &=\frac{d_{\lambda_a'}}{d_{\lambda_a}}\,
    Z_{\lambda_a}\sum_\nu q_\nu^{(\lambda_a)}h_\nu
  =\frac{d_{\lambda_a'}}{d_{\lambda_a}}\,
     Z_{\lambda_a}\overline h.
\end{align*}
The factor $d_{\lambda_a'}/d_{\lambda_a}$ appears in both the
unnormalized weight and the normalizing constant, so it cancels
when we form the conditional probability.
Thus the new distribution is obtained by multiplying the old one by
$h_{\lambda_e}$ and dividing by its average $\overline h$:
\[
  q_{\lambda_e}^{(\lambda_a')}
  =\frac{(d_{\lambda_a'}/d_{\lambda_a})h_{\lambda_e}}
         {(d_{\lambda_a'}/d_{\lambda_a})\overline h}\,
    q_{\lambda_e}^{(\lambda_a)}
  =\frac{h_{\lambda_e} q_{\lambda_e}^{(\lambda_a)}}{\overline h}.
\]
\paragraph{Bounding the influence and the gap.}
For each fixed choice in the maximum in
\eqref{eq:neighbor-label-influence-tv},
\begin{align*}
  \frac12\sum_{\lambda_e}
    \left|q_{\lambda_e}^{(\lambda_a')}-q_{\lambda_e}^{(\lambda_a)}\right|
  &=\frac{1}{2\overline h}
    \sum_{\lambda_e} q_{\lambda_e}^{(\lambda_a)}|h_{\lambda_e}-\overline h|.
\end{align*}
Set $m:=\min_{\lambda_e} h_{\lambda_e}$ and $M:=\max_{\lambda_e} h_{\lambda_e}$.  If $m=M$, the last
display vanishes.  Otherwise, since $h_{\lambda_e}\in[m,M]$ and has mean
$\overline h$ under $q^{(\lambda_a)}$, convexity gives the pointwise secant
bound
\[
  |s-\overline h|
  \le
  \frac{M-s}{M-m}(\overline h-m)
  +\frac{s-m}{M-m}(M-\overline h),
  \qquad s\in[m,M].
\]
Applying this with $s=h_{\lambda_e}$, averaging under $q^{(\lambda_a)}$, and using
$\sum_{\lambda_e} q_{\lambda_e}^{(\lambda_a)}h_{\lambda_e}=\overline h$ gives
\begin{align*}
  \frac{1}{2\overline h}
    \sum_{\lambda_e} q_{\lambda_e}^{(\lambda_a)}|h_{\lambda_e}-\overline h|
  \le
    \frac{(M-\overline h)(\overline h-m)}
         {\overline h(M-m)}
         \le
    \frac{\sqrt M-\sqrt m}{\sqrt M+\sqrt m}.
\end{align*}
The second inequality follows by maximizing the preceding expression over
$\overline h\in[m,M]$, with the maximum at
$\overline h=\sqrt{mM}$.
By \eqref{eq:vertex-weight-ratio}, both vertex-weight ratios defining
$h_{\lambda_e}$ lie in $[1,1+t]$, so
\[
  \frac1{1+t}
  \le h_{\lambda_e}
  =\frac{r_t(\lambda_e,\lambda_a',\lambda_b)}
         {r_t(\lambda_e,\lambda_a,\lambda_b)}
  \le1+t.
\]
Consequently,
\[
  m\ge\frac1{1+t},\qquad M\le1+t,
  \qquad\frac{M}{m}\le(1+t)^2.
\]
Since $(s-1)/(s+1)$ is increasing for $s\ge1$, the preceding bound becomes
\[
  \frac{\sqrt M-\sqrt m}{\sqrt M+\sqrt m}
  =\frac{\sqrt{M/m}-1}{\sqrt{M/m}+1}
  \le\frac{(1+t)-1}{(1+t)+1}
  =\frac{t}{t+2}.
\]
This estimate is independent of the neighboring labels, so taking the
maximum in \eqref{eq:neighbor-label-influence-tv} gives
$c_{ea}\le t/(t+2)$.  The same argument applies to the other three
neighboring edges, yielding
\begin{equation}
  \label{eq:uniform-label-influence}
  c_{ef}
  \le
  \begin{cases}
    \dfrac{t}{t+2}, & f\in\{a,b,c,d\},\\[5pt]
    0, & f\notin\{a,b,c,d\}.
  \end{cases}
\end{equation}
Since a cubic edge has exactly four neighboring edges,
\eqref{eq:uniform-label-influence} now gives, uniformly over $e$,
\begin{equation}
  \label{eq:label-dobrushin-total-influence}
  \sum_{f\ne e}c_{ef}
  \le \frac{4t}{t+2}
  \le \frac{4t_*}{t_*+2}
  <1,
\end{equation}
where the last inequality is precisely $t_*<2/3$.  Applying the unit-rate
heat-bath comparison from \cref{sec:prelim-dobrushin} therefore gives
\[
  \gap(\mathsf K_{\mathrm{lab}})
  \ge 1-\max_e\sum_{f\ne e}c_{ef}
  \ge 1-\frac{4t_*}{t_*+2}
  =\delta_*.
\]
On $L^2_0(\mu_{\mathrm{lab}})$ this scalar gap statement is the operator
inequality $\mathsf K_{\mathrm{lab}}\succeq\delta_*I$ in
\eqref{eq:common-label-gap}.
\end{proof}

\subsubsection{Comparison with the Physical Davies Block}
\label{sec:physical-davies-label-comparison}

We have established a uniform gap for the classical label chain in
\cref{lem:uniform-label-influence}.  In this subsection, we compare its
positive generator $\mathsf K_{\mathrm{lab}}$ with the Davies operator
compressed to the label interface, $A=\Pi K\Pi$, to obtain the interface
gap bound.

The next proposition verifies the conditions of the comparison principle
in \cref{thm:finite-block-davies-label-comparison}, without assuming
$[K,\Pi]=0$.

\begin{proposition}[Finite-group Davies-to-label comparison]
\label{prop:finite-group-davies-label-comparison}
For every finite $t_*\ge0$ and $0\le t=e^{\beta J}-1\le t_*$,
\begin{equation}
  \label{eq:physical-color-comparison}
  A=\Pi K\Pi\succeq c_{\Gamma,\mathrm{lab},*}\mathsf K_{\mathrm{lab}},
  \qquad
  c_{\Gamma,\mathrm{lab},*}
  :=\frac1{3q^2(1+t_*)^2(1+(1+t_*)^2)}.
\end{equation}
\end{proposition}

Combining
\cref{lem:uniform-label-influence,prop:finite-group-davies-label-comparison}
proves Condition~(i) of \cref{thm:abstract-two-axis-gap-assembly} with
\begin{equation}
  \label{eq:common-population-floor}
  a=a_{\Gamma,*}:=\delta_*c_{\Gamma,\mathrm{lab},*}>0.
\end{equation}
This is the only common certificate that uses the restriction $t_*<2/3$.

\begin{proof}[Proof of \cref{prop:finite-group-davies-label-comparison}]
Recall that \cref{thm:finite-block-davies-label-comparison} compares
Davies dynamics on block-scalar observables with a classical heat-bath
chain on the block labels, with comparison factor $b\gamma_{\min}\kappa$.
We restate each of its conditions (C1)--(C3) below in the local notation
before verifying it.

We compare the two dynamics one edge at a time: the Davies update on $e$,
evaluated on label observables, and the classical heat-bath update of
$\lambda_e$.  We first fix the labels on all other edges and apply the
comparison theorem to this conditional system, with blocks indexed by
the updated-edge irrep.  Averaging over the fixed labels and summing over
$e$ will give the claimed comparison.

Fix an updated edge $e$ and denote its candidate irrep label by $\lambda_e$.
The full bare coupling family is $\mathscr S_e$, while the selected subfamily in
the abstract theorem is
\begin{equation}
  \label{eq:selected-edge-comparison-couplings}
  \mathscr S_e^{\mathrm{sel}}
  :=\mathscr S_e^{\mathsf Q}=\{Q_h^{(e)}:h\in\Gamma\}
  ,
  \qquad
  w_{Q_h^{(e)}}=\frac1{3q}.
\end{equation}
Every selected coupling is self-adjoint, so this subfamily is adjoint closed.

\paragraph{The conditioned local system.}
We first isolate the degrees of freedom at the two endpoints of the
updated edge, keeping the labels on all other edges fixed.  We will
denote their local quantum state by $\rho^{\mathrm{end}}$ and the corresponding
classical distribution of the updated-edge label by $\pi_{\lambda_e}$.
The aim is to compare the Davies update on this quantum system with
the classical heat-bath update that resamples the label according to
$\pi$.  We therefore identify the state and its label blocks before
verifying conditions (C1)--(C3).

For the fixed updated edge $e=(u,v)$, choose
a realization $\boldsymbol{\xi}$ of the exterior label variable
$\boldsymbol{\lambda}_{\setminus e}$:
\[
  \boldsymbol{\xi}
  :=(\lambda_f)_{f\in E(\cG)\setminus\{e\}}
  \in\widehat\Gamma^{E(\cG)\setminus\{e\}}.
\]
Denote its four-label boundary configuration by
\[
  \boldsymbol{\partial}_e(\boldsymbol{\xi})
  :=(\lambda_a,\lambda_b,\lambda_c,\lambda_d),
\]
where $a,b$ are the two non-updated edges at $u$ and $c,d$ those at $v$.
We abbreviate this tuple to $\boldsymbol{\partial}_e$ while
$\boldsymbol{\xi}$ is fixed.  The local state depends on
$\boldsymbol{\xi}$ only through $\boldsymbol{\partial}_e$, whereas a general
label observable may retain its full $\boldsymbol{\xi}$-dependence.

Let $P_{\boldsymbol{\xi}}$ be the projector that fixes these exterior
labels and acts as the identity on the updated edge.  Its range is
\[
  \cH_{\boldsymbol{\xi}}
  :=\operatorname{Ran}P_{\boldsymbol{\xi}}
  \cong\C[\Gamma]^{(e)}
    \otimes\bigotimes_{f\ne e}
      \bigl(V_{\lambda_f}\otimes V_{\lambda_f}^*\bigr).
\]
Only the irrep labels on $f\ne e$ are fixed; their internal degrees of
freedom remain, and the updated edge retains its full regular-representation
factor.  Regrouping the tensor factors by vertex as in
\cref{eq:edge-to-vertex-regrouping} gives
\[
  \cH_{\boldsymbol{\xi}}
  \cong\cH^{\mathrm{end}}_{\boldsymbol{\partial}_e}
    \otimes\cH^{\mathrm{ext}}_{\boldsymbol{\xi}}.
\]
Here $\mathrm{end}$ denotes the factors belonging to the two endpoints
$u,v$, and $\mathrm{ext}$ denotes those belonging to all other vertices.
With duals determined by the edge orientations, the endpoint space is
\[
  \cH^{\mathrm{end}}_{\boldsymbol{\partial}_e}
  :=\C[\Gamma]^{(e)}
    \otimes V_{\lambda_a}^{(u)}
    \otimes V_{\lambda_b}^{(u)}
    \otimes V_{\lambda_c}^{(v)}
    \otimes V_{\lambda_d}^{(v)}.
\]
In particular, each of the four outer edges contributes one factor to
the endpoint space; its opposite-end factor belongs to the exterior
space.  Thus $\mathrm{ext}$ does not mean all edges other than $e$.
\Cref{fig:conditioned-endpoint-factors} illustrates this distinction.

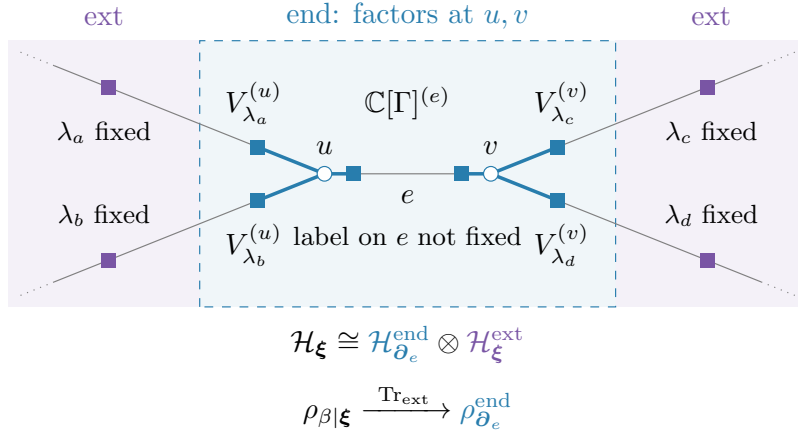
\begin{figure}[ht]
  \centering
  \scalebox{1.1}{
\begin{tikzpicture}[x=1cm,y=1cm,font=\small,
  factor/.style={rectangle,fill=figureGeometry,minimum size=5pt,inner sep=0pt},
  exterior/.style={rectangle,fill=figureComplement,minimum size=5pt,inner sep=0pt}]
  \fill[figureComplement!8] (-4.8,-1.6) rectangle (-2.5,1.6);
  \fill[figureComplement!8] (2.5,-1.6) rectangle (4.8,1.6);
  \filldraw[fill=figureGeometry!7,draw=figureGeometry,dashed]
    (-2.5,-1.6) rectangle (2.5,1.6);
  \node[text=figureGeometry] at (0,1.9) {end: factors at $u,v$};
  \node[text=figureComplement] at (-3.65,1.9) {ext};
  \node[text=figureComplement] at (3.65,1.9) {ext};

  \coordinate (u) at (-1,0);
  \coordinate (v) at (1,0);
  \draw[gray,thin] (-4.25,1.3)--(u)--(-4.25,-1.3);
  \draw[gray,thin] (4.25,1.3)--(v)--(4.25,-1.3);
  \draw[gray,thin] (u)--(v);
  \draw[figureGeometry,line width=1.2pt] (u)--(-.65,0);
  \draw[figureGeometry,line width=1.2pt] (v)--(.65,0);
  \node[above=3pt] at (u) {$u$};
  \node[above=3pt] at (v) {$v$};
  \node[below=2pt] at (0,0) {$e$};
  \node at (0,.85) {$\C[\Gamma]^{(e)}$};
  \node[font=\footnotesize] at (0,-.75) {label on $e$ not fixed};

  \foreach \sx/\sy/\lab/\vert in {-1/1/a/u,-1/-1/b/u,1/1/c/v,1/-1/d/v}{
    \draw[figureGeometry,line width=1.2pt] (\vert)--(\sx*1.8,\sy*.32);
    \node[factor] at (\sx*1.8,\sy*.32) {};
    \node[exterior] at (\sx*3.6,\sy*1.04) {};
    \node at (\sx*1.85,\sy*.85) {$V_{\lambda_\lab}^{(\vert)}$};
    \node[font=\footnotesize,fill=figureComplement!8,inner sep=1pt]
      at (\sx*3.65,\sy*.5) {$\lambda_\lab$ fixed};
    \draw[gray,dotted,thin] (\sx*4.25,\sy*1.3)--(\sx*4.7,\sy*1.48);
  }
  \node[circle,draw=figureGeometry,fill=white,minimum size=5pt,inner sep=0pt] at (u) {};
  \node[circle,draw=figureGeometry,fill=white,minimum size=5pt,inner sep=0pt] at (v) {};
  \node[factor] at (-.65,0) {};
  \node[factor] at (.65,0) {};
  \node at (0,-2.05) {$\cH_{\boldsymbol{\xi}}
    \cong {\color{figureGeometry}\cH^{\mathrm{end}}_{\boldsymbol{\partial}_e}}
    \otimes {\color{figureComplement}\cH^{\mathrm{ext}}_{\boldsymbol{\xi}}}$};
  \node at (0,-2.8) {$\rho_{\beta\mid\boldsymbol{\xi}}
    \xrightarrow{\ \operatorname{Tr}_{\mathrm{ext}}\ }
    {\color{figureGeometry}\rho^{\mathrm{end}}_{\boldsymbol{\partial}_e}}$};
\end{tikzpicture}}
  \caption{Endpoint and exterior factors after fixing the labels outside
  $e$.  Open circles mark $u,v$; thick segments connect each vertex to
  its three representation factors (colored squares).  Each outer edge
  crosses the end/ext boundary.  The central
  pair is kept together as $\C[\Gamma]^{(e)}\cong
  \bigoplus_{\lambda_e}V_{\lambda_e}\otimes V_{\lambda_e}^*$;
  tracing out ext leaves the endpoint state.}
  \label{fig:conditioned-endpoint-factors}
\end{figure}

For the fixed exterior labels, define the endpoint operators
$A_u^{\mathrm{end}},A_v^{\mathrm{end}}$ on
$\cH^{\mathrm{end}}_{\boldsymbol{\partial}_e}$ by
\[
  \left.A_w\right|_{\cH_{\boldsymbol{\xi}}}
  =A_w^{\mathrm{end}}\otimes I_{\mathrm{ext}},
  \qquad w\in\{u,v\}.
\]
The local conditional state is the normalized density
\begin{equation}
  \label{eq:conditional-endpoint-gibbs-density}
  \rho^{\mathrm{end}}_{\boldsymbol{\partial}_e}
  :=
  \frac{
    (I+tA_u^{\mathrm{end}})(I+tA_v^{\mathrm{end}})
  }{
    \operatorname{Tr}_{\cH^{\mathrm{end}}_{\boldsymbol{\partial}_e}}\!\left[
      (I+tA_u^{\mathrm{end}})(I+tA_v^{\mathrm{end}})
    \right]
  }.
\end{equation}

We remark that this local state is obtained from the global Gibbs state
by two distinct operations.  First, conditioning on the exterior labels
and normalizing gives the state on $\cH_{\boldsymbol{\xi}}$
\[
  \rho_{\beta\mid\boldsymbol{\xi}}
  :=
  \frac{
    P_{\boldsymbol{\xi}}
    \rho_\beta
    P_{\boldsymbol{\xi}}
  }{
    \operatorname{Tr}\!\left(
      \rho_\beta
      P_{\boldsymbol{\xi}}
    \right)
  }.
\]
Second, $\operatorname{Tr}_{\mathrm{ext}}$ takes the partial trace over
$\cH^{\mathrm{ext}}_{\boldsymbol{\xi}}$, leaving a state on the endpoint
space.  On the conditioned space, the factors in
$e^{-\beta H_{\cG}}=\prod_w(I+tA_w)$ with $w\ne u,v$ act only on the
exterior space.  Their traces cancel upon normalization, giving
\begin{equation}
  \label{eq:conditioned-endpoint-state-reduction}
  \operatorname{Tr}_{\mathrm{ext}}\!\left(
    \rho_{\beta\mid\boldsymbol{\xi}}
  \right)
  =\rho^{\mathrm{end}}_{\boldsymbol{\partial}_e}.
\end{equation}

Define the updated-edge label projectors on the endpoint space by
\[
  P_{\lambda_e\mid\boldsymbol{\partial}_e}
  :=P_{\lambda_e}^{(e)}\otimes I_{\boldsymbol{\partial}_e},
\]
where $I_{\boldsymbol{\partial}_e}$ is the identity on the four fixed outer factors.
The Peter--Weyl projectors on the updated edge therefore give
\[
  P_{\lambda_e\mid\boldsymbol{\partial}_e}P_{\mu_e\mid\boldsymbol{\partial}_e}
  =\mathbf1_{\lambda_e=\mu_e}P_{\lambda_e\mid\boldsymbol{\partial}_e},
  \qquad
  \sum_{\lambda_e\in\widehat\Gamma}P_{\lambda_e\mid\boldsymbol{\partial}_e}
  =I_{\cH^{\mathrm{end}}_{\boldsymbol{\partial}_e}}.
\]
Thus $\{P_{\lambda_e\mid\boldsymbol{\partial}_e}\}$ decomposes this local space according to the irrep
label on the updated edge while retaining all matrix and fusion-multiplicity
directions inside each label block.  Put
\[
  \rho^{\mathrm{end}}:=\rho^{\mathrm{end}}_{\boldsymbol{\partial}_e},
  \qquad
  \pi_{\lambda_e}
  :=\operatorname{Tr}(\rho^{\mathrm{end}} P_{\lambda_e\mid\boldsymbol{\partial}_e}),
  \qquad
  r_{\lambda_e}
  :=\operatorname{rank}P_{\lambda_e\mid\boldsymbol{\partial}_e}.
\]
Throughout the verification of (C1)--(C3), the boundary labels are fixed
and their dependence is suppressed in $\rho^{\mathrm{end}}$, $\pi_{\lambda_e}$, and
$r_{\lambda_e}$.
Thus $\pi_{\lambda_e}$ is the scalar
conditional label probability in \eqref{eq:edge-label-conditional-distribution},
obtained from $\rho^{\mathrm{end}}$ by measuring only the updated-edge irrep block.

\paragraph{Condition (C1): the blockwise Gibbs-density lower bound.}
The required bound is, for every updated-edge label,
\[
  \rho^{\mathrm{end}}P_{\lambda_e\mid\boldsymbol{\partial}_e}
  \succeq b_*\frac{\pi_{\lambda_e}}{r_{\lambda_e}}
    P_{\lambda_e\mid\boldsymbol{\partial}_e},
\]
with a uniform constant $b_*>0$.
By \cref{lem:vertex-commutativity-locality}, the two endpoint operators
$A_u^{\mathrm{end}}$ and $A_v^{\mathrm{end}}$ are commuting orthogonal projectors.
Hence $(I+tA_u^{\mathrm{end}})(I+tA_v^{\mathrm{end}})$ has eigenvalues
among $1$, $1+t$, and $(1+t)^2$.
The common normalization in
\eqref{eq:conditional-endpoint-gibbs-density} does not change eigenvalue
ratios, so, on the interval under consideration, the ratio of the largest
to the smallest eigenvalue of $\rho^{\mathrm{end}}$ is at most $(1+t_*)^2$.  Moreover,
$P_{\lambda_e\mid\boldsymbol{\partial}_e}$ commutes with $\rho^{\mathrm{end}}$, because the central updated-edge
isotypic projector commutes with the endpoint regular actions.  Let
$s_{\lambda_e,\min}$ and $s_{\lambda_e,\max}$ be the extreme eigenvalues of
$\rho^{\mathrm{end}}$ on $\operatorname{ran}P_{\lambda_e\mid\boldsymbol{\partial}_e}$.  Since
\[
  \frac{\pi_{\lambda_e}}{r_{\lambda_e}}
  =\frac{1}{r_{\lambda_e}}\operatorname{Tr}(\rho^{\mathrm{end}} P_{\lambda_e\mid\boldsymbol{\partial}_e})
\]
is the average eigenvalue in that block,
\[
  \frac{\pi_{\lambda_e}}{r_{\lambda_e}}
  \le s_{\lambda_e,\max}
  \le (1+t_*)^2s_{\lambda_e,\min}.
\]
This proves the required bound with
\begin{equation}
  \label{eq:conditional-label-density-floor}
  \rho^{\mathrm{end}} P_{\lambda_e\mid\boldsymbol{\partial}_e}
  \succeq
  b_*\frac{\pi_{\lambda_e}}{r_{\lambda_e}}P_{\lambda_e\mid\boldsymbol{\partial}_e},
  \qquad
  b_*:=(1+t_*)^{-2}.
\end{equation}

\paragraph{Condition (C2): Davies-rate lower bound.}
For the rates, we need a uniform constant $\gamma_{\min,*}>0$ satisfying
\[
  \gamma_\beta(\omega)\ge\gamma_{\min,*}
\]
for every Bohr component of each selected coupling $Q_h^{(e)}$.
The Bohr decomposition is local, as anticipated by
\cref{lem:vertex-commutativity-locality}.  With respect to the same
endpoint/exterior decomposition of $\cH_{\boldsymbol{\xi}}$,
\[
  \left.H_{\cG}\right|_{\cH_{\boldsymbol{\xi}}}
  =H_{\boldsymbol{\partial}_e}\otimes I_{\mathrm{ext}}
    +I_{\mathrm{end}}\otimes H_{\mathrm{ext}},
\]
where $H_{\boldsymbol{\partial}_e}=-J(A_u^{\mathrm{end}}+A_v^{\mathrm{end}})$.
Likewise, the selected coupling on the updated edge satisfies
\[
  \left.Q_h^{(e)}\right|_{\cH_{\boldsymbol{\xi}}}
  =Q_h^{\mathrm{end}}\otimes I_{\mathrm{ext}},
  \qquad
  Q_h^{\mathrm{end}}
  :=|h\rangle\langle h|_e\otimes I_{\boldsymbol{\partial}_e}.
\]
Since this coupling leaves the exterior energy unchanged,
\begin{equation}
  \label{eq:conditioned-edge-bohr-localization}
  \left.(Q_h^{(e)})_\omega\right|_{\cH_{\boldsymbol{\xi}}}
  =(Q_h^{\mathrm{end}})_\omega\otimes I_{\mathrm{ext}}.
\end{equation}
Thus the restriction of the positive edge update operator to this conditioned block is
exactly the finite positive Davies operator on
$\cH^{\mathrm{end}}_{\boldsymbol{\partial}_e}$.  No exterior energy
enters its Bohr frequencies or logistic rates.

Here $P_E^{\mathrm{end}}$ denotes the spectral projector of the local Hamiltonian
$H_{\boldsymbol{\partial}_e}=-J(A_u^{\mathrm{end}}+A_v^{\mathrm{end}})$
onto its eigenspace with energy $E$.
Write $P_{\lambda_e\mid\boldsymbol{\partial}_e,E}:=P_{\lambda_e\mid\boldsymbol{\partial}_e}P_E^{\mathrm{end}}$ for the joint label/energy projector in
this local system.  Its energy lies in $\{-2J,-J,0\}$, so every local
Bohr frequency has magnitude at most $2J$.
Since $e^{\beta J}=1+t\le1+t_*$, the logistic rates satisfy
\[
  \gamma_\beta(\omega)
  \ge\frac1{1+e^{\beta|\omega|}}
  \ge\frac1{1+e^{2\beta J}}
  =\frac1{1+(1+t)^2}
  \ge\frac1{1+(1+t_*)^2}.
\]
Thus we obtain the uniform lower bound
\begin{equation}
  \label{eq:local-logistic-rate-floor}
  \gamma_\beta(\omega)\ge
  \gamma_{\min,*}:=\frac1{1+(1+t_*)^2}.
\end{equation}

\paragraph{Condition (C3): overlap of the selected couplings.}
It remains to show that, for some uniform $\kappa>0$ and every pair of updated-edge labels,
\[
  \sum_{h\in\Gamma}\frac1{3q}
  \bigl\lVert P_{\mu_e\mid\boldsymbol{\partial}_e}
    Q_h^{\mathrm{end}}P_{\lambda_e\mid\boldsymbol{\partial}_e}
  \bigr\rVert_{\mathrm{HS}}^2
  \ge\kappa r_{\lambda_e}\pi_{\mu_e}.
\]
We first compute the overlap on the updated edge alone.  Set
$Q_h=|h\rangle\langle h|$, $h\in\Gamma$.  On the updated regular edge, the
isotypic projector $P_{\lambda_e}$ has constant group-basis diagonal
\[
  \langle h,P_{\lambda_e}h\rangle
  =\lVert P_{\lambda_e}|h\rangle\rVert^2
  =\frac{d_{\lambda_e}^2}{q}.
\]
Moreover,
$P_{\lambda_e}Q_hP_{\mu_e}
=|P_{\lambda_e}h\rangle\langle P_{\mu_e}h|$, so its squared
Hilbert--Schmidt norm is
$d_{\lambda_e}^2d_{\mu_e}^2/q^2$, independently of $h$.  Summing over the $q$ group
elements gives
\begin{equation}
  \label{eq:color-fourier-overlap}
  \sum_{h\in\Gamma}
  \lVert
    P_{\lambda_e}Q_hP_{\mu_e}
  \rVert_{\mathrm{HS}}^2
  =\frac{d_{\lambda_e}^2d_{\mu_e}^2}{q}>0.
\end{equation}
Thus the group-basis bath couplings connect every pair of irrep-label blocks.

We now include the weights $1/(3q)$ from
\eqref{eq:selected-edge-comparison-couplings} and the fixed outer factors.
The tensor product of the four fixed
outer representation factors has dimension
\[
  D_{\boldsymbol{\partial}_e}
  :=\prod_{f\in\{a,b,c,d\}}d_{\lambda_f}
  =d_{\lambda_a}d_{\lambda_b}d_{\lambda_c}d_{\lambda_d},
\]
which is independent of the updated label.
The fixed boundary factor contributes this same multiplicity to
every updated-edge label block, so
$r_{\lambda_e}=D_{\boldsymbol{\partial}_e}d_{\lambda_e}^2$.  Hence
\eqref{eq:color-fourier-overlap} gives
\begin{align}
  \label{eq:finite-group-abstract-overlap-verification}
  \sum_{h\in\Gamma}\frac1{3q}
    \lVert P_{\mu_e\mid\boldsymbol{\partial}_e} Q_h^{\mathrm{end}}P_{\lambda_e\mid\boldsymbol{\partial}_e}\rVert_{\mathrm{HS}}^2
  =\frac{D_{\boldsymbol{\partial}_e}d_{\lambda_e}^2d_{\mu_e}^2}{3q^2}
    \notag
  =\frac{r_{\lambda_e} d_{\mu_e}^2}{3q^2}
  \ge \frac1{3q^2}r_{\lambda_e}\pi_{\mu_e}.
\end{align}
Thus \eqref{eq:abstract-block-overlap-floor} holds with
$\kappa=1/(3q^2)$.

\paragraph{Applying the abstract comparison.}
For a global label observable $f$, fixing the exterior labels
$\boldsymbol{\xi}$ gives
\[
  f^{\boldsymbol{\xi}}=\sum_{\lambda_e\in\widehat\Gamma}
    f(\lambda_e,\boldsymbol{\xi})P_{\lambda_e\mid\boldsymbol{\partial}_e},
\]
where $f(\lambda_e,\boldsymbol{\xi})$ is its value at the full label
configuration $(\lambda_e,\boldsymbol{\xi})$.
Restoring the boundary dependence in $\rho^{\mathrm{end}}$ and $\pi$,
we apply \cref{thm:finite-block-davies-label-comparison} with
\[
  \sigma=\rho^{\mathrm{end}}_{\boldsymbol{\partial}_e},
  \qquad b=b_*,
  \qquad \gamma_{\min}=\gamma_{\min,*},
  \qquad \kappa=\frac1{3q^2},
\]
to obtain, for every exterior configuration $\boldsymbol{\xi}$,
\begin{equation}
  \label{eq:conditional-physical-label-form-comparison}
  \underbrace{
    \langle f^{\boldsymbol{\xi}},K_ef^{\boldsymbol{\xi}}\rangle_{\rho^{\mathrm{end}}_{\boldsymbol{\partial}_e}}
  }_{\mathcal E^{\mathrm{phys}}_{\boldsymbol{\partial}_e}(f^{\boldsymbol{\xi}})}
  \ge
  \frac{\gamma_{\min,*}b_*}{3q^2}
  \underbrace{
    \frac12\sum_{\lambda_e,\mu_e}
    \pi_{\boldsymbol{\partial}_e}(\lambda_e)
    \pi_{\boldsymbol{\partial}_e}(\mu_e)
    |f(\lambda_e,\boldsymbol{\xi})-f(\mu_e,\boldsymbol{\xi})|^2
  }_{\mathcal E^{\mathrm{lab}}_{\boldsymbol{\partial}_e}(f^{\boldsymbol{\xi}})}.
\end{equation}
This is exactly the stated constant
$c_{\Gamma,\mathrm{lab},*}=\gamma_{\min,*}b_*/(3q^2)$.

\paragraph{Averaging back to the global form.}
It remains to pass from the conditioned local comparison back to the global
quadratic forms.  Let
\[
  \nu_e(\boldsymbol{\xi})
  :=\mu_{\mathrm{lab}}
    (\boldsymbol{\lambda}_{\setminus e}=\boldsymbol{\xi})
  =\operatorname{Tr}(\rho_\beta P_{\boldsymbol{\xi}})
\]
be the marginal distribution of the complete exterior label configuration.
On each conditioned block, recall the Hamiltonian decomposition
\[
  \left.H_{\cG}\right|_{\cH_{\boldsymbol{\xi}}}
  =H_{\boldsymbol{\partial}_e}\otimes I_{\mathrm{ext}}
    +I_{\mathrm{end}}\otimes H_{\mathrm{ext}}.
\]
The two terms act on separate tensor factors, so the Gibbs operator is
$e^{-\beta H_{\boldsymbol{\partial}_e}}\otimes e^{-\beta H_{\mathrm{ext}}}$.
Its trace is the product of the two traces; normalizing therefore gives
\[
  \rho_{\beta\mid\boldsymbol{\xi}}
  =\rho^{\mathrm{end}}_{\boldsymbol{\partial}_e(\boldsymbol{\xi})}
    \otimes\rho^{\mathrm{ext}}_{\boldsymbol{\xi}},
  \qquad \operatorname{Tr}\rho^{\mathrm{ext}}_{\boldsymbol{\xi}}=1,
\]
where the endpoint state is the marginal in
\eqref{eq:conditioned-endpoint-state-reduction}.
For any endpoint operators $X,Y$, the KMS inner product therefore reduces as
\begin{align*}
  \langle X\otimes I_{\mathrm{ext}},Y\otimes I_{\mathrm{ext}}
    \rangle_{\rho_{\beta\mid\boldsymbol{\xi}}}
  &=\langle X,Y\rangle_
    {\rho^{\mathrm{end}}_{\boldsymbol{\partial}_e(\boldsymbol{\xi})}}
    \operatorname{Tr}\!\left[
      (\rho^{\mathrm{ext}}_{\boldsymbol{\xi}})^{1/2}
      I_{\mathrm{ext}}
      (\rho^{\mathrm{ext}}_{\boldsymbol{\xi}})^{1/2}
      I_{\mathrm{ext}}\right]\\
  &=\langle X,Y\rangle_
    {\rho^{\mathrm{end}}_{\boldsymbol{\partial}_e(\boldsymbol{\xi})}}.
\end{align*}
The exterior trace is one, so only the endpoint inner product remains.
Here $f^{\boldsymbol{\xi}}$ acts on the endpoint space, and
\eqref{eq:conditioned-edge-bohr-localization} shows that its Davies update
does too, so the identity applies with $X=f^{\boldsymbol{\xi}}$ and
$Y=K_ef^{\boldsymbol{\xi}}$.

Since the update on $e$ leaves all exterior labels unchanged, the global
quadratic form is the weighted sum of its conditional forms.  Reducing
each conditional inner product to the endpoint space as above gives
\begin{align*}
  \langle f,\Pi K_e\Pi f\rangle_{\rho_\beta}
  &=\sum_{\boldsymbol{\xi}}
    \operatorname{Tr}(\rho_\beta P_{\boldsymbol{\xi}})
    \langle f^{\boldsymbol{\xi}}\otimes I_{\mathrm{ext}},
      (K_ef^{\boldsymbol{\xi}})\otimes I_{\mathrm{ext}}\rangle_
      {\rho_{\beta\mid\boldsymbol{\xi}}}
  \\
  &=\sum_{\boldsymbol{\xi}}\nu_e(\boldsymbol{\xi})\,
    \langle f^{\boldsymbol{\xi}},K_ef^{\boldsymbol{\xi}}\rangle_
      {\rho^{\mathrm{end}}_{\boldsymbol{\partial}_e(\boldsymbol{\xi})}}
  =\sum_{\boldsymbol{\xi}}\nu_e(\boldsymbol{\xi})\,
    \mathcal E^{\mathrm{phys}}_{\boldsymbol{\partial}_e(\boldsymbol{\xi})}
      (f^{\boldsymbol{\xi}}).
\end{align*}
For the auxiliary heat-bath chain, conditional variance gives the parallel
identity
\begin{align*}
  \langle f,(I-\mathbb E_e^{\mathrm{lab}})f\rangle_{\mu_{\mathrm{lab}}}
  =\sum_{\boldsymbol{\xi}}\nu_e(\boldsymbol{\xi})\,
    \operatorname{Var}_{\pi_{\boldsymbol{\partial}_e(\boldsymbol{\xi})}}
      \bigl(f(\,\cdot\,,\boldsymbol{\xi})\bigr)
  =\sum_{\boldsymbol{\xi}}\nu_e(\boldsymbol{\xi})\,
    \mathcal E^{\mathrm{lab}}_{\boldsymbol{\partial}_e(\boldsymbol{\xi})}
      (f^{\boldsymbol{\xi}}).
\end{align*}
Applying
\eqref{eq:conditional-physical-label-form-comparison} separately for every
$\boldsymbol{\xi}$ and averaging therefore introduces no loss.  Summing the resulting
inequality over $e$ yields, for every $f\in\operatorname{ran}\Pi$,
\[
  \langle f,\Pi K\Pi f\rangle_{\rho_\beta}
  \ge\frac{\gamma_{\min,*}b_*}{3q^2}
    \langle f,\mathsf K_{\mathrm{lab}}f\rangle_{\mu_{\mathrm{lab}}},
\]
which implies the weaker stated comparison
\eqref{eq:physical-color-comparison}.
\end{proof}

\subsection{Local Quotients and the Geometric Fixed-Space Cover}
\label{sec:common-fibre-certificate}

We first identify the exact fixed space of each centered path in
\cref{sec:exact-centered-path-fixed-space}.
Using this description, \cref{sec:uniform-centered-path-quotient-gap}
proves a uniform gap above the fixed space, establishing Condition~(ii)
of \cref{thm:abstract-two-axis-gap-assembly}.
Finally, \cref{sec:global-cover-local-fixed-complements} shows that the
local fixed-space complements together cover the subspace orthogonal
to the label interface, establishing Geometric Condition~(G1).

\subsubsection{The Exact Fixed Space of a Centered Path}
\label{sec:exact-centered-path-fixed-space}

We next identify the exact kernel of the positive Davies operator restricted to a
three-edge path before taking any spectral minimum.  Fix
$p=(e_-,e_0,e_+)$.  An
observable has both ket-side and bra-side Peter--Weyl labels.  The path
operator can change the labels on the three active edges, so those labels
must remain free.  It does, however, preserve the ket and bra labels on every
outer factor.

We call each representation factor contributed by an outer edge an
\emph{outer leg}.  The \emph{path neighborhood} consists of the three active
edges and all representation factors owned by the four path vertices.  The
remaining factors form the \emph{exterior factor}, on which $K_p$ acts
trivially.

\paragraph{Ordered outer-label blocks.}
An abstract centered three-edge path has six outer representation slots as in \Cref{fig:ordered-path-kernel-factors}.
Let
\begin{equation}
  \label{eq:ordered-block-type-set}
  \Theta_{\Gamma}
  :=(\widehat\Gamma^6)_{\mathrm{ket}}
    \times(\widehat\Gamma^6)_{\mathrm{bra}}.
\end{equation}
Thus
$\boldsymbol{\theta}=(\boldsymbol{\theta}^{\mathrm{ket}},
          \boldsymbol{\theta}^{\mathrm{bra}})
\in\Theta_{\Gamma}$ is an abstract ordered outer-label type.
Each component $\boldsymbol{\theta}^{\mathrm{ket}}$ or
$\boldsymbol{\theta}^{\mathrm{bra}}$ is a six-tuple in
$\widehat\Gamma^6$.  
Once $p=(e_-,e_0,e_+)$ and its endpoint order are fixed, its two endpoint pairs
and two internal outer factors identify those six slots with the actual
outer factors surrounding $p$.  Placing the labels of $\boldsymbol{\theta}$ in these
$p$-dependent positions gives Peter--Weyl projectors
$P_{p,\boldsymbol{\theta}^{\mathrm{ket}}}$ and
$P_{p,\boldsymbol{\theta}^{\mathrm{bra}}}$.  The corresponding realized
operator block is
\begin{equation}
  \label{eq:ordered-outer-operator-block}
  \mathcal X_{p,\boldsymbol{\theta}}
  :=
  P_{p,\boldsymbol{\theta}^{\mathrm{ket}}}
  \mathcal B(\mathcal H_{\cG})
  P_{p,\boldsymbol{\theta}^{\mathrm{bra}}}.
\end{equation}
For every fixed $p$ these realized blocks give the full decomposition
\[
  \mathcal B(\mathcal H_{\cG})
  =\bigoplus_{\boldsymbol{\theta}\in\Theta_{\Gamma}}
    \mathcal X_{p,\boldsymbol{\theta}}.
\]
An element of $\mathcal X_{p,\boldsymbol{\theta}}$ maps the fixed bra outer-label sector to
the fixed ket outer-label sector, while retaining all active-edge labels and
matrix directions.  Throughout the paper, an \emph{ordered block} means
precisely this realized block indexed by the pair $(p,\boldsymbol{\theta})$, where
$\boldsymbol{\theta}=(\boldsymbol{\theta}^{\mathrm{ket}},
\boldsymbol{\theta}^{\mathrm{bra}})$ orders the ket and bra outer-label
tuples and $p$ places their six slots on the graph.  
For later readability, write
\[
  K_{p,\boldsymbol{\theta}}
  :=K_p\big|_{\mathcal X_{p,\boldsymbol{\theta}}}.
\]
Since $K_p$ preserves the fixed outer labels, it is block diagonal in the
displayed decomposition:
\begin{equation}
  \label{eq:ordered-block-kp-direct-sum}
  K_p
  =\bigoplus_{\boldsymbol{\theta}\in\Theta_{\Gamma}}
    K_{p,\boldsymbol{\theta}}.
\end{equation}

Once $(p,\boldsymbol{\theta})$ is fixed, let $v_L,v_1,v_2,v_R$ be the path vertices
from left to right.  The six outer labels are fixed, while the three active
edges remain free.

To describe the fixed space, we group the outer factors at each vertex.
Let
$\bullet$ stand for either ket or bra, and let
\[
  W_v^\bullet
  :=\bigotimes_{\substack{e\ni v\\
              e\notin\{e_-,e_0,e_+\}}}
       V_{\lambda_e^\bullet}^{(v)}
\]
be the frozen outer carrier owned by $v$.  It has two factors at an endpoint
and one at an internal vertex.  Let
$U_v^\bullet:\Gamma\to\mathcal U(W_v^\bullet)$ be the induced group action.
More explicitly, if $a_1,a_2$ are the two outer edges at $v_L$ and $b_1,b_2$
those at $v_R$, define the endpoint actions by
\begin{equation}
  \label{eq:endpoint-outer-actions}
  \begin{aligned}
    U_L^\bullet&:=U_{v_L}^\bullet
      &&\text{on }W_{v_L}^\bullet
        =V_{\lambda_{a_1}^\bullet}^{(v_L)}
         \otimes V_{\lambda_{a_2}^\bullet}^{(v_L)},\\
    U_R^\bullet&:=U_{v_R}^\bullet
      &&\text{on }W_{v_R}^\bullet
        =V_{\lambda_{b_1}^\bullet}^{(v_R)}
         \otimes V_{\lambda_{b_2}^\bullet}^{(v_R)}.
  \end{aligned}
\end{equation}
For the two internal vertices, write
\[
  W_{v_1}^\bullet=V_{\mu_1^\bullet},
  \qquad
  W_{v_2}^\bullet=V_{\mu_2^\bullet}.
\]
Thus $V_\lambda$ is one irreducible carrier, whereas $U_v$ is the group action
on the full outer carrier $W_v$ and may be a tensor-product representation.

If $X\in\ker K_{p,\boldsymbol{\theta}}$, the fixed-point characterization in
\cref{lem:davies-fixed-point-commutant} makes $X$ commute with the complete
coupling algebra on each active edge, which is
$\mathcal B(\C[\Gamma])$.  Write
\begin{equation}
  \label{eq:active-path-identity}
  I_p
  :=\bigotimes_{e\in\{e_-,e_0,e_+\}}
      I_{\C[\Gamma]^{(e)}}
\end{equation}
for the identity on the three active-edge factors.  Then every such fixed
observable has the form
\begin{equation}
  \label{eq:fixed-observable-active-identity}
  X=I_p\otimes Y.
\end{equation}
The remaining fixed-point conditions
give one intertwining condition at each path vertex.  For such a vertex $v$,
put
\[
  \mathcal I_{v,\boldsymbol{\theta}}
  :=\operatorname{Hom}_\Gamma
    (U_v^{\mathrm{bra}},U_v^{\mathrm{ket}})
  \subseteq
\operatorname{Hom}(W_v^{\mathrm{bra}},W_v^{\mathrm{ket}}).
\]
The space on the right contains all linear maps from the bra outer carrier
to the ket outer carrier.  The subscript $\Gamma$ on the left selects those
maps that are compatible with the two group actions, also called
\emph{intertwiners}.  Explicitly, $Y_v$ belongs to this space
when
\[
  U_v^{\mathrm{ket}}(g)Y_v
  =Y_vU_v^{\mathrm{bra}}(g)
  \qquad(g\in\Gamma).
\]

These intertwining conditions involve only the outer factors at the four
path vertices.  To describe the full fixed space, we also keep track of
the remaining vertex factors, on which $K_p$ acts trivially.  Accordingly,
we separate the full ordered block into its path and exterior parts:
\begin{equation}
  \label{eq:ordered-path-exterior-factorization}
  \mathcal X_{p,\boldsymbol{\theta}}
  =\mathcal X_{p,\boldsymbol{\theta}}^{\mathrm{path}}
   \otimes\mathcal X_{p,\boldsymbol{\theta}}^{\mathrm{ext}}.
\end{equation}
The first factor contains the operators on the path neighborhood;
the second contains the remaining vertex factors, on which $K_p$ acts
trivially.  \Cref{fig:ordered-path-kernel-factors} summarizes the local
intertwiner spaces and the exterior factor appearing in the following
kernel formula.

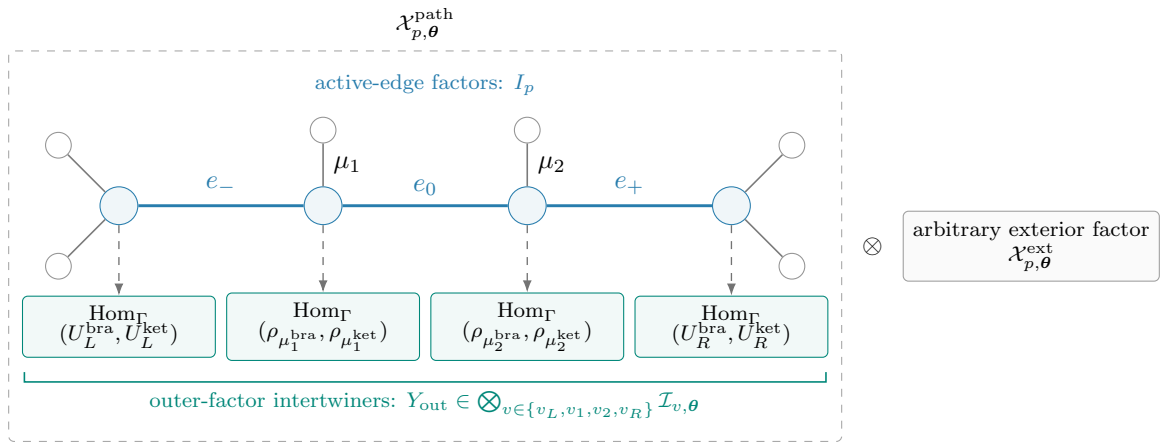
\begin{figure}[H]
  \centering
  \begin{tikzpicture}[
    path vertex/.style={circle, draw=figureGeometry, fill=figureGeometry!7,
      minimum size=5mm, inner sep=0pt},
    outer vertex/.style={circle, draw=black!42, fill=white,
      minimum size=3.4mm, inner sep=0pt},
    active edge/.style={figureGeometry, line width=1.1pt},
    outer edge/.style={black!52, line width=0.6pt},
    constraint/.style={draw=figureInterface, fill=figureInterface!5,
      rounded corners=2pt, minimum width=2.55cm, minimum height=0.62cm,
      align=center, font=\scriptsize},
    local frame/.style={draw=black!35, dashed, rounded corners=2pt,
      inner sep=5pt},
    map arrow/.style={-{Latex[length=1.6mm]}, black!55, dashed,
      line width=0.5pt},
    every node/.style={font=\small}
  ]
    \node[path vertex] (v0) at (-4.05,1.25) {};
    \node[path vertex] (v1) at (-1.35,1.25) {};
    \node[path vertex] (v2) at (1.35,1.25) {};
    \node[path vertex] (v3) at (4.05,1.25) {};
    \draw[active edge] (v0) -- node[above] {$e_-$} (v1);
    \draw[active edge] (v1) -- node[above] {$e_0$} (v2);
    \draw[active edge] (v2) -- node[above] {$e_+$} (v3);

    \node[outer vertex] (l1) at (-4.85,2.05) {};
    \node[outer vertex] (l2) at (-4.85,0.45) {};
    \draw[outer edge] (v0) -- (l1);
    \draw[outer edge] (v0) -- (l2);

    \node[outer vertex] (m1) at (-1.35,2.25) {};
    \node[outer vertex] (m2) at (1.35,2.25) {};
    \draw[outer edge] (v1) -- node[right, text=black] {$\mu_1$} (m1);
    \draw[outer edge] (v2) -- node[right, text=black] {$\mu_2$} (m2);

    \node[outer vertex] (r1) at (4.85,2.05) {};
    \node[outer vertex] (r2) at (4.85,0.45) {};
    \draw[outer edge] (v3) -- (r1);
    \draw[outer edge] (v3) -- (r2);

    \node[constraint] (hL) at (-4.05,-0.35)
      {$\operatorname{Hom}_\Gamma$\\[-2pt]
       $(U_L^{\mathrm{bra}},U_L^{\mathrm{ket}})$};
    \node[constraint] (h1) at (-1.35,-0.35)
      {$\operatorname{Hom}_\Gamma$\\[-2pt]
       $(\rho_{\mu_1^{\mathrm{bra}}},\rho_{\mu_1^{\mathrm{ket}}})$};
    \node[constraint] (h2) at (1.35,-0.35)
      {$\operatorname{Hom}_\Gamma$\\[-2pt]
       $(\rho_{\mu_2^{\mathrm{bra}}},\rho_{\mu_2^{\mathrm{ket}}})$};
    \node[constraint] (hR) at (4.05,-0.35)
      {$\operatorname{Hom}_\Gamma$\\[-2pt]
       $(U_R^{\mathrm{bra}},U_R^{\mathrm{ket}})$};
    \draw[map arrow] (v0) -- (hL.north);
    \draw[map arrow] (v1) -- (h1.north);
    \draw[map arrow] (v2) -- (h2.north);
    \draw[map arrow] (v3) -- (hR.north);

    \node[font=\scriptsize, text=figureGeometry] (activeLabel) at (0,2.85)
      {active-edge factors: $I_p$};
    \draw[figureInterface, line width=0.6pt]
      (-5.3,-0.95) -- (-5.3,-1.08) -- (5.3,-1.08) -- (5.3,-0.95);
    \node[font=\scriptsize, text=figureInterface] (outLabel) at (0,-1.38)
      {outer-factor intertwiners: $Y_{\mathrm{out}}\in
       \bigotimes_{v\in\{v_L,v_1,v_2,v_R\}}\mathcal I_{v,\boldsymbol{\theta}}$};
    \node[local frame, fit=(l1)(l2)(m1)(m2)(r1)(r2)(hL)(hR)(activeLabel)(outLabel),
      label={[font=\scriptsize]above:$\mathcal X_{p,\boldsymbol{\theta}}^{\mathrm{path}}$}]
      (local) {};
    \node[draw=black!38, fill=black!2, rounded corners=2pt,
      right=0.8cm of local, align=center, font=\scriptsize]
      (ext) {arbitrary exterior factor\\$\mathcal X_{p,\boldsymbol{\theta}}^{\mathrm{ext}}$};
    \node at ($(local.east)!0.5!(ext.west)$) {$\otimes$};
  \end{tikzpicture}
  \caption{Local factors entering the exact centered-path kernel.  The complete
  bath acts on the three blue path edges and leaves four equivariance
  conditions, one at each touching vertex.  On a fixed ordered block these
  become the four displayed intertwiner spaces, whose tensor product
  contains $Y_{\mathrm{out}}$.  Here ``out'' refers to the outer-edge
  factors owned by the four path vertices, inside the dashed path frame;
  it is distinct from the exterior factor ``ext''.  All remaining
  factors associated with the vertices form the arbitrary exterior factor on which $K_p$
  acts trivially.}
  \label{fig:ordered-path-kernel-factors}
\end{figure}

\begin{lemma}[Exact centered-path kernel]
\label{lem:exact-ordered-path-kernel}
On the ordered block $\mathcal X_{p,\boldsymbol{\theta}}$, the fixed space of $K_p$ is
\begin{equation}
  \label{eq:exact-ordered-path-kernel-with-spectator}
  \ker K_{p,\boldsymbol{\theta}}
  =
  \underbrace{
    I_p\otimes
    \bigotimes_{v\in\{v_L,v_1,v_2,v_R\}}
      \mathcal I_{v,\boldsymbol{\theta}}
  }_{\displaystyle \mathcal F_{p,\boldsymbol{\theta}}}
  \otimes\mathcal X_{p,\boldsymbol{\theta}}^{\mathrm{ext}}.
\end{equation}
\end{lemma}
Thus the active-edge factor is the identity, the four local vertex factors
belong to their intertwiner spaces, and all remaining degrees of freedom form
an arbitrary exterior factor.  In particular, this kernel is independent
of $t$ on every closed, bounded temperature interval.

The proof, including the coefficient-separation argument and the blockwise
identification of the fixed space, is given in
\cref{app:exact-ordered-path-kernel-proof}.

The order in $\operatorname{Hom}_\Gamma
(U_v^{\mathrm{bra}},U_v^{\mathrm{ket}})$ records that its operators map the
bra representation to the ket representation.  At the two internal vertices
$U_{v_i}^\bullet=\rho_{\mu_i^\bullet}$, whereas $U_L$ and $U_R$ combine two
outer-leg actions.  Thus the local factor
in \eqref{eq:exact-ordered-path-kernel-with-spectator} can also be written as
\begin{align}
  \label{eq:exact-ordered-path-kernel}
  \mathcal F_{p,\boldsymbol{\theta}}={}&
  I_p
  \otimes
  \operatorname{Hom}_\Gamma
    (U_L^{\mathrm{bra}},U_L^{\mathrm{ket}})
  \otimes
  \operatorname{Hom}_\Gamma
    (\rho_{\mu_1^{\mathrm{bra}}},\rho_{\mu_1^{\mathrm{ket}}})
  \otimes
  \operatorname{Hom}_\Gamma
    (\rho_{\mu_2^{\mathrm{bra}}},\rho_{\mu_2^{\mathrm{ket}}})
  \otimes
  \operatorname{Hom}_\Gamma
    (U_R^{\mathrm{bra}},U_R^{\mathrm{ket}}).
\end{align}

For the remainder of this section, only three consequences are used:
\begin{itemize}
  \item The kernel is known before any spectral minimum is taken.
  \item It is independent of $\beta$, or equivalently of $t$.
  \item By \eqref{eq:schur-intertwiner-dimension}, a ket/bra label
  mismatch at an internal outer leg makes the corresponding Hom factor vanish.
\end{itemize}

\subsubsection{A Uniform Local Gap Above the Fixed Space}
\label{sec:uniform-centered-path-quotient-gap}

We now combine the exact path kernel with the above-kernel Poincar\'e
formulation in \cref{sec:prelim-above-kernel-poincare}.  We use
\cref{lem:finite-compactness-fixed-kernel} to show that the positive gap
above this fixed kernel has a positive minimum over $t\in[0,t_*]$.

For each ordered block $\mathcal X_{p,\boldsymbol{\theta}}$ and each current
value of $t$, let $E_{p,\boldsymbol{\theta}}$ be the path-KMS orthogonal
projection onto $\mathcal F_{p,\boldsymbol{\theta}}$, and put
$E_{p,\boldsymbol{\theta}}^\perp:=I-E_{p,\boldsymbol{\theta}}$.  We suppress the
$t$-dependence of these KMS-orthogonal projections.  The ordered blocks are
KMS-orthogonal, and $K_p$ is block diagonal by
\eqref{eq:ordered-block-kp-direct-sum}.  Hence the blockwise kernel identity in
\cref{lem:exact-ordered-path-kernel} gives
\begin{equation}
  \label{eq:full-path-complement-from-ordered-blocks}
  E_p^\perp
  =\bigoplus_{\boldsymbol{\theta}}
    \left(
      E_{p,\boldsymbol{\theta}}^\perp
      \otimes I_{\mathrm{ext},\boldsymbol{\theta}}
    \right),
\end{equation}
where the direct sum is restricted to $L^2_0(\rho_\beta)$.  Equivalently,
$E_p$ is the restriction to the mean-zero space of the full-kernel
projection, because $I\in\ker K_p$ and the identity is KMS-orthogonal to
$L^2_0(\rho_\beta)$.

\begin{lemma}[Uniform local quotient gap]
\label{lem:uniform-ordered-path-quotient-gap}
Fix a nontrivial finite group $\Gamma$ and $0\le t_*<\infty$.  There is a constant
$g_{\Gamma,\mathrm{loc},*}>0$, depending only on $\Gamma$ and $t_*$, such
that every centered path satisfies
\begin{equation}
  \label{eq:common-ordered-packet-floor}
  K_p\succeq g_{\Gamma,\mathrm{loc},*}E_p^\perp.
\end{equation}
\end{lemma}

\begin{proof}
\medskip\noindent\textbf{Passage to Finite Path Problems.}\par\smallskip
Recall that \cref{lem:finite-compactness-fixed-kernel} gives, on a fixed
finite-dimensional space $V$,
\[
  \langle X,D_tX\rangle_t
  \ge g\,\operatorname{dist}_t(X,F)^2,
  \qquad X\in V,
\]
for a continuous family of positive operators with parameter-independent
kernel $F$ and positive-definite inner products.  Below,
$V=\mathcal X_{p,\boldsymbol{\theta}}^{\mathrm{path}}$,
$F=\mathcal F_{p,\boldsymbol{\theta}}$, and $\lVert\cdot\rVert_t$ is the path
KMS norm; $D_t$ will be the path operator defined in
\eqref{eq:ordered-block-generator-factorization}.  We first remove the
exterior factor and then verify these hypotheses.

Fix $\boldsymbol{\theta}\in\Theta_{\Gamma}$.  By
\eqref{eq:ordered-block-kp-direct-sum}, $K_p$ acts separately on
$\mathcal X_{p,\boldsymbol{\theta}}$ because every active coupling preserves
its six fixed outer labels.  Moreover,
$\Theta_{\Gamma}$ is finite, and changing $p$ only relabels these
outer slots.  It therefore suffices to control one finite path problem for
each block type, uniformly in $t\in[0,t_*]$.

The orientation pattern of the local edges does not create additional
spectral types.  On one regular edge, the unitary
\[
  W|h\rangle:=|h^{-1}\rangle
\]
satisfies
\[
  WL_gW^\dagger=R_g,
  \qquad
  WR_gW^\dagger=L_g,
  \qquad
  WQ_hW^\dagger=Q_{h^{-1}}.
\]
Conjugating by $W$ on every edge whose orientation is reversed maps the
local star Hamiltonian, its Gibbs factors, and the complete equal-clock bath
to those of a fixed reference orientation.  It also maps the corresponding
kernel projection to the reference kernel projection.  Thus all local
orientation signatures are unitarily equivalent in the KMS geometry, and
the constants below are uniform in both $p$ and its orientation pattern.

The ordered-block projections commute with the Gibbs state and give the
KMS-orthogonal decomposition
\[
  L^2(\rho_\beta)
  =
  \bigoplus_{\boldsymbol{\theta}}
  \left(
    \mathcal X_{p,\boldsymbol{\theta}}^{\mathrm{path}}
    \otimes
    \mathcal X_{p,\boldsymbol{\theta}}^{\mathrm{ext}}
  \right).
\]
The path factor contains the active edges and the factors owned by the four
path vertices; the rest is an exterior factor.
We next factor the Gibbs operator to obtain a tensor-product KMS inner
product, which will allow us to transfer a gap bound on the finite path
factor to the full ordered block.
An ordered operator block maps the bra sector to the ket sector, so it
need not itself carry a density.  Here we instead restrict the Gibbs
state to each sector separately: both projectors below refer to the
same sector, with $\bullet=\mathrm{ket}$ or $\mathrm{bra}$.

As in the endpoint/exterior decomposition used for the single-edge
comparison, the Hamiltonian on each sector splits into terms acting on
separate tensor factors:
\[
  \left.H_{\cG}\right|_{\operatorname{Ran}P_{p,\boldsymbol{\theta}^{\bullet}}}
  =H^{\bullet,\mathrm{path}}\otimes I_{\mathrm{ext}}
   +I_{\mathrm{path}}\otimes H^{\bullet,\mathrm{ext}}.
\]
Here the path Hamiltonian contains the four path-vertex terms and the
exterior Hamiltonian contains the remaining vertex terms, with the
chosen sector labels fixed.  Their exponentials therefore factor as
$e^{-\beta H^{\bullet,\mathrm{path}}}\otimes
e^{-\beta H^{\bullet,\mathrm{ext}}}$.
Absorbing the global factor $Z_\beta^{-1}$ into one tensor factor, write
\begin{equation}
  \label{eq:ordered-block-gibbs-factorization}
  P_{p,\boldsymbol{\theta}^{\bullet}}\rho_\beta
  P_{p,\boldsymbol{\theta}^{\bullet}}
  =\rho_{\boldsymbol{\theta},\beta}^{\bullet,\mathrm{path}}
   \otimes\rho_{\boldsymbol{\theta},\beta}^{\bullet,\mathrm{ext}}.
\end{equation}
These are unnormalized Gibbs factors: $P\rho_\beta P$ has not been
divided by its trace, and the two factors are not assumed to have unit
trace.  The factorization holds separately for the ket and bra sectors.
Every active Bohr component has the form
$S_\omega^{\mathrm{path}}\otimes I_{\mathrm{ext}}$.  Hence the restriction
in \eqref{eq:ordered-block-kp-direct-sum} acts trivially on the exterior
factor.  
After identifying the tensor factors of each path neighborhood with those of a fixed reference path, denote its positive
path factor on $\mathcal X_{p,\boldsymbol{\theta}}^{\mathrm{path}}$ by
$K^{\mathrm{path}}_{\boldsymbol{\theta},t}$, so that
\begin{equation}
  \label{eq:ordered-block-generator-factorization}
  K_{p,\boldsymbol{\theta}}
  =K^{\mathrm{path}}_{\boldsymbol{\theta},t}\otimes I_{\mathrm{ext}}.
\end{equation}
Here the $t$-dependence inherited from the Davies rates is displayed because
we will take a minimum over $t\in[0,t_*]$; it remains suppressed in
$K_p$ and $K_{p,\boldsymbol{\theta}}$.  The path label $p$ is omitted from the
local factor because changing the centered path only relabels the same fixed
local slots.  This is the family $D_t$ in the uniform gap estimate above.

The Gibbs factorization \eqref{eq:ordered-block-gibbs-factorization} makes the
KMS inner product on this ordered block the tensor product of its path and
exterior inner products.  Together with
\eqref{eq:ordered-block-generator-factorization}, this shows that the exterior
factor only repeats eigenvalues: $K_{p,\boldsymbol{\theta}}$ and
$K^{\mathrm{path}}_{\boldsymbol{\theta},t}$ have the same nonzero spectrum, up to
multiplicity.

\medskip\noindent\textbf{A Uniform Gap for Each Block Type.}\par\smallskip
The exact kernel identity
\eqref{eq:exact-ordered-path-kernel-with-spectator} and the factorization
\eqref{eq:ordered-block-generator-factorization} imply
\[
  \ker K^{\mathrm{path}}_{\boldsymbol{\theta},t}
  =\mathcal F_{p,\boldsymbol{\theta}}
  \qquad(0\le t\le t_*).
\]
This kernel is independent of $t$.  For every
$X\in\mathcal X_{p,\boldsymbol{\theta}}^{\mathrm{path}}$, orthogonal projection
gives
\[
  \operatorname{dist}_t
    (X,\mathcal F_{p,\boldsymbol{\theta}})^2
  =\lVert E_{p,\boldsymbol{\theta}}^\perp X\rVert_t^2
  =\langle X,E_{p,\boldsymbol{\theta}}^\perp X\rangle_t.
\]
The path KMS inner product and the
quadratic form of $K^{\mathrm{path}}_{\boldsymbol{\theta},t}$ vary continuously
on $[0,t_*]$.  The Gibbs factors are strictly positive, so the path KMS
inner product is positive definite.  For each block type with a nonzero
kernel complement, the gap above the fixed kernel is positive and varies
continuously with $t$.  Its minimum on the closed, bounded interval
$[0,t_*]$ is therefore strictly positive.  By
\cref{lem:finite-compactness-fixed-kernel}, this gives
\[
  K^{\mathrm{path}}_{\boldsymbol{\theta},t}
  \succeq
  g_{\boldsymbol{\theta},*}E_{p,\boldsymbol{\theta}}^\perp
  \qquad(0\le t\le t_*)
\]
for some $g_{\boldsymbol{\theta},*}>0$.
Let $\Theta_{\Gamma}^+
\subseteq\Theta_{\Gamma}$ be the
block types whose kernel complement is nonzero.  Taking the minimum of the
blockwise lower bounds gives
\begin{equation}
  \label{eq:ordered-path-compactness-floor}
  g_{\Gamma,\mathrm{loc},*}
  :=\min_{\boldsymbol{\theta}\in\Theta_{\Gamma}^+}g_{\boldsymbol{\theta},*}>0.
\end{equation}

\Needspace{7\baselineskip}
\medskip\noindent\textbf{Assembly of the Blocks.}\par\smallskip
Since \(K_{p,\boldsymbol{\theta}}=K^{\mathrm{path}}_{\boldsymbol{\theta},t}\otimes I_{\mathrm{ext}}\) and the KMS inner product factorizes, the same lower bound holds on the full ordered block.  Using the KMS-orthogonal block decomposition and
\eqref{eq:full-path-complement-from-ordered-blocks}, we obtain
\[
  K_p
  \succeq
  g_{\Gamma,\mathrm{loc},*}
  \bigoplus_{\boldsymbol{\theta}}
  \left(
    E_{p,\boldsymbol{\theta}}^\perp
    \otimes I_{\mathrm{ext},\boldsymbol{\theta}}
  \right)
  =g_{\Gamma,\mathrm{loc},*}E_p^\perp.
\]
This is \eqref{eq:common-ordered-packet-floor}.
\end{proof}

\subsubsection{Global Cover by Local Fixed-Space Complements}
\label{sec:global-cover-local-fixed-complements}

\begin{lemma}[Overlap of local fixed-space complements]
\label{lem:cubic-defect-frame}
The orthogonal-complement projections $E_p^\perp=I-E_p$ of the local fixed
spaces satisfy
\begin{equation}
  \label{eq:common-cubic-defect-frame}
  \Pi^\perp\Big(\sum_{p\in\PathFam}E_p^\perp\Big)\Pi^\perp
  \succeq6\Pi^\perp.
\end{equation}
\end{lemma}

\begin{proof}
Set
\[
  \mathcal C_{\mathrm{cov}}
  :=\Pi^\perp\Big(\sum_{p\in\PathFam}E_p^\perp\Big)\Pi^\perp.
\]
For $X\in\operatorname{ran}\Pi^\perp$, the quadratic form of this operator is
$\sum_{p\in\PathFam}\lVert E_p^\perp X\rVert_{\rho_\beta}^2$.
We decompose the observable space into two types of
Peter--Weyl blocks; an arbitrary observable may have components of both
types.  On a global Peter--Weyl block
\[
  P_{\boldsymbol\lambda^{\mathrm{ket}}}
  \mathcal B(\mathcal H_{\cG})
  P_{\boldsymbol\lambda^{\mathrm{bra}}},
\]
we define the terms \emph{matching} and \emph{mismatched} by
\[
  \begin{aligned}
  \text{matching:}\quad
  &\boldsymbol\lambda^{\mathrm{ket}}
    =\boldsymbol\lambda^{\mathrm{bra}}
  &&\Longleftrightarrow\quad
    \lambda_e^{\mathrm{ket}}=\lambda_e^{\mathrm{bra}}
    \quad\text{for every }e,\\
  \text{mismatched:}\quad
  &\boldsymbol\lambda^{\mathrm{ket}}
    \ne\boldsymbol\lambda^{\mathrm{bra}}
  &&\Longleftrightarrow\quad
    \lambda_e^{\mathrm{ket}}\ne\lambda_e^{\mathrm{bra}}
    \quad\text{for some }e.
  \end{aligned}
\]
In the matching case, write the common tuple as $\boldsymbol\lambda$.
Fixing it does not force an observable
$X\in P_{\boldsymbol\lambda}\mathcal B(\mathcal H_{\cG})
P_{\boldsymbol\lambda}$ to lie in $\mathbb C P_{\boldsymbol\lambda}$: it may
still be non-scalar in a matrix or fusion factor associated with a vertex.  Thus
``matching'' does not mean label-only.  Once we restrict to the orthogonal
complement of the interface, $\operatorname{ran}\Pi^\perp$, as we do below, its matching part consists
precisely of the non-label directions inside the diagonal Peter--Weyl
blocks.  

To separate the matching and mismatched parts, introduce the global
Peter--Weyl pinching
\[
  \mathcal D(X)
  :=\sum_{\boldsymbol\lambda\in\widehat\Gamma^{E(\cG)}}
    P_{\boldsymbol\lambda}XP_{\boldsymbol\lambda}.
\]
The map $\mathcal D$ keeps the diagonal Peter--Weyl blocks and deletes all
mismatched ones.  Because every $P_{\boldsymbol\lambda}$ commutes with
$\rho_\beta$, it is a KMS-orthogonal projection.  Moreover, $\mathcal D$
commutes with $\Pi$, which extracts the block-scalar label part from the
matching space.  Hence
\begin{equation}
  \label{eq:matching-mismatched-transverse-projections}
  \Pi^\perp_{=}:=\mathcal D\Pi^\perp,
  \qquad
  \Pi^\perp_{\neq}:=(I-\mathcal D)\Pi^\perp
\end{equation}
are the KMS-orthogonal projections onto the matching and mismatched classes
inside $\Pi^\perp$.  In particular,
\[
  \operatorname{ran}\Pi^\perp_{=}
  =\operatorname{ran}\mathcal D\cap\operatorname{ran}\Pi^\perp
\]
is exactly the diagonal non-label sector, and
\begin{equation}
  \label{eq:diag-rect-reducing-split}
  \Pi^\perp=\Pi^\perp_{=}+\Pi^\perp_{\neq}.
\end{equation}

To check compatibility with the local fixed spaces, factor this global map
into the edgewise pinchings
\[
  \mathcal D_e(X)
  :=\sum_{\lambda_e\in\widehat\Gamma}
    P_{\lambda_e}XP_{\lambda_e},
  \qquad
  \mathcal D=\prod_{e\in E(\cG)}\mathcal D_e.
\]
The only consequence needed below is the following compatibility statement.

\begin{claim}[Pinching compatibility]
\label{clm:edgewise-pinching-compatibility}
For every $p\in\PathFam$ and $e\in E(\cG)$,
\[
  [E_p,\mathcal D_e]=0,
  \qquad
  [E_p,\mathcal D]=0.
\]
\end{claim}

\begin{proof}[Proof of the claim]
The ordered blocks are KMS-orthogonal.  By
\eqref{eq:ordered-block-kp-direct-sum},
\[
  \ker K_p
  =\bigoplus_{\boldsymbol\theta}\ker K_{p,\boldsymbol\theta}.
\]
Recall from \eqref{eq:exact-ordered-path-kernel-with-spectator} that, on each
ordered block,
\[
  \ker K_{p,\boldsymbol\theta}
  =I_p\otimes\left(
    \bigotimes_{v\in\{v_L,v_1,v_2,v_R\}}
      \mathcal I_{v,\boldsymbol\theta}
    \right)
    \otimes\mathcal X_{p,\boldsymbol\theta}^{\mathrm{ext}}.
\]
Pinching an
active edge fixes this identity.  On an outer edge it preserves the local
intertwiner condition, and on an exterior edge it acts within the arbitrary
spectator factor.  Taking the direct sum over the ordered blocks gives
\[
  \mathcal D_e(\ker K_p)\subseteq\ker K_p
  \qquad(e\in E(\cG),\ p\in\PathFam).
\]
Since $\mathcal D_e$ is an orthogonal projection, this invariance implies
$[E_p,\mathcal D_e]=0$.  Since the edgewise pinchings commute and their
product is $\mathcal D$, it also implies $[E_p,\mathcal D]=0$.
\end{proof}

Since $[\mathcal D,\Pi^\perp]=0$, \Cref{clm:edgewise-pinching-compatibility}
implies
\begin{equation}
  \label{eq:cover-commutes-with-matching-split}
  [\mathcal C_{\mathrm{cov}},\Pi^\perp_{=}]=0,
  \qquad
  [\mathcal C_{\mathrm{cov}},\Pi^\perp_{\neq}]=0.
\end{equation}
Thus the cover operator is block diagonal with respect to
\eqref{eq:diag-rect-reducing-split}.

The two sectors have different local witnesses.  A matching observable in
$\operatorname{ran}\Pi^\perp$ is non-scalar on some factor $w$ associated with a vertex, whereas a mismatched
block contains an edge $e$ with unequal ket and bra labels.  We fix such a
$w$ or $e$ and count the centered paths that detect it; the uniform
degree-three counts are shown in
\Cref{fig:transverse-cover-incidence-counts}.

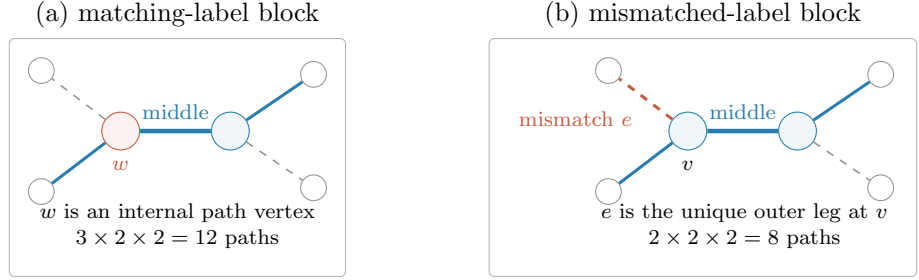
\begin{figure}[H]
  \centering
  \begin{tikzpicture}[
    vertex/.style={circle, draw=figureGeometry, fill=figureGeometry!7,
      minimum size=5mm, inner sep=0pt},
    witness vertex/.style={vertex, draw=figureCoupling, fill=figureCoupling!8},
    outer vertex/.style={circle, draw=black!42, fill=white,
      minimum size=3.4mm, inner sep=0pt},
    active edge/.style={figureGeometry, line width=1.15pt},
    middle edge/.style={figureGeometry, line width=1.7pt},
    unused edge/.style={black!42, dashed, line width=0.6pt},
    mismatch edge/.style={figureCoupling, dashed, line width=1.15pt},
    panel/.style={draw=black!30, rounded corners=2pt, inner sep=7pt},
    every node/.style={font=\small}
  ]
    \begin{scope}[xshift=-4.15cm]
      \node[witness vertex,
        label={[font=\scriptsize, text=figureCoupling]below:$w$}]
        (vA) at (0,0) {};
      \node[vertex] (wA) at (1.45,0) {};
      \node[outer vertex] (lA) at (-1.05,-0.8) {};
      \node[outer vertex] (oA) at (-1.05,0.8) {};
      \node[outer vertex] (zA) at (2.55,0.75) {};
      \node[outer vertex] (qA) at (2.55,-0.75) {};
      \draw[active edge] (lA) -- (vA);
      \draw[middle edge] (vA) -- node[above, font=\scriptsize] {middle} (wA);
      \draw[active edge] (wA) -- (zA);
      \draw[unused edge] (vA) -- (oA);
      \draw[unused edge] (wA) -- (qA);
      \node[font=\scriptsize, align=center] (noteA) at (0.75,-1.25)
        {$w$ is an internal path vertex\\
         $3\times2\times2=12$ paths};
      \node[panel, fit=(lA)(oA)(zA)(qA)(noteA),
        label={[font=\small]above:(a) matching-label block}] {};
    \end{scope}

    \begin{scope}[xshift=3.35cm]
      \node[vertex, label={[font=\scriptsize]below:$v$}] (vB) at (0,0) {};
      \node[vertex] (wB) at (1.45,0) {};
      \node[outer vertex] (lB) at (-1.05,-0.8) {};
      \node[outer vertex] (eB) at (-1.05,0.8) {};
      \node[outer vertex] (zB) at (2.55,0.75) {};
      \node[outer vertex] (qB) at (2.55,-0.75) {};
      \draw[active edge] (lB) -- (vB);
      \draw[middle edge] (vB) -- node[above, font=\scriptsize] {middle} (wB);
      \draw[active edge] (wB) -- (zB);
      \draw[mismatch edge] (vB) --
        node[pos=.55, below left, yshift=-1pt,
          text=figureCoupling, font=\scriptsize] (mismatchLabel)
        {mismatch $e$} (eB);
      \draw[unused edge] (wB) -- (qB);
      \node[font=\scriptsize, align=center] (noteB) at (0.75,-1.25)
        {$e$ is the unique outer leg at $v$\\
         $2\times2\times2=8$ paths};
      \node[panel, fit=(lB)(eB)(zB)(qB)(noteB)(mismatchLabel),
        label={[font=\small]above:(b) mismatched-label block}] {};
    \end{scope}
  \end{tikzpicture}
  \caption{The two incidence counts used below.  Blue edges form the
  centered three-edge path and the thick blue edge is its middle edge; gray
  dashed edges are outside that path.  In the matching-label case, the fixed
  vertex is internal in $3\cdot2\cdot2=12$ paths.  In the mismatched-label
  case, the fixed mismatched edge is the unique outer edge at an internal
  vertex in $2\cdot2\cdot2=8$ paths.}
  \label{fig:transverse-cover-incidence-counts}
\end{figure}

\medskip\noindent\textbf{Matching Labels: Detecting a Non-scalar Vertex
Factor.}\par\smallskip
Fix a complete global label configuration $\boldsymbol\lambda$.  Its
fixed-label Hilbert block, identified with
$\cH_{\boldsymbol\lambda}=P_{\boldsymbol\lambda}\mathcal H_{\cG}$, regroups as
\[
  \cH_{\boldsymbol\lambda}
  \overset{\eqref{eq:edge-to-vertex-regrouping}}{\cong}
  \bigotimes_{v\in V(\cG)}\mathcal V_v(\boldsymbol\lambda).
\]
In finite dimensions, the operator space on a tensor product is the
tensor product of the individual operator spaces.  Applying this to the
vertex decomposition above gives the matching operator block
\[
  P_{\boldsymbol\lambda}\mathcal B(\mathcal H_{\cG})
    P_{\boldsymbol\lambda}
  \cong\mathcal B(P_{\boldsymbol\lambda}\mathcal H_{\cG})
  \cong
  \bigotimes_{v\in V(\cG)}
    \mathcal B\bigl(\mathcal V_v(\boldsymbol\lambda)\bigr).
\]
As in the product formula
\eqref{eq:global-label-distribution}, the conditional Gibbs density has
local factors proportional to $I+tA_v$.  To pass from the exact kernels of
$K_{p,\boldsymbol{\theta}}$ to the fixed-space cover on $\operatorname{ran}\Pi^\perp$, we use the local analogue
of the mean-zero projection $X\mapsto X-\operatorname{Tr}(\rho X)I$.  For each
vertex $w\in V(\cG)$, define
\[
  R^{\mathrm{loc}}_{w,\boldsymbol\lambda}(X_w)
  :=X_w-
  \frac{\operatorname{Tr}_{\mathcal V_w(\boldsymbol\lambda)}
    \bigl((I+tA_w)X_w\bigr)}
       {\operatorname{Tr}_{\mathcal V_w(\boldsymbol\lambda)}(I+tA_w)}
  I_{\mathcal V_w(\boldsymbol\lambda)},
  \qquad
  X_w\in\mathcal B\bigl(\mathcal V_w(\boldsymbol\lambda)\bigr),
\]
where $A_w$ is restricted to $\mathcal V_w(\boldsymbol\lambda)$.
Thus $R^{\mathrm{loc}}_{w,\boldsymbol\lambda}$ removes the local scalar
component and detects non-label matrix directions at $w$.  It is the local
KMS-orthogonal projection onto the complement of the scalar line.  Extend it
to the full matching block by leaving every other vertex factor unchanged:
\[
  R_{w,\boldsymbol\lambda}
  :=R^{\mathrm{loc}}_{w,\boldsymbol\lambda}
    \otimes
    \bigotimes_{v\ne w}
      \operatorname{id}_{\mathcal B(\mathcal V_v(\boldsymbol\lambda))}.
\]
Finally, define a path-independent projection on the full GNS space by
\[
  R_w:=\bigoplus_{\boldsymbol\lambda\in\widehat\Gamma^E}
       R_{w,\boldsymbol\lambda}
\]
on the matching space, and extend it by zero on the mismatched space.  Since
$R_wI=0$, it restricts to the global mean-zero space $L^2_0(\rho_\beta)$ used
below.

\begin{claim}[Matching-sector vertex detectors]
\label{clm:matching-sector-vertex-detectors}
The projections $R_w$ commute and satisfy
\[
  \sum_{w\in V(\cG)}R_w\succeq\Pi^\perp_{=},
  \qquad
  E_p^\perp\succeq
  \frac12\bigl(R_{v_1(p)}+R_{v_2(p)}\bigr)
  \quad(p\in\PathFam),
\]
where $v_1(p)$ and $v_2(p)$ are the two internal vertices of $p$.
\end{claim}

\begin{proof}
Recall from \eqref{eq:matching-mismatched-transverse-projections} that
$\Pi^\perp_{=}=\mathcal D\Pi^\perp$.  Since
$\Pi^\perp=I-\Pi$ on $L^2_0(\rho_\beta)$ and
$\mathcal D\Pi=\Pi$,
\[
  \Pi^\perp_{=}=\mathcal D-\Pi.
\]
$R_w$ commute because they act on distinct vertex factors.  A matching
observable $X$ is label-only precisely when each fixed-label block is scalar:
\[
  X\in\mathcal M
  \quad\Longleftrightarrow\quad
  P_{\boldsymbol\lambda}XP_{\boldsymbol\lambda}
    =f(\boldsymbol\lambda)P_{\boldsymbol\lambda}
  \ \text{for every }\boldsymbol\lambda
  \quad\Longleftrightarrow\quad
  R_wX=0
  \ \text{for every }w.
\]
Because $R_w$ commute, $\prod_w(I-R_w)$ is the projection onto their
common kernel.  Concretely, $\mathcal D$ first removes the mismatched blocks.
On each remaining fixed-label block, applying $\prod_w(I-R_w)$
replaces every vertex factor by its Gibbs-weighted scalar part,
leaving a scalar multiple of $P_{\boldsymbol\lambda}$.
Thus the resulting observable depends only on the labels.
Since $\mathcal D$ and the $I-R_w$ are commuting KMS-orthogonal
projections and preserve the mean-zero space, their product is precisely
the orthogonal projection $\Pi$ onto $\mathcal M\cap L^2_0(\rho_\beta)$:
\[
  \prod_{w\in V(\cG)}(I-R_w)\mathcal D=\Pi.
\]
Simultaneous diagonalization of $R_w$, together with
$1-\prod_w(1-r_w)\le\sum_w r_w$ for $r_w\in\{0,1\}$, now gives
\[
  \Pi^\perp_{=}
  =\left[I-\prod_{w\in V(\cG)}(I-R_w)\right]\mathcal D
  \preceq\sum_{w\in V(\cG)}R_w.
\]

For a centered path $p$, the exact kernel formula
\eqref{eq:exact-ordered-path-kernel} contains the identity $I_p$ on the
three active edges.  At each internal vertex $v_i(p)$, the single outer
factor has the same ket and bra label $\mu_i$ on a matching component.
Schur's lemma, in the form \eqref{eq:schur-intertwiner-dimension}, gives
\[
  \operatorname{Hom}_\Gamma(V_{\mu_i},V_{\mu_i})
  =\mathbb C I_{V_{\mu_i}},
  \qquad i=1,2.
\]
Thus both the active-edge factors and the outer factor at each internal
vertex are proportional to the identity.  On each matching component,
every fixed observable is therefore scalar on the full vertex factor
at $v_1(p)$ and $v_2(p)$.
Since $R_w$ removes the scalar component at $w$ (and is zero on
mismatched components), it annihilates the range of $E_p$ at these vertices:
\[
  R_{v_1(p)}E_p=R_{v_2(p)}E_p=0.
\]
Equivalently, $\operatorname{ran}E_p$ lies in
$\ker R_{v_1(p)}\cap\ker R_{v_2(p)}$.  Since $R_w$ commute, the projection
onto this intersection is $(I-R_{v_1(p)})(I-R_{v_2(p)})$, and hence
\[
  E_p\preceq(I-R_{v_1(p)})(I-R_{v_2(p)}).
\]
Taking complements gives
\[
  E_p^\perp
  \succeq R_{v_1(p)}+R_{v_2(p)}-R_{v_1(p)}R_{v_2(p)}
  \succeq\frac12\bigl(R_{v_1(p)}+R_{v_2(p)}\bigr),
\]
where the last step uses
$R_{v_1(p)}+R_{v_2(p)}-2R_{v_1(p)}R_{v_2(p)}
=(R_{v_1(p)}-R_{v_2(p)})^2\succeq0$.
\end{proof}

Every graph vertex is internal in twelve centered paths, as illustrated in
panel~(a) of \Cref{fig:transverse-cover-incidence-counts}.  Summing
the second inequality in \Cref{clm:matching-sector-vertex-detectors} over
$p$ and then using the first yields
\[
  \sum_{p\in\PathFam}E_p^\perp
  \succeq6\sum_{w\in V(\cG)}R_w
  \succeq6\Pi^\perp_{=}.
\]
Consequently,
\begin{equation}
  \label{eq:diagonal-defect-cover}
  \Pi^\perp_{=}\mathcal C_{\mathrm{cov}}
  \Pi^\perp_{=}
  \succeq6\Pi^\perp_{=}.
\end{equation}

\medskip\noindent\textbf{Mismatched Labels: Detecting an Edge-label
Mismatch.}\par\smallskip
The commuting projections $\mathcal D_e$ refine the mismatched space into
orthogonal sectors indexed by the nonempty set of edges on which the ket and
bra labels differ.  By
\Cref{clm:edgewise-pinching-compatibility}, $E_p$ commutes with the
projection onto each such sector.  Fix one mismatch set, choose $e$ in that
set, and let $Q$ be the projection onto the resulting mismatch sector.  Also
write
\[
  \PathFam(e)
  :=\left\{
    p\in\PathFam:
    \text{$e$ is the unique outer edge at an internal vertex of $p$}
  \right\}.
\]
For every $p\in\PathFam(e)$, the label of $e$ is among the fixed outer data
and is mismatched.  Schur's lemma in
\eqref{eq:schur-intertwiner-dimension} therefore makes the corresponding internal
factor $\mathcal I_{v,\boldsymbol{\theta}}$ in
\eqref{eq:exact-ordered-path-kernel-with-spectator} vanish.  On this mismatch
sector,
\[
  E_pQ=0,
  \qquad
  E_p^\perp Q=Q
  \qquad (p\in\PathFam(e)).
\]
The corresponding count, illustrated in panel~(b) of
\Cref{fig:transverse-cover-incidence-counts}, is
\[
  |\PathFam(e)|
  =\underbrace{2}_{\text{endpoint of $e$,}}
   \underbrace{2}_{\text{middle edge,}}
  \underbrace{2}_{\text{continuation}}
  =8.
\]
Because $Q\preceq\Pi^\perp_{\neq}\preceq\Pi^\perp$, the definition
$\mathcal C_{\mathrm{cov}}
=\Pi^\perp(\sum_{p\in\PathFam}E_p^\perp)\Pi^\perp$ gives
\[
  Q\mathcal C_{\mathrm{cov}}Q
  \succeq
  \sum_{p\in\PathFam(e)}QE_p^\perp Q
  =|\PathFam(e)|Q
  =8Q.
\]
The same bound holds on every nonempty mismatch sector.  Taking their
orthogonal direct sum therefore yields
\begin{equation}
  \label{eq:rectangular-defect-cover}
  \Pi^\perp_{\neq}\mathcal C_{\mathrm{cov}}
  \Pi^\perp_{\neq}
  \succeq8\Pi^\perp_{\neq}.
\end{equation}

\medskip\noindent\textbf{Orthogonal Recombination.}\par\smallskip
By \eqref{eq:cover-commutes-with-matching-split}, the cover operator is block
diagonal across the matching/mismatched split.  For $X\in\operatorname{ran}\Pi^\perp$,
set
$X^{=}:=\Pi^\perp_{=}X$ and
$X^{\neq}:=\Pi^\perp_{\neq}X$.  Then
$X=X^{=}+X^{\neq}$ is the corresponding orthogonal decomposition, and
\begin{align*}
  \langle X,\mathcal C_{\mathrm{cov}}X\rangle
  &={}
  \langle X^{=},
    \mathcal C_{\mathrm{cov}}X^{=}\rangle
  +\langle X^{\neq},
    \mathcal C_{\mathrm{cov}}X^{\neq}\rangle \\
  &\ge 6\|X^{=}\|^2+8\|X^{\neq}\|^2
   \ge 6\|X\|^2.
\end{align*}
In particular, the two blockwise bounds combine by taking their minimum;
no additional loss is incurred in passing from the two orthogonal classes to
the full complementary subspace $\operatorname{ran}\Pi^\perp$.  Equivalently,
\[
  \Pi^\perp\Big(\sum_{p\in\PathFam}E_p^\perp\Big)\Pi^\perp\succeq6\Pi^\perp.
\]
This is \eqref{eq:common-cubic-defect-frame}.
\end{proof}

Thus an individual path may have nonzero fixed observables orthogonal to the
interface, but no nonzero observable in $\operatorname{ran}\Pi^\perp$ can
remain fixed for all centered paths.  This
is why the proof needs the overlap frame rather than a gap for each $K_p$ on
the whole local subspace orthogonal to the interface.

Thus \cref{lem:uniform-ordered-path-quotient-gap} verifies Model
Condition~(ii) of \cref{thm:abstract-two-axis-gap-assembly}, while
\cref{lem:cubic-defect-frame} verifies Geometric Condition~\textnormal{(G1)},
with
\begin{equation}
  \label{eq:common-fibre-constants}
  g_{\mathrm{loc}}=g_{\Gamma,\mathrm{loc},*},
  \qquad
  \kappa_{\mathrm{cov}}=6.
\end{equation}

\subsection{The Schur Condition: Middle-Edge Decoupling}
\label{sec:common-schur-certificate}

This subsection verifies Condition~(iii) of \cref{thm:abstract-two-axis-gap-assembly}.  We
first show that the label projection of the centered path fixed space is
invisible to the middle-edge update, and then combine this shielding property
with the local quotient gap to obtain the required shorted-form bound.  The
only representation-theoretic output needed in the second step is the
following shielding inclusion.  Only
\eqref{eq:middle-edge-shielding} is used in the later arguments and the proof of the following proposition may be skipped on a first reading.

Recall that $E_p$ projects onto the mean-zero fixed space of the three-edge
path operator, while $\Pi$ projects onto mean-zero observables depending
only on the irrep labels.  Thus $\operatorname{ran}(\Pi E_p)$ is the
irrep-label component of the path fixed space.

\begin{proposition}[Middle-edge decoupling]
\label{prop:middle-edge-shielding}
For every fixed finite group and every centered path,
\begin{equation}
  \label{eq:middle-edge-shielding}
  \operatorname{ran}(\Pi E_p)
  \subseteq\ker A_{\operatorname{mid}(p)}.
\end{equation}
\end{proposition}

\begin{proof}
The first step proves cancellation of the middle-label dependence under
$\mathcal E_{\mathrm{lab}}$; the second converts that cancellation into
fixedness for the physical middle-edge update.  We begin with the
representation-theoretic notation used in the cancellation step.

The cancellation argument uses only Peter--Weyl blocks diagonal in their
ket/bra labels.  By \eqref{eq:endpoint-outer-actions}, matching labels
canonically identify the endpoint actions as
\[
  U_L^{\mathrm{bra}}=U_L^{\mathrm{ket}}=:U_L,
  \qquad
  U_R^{\mathrm{bra}}=U_R^{\mathrm{ket}}=:U_R.
\]
We likewise write $W_{v_L}$ and $W_{v_R}$ for their common carrier spaces.
Thus $U_L$ and $U_R$ are the tensor-product representations carried jointly
by the two frozen outer legs at the left and right endpoints.

Orient the central path from left to right and write its labels as
$\lambda_{e_-},\lambda_{e_0},\lambda_{e_+}$.  Let $\mu_1,\mu_2$ be the single outer
irreps at the two internal vertices.  The geometric placement of these
labels and of the endpoint representations is shown in
\Cref{fig:path-channel-labels}.

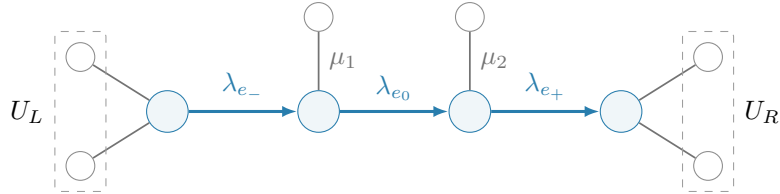
\begin{figure}[H]
  \centering
  \begin{tikzpicture}[
    path vertex/.style={circle, draw=figureGeometry, fill=figureGeometry!7,
      minimum size=5.5mm, inner sep=0pt},
    outer vertex/.style={circle, draw=black!45, fill=white,
      minimum size=3.8mm, inner sep=0pt},
    path edge/.style={-{Latex[length=2mm]}, figureGeometry,
      line width=1.15pt},
    outer edge/.style={black!55, line width=0.65pt},
    endpoint representation/.style={draw=black!38, dashed, rounded corners=2pt,
      inner sep=4pt},
    every node/.style={font=\small}
  ]
    \node[path vertex] (v0) at (-3,0) {};
    \node[path vertex] (v1) at (-1,0) {};
    \node[path vertex] (v2) at (1,0) {};
    \node[path vertex] (v3) at (3,0) {};

    \draw[path edge] (v0) -- node[above] {$\lambda_{e_-}$} (v1);
    \draw[path edge] (v1) -- node[above] {$\lambda_{e_0}$} (v2);
    \draw[path edge] (v2) -- node[above] {$\lambda_{e_+}$} (v3);

    \node[outer vertex] (mu1) at (-1,1.25) {};
    \node[outer vertex] (mu2) at (1,1.25) {};
    \draw[outer edge] (v1) -- node[right] {$\mu_1$} (mu1);
    \draw[outer edge] (v2) -- node[right] {$\mu_2$} (mu2);

    \node[outer vertex] (la1) at (-4.15,0.72) {};
    \node[outer vertex] (la2) at (-4.15,-0.72) {};
    \draw[outer edge] (v0) -- (la1);
    \draw[outer edge] (v0) -- (la2);
    \node[endpoint representation, fit=(la1)(la2),
      label={[font=\small]left:$U_L$}] {};

    \node[outer vertex] (rb1) at (4.15,0.72) {};
    \node[outer vertex] (rb2) at (4.15,-0.72) {};
    \draw[outer edge] (v3) -- (rb1);
    \draw[outer edge] (v3) -- (rb2);
    \node[endpoint representation, fit=(rb1)(rb2),
      label={[font=\small]right:$U_R$}] {};
  \end{tikzpicture}
  \caption{Representation labels around one centered three-edge path.  The
  blue arrows are the oriented path labels
  $\lambda_{e_-},\lambda_{e_0},\lambda_{e_+}$.  Each internal vertex has one outer irrep
  $\mu_i$, whereas the two frozen outer legs at the left and right endpoints
  jointly carry $U_L$ and $U_R$, respectively.  Their later isotypic
  decompositions and possible multiplicities are suppressed in the schematic.
  Orientation-dependent duals are displayed explicitly in the formulas below
  and suppressed in this schematic.}
  \label{fig:path-channel-labels}
\end{figure}

Write $d_\nu:=\dim V_\nu$ for the dimension of an irrep and recall that
$t=e^{\beta J}-1$.  For a candidate middle-edge label $\nu$, put
\begin{align*}
  m_1(\nu)&:=\dim\operatorname{Inv}
    (V_{\lambda_{e_-}}^*\otimes V_\nu\otimes V_{\mu_1}),
  &m_2(\nu)&:=\dim\operatorname{Inv}
    (V_\nu^*\otimes V_{\lambda_{e_+}}\otimes V_{\mu_2}).
\end{align*}
They are the fusion multiplicities at the two internal vertices in the sense
of \cref{sec:prelim-invariant-spaces-intertwiners}, and are the only
representation-theoretic data needed to track the dependence on the middle
label.

\medskip\noindent\textbf{Cancellation of the Middle-label Factor.}\par\smallskip
We first show that, for every $X\in\ker K_p$, the label observable
$\mathcal E_{\mathrm{lab}}X$ is independent of the middle-edge irrep label
$\lambda_{e_0}$.
Fix $\boldsymbol{\lambda}_{\setminus e_0}$.  For a candidate value for the middle edge $\nu\in\widehat\Gamma$, let
$\boldsymbol{\lambda}^{(\nu)}$ be its completion with
$\lambda_{e_0}^{(\nu)}=\nu$ and
$\lambda_e^{(\nu)}=\lambda_e$ for $e\ne e_0$.  Because
$P_{\boldsymbol{\lambda}^{(\nu)}}$ commutes with the
Gibbs state, \eqref{eq:label-coarse-graining-formula} gives, at the completion
$\boldsymbol{\lambda}^{(\nu)}$,
\begin{equation}
  \label{eq:label-kms-projection-formula}
  (\mathcal E_{\mathrm{lab}}X)(\boldsymbol{\lambda}^{(\nu)})
  =
  \frac{\operatorname{Tr}
    (\rho_\beta P_{\boldsymbol{\lambda}^{(\nu)}}X)}
       {\operatorname{Tr}
    (\rho_\beta P_{\boldsymbol{\lambda}^{(\nu)}})}.
\end{equation}
The point is that the numerator and denominator in
\eqref{eq:label-kms-projection-formula} contain the same positive factor
$w_{\mathrm{mid}}(\nu)$ from the two internal vertices.  Explicitly,
\[
  w_{\mathrm{mid}}(\nu)
  :=z_t(\lambda_{e_-}^*,\nu,\mu_1)
    z_t(\nu^*,\lambda_{e_+},\mu_2),
\]
the product of their two local vertex weights from
\eqref{eq:vertex-label-weight}.  Expanding those two weights gives
\begin{align*}
  w_{\mathrm{mid}}(\nu)
  ={}&d_{\lambda_{e_-}}d_\nu^2d_{\lambda_{e_+}}d_{\mu_1}d_{\mu_2}
  +t\left(
      d_{\lambda_{e_-}}d_\nu d_{\mu_1}m_2(\nu)
      +d_\nu d_{\lambda_{e_+}}d_{\mu_2}m_1(\nu)
    \right)
    +t^2m_1(\nu)m_2(\nu).
\end{align*}
The three terms arise by selecting neither, exactly one, or both of the two
internal vertex projectors in
$(I+tA_{v_1})(I+tA_{v_2})$.  The first term is strictly positive, so
$w_{\mathrm{mid}}(\nu)>0$.

Cancellation reduces to finding quantities
$W_X$ and $W_I>0$, both independent of $\nu$, for which
\begin{equation}
  \label{eq:middle-label-trace-factorization-targets}
  \begin{aligned}
  &
    \operatorname{Tr}
      (\rho_\beta P_{\boldsymbol{\lambda}^{(\nu)}}X)
    =Z_\beta^{-1}w_{\mathrm{mid}}(\nu)W_X
  ,\\[1ex]
  &
    \operatorname{Tr}
      (\rho_\beta P_{\boldsymbol{\lambda}^{(\nu)}})
    =Z_\beta^{-1}w_{\mathrm{mid}}(\nu)W_I
  .
  \end{aligned}
\end{equation}
We will show these factorizations.

\smallskip\noindent\textbf{Denominator Factorization.}\par\smallskip
Factorization for the denominator is straightforward.
Using the local labels at each vertex, set
\[
  W_I:=\prod_{v\notin\{v_1,v_2\}}
    z_t(\lambda_{v,1},\lambda_{v,2},\lambda_{v,3})>0.
\]
The second line of \eqref{eq:middle-label-trace-factorization-targets} follows
directly from \eqref{eq:global-label-distribution}.  Only the two omitted
vertex weights involve $\nu$, so $W_I$ is independent of $\nu$.

\Needspace{5\baselineskip}
\smallskip\noindent\textbf{Numerator Factorization.}\par\smallskip
It remains to prove the numerator identity in
\eqref{eq:middle-label-trace-factorization-targets}.
Only matching blocks contribute to
$\operatorname{Tr}(\rho_\beta P_{\boldsymbol\lambda^{(\nu)}}X)$.
On such a block, the exact kernel formula
\eqref{eq:exact-ordered-path-kernel-with-spectator} and Schur's lemma at
the two internal vertices give a finite expansion
\[
  P_{\boldsymbol\lambda^{(\nu)}}XP_{\boldsymbol\lambda^{(\nu)}}
  =\sum_j I_p^{(\nu)}\otimes Y_j^L
    \otimes I_{V_{\mu_1}}\otimes I_{V_{\mu_2}}
    \otimes Y_j^R\otimes Z_j,
\]
where $Y_j^L\in\operatorname{Hom}_\Gamma(U_L,U_L)$,
$Y_j^R\in\operatorname{Hom}_\Gamma(U_R,U_R)$, and $Z_j$ acts on the
exterior factor; coefficients are absorbed into these operators.
Here $I_p^{(\nu)}$ is the identity on the
$(\lambda_{e_-},\nu,\lambda_{e_+})$ block of the three active edges.
The operators $Y_j^L,Y_j^R,Z_j$ can be chosen independently of $\nu$:
before selecting the active-edge labels, $X$ is the identity on the
full active-edge factor.

For each vertex, set
$G_v^{(\nu)}:=I+tA_v|_{\mathcal V_v(\boldsymbol\lambda^{(\nu)})}$.
Only $G_{v_1}^{(\nu)}$ and $G_{v_2}^{(\nu)}$ depend on $\nu$; write
$G_v$ for the other factors and
$G_{\mathrm{ext}}:=\bigotimes_{v\notin\{v_L,v_1,v_2,v_R\}}G_v$.
Regrouping each summand by vertex and taking the tensor-product trace gives
\[
  Z_\beta\operatorname{Tr}(\rho_\beta P_{\boldsymbol\lambda^{(\nu)}}X)
  =\left(\operatorname{Tr}_{\mathcal V_{v_1}}G_{v_1}^{(\nu)}\right)
   \left(\operatorname{Tr}_{\mathcal V_{v_2}}G_{v_2}^{(\nu)}\right)W_X,
\]
where the remaining traces are collected in
\[
\begin{aligned}
  W_X:=\sum_j{}
    \operatorname{Tr}_{\mathcal V_{v_L}}
      \!\left[G_{v_L}(I\otimes Y_j^L)\right]
    \operatorname{Tr}_{\mathcal V_{v_R}}
      \!\left[G_{v_R}(I\otimes Y_j^R)\right]
\operatorname{Tr}_{\mathrm{ext}}(G_{\mathrm{ext}}Z_j).
\end{aligned}
\]
Here the endpoint identities act on the incident active-edge factors.
Every term in $W_X$ is independent of $\nu$.  At the internal vertices,
the inserted operators are identities, so
\[
  \left(\operatorname{Tr}_{\mathcal V_{v_1}}G_{v_1}^{(\nu)}\right)
  \left(\operatorname{Tr}_{\mathcal V_{v_2}}G_{v_2}^{(\nu)}\right)
  =z_t(\lambda_{e_-}^*,\nu,\mu_1)
   z_t(\nu^*,\lambda_{e_+},\mu_2)
  =w_{\mathrm{mid}}(\nu).
\]
This proves the numerator identity in
\eqref{eq:middle-label-trace-factorization-targets}.
Together with the denominator identity and
\eqref{eq:label-kms-projection-formula}, it yields
$(\mathcal E_{\mathrm{lab}}X)(\boldsymbol\lambda^{(\nu)})=W_X/W_I$,
independent of $\nu$.

\medskip\noindent\textbf{From Cancellation to Middle-edge Fixedness.}\par\smallskip
The cancellation established above may be summarized as
\[
  \mathcal E_{\mathrm{lab}}(\ker K_p)
  \subseteq
  \{\text{label observables independent of }\lambda_{e_0}\}.
\]
We now apply this statement to the projected fixed space.  Let
$f\in\operatorname{ran}(\Pi E_p)$.  Then there exists
$X\in\operatorname{ran}E_p=\ker K_p\cap L^2_0(\rho_\beta)$ such that
$f=\Pi X$.  Since $\Pi$ is the restriction of
$\mathcal E_{\mathrm{lab}}$ to $L^2_0(\rho_\beta)$, we have
\[
  f=\Pi X=\mathcal E_{\mathrm{lab}}X.
\]
The cancellation statement therefore implies that $f$ is independent of
the middle-edge label $\lambda_{e_0}$.
Equivalently, there exists a label observable $f_{\setminus e_0}$ on the
edges other than $e_0$ such that
\[
  f=f_{\setminus e_0}\otimes I_{e_0}.
\]
Thus $f$ may depend nontrivially on all labels outside $e_0$, but acts as the
identity on the entire middle-edge factor.  

Every label observable commutes
with $H_{\cG}$, and this factorization shows, more strongly, that $f$ commutes
with every bare middle-edge coupling
$L_g$, $R_g$, and $|h\rangle\langle h|$.  Since it also commutes with the
spectral projectors of $H_{\cG}$, it commutes with every Bohr component of
those couplings.  Applying
\cref{lem:davies-fixed-point-commutant} to the positive middle-edge Davies
operator gives $K_{e_0}f=0$, and hence
$A_{e_0}f=\Pi K_{e_0}\Pi f=0$.  Thus
$\operatorname{ran}(\Pi E_p)\subseteq\ker A_{e_0}$, which is
\eqref{eq:middle-edge-shielding}.
\end{proof}

\subsubsection{From Decoupling to the Local Schur Certificate}
\label{sec:decoupling-to-local-schur-certificate}

The shielding statement concerns the matching blocks reached from $\Pi$, but
the variational short allows an arbitrary $\Pi^\perp$-component.  We therefore use
the all-block quotient lower bound
$g_{\Gamma,\mathrm{loc},*}$ from
\cref{lem:uniform-ordered-path-quotient-gap}, rather than introducing a
matching-only gap.  Using the ordered-block type set
$\Theta_{\Gamma}$ from
\eqref{eq:ordered-block-type-set}, define
\begin{equation}
  \label{eq:middle-comparison-ceiling}
  C_{\Gamma,\mathrm{mid},*}
  :=\max\left\{1,
    \max_{\substack{\boldsymbol{\theta}\in\Theta_{\Gamma}\\0\le t\le t_*}}
    \lVert K^{\mathrm{path}}_{\boldsymbol{\theta},t}\rVert
  \right\}.
\end{equation}
The norm in the maximum is the operator norm in the local KMS Hilbert space
of the ordered block.

\begin{lemma}[Uniform middle-edge upper bound]
\label{lem:uniform-middle-edge-ceiling}
The constant in \eqref{eq:middle-comparison-ceiling} is finite and, uniformly
over all graphs in the theorem family, all centered paths, and all
$t\in[0,t_*]$, satisfies
\[
  \lVert A_{\operatorname{mid}(p)}\rVert
  \le C_{\Gamma,\mathrm{mid},*}.
\]
\end{lemma}

\begin{proof}
For fixed $\Gamma$, the set $\Theta_{\Gamma}$ is finite, and
each local matrix depends continuously on $t\in[0,t_*]$.
Its norm attains a finite maximum on this closed, bounded interval, so
\eqref{eq:middle-comparison-ceiling} is finite.  For
every centered path,
\[
  \lVert A_{\operatorname{mid}(p)}\rVert
  \le \lVert K_{\operatorname{mid}(p)}\rVert
  \le \lVert K_p\rVert
  =\max_{\boldsymbol{\theta}}
    \lVert K^{\mathrm{path}}_{\boldsymbol{\theta},t}\rVert
  \le C_{\Gamma,\mathrm{mid},*}.
\]
The first inequality is contractivity of orthogonal compression, the second
uses $0\preceq K_{\operatorname{mid}(p)}\preceq K_p$, and the equality follows
from the ordered-block decomposition with exterior tensor factors.  The
orientation conjugacy used in the proof of
\cref{lem:uniform-ordered-path-quotient-gap} shows that the same finite
collection controls every centered path and orientation pattern, independently
of the size of the ambient graph.
\end{proof}

Combining this uniform upper bound with the shielding inclusion and the uniform
local quotient lower bound now gives the required pathwise shorted-form comparison.

\begin{proposition}[Common local Schur certificate]
\label{prop:common-local-schur-certificate}
Every centered path satisfies
\begin{equation}
  \label{eq:common-local-schur-certificate}
  \mathcal S_{\Pi}(K_p)
  \succeq c_{\Gamma,3,*}A_{\operatorname{mid}(p)},
  \qquad
  c_{\Gamma,3,*}
  :=\frac{g_{\Gamma,\mathrm{loc},*}}
          {C_{\Gamma,\mathrm{mid},*}}>0.
\end{equation}
By \cref{lem:uniform-ordered-path-quotient-gap,lem:uniform-middle-edge-ceiling},
both constants defining
$c_{\Gamma,3,*}$ are uniform in the centered path and the ambient graph.
Hence the resulting positive shorted-form constant is uniform over the
entire theorem family and over $t\in[0,t_*]$.
\end{proposition}

\begin{proof}
For $f\in\operatorname{ran}\Pi$ and $y\in\operatorname{ran}\Pi^\perp$, the above-kernel
Poincar\'e form \eqref{eq:above-kernel-poincare} of the local quotient bound
in \cref{lem:uniform-ordered-path-quotient-gap} gives
\[
  \langle f+y,K_p(f+y)\rangle
  \ge g_{\Gamma,\mathrm{loc},*}
  \operatorname{dist}_{\mathrm{KMS}}
    (f+y,\operatorname{ran}E_p)^2.
\]
The passage from the centered fixed space to its label projection is exact.
Indeed, for $z\in\operatorname{ran}E_p$, KMS orthogonality of the
$\Pi\oplus\Pi^\perp$ decomposition gives
\[
  \lVert f+y-z\rVert_{\mathrm{KMS}}^2
  =\lVert f-\Pi z\rVert_{\mathrm{KMS}}^2
   +\lVert y-\Pi^\perp z\rVert_{\mathrm{KMS}}^2.
\]
For fixed $z$, the second term is minimized by $y=\Pi^\perp z$; taking the
infimum over $z\in\operatorname{ran}E_p$ therefore yields
\[
  \inf_{y\in\operatorname{ran}\Pi^\perp}
  \operatorname{dist}_{\mathrm{KMS}}
    (f+y,\operatorname{ran}E_p)^2
  =\operatorname{dist}_{\mathrm{KMS}}
    \bigl(f,\operatorname{ran}(\Pi E_p)\bigr)^2.
\]
By
\cref{prop:middle-edge-shielding}, this projected fixed space lies in the
kernel of $A_{\operatorname{mid}(p)}$.  Hence, for every
$z\in\operatorname{ran}(\Pi E_p)$,
\[
  \langle f,A_{\operatorname{mid}(p)}f\rangle
  =\langle f-z,A_{\operatorname{mid}(p)}(f-z)\rangle
  \le C_{\Gamma,\mathrm{mid},*}
       \lVert f-z\rVert_{\mathrm{KMS}}^2.
\]
Taking the infimum over $z$ gives
\[
  \langle f,A_{\operatorname{mid}(p)}f\rangle
  \le C_{\Gamma,\mathrm{mid},*}
  \operatorname{dist}_{\mathrm{KMS}}
    \bigl(f,\operatorname{ran}(\Pi E_p)\bigr)^2.
\]
We now minimize the original quadratic form over the complementary component
$y$.  The variational characterization
\eqref{eq:short-variational-definition}, followed by the local quotient
bound and the distance identity above, gives
\[
\begin{aligned}
  \langle f,\mathcal S_{\Pi}(K_p)f\rangle
  &=\inf_{y\in\operatorname{ran}\Pi^\perp}
    \langle f+y,K_p(f+y)\rangle\\
  &\ge g_{\Gamma,\mathrm{loc},*}
    \inf_{y\in\operatorname{ran}\Pi^\perp}
    \operatorname{dist}_{\mathrm{KMS}}
      (f+y,\operatorname{ran}E_p)^2\\
  &=g_{\Gamma,\mathrm{loc},*}
    \operatorname{dist}_{\mathrm{KMS}}
      \bigl(f,\operatorname{ran}(\Pi E_p)\bigr)^2\\
  &\ge\frac{g_{\Gamma,\mathrm{loc},*}}
              {C_{\Gamma,\mathrm{mid},*}}
    \langle f,A_{\operatorname{mid}(p)}f\rangle.
\end{aligned}
\]
The last step uses the preceding upper bound on the middle-edge form,
rearranged as a lower bound on the squared distance.  Since this holds for
every $f\in\operatorname{ran}\Pi$, it is precisely the operator inequality
\eqref{eq:common-local-schur-certificate}.
\end{proof}

This proves Condition~(iii) of \cref{thm:abstract-two-axis-gap-assembly} with
$c_{\mathrm{Sch}}=c_{\Gamma,3,*}$.

\section{The \texorpdfstring{$S_3$}{S3} Strengthening Beyond One-Site Influence}
\label{sec:s3-strengthening}

The common fixed-group theorem isolates the temperature restriction in the
label-subspace condition.  For $S_3$, we replace the one-site estimate by a
conditioned vertex-star certificate.  We first state the resulting conditional
criterion and then verify its two hypotheses on $0\le t\le2$ using the finite
calculation in \cref{app:s3-finite-certificates}, which may be treated as a
black box on a first reading.  The local quotient and Schur conditions,
together with all three geometric conditions, remain unchanged.  

\subsection{A Conditional Vertex-Star Criterion}
\label{sec:s3-star-certificate}

Write $\widehat{S_3}$ for its three irreducible representation types and
$d_\lambda=\dim V_\lambda$.
By \cref{lem:group-averaging-invariants}, $A_v$ projects onto
$\operatorname{Inv}\mathcal V_v$, so its trace equals the dimension of
this invariant space.  The vertex weight in
\eqref{eq:vertex-label-weight} therefore becomes
\[
 z_t(\lambda_1,\lambda_2,\lambda_3)
 =d_{\lambda_1}d_{\lambda_2}d_{\lambda_3}
  +t\,\dim\operatorname{Inv}
   (V_{\lambda_1}\otimes V_{\lambda_2}\otimes V_{\lambda_3}).
\]
Here $\operatorname{Inv}$ is the subspace fixed by the simultaneous group
action.  For this commuting model, the Peter--Weyl
decomposition separates the vertex factors at fixed edge labels.
The global irrep-label distribution is therefore proportional to the
product of these vertex weights, with dual labels at incoming endpoints.
Let $\mathsf K_{\mathrm{lab}}$ be its unit-rate single-edge heat-bath
generator, defined in \eqref{eq:label-heat-bath-generator}.

For a vertex $v$, let $E(v)$ be its three incident edges and let
$\mathbb E_v^\star$ denote conditional expectation over the three labels in
$E(v)$, with all other edge labels fixed.  The unit-rate vertex-star
heat-bath generator is
\begin{equation}
  \label{eq:s3-star-block-generator}
  \mathsf K^\star:=\sum_{v\in V(\cG)}(I-\mathbb E_v^\star).
\end{equation}
Fixing the labels outside $E(v)$ leaves a conditional distribution on the
three star labels
$\boldsymbol{\lambda}_v
:=(\lambda_1,\lambda_2,\lambda_3)\in\widehat{S_3}^{3}$.
For the $i$th star edge, let $b_{i,1},b_{i,2}$ be the labels on the other two
edges at its outer endpoint, and collect these six boundary labels as
$\boldsymbol b=(b_{i,j})_{1\le i\le3,\,1\le j\le2}
\in\widehat{S_3}^{6}$, see \Cref{fig:s3-conditioned-star-geometry}.  

The purpose of this star is to reduce the interface-gap estimate to
fixed-size conditional problems, uniformly over their boundary labels.

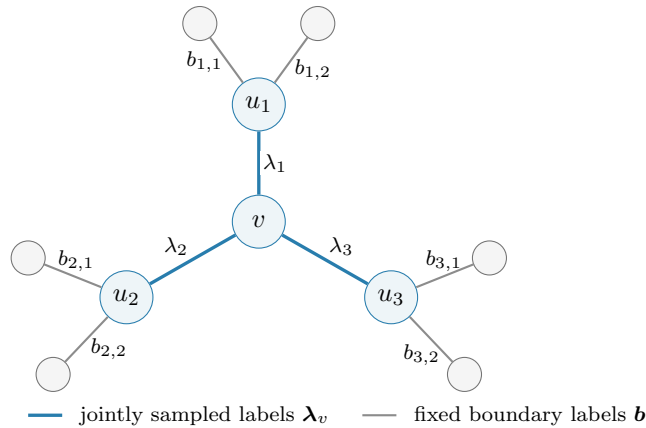
\begin{figure}[bht]
  \centering
  \begin{tikzpicture}[
    star vertex/.style={circle, draw=figureGeometry, fill=figureGeometry!8,
      minimum size=7mm, inner sep=0pt, font=\small},
    boundary vertex/.style={circle, draw=black!55, fill=black!4,
      minimum size=4.5mm, inner sep=0pt},
    star edge/.style={draw=figureGeometry, line width=1.25pt},
    boundary edge/.style={draw=black!45, line width=.8pt,
      shorten >=2.2mm},
    edge label/.style={font=\scriptsize, fill=white, inner sep=1pt},
    boundary label/.style={font=\scriptsize, inner sep=0pt},
    every node/.style={font=\small}
  ]
    \node[star vertex] (v) at (0,0) {$v$};
    \node[star vertex] (u1) at (0,1.55) {$u_1$};
    \node[star vertex] (u2) at (-1.75,-1.00) {$u_2$};
    \node[star vertex] (u3) at (1.75,-1.00) {$u_3$};

    \coordinate (b11) at (-.78,2.62);
    \coordinate (b12) at (.78,2.62);
    \coordinate (b21) at (-3.05,-.48);
    \coordinate (b22) at (-2.72,-2.02);
    \coordinate (b31) at (3.05,-.48);
    \coordinate (b32) at (2.72,-2.02);
    \foreach \x in {b11,b12,b21,b22,b31,b32} {
      \node[boundary vertex] at (\x) {};
    }

    \draw[boundary edge] (u1) --
      node[pos=.58, boundary label, xshift=-5pt, yshift=-5pt] {$b_{1,1}$} (b11);
    \draw[boundary edge] (u1) --
      node[pos=.58, boundary label, xshift=5pt, yshift=-8pt] {$b_{1,2}$} (b12);
    \draw[boundary edge] (u2) --
      node[pos=.56, boundary label, xshift=6pt, yshift=4pt] {$b_{2,1}$} (b21);
    \draw[boundary edge] (u2) --
      node[pos=.56, boundary label, xshift=12pt, yshift=0pt] {$b_{2,2}$} (b22);
    \draw[boundary edge] (u3) --
      node[pos=.56, boundary label, xshift=-5pt, yshift=4pt] {$b_{3,1}$} (b31);
    \draw[boundary edge] (u3) --
      node[pos=.56, boundary label, xshift=-7pt, yshift=-2pt] {$b_{3,2}$} (b32);

    \draw[star edge] (v) -- node[edge label, right] {$\lambda_1$} (u1);
    \draw[star edge] (v) -- node[edge label, above left] {$\lambda_2$} (u2);
    \draw[star edge] (v) -- node[edge label, above right] {$\lambda_3$} (u3);

    \draw[star edge] (-3.05,-2.56) -- (-2.60,-2.56);
    \node[anchor=west, font=\scriptsize] (sampled-legend) at (-2.50,-2.56)
      {jointly sampled labels $\boldsymbol{\lambda}_v$};
    \draw[draw=black!45, line width=.8pt]
      ([xshift=3mm]sampled-legend.east) -- ++(.45,0)
      coordinate (boundary-legend-end);
    \node[anchor=west, font=\scriptsize, xshift=1mm] at (boundary-legend-end)
      {fixed boundary labels $\boldsymbol b$};
  \end{tikzpicture}
  \caption{The conditioned vertex star.  The three blue edges form $E(v)$
  and carry the jointly sampled labels $\lambda_1,\lambda_2,\lambda_3$.  The six gray
  edges carry the fixed boundary condition $\boldsymbol b$.  Given
  $\boldsymbol b$,
  \eqref{eq:s3-conditioned-star-distribution} is the distribution on the
  three blue labels.}
  \label{fig:s3-conditioned-star-geometry}
\end{figure}

Why does this neighborhood suffice?  After conditioning on all labels
outside $E(v)$, only the weight at $v$
and those at its three neighbors can vary; all other weights cancel
in the conditional normalization.  Thus the three blue labels and six
gray boundary labels determine the entire conditional problem, regardless
of the size of the graph.  For $S_3$, this leaves only $3^3=27$ internal label 
configurations and $3^6=729$ boundary label configurations.

The global step uses the cubic geometry and the girth assumption.
Consider two global label configurations differing only on an edge
$e=\{u,v\}$.  As illustrated in \Cref{fig:s3-star-update-roles},
this edge is updated by an endpoint star, but is a boundary edge for
a star centered at another neighbor of either endpoint.
A disagreement on one edge is removed by either of its two endpoint-star
updates, and can create new disagreements through only four other star
updates.  A uniform bound on the response to one changed boundary label
therefore gives contraction of the global star chain.  A second local
input, a uniform gap for the three single-edge updates within a
conditioned star, transfers this star-chain gap to
$\mathsf K_{\mathrm{lab}}$; each edge belongs to two stars, so this
comparison also has a size-independent constant.
Finally, \cref{prop:finite-group-davies-label-comparison} transfers the
label-chain bound to the interface compression $A=\Pi K\Pi$.

The two local inputs are stated precisely in
\cref{prop:s3-finite-certificate-to-gns}.  Their finite verification uses
exact polynomial inequalities to certify the whole interval
$0\le t\le2$, not just a finite sample of temperatures.

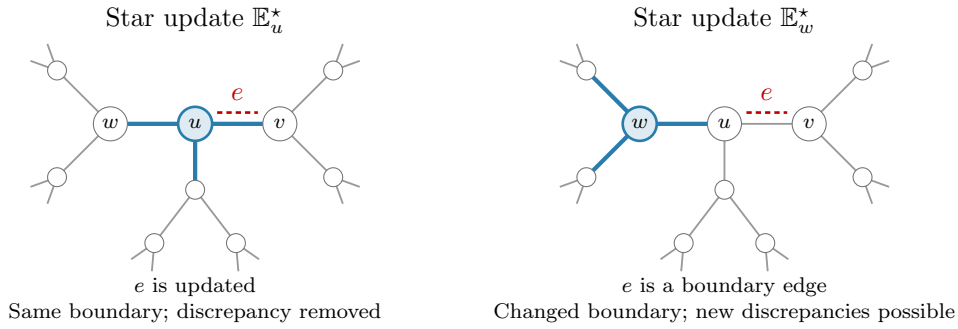
\begin{figure}[htbp]
  \centering
  \begin{tikzpicture}[
    x=.83cm, y=.83cm,
    fixed edge/.style={draw=black!40, line width=.7pt},
    updated edge/.style={draw=figureGeometry, line width=1.6pt},
    vertex/.style={circle, draw=black!55, fill=white,
      minimum size=4.5mm, inner sep=0pt, font=\scriptsize},
    active center/.style={vertex, draw=figureGeometry,
      fill=figureGeometry!15, line width=1pt},
    every node/.style={font=\small}
  ]
    \foreach \panel/\offset in {1/0,2/7} {
      \begin{scope}[xshift=\offset cm]
        \coordinate (vm) at (0,0);
        \coordinate (vp) at (1.35,0);
        \coordinate (u1) at (-1.35,0);
        \coordinate (u2) at (0,-1.05);
        \coordinate (u3) at (2.2,.85);
        \coordinate (u4) at (2.2,-.85);
        \coordinate (a) at (-2.2,.85);
        \coordinate (b) at (-2.2,-.85);
        \coordinate (r) at (-.65,-1.9);
        \coordinate (s) at (.65,-1.9);
        \foreach \left/\right in {vm/vp,vm/u1,vm/u2,vp/u3,vp/u4,u1/a,u1/b,u2/r,u2/s}
          \draw[fixed edge] (\left) -- (\right);
        \foreach \leaf/\angle in {a/135,b/225,u3/45,u4/315,r/240,s/300} {
          \draw[fixed edge] (\leaf) -- ++({\angle-25}:.45);
          \draw[fixed edge] (\leaf) -- ++({\angle+25}:.45);
        }
        \ifnum\panel=1
          \foreach \neighbor in {vp,u1,u2}
            \draw[updated edge] (vm) -- (\neighbor);
        \else
          \foreach \neighbor in {vm,a,b}
            \draw[updated edge] (u1) -- (\neighbor);
        \fi
        \foreach \point in {a,b,r,s,u2,u3,u4}
          \node[vertex, minimum size=2.5mm] at (\point) {};
        \node[vertex] at (vm) {$u$};
        \node[vertex] at (vp) {$v$};
        \node[vertex] at (u1) {$w$};
        \ifnum\panel=1
          \node[active center] at (vm) {$u$};
        \else
          \node[active center] at (u1) {$w$};
        \fi
        \draw[draw=red!75!black, line width=1pt,
          dash pattern=on 2pt off 1.5pt] (.35,.18) -- (1,.18);
        \node[text=red!75!black, above] at (.675,.22) {$e$};
        \ifnum\panel=1
          \node at (0,1.65) {Star update $\mathbb E_u^\star$};
          \node[align=center, font=\scriptsize] at (0,-2.8)
            {$e$ is updated\\Same boundary; discrepancy removed};
        \else
          \node at (0,1.65) {Star update $\mathbb E_w^\star$};
          \node[align=center, font=\scriptsize] at (0,-2.8)
            {$e$ is a boundary edge\\Changed boundary; new discrepancies possible};
        \fi
      \end{scope}
    }
  \end{tikzpicture}
  \caption{Two roles of the same discrepant edge $e$.
  Both panels show the same local graph; blue edges are jointly updated,
  gray edges are fixed, and the red dashed mark above $e$ indicates the only initial
  disagreement between two global label configurations.
  Left: the star at $u$ contains $e$, and identical conditional samples
  remove the disagreement.  Right: the star at $w$ sees $e$ in its
  boundary; the disagreement on $e$ remains and can affect the updated
  labels.  Here $u,v,w$ correspond respectively to $v_-,v_+,u_1$ in
  \Cref{fig:s3-star-coupling-count}.  The rest of the graph is omitted.}
  \label{fig:s3-star-update-roles}
\end{figure}

Concretely, let $\boldsymbol{\xi}_{\boldsymbol b}$ be any label configuration
on $E\setminus E(v)$ whose six boundary labels are $\boldsymbol b$.
The conditional distribution is independent of the remaining exterior
labels and is given by
\begin{equation}
  \label{eq:s3-conditioned-star-distribution}
  \begin{aligned}
  \mu_{\boldsymbol b}^t(\boldsymbol{\lambda}_v)
  &=\mu_{\mathrm{lab}}\!\left(
    \boldsymbol{\lambda}_{E(v)}=\boldsymbol{\lambda}_v
    \,\middle|\,
    \boldsymbol{\lambda}_{E\setminus E(v)}=\boldsymbol{\xi}_{\boldsymbol b}
    \right)
   =\frac{w_{\boldsymbol b}^t(\boldsymbol{\lambda}_v)}{Z_{\boldsymbol b}(t)},\\
  w_{\boldsymbol b}^t(\boldsymbol{\lambda}_v)
  &:=z_t(\lambda_1,\lambda_2,\lambda_3)
    \prod_{i=1}^3 z_t(\lambda_i^*,b_{i,1},b_{i,2}),
  \end{aligned}
\end{equation}
where
$Z_{\boldsymbol b}(t):=\sum_{\boldsymbol{\eta}\in\widehat{S_3}^{3}}
w_{\boldsymbol b}^t(\boldsymbol{\eta})$.
The first factor in $w_{\boldsymbol b}^t$ is the weight of the central vertex, while the
$i$th factor in the product is the weight of the outer endpoint of the $i$th
star edge.
Although every $S_3$ irrep is self-dual, the duals retain the orientation
convention.

We equip label configurations with the edgewise Hamming distance
$d_{\mathrm H}$, which counts the edges on which two configurations have
different labels. 
Explicitly, $d_{\mathrm H}(\boldsymbol{\lambda},\boldsymbol{\lambda}')
=\sum_e \mathbf 1_{\{\lambda_e\ne\lambda'_e\}}.$ 
For probability distributions $\mu$ and $\nu$ on the same
configuration space, the associated Wasserstein distance is
\begin{equation}
  \label{eq:s3-hamming-wasserstein}
  W_1(\mu,\nu)
  :=\min_{\pi\in\mathcal C(\mu,\nu)}
    \sum_{\boldsymbol\lambda,\boldsymbol\lambda'}
    \pi(\boldsymbol\lambda,\boldsymbol\lambda')
    d_{\mathrm H}(\boldsymbol\lambda,\boldsymbol\lambda'),
\end{equation}
where $\mathcal C(\mu,\nu)$ is the set of couplings, namely joint
distributions with marginals $\mu$ and $\nu$, and
$\pi(\boldsymbol\lambda,\boldsymbol\lambda')$ is the joint probability
assigned to the pair of configurations
$(\boldsymbol\lambda,\boldsymbol\lambda')$.
For the conditioned stars, this specializes to
\[
  W_1(\mu_{\boldsymbol b}^t,\mu_{\boldsymbol b'}^t)
  =\min_{\pi\in\mathcal C(\mu_{\boldsymbol b}^t,\mu_{\boldsymbol b'}^t)}
    \sum_{\boldsymbol\lambda,\boldsymbol\lambda'\in\widehat{S_3}^{\,3}}
    \pi(\boldsymbol\lambda,\boldsymbol\lambda')
    \sum_{i=1}^3\mathbf 1_{\{\lambda_i\ne\lambda_i'\}}.
\]
Thus this distance is the minimum expected number of disagreements among
the three updated edge labels; the fixed boundary labels are not counted.
We write $\boldsymbol b\sim\boldsymbol b'$ when the two boundary
configurations differ in exactly one of their six entries.

\Needspace{8\baselineskip}
For fixed $t$ and $\boldsymbol b$, let
$\mathbb E_{i\mid\boldsymbol b}^{t,\star}$ be conditional expectation under
$\mu_{\boldsymbol b}^t$ that resamples only $\lambda_i$, conditioned on the
other two star labels.  Thus the six labels in $\boldsymbol b$ remain fixed,
while each of the three star labels has its own rate-one update.  The
corresponding internal single-edge heat-bath generator is
\begin{equation}
  \label{eq:s3-conditioned-internal-generator}
  \mathsf K_{\mathrm{int},\boldsymbol b}^{t,\star}
  :=\sum_{i=1}^3\bigl(I-\mathbb E_{i\mid\boldsymbol b}^{t,\star}\bigr)
  \quad\text{on }L^2(\mu_{\boldsymbol b}^t).
\end{equation}

The following proposition converts two finite conditions on these
conditioned stars into a label-chain gap.  Its corollary transfers this
bound to the full quantum system.

\begin{proposition}[Finite star certificate for the label-chain gap]
\label{prop:s3-finite-certificate-to-gns}
Fix $T<\infty$.  Suppose that there are constants $B_T<1/2$ and
$\delta_{\mathrm{int},T}>0$ such that, uniformly for $0\le t\le T$, every
conditioned $S_3$ vertex star satisfies
\begin{align}
  W_1(\mu_{\boldsymbol b}^t,\mu_{\boldsymbol b'}^t)
  &\le B_T
  \quad(\boldsymbol b\sim\boldsymbol b'),
  \label{eq:s3-certificate-wasserstein-assumption}\\
  \gap(\mathsf K_{\mathrm{int},\boldsymbol b}^{t,\star})
  &\ge\delta_{\mathrm{int},T}
  \quad\text{for every }\boldsymbol b.
  \label{eq:s3-certificate-internal-gap-assumption}
\end{align}
Then, for every finite simple cubic graph of girth at least six and every
$0\le t\le T$,
\begin{equation}
  \label{eq:s3-certificate-label-consequence}
  \gap(\mathsf K_{\mathrm{lab}})
  \ge \frac{\delta_{\mathrm{int},T}}2(2-4B_T)>0.
\end{equation}
\end{proposition}

The proof is postponed to \cref{sec:proof-s3-certificate-to-gns}.

\begin{corollary}[Full-GNS gap from the star certificate]
\label{cor:s3-certificate-to-gns}
Under the assumptions of \cref{prop:s3-finite-certificate-to-gns}, there is
a constant $c_{S_3,T}>0$, uniform over $0\le t\le T$ and all finite simple
cubic graphs of girth at least six,
such that the complete positive Davies operator satisfies
\begin{equation}
  \label{eq:s3-certificate-gns-consequence}
  \gap_{\mathrm{GNS}}(K_{\cG,\beta})\ge c_{S_3,T}.
\end{equation}
\end{corollary}

\begin{proof}
The label-chain bound \eqref{eq:s3-certificate-label-consequence} and the
Davies-to-label comparison in
\cref{prop:finite-group-davies-label-comparison} give
\[
  A=\Pi K\Pi\succeq a_T\Pi,
  \qquad
  a_T:=c_{\mathrm{cmp},T}
       \frac{\delta_{\mathrm{int},T}}2(2-4B_T)>0,
\]
where $c_{\mathrm{cmp},T}>0$ is the comparison constant with $t_*=T$.
This verifies Condition~(i) of
\cref{thm:abstract-two-axis-gap-assembly}, without assuming $[K,\Pi]=0$.
The local quotient and Schur bounds in
\cref{lem:uniform-ordered-path-quotient-gap,prop:common-local-schur-certificate}
hold on every fixed finite interval $[0,T]$: their proofs use positive
Davies rates and continuous positive local bounds for finitely many path
types, not the one-site influence restriction $T<2/3$.
Together with the fixed-space cover \eqref{eq:common-cubic-defect-frame}
and the incidence identities in \cref{lem:cubic-path-ownership}, these
verify the remaining hypotheses of
\cref{thm:abstract-two-axis-gap-assembly} and give a uniform positive KMS
gap on $L^2_0(\rho_\beta)$.  Centering any fixed observable shows that
$\ker K_{\cG,\beta}=\C I$, and
\cref{prop:davies-kms-gns-gap-bridge} gives the stated full-GNS gap.
\end{proof}

\subsection{The Exact Certificate at \texorpdfstring{$t=2$}{t=2}}
\label{sec:s3-exact-certificate-application}

The exact finite certificate in \cref{prop:exact-s3-star-certificate}
of \cref{app:s3-interval-certificate} verifies the two hypotheses of
\cref{prop:s3-finite-certificate-to-gns} on $0\le t\le2$, with
\[
  B_2=\frac{2818}{7395}<\frac12,
  \qquad
  \delta_{\mathrm{int},2}=\frac{11}{35}>0.
\]
Thus we take $T=2$, $B_T=B_2$, and
$\delta_{\mathrm{int},T}=\delta_{\mathrm{int},2}$.
All auxiliary transport and influence calculations are confined to
\cref{app:s3-finite-certificates}.
At $t=2$, the one-site influence criterion is inconclusive, whereas
the star criterion yields a positive gap; see
\eqref{eq:s3-endpoint-certificate-contrast} for the quantitative comparison.

\Needspace{8\baselineskip}
\medskip\noindent\textbf{Quantitative bounds.}\par\smallskip
Before concluding the proof of \cref{thm:finite-group-davies-gaps}\textnormal{(ii)},
we record the numerical bounds supplied by the certificate.
Throughout $0\le t\le2$, the label-gap conclusion
\eqref{eq:s3-certificate-label-consequence} of
\cref{prop:s3-finite-certificate-to-gns} gives
\begin{equation}
  \label{eq:exact-s3-label-gap-floor}
  \gap(\mathsf K_{\mathrm{lab}})
  \ge \frac{\delta_{\mathrm{int},2}}2(2-4B_2)
  =\frac12\frac{11}{35}\frac{3518}{7395}
  =\frac{19349}{258825}.
\end{equation}
With $t_*=2$ and $|S_3|=6$, the comparison constant in
\cref{prop:finite-group-davies-label-comparison} is
\begin{equation}
  \label{eq:s3-physical-label-comparison-constant}
  c_{S_3,\mathrm{cmp},2}
  =\frac1{3\cdot6^2\cdot3^2\cdot(1+3^2)}
  =\frac1{9720}.
\end{equation}
Consequently, the physical label compression satisfies
\begin{equation}
  \label{eq:s3-physical-p-sector-floor}
  A=\Pi K\Pi
  \succeq \frac{19349}{2515779000}\Pi.
\end{equation}
These bounds are uniform in the graph size and in $0\le t\le2$.

\Needspace{8\baselineskip}
\begin{proof}[Proof of \cref{thm:finite-group-davies-gaps}\textnormal{(ii)}]
By \cref{prop:exact-s3-star-certificate}, uniformly for $0\le t\le2$ and
boundary conditions $\boldsymbol b\sim\boldsymbol b'$,
\[
  W_1(\mu_{\boldsymbol b}^t,\mu_{\boldsymbol b'}^t)
  \le B_2=\frac{2818}{7395}<\frac12,
  \qquad
  \gap(\mathsf K_{\mathrm{int},\boldsymbol b}^{t,\star})
  \ge\delta_{\mathrm{int},2}=\frac{11}{35}>0.
\]
Hence assumptions \eqref{eq:s3-certificate-wasserstein-assumption} and
\eqref{eq:s3-certificate-internal-gap-assumption} of
\cref{prop:s3-finite-certificate-to-gns} hold with $T=2$.  Applying
\cref{cor:s3-certificate-to-gns} yields the volume-uniform full-GNS gap throughout
$0\le t\le2$, equivalently $0\le\beta J\le\log3$, as asserted in
\cref{thm:finite-group-davies-gaps}\textnormal{(ii)}.
\end{proof}

\subsection{Proof of the Label-Chain Certificate}
\label{sec:proof-s3-certificate-to-gns}

Finally, we prove the classical label-chain bound used above.

\begin{proof}[Proof of \cref{prop:s3-finite-certificate-to-gns}]
Recall the two positive heat-bath operators defined in
\eqref{eq:s3-star-block-generator} and
\eqref{eq:label-heat-bath-generator}, respectively:
\[
  \mathsf K^\star=\sum_{v\in V(\cG)}(I-\mathbb E_v^\star),
  \qquad
  \mathsf K_{\mathrm{lab}}
  =\sum_{e\in E(\cG)}(I-\mathbb E_e^{\mathrm{lab}}).
\]
Here $\mathsf K^\star$ jointly resamples the three labels incident to a
vertex at rate one per vertex, whereas $\mathsf K_{\mathrm{lab}}$ resamples
one edge label at rate one per edge.  In both cases, all labels outside
the updated set are held fixed.

We use path coupling to prove a gap for the star chain, using the
Wasserstein assumption \eqref{eq:s3-certificate-wasserstein-assumption}.
It suffices to construct a uniformly contracting coupling for pairs of
label configurations that \emph{differ on a single edge}: any two configurations
can be joined by changing one edge label at a time, and the local
contraction extends along this sequence to give contraction in Hamming
Wasserstein distance.  This yields a spectral gap for $\mathsf K^\star$.
We then use the conditioned internal-gap assumption
\eqref{eq:s3-certificate-internal-gap-assumption} to compare the star
chain with the single-edge label chain.

\medskip\noindent\textbf{From the finite certificate to star-chain
contraction.}\par\smallskip
Start two copies $X_s,Y_s$ from configurations that differ only on one
edge $e$, so $d_{\mathrm H}(X_0,Y_0)=1$.  Here $s$ denotes Markov-chain
time, not the temperature parameter $t$.  Couple the rate-one clocks at
each vertex synchronously.  There are three types of updates:
\begin{itemize}
\item \emph{The two endpoint stars of $e$.}
Their exterior labels agree, so their conditional distributions are
identical.  Using the same sample in both copies removes the sole
discrepancy, changing the distance from $1$ to $0$.
\item \emph{The four other stars of $u_1, u_2, u_3, u_4$ having $e$ in their boundary.}
Their boundary configurations differ in exactly one entry.  The
Wasserstein assumption \eqref{eq:s3-certificate-wasserstein-assumption}
supplies a coupling with at most $B_T$ expected
disagreements among the three resampled labels.  Since $e$ itself is not
updated, its discrepancy remains: the expected distance after the update
is at most $1+B_T$, an increase of at most $B_T$.
\item \emph{All remaining stars.}
Their boundary labels agree and they do not update $e$.  Identical samples
therefore leave the distance equal to $1$, with no change.
\end{itemize}
Cubicity and girth at least six ensure that the four affected outer stars
are distinct.  The two removing stars and the four potentially creating
stars are shown in \Cref{fig:s3-star-coupling-count}.

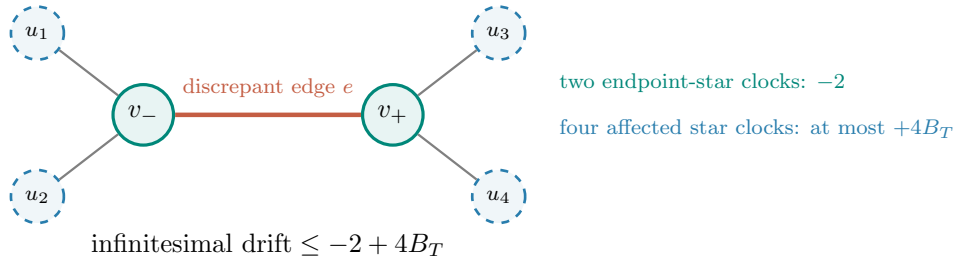
\begin{figure}[H]
  \centering
  \begin{tikzpicture}[
    removing/.style={circle, draw=figureInterface, fill=figureInterface!9,
      line width=1.2pt, minimum size=8mm, inner sep=0pt, font=\small},
    affected/.style={circle, draw=figureGeometry, fill=figureGeometry!7,
      dashed, line width=1pt, minimum size=7mm, inner sep=0pt,
      font=\scriptsize},
    ordinary edge/.style={draw=black!50, line width=.85pt},
    discrepant edge/.style={draw=figureCoupling, line width=1.7pt},
    legend/.style={font=\scriptsize, anchor=west},
    every node/.style={font=\small}
  ]
    \node[removing] (vm) at (-1.65,0) {$v_-$};
    \node[removing] (vp) at (1.65,0) {$v_+$};
    \node[affected] (u1) at (-3.05,1.08) {$u_1$};
    \node[affected] (u2) at (-3.05,-1.08) {$u_2$};
    \node[affected] (u3) at (3.05,1.08) {$u_3$};
    \node[affected] (u4) at (3.05,-1.08) {$u_4$};

    \draw[ordinary edge] (vm) -- (u1);
    \draw[ordinary edge] (vm) -- (u2);
    \draw[ordinary edge] (vp) -- (u3);
    \draw[ordinary edge] (vp) -- (u4);
    \draw[discrepant edge] (vm) -- node[above=2pt, font=\scriptsize,
      text=figureCoupling] {discrepant edge $e$} (vp);

    \node[legend, text=figureInterface] at (3.72,.42)
      {two endpoint-star clocks: $-2$};
    \node[legend, text=figureGeometry] at (3.72,-.18)
      {four affected star clocks: at most $+4B_T$};
    \node[font=\small] at (0,-1.72)
      {$\displaystyle \text{infinitesimal drift}\le -2+4B_T$};
  \end{tikzpicture}
  \caption{The path-coupling count for one discrepant edge.  Updating either
  green endpoint star includes $e$ and removes the discrepancy.  The same
  edge appears in the boundary of the four dashed blue stars, each of which
  creates at most $B_T$ expected discrepancies.}
  \label{fig:s3-star-coupling-count}
\end{figure}

Each clock has rate one, so we sum the expected changes in distance,
each measured relative to the initial distance $1$:
\begin{equation}
  \label{eq:s3-star-path-coupling-drift}
  \left.\frac{d}{ds}\right|_{s=0}
  \mathbb E[d_{\mathrm H}(X_s,Y_s)]
  \le 2\cdot(0-1)+4\cdot\bigl((1+B_T)-1\bigr)
  =-2+4B_T<0.
\end{equation}
Here, $\mathbb E$ denotes expectation under the coupled dynamics started from the fixed pair $(X_0,Y_0)$.
The first term accounts for the two updates that remove the discrepancy;
the second accounts for the four updates that can create new discrepancies.
Since $B_T<1/2$, this gives the \emph{positive contraction rate}
\[
  \kappa_{\star,T}:=2-4B_T>0.
\]
With this contraction rate for configurations differing on one edge,
the path-coupling extension gives contraction by
$e^{-\kappa_{\star,T}s}$ for arbitrary initial configurations: concatenate
optimal couplings along a shortest Hamming path and use the triangle
inequality for $W_1$.  
Applying that contraction to a nonconstant
eigenfunction of $e^{-s\mathsf K^\star}$ shows that
\begin{equation}
  \label{eq:s3-general-star-gap}
  \gap(\mathsf K^\star)\ge\kappa_{\star,T}.
\end{equation}

\medskip\noindent\textbf{From the star chain to the edge-update label
chain.}\par\smallskip
Fix the labels outside $E(v)$ and write $\mu=\mu_{\boldsymbol b}^t$ for
the resulting three-label distribution.  In the following conditional
calculation, $f$ denotes the restriction of the global observable to these
fixed exterior labels.  The internal-gap assumption
\eqref{eq:s3-certificate-internal-gap-assumption}, in Poincar\'e form, gives
\[
  \langle f,\mathsf K_{\mathrm{int},\boldsymbol b}^{t,\star} f\rangle_\mu
  \ge\delta_{\mathrm{int},T}\operatorname{Var}_\mu(f).
\]
By \eqref{eq:s3-conditioned-internal-generator}, the operator on the left
is the sum of the three single-edge updates.  With the exterior labels
fixed, each internal conditional expectation agrees with the corresponding
$\mathbb E_e^{\mathrm{lab}}$.  On the right, $\mathbb E_v^\star f$ is
the constant $\mu(f)$ on this conditional space, so
\[
  \operatorname{Var}_\mu(f)
  =\lVert f-\mathbb E_v^\star f\rVert_\mu^2
  =\langle f,(I-\mathbb E_v^\star)f\rangle_\mu.
\]
Substituting these two identities yields
\begin{equation}
  \label{eq:s3-conditioned-star-to-edge-form}
  \sum_{e\in E(v)}
  \langle f,(I-\mathbb E_e^{\mathrm{lab}})f\rangle_\mu
  \ge
  \delta_{\mathrm{int},T}
  \langle f,(I-\mathbb E_v^\star)f\rangle_\mu.
\end{equation}
Averaging over the exterior labels under $\mu_{\mathrm{lab}}$ turns these
conditional inner products into global $L^2(\mu_{\mathrm{lab}})$ inner
products.  Summing over $v$ then counts each edge twice, once for each
endpoint star, and gives
\begin{equation}
  \label{eq:s3-edge-to-star-form-comparison}
  2\mathsf K_{\mathrm{lab}}
  \succeq\delta_{\mathrm{int},T}\mathsf K^\star.
\end{equation}
Combining this form comparison with
\eqref{eq:s3-general-star-gap} gives
\begin{equation}
  \label{eq:s3-general-label-gap}
  \mathsf K_{\mathrm{lab}}
  \succeq \delta_{\mathrm{lab},T}I
  \quad\text{on }L^2_0(\mu_{\mathrm{lab}}),
  \qquad
  \delta_{\mathrm{lab},T}
  :=\frac{\delta_{\mathrm{int},T}}2\kappa_{\star,T}>0.
\end{equation}
This is \eqref{eq:s3-certificate-label-consequence}.
\end{proof}

\section*{Acknowledgement}
RH was supported by JST PRESTO Grant Number JPMJPR23F9 and JST ASPIRE Grant Number JPMJAP26A4, Japan. 

\section*{AI Usage}
OpenAI GPT-5.5, GPT-5.6 Sol, and GPT-6 Astra (via Codex), and Anthropic
Claude Opus 4.8 and Claude Fable 5 (via Claude Code), were used 
to assist with research development, manuscript preparation, and review.
The Python code for the finite verification was generated by GPT-5.6 Sol
and verified by the authors.  The authors assessed and revised the
AI-assisted work and take full responsibility for the results, code,
and text.

\bibliographystyle{alpha}
\bibliography{references}

\appendix

\section{Proof of the Finite-Block Davies-to-Label Comparison}
\label{app:finite-block-davies-label-comparison-proof}

This appendix proves \cref{thm:finite-block-davies-label-comparison}.  The
argument is a finite-dimensional conductance comparison; all model-specific
work enters only through the three displayed hypotheses of that theorem.

Recall that $\mathscr S$ is the full weighted family of bare couplings of
$K$, and that
$\mathscr S^{\mathrm{sel}}\subseteq\mathscr S$ is the adjoint-closed
subfamily selected for the comparison.  Thus the full Davies form below is
summed over $S\in\mathscr S$, while the hypotheses and the resulting lower
bound retain only $S\in\mathscr S^{\mathrm{sel}}$.

For convenience, we restate the three conditions in the notation and under the
names used in the proof.  For every pair of block labels $i,j\in\mathsf I$:
\begin{enumerate}[label=(C\arabic*),leftmargin=2.5em]
\item \emph{Block-density lower bound.}  The Gibbs density is uniformly spread
  inside each block:
  $
    \sigma P_i\succeq b\frac{\pi_i}{r_i}P_i.
  $
\item \emph{Davies-rate lower bound.}  Every Bohr component of every
  $S\in\mathscr S^{\mathrm{sel}}$ has rate
  $
    \gamma_\beta(\omega)\ge \gamma_{\min}.
  $
\item \emph{Block-overlap lower bound.}  The family
  $\mathscr S^{\mathrm{sel}}$ carries sufficient Hilbert--Schmidt weight
  from block $i$ to block $j$:
  $
    \sum_{S\in\mathscr S^{\mathrm{sel}}}w_S
      \lVert P_jSP_i\rVert_{\mathrm{HS}}^2
    \ge \kappa r_i\pi_j.
  $
\end{enumerate}
Here $b,\gamma_{\min},\kappa>0$ are the same constants as in
\cref{eq:abstract-block-density-floor,eq:abstract-block-rate-floor,eq:abstract-block-overlap-floor}.
The proof first uses (C1)--(C2) to lower-bound each energy-resolved physical
transition and then uses (C3) to sum those transitions into the label
heat-bath form.

\begin{proof}[Proof of \cref{thm:finite-block-davies-label-comparison}]
Let $H_0=\sum_E EP_E$ be the spectral resolution of the Hamiltonian and put
$P_{i,E}:=P_iP_E$.  This is again a family of orthogonal projectors because
$[P_i,H_0]=0$.  A block-scalar observable
$f=\sum_i f_iP_i$ commutes with $H_0$ and with $\sigma$, so its KMS pairing
becomes
\[
  \langle f,Y\rangle_\sigma
  =\operatorname{Tr}(\sigma f^\dagger Y).
\]
By the carr\'e-du-champ identity
\eqref{eq:lindblad-carre-du-champ}, each Lindblad jump $V$ satisfies
\[
  \mathcal L_V(f^\dagger f)
  -\mathcal L_V(f^\dagger)f
  -f^\dagger\mathcal L_V(f)
  =[V,f]^\dagger[V,f].
\]
Apply this identity to $V=S_\omega$ and sum it with the Davies weights.
Since
\[
  \mathcal L
  =\sum_{S\in\mathscr S,\omega}
    w_S\gamma_\beta(\omega)\mathcal L_{S_\omega},
\]
the resulting operator identity is
\begin{align*}
  \sum_{S\in\mathscr S,\omega}
    w_S\gamma_\beta(\omega)
    [S_\omega,f]^\dagger[S_\omega,f]
  =\mathcal L(f^\dagger f)
   -\mathcal L(f^\dagger)f
   -f^\dagger\mathcal L(f).
\end{align*}
Now take expectation in $\sigma$.  Stationarity $\operatorname{Tr}\!\left(
    \sigma\mathcal L(f^\dagger f)
  \right)=0$ is used precisely in the
first term on the right.
Consequently,
\begin{align*}
  \sum_{S\in\mathscr S,\omega}
    w_S\gamma_\beta(\omega)
    \operatorname{Tr}\!\left(
      \sigma[S_\omega,f]^\dagger[S_\omega,f]
    \right)
  &
  =-\operatorname{Tr}\!\left(
      \sigma\mathcal L(f^\dagger)f
    \right)
   -\operatorname{Tr}\!\left(
      \sigma f^\dagger\mathcal L(f)
    \right)
  \\
  &
  =-2\operatorname{Tr}\!\left(
      \sigma f^\dagger\mathcal L(f)
    \right).
\end{align*}
The two traces in the first equality are complex conjugates because
$\mathcal L$ is $*$-preserving.  Since $[f,\sigma]=0$, the
second trace equals
$\langle f,\mathcal Lf\rangle_\sigma$, which is real by KMS
symmetry; hence the two traces coincide.  Writing
$K=-\mathcal L$ therefore gives
\begin{equation}
  \label{eq:abstract-label-davies-dirichlet-form}
  \langle f,Kf\rangle_\sigma
  =\frac12\sum_{S\in\mathscr S,\omega}
    w_S\gamma_\beta(\omega)
    \operatorname{Tr}\!\left(
      \sigma[S_\omega,f]^\dagger
             [S_\omega,f]
    \right).
\end{equation}
Every summand in \eqref{eq:abstract-label-davies-dirichlet-form} is
nonnegative, so we may restrict directly to
$S\in\mathscr S^{\mathrm{sel}}$.  Since $\sigma$ commutes with the joint
projectors $P_{i,E}$, the restricted form can then be decomposed by inserting
their resolutions of the identity on the source and target sides.  Only the
component with $\omega=E'-E$ survives, and hence
\begin{equation}
  \label{eq:abstract-selected-davies-dirichlet-form}
  \begin{aligned}
    \langle f,Kf\rangle_\sigma
    &\ge \frac12\sum_{S\in\mathscr S^{\mathrm{sel}},\omega}
      w_S\gamma_\beta(\omega)
      \operatorname{Tr}\!\left(
        \sigma[S_\omega,f]^\dagger[S_\omega,f]
      \right)
      \\
    &=\frac12
      \sum_{S\in\mathscr S^{\mathrm{sel}}}w_S
      \sum_{i,j}\sum_{E,E'}\gamma_\beta(E'-E)
      \operatorname{Tr}\!\left(
        \sigma
        \bigl(P_{j,E'}[S_{E'-E},f]P_{i,E}\bigr)^\dagger
        \bigl(P_{j,E'}[S_{E'-E},f]P_{i,E}\bigr)
      \right).
  \end{aligned}
\end{equation}
For $S\in\mathscr S^{\mathrm{sel}}$ and $\omega=E'-E$, taking the
$(i,E)\to(j,E')$ block of $[S_\omega,f]$ gives
\begin{equation}
  \label{eq:abstract-resolved-commutator-block}
  P_{j,E'}[S_\omega,f]P_{i,E}
  =(f_i-f_j)P_{j,E'}SP_{i,E}.
\end{equation}
For each energy-resolved block,
$P_{j,E'}SP_{i,E}$ has source in $\operatorname{ran}P_i$.  Consequently,
the positive operator
$(P_{j,E'}SP_{i,E})^\dagger(P_{j,E'}SP_{i,E})$ is supported on
$\operatorname{ran}P_i$, and the block-density lower bound (C1) gives
\begin{equation}
  \label{eq:abstract-resolved-density-floor}
  \operatorname{Tr}\!\left(
    \sigma
    (P_{j,E'}SP_{i,E})^\dagger
    (P_{j,E'}SP_{i,E})
  \right)
  \ge b\frac{\pi_i}{r_i}
    \lVert P_{j,E'}SP_{i,E}\rVert_{\mathrm{HS}}^2.
\end{equation}
The Davies-rate lower bound (C2) supplies
$\gamma_\beta(E'-E)\ge \gamma_{\min}$.  Substituting
\cref{eq:abstract-resolved-commutator-block,eq:abstract-resolved-density-floor}
into \eqref{eq:abstract-selected-davies-dirichlet-form} therefore yields
\begin{equation}
  \label{eq:abstract-energy-resolved-davies-lower-bound}
  \begin{aligned}
  \langle f,Kf\rangle_\sigma
  &\ge \frac{b\gamma_{\min}}{2}
    \sum_{i,j}|f_i-f_j|^2\frac{\pi_i}{r_i}
    \sum_{S\in\mathscr S^{\mathrm{sel}}}w_S
    \sum_{E,E'}
      \lVert P_{j,E'}SP_{i,E}\rVert_{\mathrm{HS}}^2.
  \end{aligned}
\end{equation}
For fixed $i,j,S$, the joint energy projectors give
\[
  P_jSP_i=\sum_{E,E'}P_{j,E'}SP_{i,E}.
\]
The summands have orthogonal source or target energy spaces for distinct
pairs $(E,E')$ and are therefore Hilbert--Schmidt orthogonal.  Hence
\begin{equation}
  \label{eq:abstract-energy-block-hs-orthogonality}
  \sum_{E,E'}
    \lVert P_{j,E'}SP_{i,E}\rVert_{\mathrm{HS}}^2
  =\lVert P_jSP_i\rVert_{\mathrm{HS}}^2.
\end{equation}
Substituting \eqref{eq:abstract-energy-block-hs-orthogonality} into
\eqref{eq:abstract-energy-resolved-davies-lower-bound} gives
\begin{equation}
  \label{eq:abstract-block-resolved-davies-lower-bound}
  \langle f,Kf\rangle_\sigma
  \ge \frac{b\gamma_{\min}}{2}
    \sum_{i,j}|f_i-f_j|^2\frac{\pi_i}{r_i}
    \sum_{S\in\mathscr S^{\mathrm{sel}}}w_S
      \lVert P_jSP_i\rVert_{\mathrm{HS}}^2.
\end{equation}
The block-overlap lower bound (C3) now replaces the weighted coupling
sum in \eqref{eq:abstract-block-resolved-davies-lower-bound} by
$\kappa r_i\pi_j$.  The factor $r_i$ cancels the denominator in
$\pi_i/r_i$, and therefore
\begin{align*}
  \langle f,Kf\rangle_\sigma
  &\ge \frac{b\gamma_{\min}\kappa}{2}
    \sum_{i,j}\pi_i\pi_j|f_i-f_j|^2
   =b\gamma_{\min}\kappa\,
    \langle f,\mathsf K_{\mathrm{hb},\pi}f\rangle_\pi,
\end{align*}
which is \eqref{eq:abstract-davies-label-comparison}.
\end{proof}

\section{Proof of the Exact Centered-Path Kernel}
\label{app:exact-ordered-path-kernel-proof}

This appendix provides the proof of \cref{lem:exact-ordered-path-kernel}.

\begin{proof}
\medskip\noindent\textbf{From Fixed Points to the Bohr Commutant.}\par\smallskip
The complete star bath is adjoint closed and all logistic rates are
strictly positive.  Specializing
\cref{lem:davies-fixed-point-commutant} to the three edge update operators in $K_p$ gives
the exact starting point
\begin{equation}
  \label{eq:path-kernel-bohr-commutant}
  \ker K_p
  =
  \left\{
    X:\ [X,S_\omega]=0
    \ \text{for every }e\in\{e_-,e_0,e_+\},\
       S\in\mathscr S_e,\text{ and every }\omega
  \right\}.
\end{equation}
Recall the full bare coupling family
\[
\mathscr S_e=
\{L_g^{(e)}:g\in\Gamma\}
  \sqcup\{R_g^{(e)}:g\in\Gamma\}
  \sqcup\{Q_h^{(e)}:h\in\Gamma\}
\]
from \eqref{eq:active-edge-bare-coupling-family}.  Here $S_\omega$ is the
Bohr component of $S\in\mathscr S_e$.  Since $S=\sum_\omega S_\omega$,
\eqref{eq:path-kernel-bohr-commutant} gives, for every active bare coupling,
\[
  [X,S]=\sum_\omega[X,S_\omega]=0.
\]

\medskip\noindent\textbf{Removing the Active-Edge Factors.}\par\smallskip
On one active edge, the full coupling family $\mathscr S_e$ generates every
matrix unit because
\[
  Q_hL_{hk^{-1}}Q_k
  =|h\rangle\!\langle k|.
\]
Since these matrix units span the full edge algebra,
\begin{equation}
  \label{eq:active-edge-generated-algebra}
  \operatorname{Alg}(\mathscr S_e)
  =\mathcal B(\C[\Gamma]).
\end{equation}
Applying this to all three active edges, $X\in\ker K_p$ commutes with
their full matrix algebra and therefore acts as the identity on those
factors.  This recovers
\eqref{eq:fixed-observable-active-identity}, with $I_p$ defined in
\eqref{eq:active-path-identity}:
\[
  X=I_p\otimes Y.
\]
Here $Y$ includes both the outer-edge factors owned by the four path
vertices and the exterior factors owned by the remaining vertices.
On a fixed ordered block, expand the exterior factor in a basis
$\{Z_\ell\}_\ell$ of
$\mathcal X_{p,\boldsymbol{\theta}}^{\mathrm{ext}}$, using
\eqref{eq:ordered-path-exterior-factorization}.  Then
\[
  Y=\sum_\ell Y_{\mathrm{out},\ell}\otimes Z_\ell,
  \qquad
  X=\sum_\ell (I_p\otimes Y_{\mathrm{out},\ell})\otimes Z_\ell.
\]
Since $K_p$ acts trivially on the exterior factor, linear independence of
the $Z_\ell$ reduces the fixed-point condition to each path coefficient.
Fix one such coefficient and suppress $\ell$, writing
\begin{equation}
  \label{eq:active-edge-scalar-fixed-observable}
  X_{\mathrm{path}}
  =I_p\otimes Y_{\mathrm{out}}.
\end{equation}
The operator
$Y_{\mathrm{out}}$ acts on the outer-edge representation factors owned by the
four path vertices, not on the exterior factor; see the distinction between
``out'' and ``ext'' in \Cref{fig:ordered-path-kernel-factors}.
This is only the constraint coming from the
active-edge matrix algebras; the four vertex-commutation conditions below
will further restrict $Y_{\mathrm{out}}$ to the corresponding intertwiner
spaces.

\medskip\noindent\textbf{Separating the Four Vertex Constraints.}\par\smallskip
Recall from \eqref{eq:explicit-local-gauge-action} that $U_v(g)$ is the
tensor product of the left or right actions on the edges incident to $v$.
Grouping those actions according to whether the edge is active or outer
gives, on the path neighborhood,
\begin{equation}
  \label{eq:path-vertex-active-outer-action}
  U_v(g)=C_v(g)\otimes U_v^{\mathrm{out}}(g).
\end{equation}
Here $C_v(g)$ acts on the active path edges incident to $v$, with the
identity on the other active edges.  The operator $U_v^{\mathrm{out}}(g)$
acts on the outer-edge representation factors owned by $v$, with the
identity on the outer factors owned by the other three path vertices.
Thus $Y_{\mathrm{out}}$ acts on all four vertices' outer factors, whereas
each $U_v^{\mathrm{out}}(g)$ acts only on the part belonging to $v$.

\begin{claim}[Separated outer-vertex constraints]
\label{clm:separated-outer-vertex-constraints}
Let $X\in\ker K_p$ and write its restriction to the path neighborhood as
\[
  X_{\mathrm{path}}
  =I_p\otimes Y_{\mathrm{out}}
\]
as in \eqref{eq:active-edge-scalar-fixed-observable}.  Then $Y_{\mathrm{out}}$
commutes separately with the outer-edge part of the vertex action at each path vertex:
\[
  [Y_{\mathrm{out}},U_v^{\mathrm{out}}(g)]=0
  \quad\text{for every }g\in\Gamma,
  \qquad v\in\{v_L,v_1,v_2,v_R\}.
\]
\end{claim}

\begin{proof}
Commutation with every Bohr component also gives
\[
  [X,[H_{\cG},S]]
  =\sum_\omega\omega[X,S_\omega]=0
\]
for each active bare coupling $S$.  By
\eqref{eq:edge-coupling-commutator-locality}, Hamiltonian terms not incident
to the active edges commute with $S$, so
$[H_{\cG},S]=[H_{\mathrm{path}},S]$ and only the four path vertices enter.
Since the preceding active-edge argument also gives $[X,S]=0$, Jacobi's
identity yields the explicit
implication
\[
  0=[X,[H_{\mathrm{path}},S]]
   =[[X,H_{\mathrm{path}}],S]
     +[H_{\mathrm{path}},[X,S]]
   =[[X,H_{\mathrm{path}}],S].
\]
For each of the four path vertices
$v\in\{v_L,v_1,v_2,v_R\}$, recall that its star projector is the group
average
\[
  A_v=\frac1q\sum_{g\in\Gamma}U_v(g),
  \qquad q=|\Gamma|.
\]
Using the active/outer decomposition
\eqref{eq:path-vertex-active-outer-action} and the fact that $X$ is the
identity on all three active edges by
\eqref{eq:active-edge-scalar-fixed-observable}, the path Hamiltonian and its
commutator take the form
\[
  H_{\mathrm{path}}
  =-\frac{J}{q}\sum_{v\in\{v_L,v_1,v_2,v_R\}}
     \sum_{g\in\Gamma}C_v(g)\otimes U_v^{\mathrm{out}}(g),
\]
\[
  [X_{\mathrm{path}},H_{\mathrm{path}}]
  =-\frac{J}{q}\sum_{v,g}C_v(g)\otimes
    [Y_{\mathrm{out}},U_v^{\mathrm{out}}(g)].
\]
The preceding Jacobi identity says that this commutator commutes with all
active-edge couplings.  Those couplings generate the full active-edge
matrix algebra, so there is an operator $Z_{\mathrm{out}}$ on the outer factors such that
\[
  [X_{\mathrm{path}},H_{\mathrm{path}}]
  =I_p\otimes Z_{\mathrm{out}}.
\]
We now take the partial trace over the three active-edge factors
$e_-,e_0,e_+$.  This operation removes those three factors and leaves an
operator on the outer-edge factors.  Applied to the expanded commutator in
the preceding display, it gives zero.  Indeed, when $g\ne1_\Gamma$, the
active-edge operator $C_v(g)$ contains a nonidentity left or right regular
translation and therefore
$\operatorname{Tr}_{e_-e_0e_+}C_v(g)=0$.  When $g=1_\Gamma$, the outer
commutator vanishes instead:
\[
  [Y_{\mathrm{out}},U_v^{\mathrm{out}}(1_\Gamma)]
  =[Y_{\mathrm{out}},I]=0.
\]
Consequently
\[
  \operatorname{Tr}_{e_-e_0e_+}
  [X_{\mathrm{path}},H_{\mathrm{path}}]=0.
\]
On the other hand, applying the same partial trace to $I_p\otimes Z_{\mathrm{out}}$ gives
\[
  \operatorname{Tr}_{e_-e_0e_+}
  (I_p\otimes Z_{\mathrm{out}})=q^3Z_{\mathrm{out}},
\]
because each active-edge space has dimension $q=|\Gamma|$.
Hence $q^3Z_{\mathrm{out}}=0$, so $Z_{\mathrm{out}}=0$ and
\[
  [X_{\mathrm{path}},H_{\mathrm{path}}]=0.
\]
Consequently
\begin{equation}
  \label{eq:path-vertex-coefficient-expansion}
  \sum_{v,g}C_v(g)\otimes
  [Y_{\mathrm{out}},U_v^{\mathrm{out}}(g)]=0.
\end{equation}

It remains to separate
\eqref{eq:path-vertex-coefficient-expansion} into these four vertex
conditions.
Recall that $q=|\Gamma|$.  The coefficient family
$\{C_v(g)\}_{v,g}$ has rank $4q-3$.  Its only linear
dependencies identify the four copies of the identity coefficient
$C_v(1_\Gamma)=I$.
Within one vertex family, orthogonality of the regular representation gives
\[
  \langle C_v(g),C_v(h)\rangle_{\mathrm{HS}}
  =q^3\mathbf1_{g=h}.
\]
For two distinct path vertices $v\ne w$, at least one active edge is acted on
by one coefficient family and not by the other.  Taking the trace on that
edge makes the cross inner product zero unless the corresponding group
element is $1_\Gamma$; the shared or remaining active edge then forces the
other group element to be $1_\Gamma$ as well.  Hence
\[
  \langle C_v(g),C_w(h)\rangle_{\mathrm{HS}}=0
  \qquad\text{unless }g=h=1_\Gamma.
\]
The four identity columns coincide and account for exactly three
dependencies, while all nonidentity columns are mutually independent and
independent of that common identity column.

Since $[Y_{\mathrm{out}},U_v^{\mathrm{out}}(1_\Gamma)]=0$ automatically,
applying this linear independence
to \eqref{eq:path-vertex-coefficient-expansion} shows that all remaining
coefficients vanish separately.  This proves
\[
  [Y_{\mathrm{out}},U_v^{\mathrm{out}}(g)]=0
  \qquad(v\text{ a path vertex},\ g\in\Gamma).
\]
\end{proof}

\medskip\noindent\textbf{Identifying the Local Factors.}\par\smallskip
By \Cref{clm:separated-outer-vertex-constraints}, identifying the remaining
$K_p$-fixed directions amounts to taking the simultaneous fixed space of the
four outer vertex actions.  We now make the fixed bra and ket labels
explicit.  On the ordered block $\mathcal X_{p,\boldsymbol{\theta}}$,
$Y_{\mathrm{out}}$ maps the bra outer space to the ket outer space.
The preceding commutation relation therefore reads
\[
  U_v^{\mathrm{out,ket}}(g)Y_{\mathrm{out}}
  =Y_{\mathrm{out}}U_v^{\mathrm{out,bra}}(g).
\]
Here the superscripts indicate the restrictions of the same vertex action
to the output (ket) and input (bra) spaces, respectively; the actions on
the other three vertices are identities.  Thus this is the same condition
as before, with its two block labels displayed.

The outer operator space factors over the four path vertices, so every
$Y_{\mathrm{out}}$ has a finite expansion
\[
  Y_{\mathrm{out}}
  =\sum_\alpha
    Y_{v_L,\alpha}\otimes Y_{v_1,\alpha}
    \otimes Y_{v_2,\alpha}\otimes Y_{v_R,\alpha}.
\]
Each $Y_{v,\alpha}$ is a map
between the bra and ket outer carriers at one vertex $v$, belonging to
\[
  \operatorname{Hom}
  \bigl(U_v^{\mathrm{out,bra}},U_v^{\mathrm{out,ket}}\bigr).
\]
Writing $Y_v$ for a generic such local map, the vertex action fixes it
precisely when
\[
  U_v^{\mathrm{out,ket}}(g)Y_v
  =Y_vU_v^{\mathrm{out,bra}}(g)
  \qquad(g\in\Gamma).
\]
Thus its fixed subspace is
$\operatorname{Hom}_\Gamma
  (U_v^{\mathrm{out,bra}},U_v^{\mathrm{out,ket}})$.
To see explicitly how this constrains the finite sum, define the local
group average on a bra-to-ket map by
\[
  \mathcal P_v^{\mathrm{out}}(Y_v)
  :=\frac1{|\Gamma|}\sum_{g\in\Gamma}
    U_v^{\mathrm{out,ket}}(g)Y_v
    U_v^{\mathrm{out,bra}}(g)^{-1}.
\]
For any $h\in\Gamma$, conjugating this average by the ket and bra actions
of $h$ merely replaces the summation index $g$ by $hg$.  Hence its output
satisfies the intertwining condition.  Conversely, if $Y_v$ already
intertwines the two actions, every summand equals $Y_v$, so the average
leaves it unchanged.  Thus $\mathcal P_v^{\mathrm{out}}$ is a projection
onto the local intertwiner space.

The four averages act on distinct operator factors.  Applying all of them
therefore replaces each factor in a simple tensor by its own group average.
Since $Y_{\mathrm{out}}$ is fixed by each vertex action, it is unchanged by
this operation, and its finite expansion gives
\[
\begin{aligned}
  Y_{\mathrm{out}}
  &=\left(\mathcal P_{v_L}^{\mathrm{out}}\otimes
          \mathcal P_{v_1}^{\mathrm{out}}\otimes
          \mathcal P_{v_2}^{\mathrm{out}}\otimes
          \mathcal P_{v_R}^{\mathrm{out}}\right)(Y_{\mathrm{out}})\\
  &=\sum_\alpha
    \mathcal P_{v_L}^{\mathrm{out}}(Y_{v_L,\alpha})\otimes
    \mathcal P_{v_1}^{\mathrm{out}}(Y_{v_1,\alpha})\otimes
    \mathcal P_{v_2}^{\mathrm{out}}(Y_{v_2,\alpha})\otimes
    \mathcal P_{v_R}^{\mathrm{out}}(Y_{v_R,\alpha}).
\end{aligned}
\]
Every factor on the last line is an intertwiner.  This shows that the
original factors need not individually have been intertwiners, but they
can be replaced by intertwiners without changing $Y_{\mathrm{out}}$.

It remains to identify the four resulting local spaces.  Each endpoint
has two outer legs, whose joint actions are $U_L$ and $U_R$; their factors
are therefore
$\operatorname{Hom}_\Gamma(U_L^{\mathrm{bra}},U_L^{\mathrm{ket}})$ and
$\operatorname{Hom}_\Gamma(U_R^{\mathrm{bra}},U_R^{\mathrm{ket}})$.
At each internal degree-three vertex, two incident edges are active, so
only one outer irrep remains.  Schur's lemma gives
\[
  \dim\operatorname{Hom}_\Gamma(V_\mu,V_\lambda)
  =\mathbf1_{\mu=\lambda}.
\]
Thus an internal factor vanishes if its bra and ket labels differ, and is
the scalar identity space if they agree.  These are the two middle factors
in the stated kernel formula.  Consequently,
$I_p\otimes Y_{\mathrm{out}}\in\mathcal F_{p,\boldsymbol{\theta}}$, proving
the inclusion from the path fixed space into the displayed tensor product.

\medskip\noindent\textbf{The Converse Inclusion.}\par\smallskip
Conversely, let $Y$ belong to the tensor product of the four displayed Hom spaces and put
$X_{\mathrm{path}}=I_p\otimes Y$.  It commutes with every active
bare coupling and with the actions at each of the four path vertices, hence with
the path Hamiltonian.  On each ordered block the factorization over vertices
gives
\[
  H_{\cG}=H_{\mathrm{path}}\otimes I_{\mathrm{ext}}
    +I_{\mathrm{path}}\otimes H_{\mathrm{ext}},
  \qquad
  S=S^{\mathrm{path}}\otimes I_{\mathrm{ext}}
\]
for every active coupling $S$.  Writing $P_\epsilon^{\mathrm{path}}$ and
$P_\eta^{\mathrm{ext}}$ for the spectral projections, the unchanged
exterior energy gives
\[
  S_\omega
  =\sum_{\epsilon-\epsilon'=\omega}\sum_\eta
    \bigl(P_\epsilon^{\mathrm{path}}S^{\mathrm{path}}
          P_{\epsilon'}^{\mathrm{path}}\bigr)
    \otimes P_\eta^{\mathrm{ext}}
  =S_\omega^{\mathrm{path}}\otimes I_{\mathrm{ext}}.
\]
Since $X_{\mathrm{path}}$ commutes with $H_{\mathrm{path}}$ and
$S^{\mathrm{path}}$, it commutes with the spectral projections and hence
with every $S_\omega^{\mathrm{path}}$.  Thus
$X_{\mathrm{path}}\otimes Z_{\mathrm{ext}}\in\ker K_p$ for arbitrary
$Z_{\mathrm{ext}}$, by \cref{lem:davies-fixed-point-commutant}.
Extending by linearity to finite sums proves the reverse inclusion in
\eqref{eq:exact-ordered-path-kernel-with-spectator}.  The resulting kernel
depends only on the representation labels, not on the positive rates,
and is therefore independent of $t$.
\end{proof}

\section{Exact Finite Certificates for the
  \texorpdfstring{$S_3$}{S3} Star Block}
\label{app:s3-finite-certificates}

This appendix is for the finite, reproducible part of the $S_3$
strengthening.  Its role is to certify the local constants used in
\cref{sec:s3-strengthening}.

\subsection{Fusion Rules and Boundary Comparison}
\label{app:s3-fusion-inventory}

We recall the fusion rules from the
\hyperref[ex:s3-fusion-vertex-weights]{$S_3$ fusion-rule example in the preliminaries}.
The three irreducible representations of $S_3$ are the trivial
representation $\mathbf 1$, the sign representation $\mathrm{sgn}$, and the
two-dimensional standard representation $\tau$.  They are all self-dual,
and their dimensions and fusion rules are
\begin{equation}
  \label{eq:s3-fusion-inventory}
  d_{\mathbf 1}=d_{\mathrm{sgn}}=1,
  \qquad d_\tau=2,
  \qquad
  \begin{array}{c|ccc}
    \otimes & \mathbf 1 & \mathrm{sgn} & \tau\\ \hline
    \mathbf 1 & \mathbf 1 & \mathrm{sgn} & \tau\\
    \mathrm{sgn} & \mathrm{sgn} & \mathbf 1 & \tau\\
    \tau & \tau & \tau & \mathbf 1\oplus\mathrm{sgn}\oplus\tau
  \end{array}.
\end{equation}
In particular, all fusion multiplicities are zero or one.  Put
\[
  N(\lambda_1,\lambda_2,\lambda_3)
  :=\dim\operatorname{Inv}
    (V_{\lambda_1}\otimes V_{\lambda_2}\otimes V_{\lambda_3}).
\]
Tracing one vertex factor in the Peter--Weyl block with incident labels
$\lambda_1,\lambda_2,\lambda_3$ gives the unnormalized local weight
\begin{equation}
  \label{eq:s3-local-fusion-weight}
  z_t(\lambda_1,\lambda_2,\lambda_3)
  :=d_{\lambda_1}d_{\lambda_2}d_{\lambda_3}
    +tN(\lambda_1,\lambda_2,\lambda_3).
\end{equation}
This is the same vertex weight that appears in
\eqref{eq:edge-label-conditional-distribution}; the present appendix keeps its full
three-label dependence instead of passing immediately to a one-edge conditional
distribution.

For the star labels $\boldsymbol{\lambda}_v$ and the six-label boundary
condition $\boldsymbol b$ introduced in \cref{sec:s3-star-certificate}, the exact
conditioned distribution is \eqref{eq:s3-conditioned-star-distribution}.
Every factor in \eqref{eq:s3-conditioned-star-distribution} is an integer polynomial
in $t$, and $Z_{\boldsymbol b}(t)>0$ for $t\ge0$.

The three internal labels $(\lambda_1,\lambda_2,\lambda_3)$ have
$3^3=27$ possible configurations, while the six boundary labels
$\boldsymbol b=(b_{i,j})_{1\le i\le3,\,1\le j\le2}$ have $3^6=729$.
To compare two boundary conditions differing in one label, choose a
coordinate $b_{i,j}$ and keep the other five labels identical.
For example, if the changed coordinate is $b_{1,1}$, then
\[
  \begin{aligned}
    \boldsymbol b&=(a,b_{1,2},b_{2,1},b_{2,2},b_{3,1},b_{3,2}),\\
    \boldsymbol b'&=(a',b_{1,2},b_{2,1},b_{2,2},b_{3,1},b_{3,2}),
  \end{aligned}
\]
where $a\ne a'$ belong to $\{\mathbf 1,\mathrm{sgn},\tau\}$.
There are six choices of the changed coordinate, $3^5$ assignments of
the unchanged labels, and $\binom32$ unordered choices of $\{a,a'\}$.
Since the comparison is symmetric in $\boldsymbol b,\boldsymbol b'$,
each unordered pair need only be checked once.  Thus the exhaustive
inventory contains
\begin{equation}
  \label{eq:s3-boundary-comparison-count}
  6\cdot3^5\cdot\binom32=4374
\end{equation}
one-label boundary comparisons.

\subsection{Exact Interval Certificate on
  \texorpdfstring{$0\le t\le2$}{0 <= t <= 2}}
\label{app:s3-interval-certificate}

The first local constant controls the response of the star distribution
to changing one boundary label, as required by the Wasserstein assumption
\eqref{eq:s3-certificate-wasserstein-assumption}.
Recall that $\mu_{\boldsymbol b}^t$ in
\eqref{eq:s3-conditioned-star-distribution} is the joint distribution of
the three star labels with the six boundary labels fixed to $\boldsymbol b$.
Suppose $\boldsymbol b$ and $\boldsymbol b'$ differ in one outer label
attached to the $i$th star edge, and let
$\overline\mu_{\boldsymbol b,i}^t$ be the marginal of
$\mu_{\boldsymbol b}^t$ on the
other two internal labels; for example,
\[
  \overline\mu_{\boldsymbol b,1}^t(\lambda_2,\lambda_3)
  =\sum_{\lambda_1\in\widehat{S_3}}
    \mu_{\boldsymbol b}^t(\lambda_1,\lambda_2,\lambda_3).
\]
Define the residual transport bound
\begin{equation}
  \label{eq:s3-residual-transport-bound}
  \begin{aligned}
  \mathcal B_{\boldsymbol b,\boldsymbol b'}(t)
  &:=\lVert\mu_{\boldsymbol b}^t-\mu_{\boldsymbol b'}^t\rVert_{\mathrm{TV}}
    +2\lVert
      \overline\mu_{\boldsymbol b,i}^t
      -\overline\mu_{\boldsymbol b',i}^t
    \rVert_{\mathrm{TV}},\\
  B_2&:=
  \max_{\substack{\boldsymbol b,\boldsymbol b'\in\widehat{S_3}^{\,6}\\
                   \boldsymbol b\sim\boldsymbol b',\ 0\le t\le2}}
    \mathcal B_{\boldsymbol b,\boldsymbol b'}(t).
  \end{aligned}
\end{equation}
Thus $B_2$ is the maximum of the residual transport bound over all
boundary pairs differing in exactly one label and over the entire interval
$0\le t\le2$.  The boundary pairs are precisely the $4374$ comparisons
counted in \eqref{eq:s3-boundary-comparison-count}.
The coupling estimate \eqref{eq:s3-star-wasserstein-from-tv}, proved below
in the proof of \cref{prop:exact-s3-star-certificate}, gives
\[
  W_1(\mu_{\boldsymbol b}^t,\mu_{\boldsymbol b'}^t)
  \le \mathcal B_{\boldsymbol b,\boldsymbol b'}(t)
  \le B_2
  \qquad(\boldsymbol b\sim\boldsymbol b',\ 0\le t\le2).
\]
Thus showing $B_2<1/2$ verifies the Wasserstein assumption
\eqref{eq:s3-certificate-wasserstein-assumption} used in the star path
coupling.

The second local constant concerns the single-edge heat-bath chain inside a
conditioned star.  Consider updating the first star edge, with all six
boundary labels fixed and the other two internal labels set to
$\lambda_2=\alpha$ and $\lambda_3=\gamma$.
Write $\mu_{\boldsymbol b,1}^t(\,\cdot\mid\alpha,\gamma)$ for this
single-edge conditional distribution, where the subscript $1$ denotes
the updated edge:
\begin{equation}
  \label{eq:s3-internal-conditional-distribution}
  \begin{aligned}
  \mu_{\boldsymbol b,1}^t(\lambda\mid\alpha,\gamma)
  &:=\mu_{\boldsymbol b}^t(\lambda_1=\lambda
      \mid\lambda_2=\alpha,\lambda_3=\gamma)\\
  &\propto
  z_t(\lambda,\alpha,\gamma)
  z_t(\lambda^*,b_{1,1},b_{1,2}).
  \end{aligned}
\end{equation}
The factors involving the other four boundary labels do not depend on
$\lambda$ and cancel in the conditional normalization, so the dependence
on $\boldsymbol b$ is only through $(b_{1,1},b_{1,2})$.
Permuting the three star edges gives the same family of conditional
distributions for either of the other updated edges.
Unlike the boundary comparison above, we now keep the boundary fixed
and change one other internal label, $\alpha$.
The maximum response to this change is
\begin{equation}
  \label{eq:s3-internal-star-floor-definition}
  \begin{aligned}
  c_{\mathrm{int},2}
  &:=\max_{\substack{\alpha\ne\alpha',\,\gamma,\boldsymbol b\\0\le t\le2}}
    \lVert
      \mu_{\boldsymbol b,1}^t(\,\cdot\mid\alpha,\gamma)
      -\mu_{\boldsymbol b,1}^t(\,\cdot\mid\alpha',\gamma)
    \rVert_{\mathrm{TV}},\\
  \delta_{\mathrm{int},2}
  &:=1-2c_{\mathrm{int},2}.
  \end{aligned}
\end{equation}
Only the two boundary labels $(b_{1,1},b_{1,2})$ need to be varied,
giving $3^2$ choices.  Together with the three choices of $\gamma$ and
the $\binom32$ unordered pairs $\{\alpha,\alpha'\}$, this gives
$3^2\cdot3\cdot\binom32=81$ comparisons in the internal inventory.

\begin{proposition}[Exact $S_3$ star certificate]
\label{prop:exact-s3-star-certificate}
The two hypotheses of \cref{prop:s3-finite-certificate-to-gns} hold with
$T=2$, $B_T=B_2$, and
$\delta_{\mathrm{int},T}=\delta_{\mathrm{int},2}$.
Uniformly for $0\le t\le2$, the boundary-response bound is
\begin{equation}
  \label{eq:exact-s3-star-constants}
  W_1(\mu_{\boldsymbol b}^t,\mu_{\boldsymbol b'}^t)
  \le B_2=\frac{2818}{7395}<\frac12
  \qquad
  (\boldsymbol b\sim\boldsymbol b'),
\end{equation}
and the conditioned internal heat-bath gap satisfies
\begin{equation}
  \label{eq:exact-s3-internal-constants}
  \gap(\mathsf K_{\mathrm{int},\boldsymbol b}^{t,\star})
  \ge\delta_{\mathrm{int},2}=\frac{11}{35}>0
  \qquad\text{for every }\boldsymbol b.
\end{equation}
These supply the inputs $B_T$ and $\delta_{\mathrm{int},T}$ in
\eqref{eq:s3-certificate-wasserstein-assumption} and
\eqref{eq:s3-certificate-internal-gap-assumption}, respectively.
\end{proposition}

Before proving the proposition, we list the conditions and values verified by exact-arithmetic computation and used in the proof.
We first introduce the polynomial expressions needed to state these checks.
For a state $\boldsymbol{\lambda}_v$ and adjacent boundary conditions
$\boldsymbol b,\boldsymbol b'$, clearing the positive denominators gives
\begin{equation}
  \label{eq:s3-probability-difference-polynomial}
  \mu_{\boldsymbol b}^t(\boldsymbol{\lambda}_v)
  -\mu_{\boldsymbol b'}^t(\boldsymbol{\lambda}_v)
  =
  \frac{
    w_{\boldsymbol b}^t(\boldsymbol{\lambda}_v)Z_{\boldsymbol b'}(t)
    -w_{\boldsymbol b'}^t(\boldsymbol{\lambda}_v)Z_{\boldsymbol b}(t)
  }{Z_{\boldsymbol b}(t)Z_{\boldsymbol b'}(t)}.
\end{equation}
For this fixed boundary pair, denote the numerator in
\eqref{eq:s3-probability-difference-polynomial} by
$q_{\boldsymbol\lambda}(t)$, suppressing its dependence on
$\boldsymbol b,\boldsymbol b'$ and writing $\boldsymbol\lambda$ for the
three star labels.  If $\boldsymbol\eta$ denotes the two labels other than
the $i$th, summing over $\lambda_i$ gives the difference between their
marginal distributions:
\[
  \overline\mu_{\boldsymbol b,i}^t(\boldsymbol\eta)
  -\overline\mu_{\boldsymbol b',i}^t(\boldsymbol\eta)
  =\frac{\displaystyle\sum_{\lambda_i\in\widehat{S_3}}
      q_{\boldsymbol\lambda}(t)}
    {Z_{\boldsymbol b}(t)Z_{\boldsymbol b'}(t)}
  =\frac{\bar q_{\boldsymbol\eta}(t)}
    {Z_{\boldsymbol b}(t)Z_{\boldsymbol b'}(t)},
\]
where $\bar q_{\boldsymbol\eta}(t):=
\sum_{\lambda_i\in\widehat{S_3}}q_{\boldsymbol\lambda}(t)$,
with $\boldsymbol\eta$ held fixed in the sum.
The denominator is independent of $\lambda_i$, so only the numerators
need to be summed.
The total-variation formula \eqref{eq:total-variation-distance} and
\eqref{eq:s3-residual-transport-bound} therefore give
\[
  \mathcal B_{\boldsymbol b,\boldsymbol b'}(t)
  =
  \frac{
    \displaystyle\sum_{\boldsymbol\lambda\in\widehat{S_3}^{\,3}}
      |q_{\boldsymbol\lambda}(t)|
    +2\displaystyle\sum_{\boldsymbol\eta\in\widehat{S_3}^{\,2}}
      |\bar q_{\boldsymbol\eta}(t)|
  }{2Z_{\boldsymbol b}(t)Z_{\boldsymbol b'}(t)}.
\]
Thus the individual probability differences are summed over all $27$
states and all $9$ marginal states.  Once $q_{\boldsymbol\lambda}$ and
$\bar q_{\boldsymbol\eta}$ have fixed signs on $[0,2]$, i.e., each
polynomial is either nonnegative throughout the interval or nonpositive
throughout the interval, each absolute value can be replaced by the
corresponding polynomial or its negative.  Choose signs
$\varepsilon_{\boldsymbol\lambda},\bar\varepsilon_{\boldsymbol\eta}
\in\{-1,1\}$, independent of $t$, so that
$\varepsilon_{\boldsymbol\lambda}q_{\boldsymbol\lambda}(t)\ge0$ and
$\bar\varepsilon_{\boldsymbol\eta}\bar q_{\boldsymbol\eta}(t)\ge0$
throughout $[0,2]$; either sign may be chosen for an identically zero
polynomial.  These signs depend on the fixed boundary pair, whose indices
are suppressed.  The resulting numerator and denominator are
\[
  \begin{aligned}
  N_{\boldsymbol b,\boldsymbol b'}(t)
  &:=\sum_{\boldsymbol\lambda\in\widehat{S_3}^{\,3}}
      \varepsilon_{\boldsymbol\lambda}q_{\boldsymbol\lambda}(t)
    +2\sum_{\boldsymbol\eta\in\widehat{S_3}^{\,2}}
      \bar\varepsilon_{\boldsymbol\eta}\bar q_{\boldsymbol\eta}(t),\\
  D_{\boldsymbol b,\boldsymbol b'}(t)
  &:=2Z_{\boldsymbol b}(t)Z_{\boldsymbol b'}(t).
  \end{aligned}
\]
These are integer polynomials, and
\[
  \mathcal B_{\boldsymbol b,\boldsymbol b'}(t)
  =\frac{N_{\boldsymbol b,\boldsymbol b'}(t)}
         {D_{\boldsymbol b,\boldsymbol b'}(t)},
  \qquad D_{\boldsymbol b,\boldsymbol b'}(t)>0.
\]
Its derivative has the sign of the polynomial
\[
  N_{\boldsymbol b,\boldsymbol b'}'(t)D_{\boldsymbol b,\boldsymbol b'}(t)
  -N_{\boldsymbol b,\boldsymbol b'}(t)D_{\boldsymbol b,\boldsymbol b'}'(t).
\]
For the internal comparisons, the same construction is applied to
$\lVert\mu_{\boldsymbol b,1}^t(\,\cdot\mid\alpha,\gamma)-
\mu_{\boldsymbol b,1}^t(\,\cdot\mid\alpha',\gamma)\rVert_{\mathrm{TV}}$.

The exact computation documented in \cref{app:s3-reproducible-computation}
certifies the following four facts on $[0,2]$.
\begin{enumerate}[label=(F\arabic*),ref=F\arabic*,leftmargin=3em]
\item \label{item:s3-f1-fixed-signs}
\emph{Fixed signs.}
For every boundary comparison, each numerator in
\eqref{eq:s3-probability-difference-polynomial}, and each corresponding
two-label marginal numerator, has a fixed sign on $[0,2]$.
This is certified by its Bernstein coefficients, which are either all
nonnegative or all nonpositive.

\item \label{item:s3-f2-boundary-monotonicity}
\emph{Boundary-response derivatives.}
After the signs in (\ref{item:s3-f1-fixed-signs}) are used to form $N$ and $D$,
all Bernstein coefficients of $N'D-ND'$ on $[0,2]$ are nonnegative.
The enumeration yields $288$ distinct comparison data sets after identical
polynomial data are deduplicated; these cover all $4374$ boundary comparisons,
as explained in
\cref{app:s3-reproducible-computation}.

\item \label{item:s3-f3-boundary-endpoint}
\emph{Boundary endpoint maximum.}
Exact evaluation gives
\[
  \max_{\boldsymbol b\sim\boldsymbol b'}
  \mathcal B_{\boldsymbol b,\boldsymbol b'}(2)=\frac{2818}{7395}.
\]
One maximizing pair is
$\boldsymbol b=(\mathbf1,\mathbf1,\mathbf1,\mathbf1,\mathbf1,\mathbf1)$ and
$\boldsymbol b'=(\tau,\mathbf1,\mathbf1,\mathbf1,\mathbf1,\mathbf1)$.

\item \label{item:s3-f4-internal-response}
\emph{Internal-response checks.}
For all $81$ internal comparisons, the probability-difference numerators
have fixed signs and the resulting total-variation derivative numerators
have nonnegative Bernstein coefficients on $[0,2]$.
Exact endpoint evaluation gives
\[
  \max_{\alpha\ne\alpha',\,\gamma,\boldsymbol b}
  \lVert\mu_{\boldsymbol b,1}^{2}(\,\cdot\mid\alpha,\gamma)
       -\mu_{\boldsymbol b,1}^{2}(\,\cdot\mid\alpha',\gamma)\rVert_{\mathrm{TV}}
  =\frac{12}{35}.
\]
\end{enumerate}
No interval subdivision is needed for these $S_3$ checks.
The machine-readable record also gives the number of extremizers and
representative witnesses.

\begin{proof}[Proof of \cref{prop:exact-s3-star-certificate}]
We first relate the transport bound to the required Wasserstein input.
Fix adjacent boundary conditions differing at a label attached to the
$i$th star edge.
We construct a joint distribution of two star configurations with
marginals $\mu_{\boldsymbol b}^t$ and $\mu_{\boldsymbol b'}^t$.
Both boundary conditions and both marginal distributions remain fixed:
we only choose how to assign joint probability to pairs of configurations,
not how to evolve either configuration.
By the definition of $W_1$ in \eqref{eq:s3-hamming-wasserstein}, it suffices
to construct a coupling whose expected Hamming cost is at most the
right-hand side of \eqref{eq:s3-star-wasserstein-from-tv} below.
We first assign probability to identical pairs at zero cost, then
allocate the remaining mass to favor pairs sharing the other two labels.

For each configuration $\boldsymbol\lambda$, assign joint probability
$c(\boldsymbol\lambda):=\min\{\mu_{\boldsymbol b}^t(\boldsymbol\lambda),
\mu_{\boldsymbol b'}^t(\boldsymbol\lambda)\}$ to the identical pair
$(\boldsymbol\lambda,\boldsymbol\lambda)$, at zero Hamming cost.
This assignment is what we mean by matching the common mass.
In what follows, we suppress the argument $\boldsymbol\lambda$ and write
$c$ for the resulting common-mass measure.
The residual measures
$r:=\mu_{\boldsymbol b}^t-c$ and $s:=\mu_{\boldsymbol b'}^t-c$
each have total mass
$m:=\lVert\mu_{\boldsymbol b}^t-\mu_{\boldsymbol b'}^t\rVert_{\mathrm{TV}}$.
They record the mass still to be assigned to pairs, rather than a change
to either original distribution.
Let $\bar r,\bar s,\bar c$ denote the corresponding marginals on the
two star labels other than the $i$th.  Since the same common mass was
removed from both sides,
\[
  \bar r-\bar s
  =(\overline\mu_{\boldsymbol b,i}^t-\bar c)
   -(\overline\mu_{\boldsymbol b',i}^t-\bar c)
  =\overline\mu_{\boldsymbol b,i}^t-\overline\mu_{\boldsymbol b',i}^t.
\]
Next construct a joint measure with marginals $r$ and $s$ that assigns
as much mass as possible to pairs sharing the two labels other than the
$i$th.  For each two-label configuration $\boldsymbol\eta$, assign total
mass $\min\{\bar r(\boldsymbol\eta),\bar s(\boldsymbol\eta)\}$ to such
pairs; their $i$th labels need not agree.
Since both residual measures have total mass $m$, the mass for which
the two-label configurations fail to agree is
\[
  \begin{aligned}
  d&:=m-\sum_{\boldsymbol\eta\in\widehat{S_3}^{\,2}}
       \min\{\bar r(\boldsymbol\eta),\bar s(\boldsymbol\eta)\}\\
   &=\frac12\sum_{\boldsymbol\eta\in\widehat{S_3}^{\,2}}
       |\bar r(\boldsymbol\eta)-\bar s(\boldsymbol\eta)|
    =\lVert\overline\mu_{\boldsymbol b,i}^t
       -\overline\mu_{\boldsymbol b',i}^t\rVert_{\mathrm{TV}}.
  \end{aligned}
\]
Every residual pair costs at most one for the $i$th coordinate; only
the mass $d$ can cost up to two more for the other coordinates.
The total expected Hamming cost is thus at most $m+2d$, giving
\begin{equation}
  \label{eq:s3-star-wasserstein-from-tv}
  W_1(\mu_{\boldsymbol b}^t,\mu_{\boldsymbol b'}^t)
  \le
  \lVert\mu_{\boldsymbol b}^t-\mu_{\boldsymbol b'}^t\rVert_{\mathrm{TV}}
  +2\lVert
    \overline\mu_{\boldsymbol b,i}^t
    -\overline\mu_{\boldsymbol b',i}^t
  \rVert_{\mathrm{TV}}.
\end{equation}

By (\ref{item:s3-f1-fixed-signs}), the absolute values in both
total-variation terms of $\mathcal B_{\boldsymbol b,\boldsymbol b'}$
can be removed with a fixed choice of signs throughout $[0,2]$.
Since $D>0$, its derivative is $(N'D-ND')/D^2$;
(\ref{item:s3-f2-boundary-monotonicity}) therefore makes every boundary
response nondecreasing.  Each interval maximum is attained at $t=2$,
so (\ref{item:s3-f3-boundary-endpoint}) and the definition
\eqref{eq:s3-residual-transport-bound} give
$B_2=2818/7395$.
The coupling estimate \eqref{eq:s3-star-wasserstein-from-tv} then gives
\eqref{eq:exact-s3-star-constants}.

Likewise, the sign and derivative checks in
(\ref{item:s3-f4-internal-response}) make each internal response
nondecreasing.  Its endpoint evaluation thus gives
$c_{\mathrm{int},2}=12/35$.
Each internal edge has only two other internal labels that can influence
its conditional distribution, so the unit-rate Dobrushin bound gives
an internal heat-bath gap of at least
$1-2c_{\mathrm{int},2}$, uniformly in the boundary condition.

Consequently,
\[
  \delta_{\mathrm{int},2}
  =1-2c_{\mathrm{int},2}=\frac{11}{35}>0,
\]
which proves \eqref{eq:exact-s3-internal-constants}.
\end{proof}

\paragraph{Comparison with one-site influence.}
Each internal edge has two other edges inside its star, whereas a global
edge has four neighboring edges.  The internal influence maximum is
attained at $t=2$, so the value computed above gives the contrast
\begin{equation}
  \label{eq:s3-endpoint-certificate-contrast}
  \underbrace{4c_{\mathrm{int},2}}_{\text{one-site total influence}}
  =\frac{48}{35}>1,
  \qquad
  \underbrace{2-4B_2}_{\text{star-coupling contraction margin}}
  =\frac{3518}{7395}>0.
\end{equation}
Thus the one-site criterion is inconclusive, while the star criterion
still yields a positive gap.  These are consequences of the certificate,
not additional hypotheses of
\cref{prop:s3-finite-certificate-to-gns}.

\subsection{Reproducible Computation}
\label{app:s3-reproducible-computation}

The verification package is available on GitHub as version
v1.0.0~\cite{hayakawa2026s3certificate}.  Relative to the repository root,
the self-contained exact-arithmetic script and machine-readable record are:
\begin{center}
\small
\path{s3-star/verify_s3_star_certificate.py}\\
\path{s3-star/s3_star_certificate.json}.
\end{center}
The script is specific to $S_3$, includes its character data, and uses only
the Python standard library.  It uses integer polynomials and rational
arithmetic for every sign, derivative, endpoint, and displayed constant.
Floating-point values in the record are included only as readable metadata.

Running the following command from the repository root
\begin{center}
\small
\texttt{python3 }%
\path{s3-star/verify_s3_star_certificate.py}%
\texttt{ --check-only}
\end{center}
recomputes the complete inventory and fails unless it agrees with the stored
record.

For $S_3$, all $729$ boundary conditions and $4374$ adjacent-condition
comparisons are enumerated.  Comparisons with identical weight and
partition polynomials on both sides and the same coordinate $i$ summed
out for the two-label marginal reuse a single verification.  In this way,
the script identifies $288$ distinct polynomial checks.  Only repeated computations are
avoided; these checks cover all $4374$ boundary comparisons.
The script also checks all $81$ internal-influence comparisons.
The record gives
these counts, zero unresolved sign or derivative failures, and representative
extremal boundary conditions supporting \eqref{eq:exact-s3-star-constants}.

The same directory includes a concise reproduction guide and SHA-256 hashes
of the distributed files.
The analytic passage from the model to
\eqref{eq:s3-conditioned-star-distribution}, the star path-coupling estimate, and the
edge-to-star comparison remain in the paper.  The verification package verifies only
the exhaustive finite polynomial statement isolated in the proof of
\cref{prop:exact-s3-star-certificate}.

\section{Quantitative Comparisons with Existing Bounds}
\label{app:baseline-calibrations}

This appendix records two quantitative comparisons used in
\cref{sec:new-intro-significance}: the quantum-to-classical Davies
comparison of \cite{basso2025quantum} and the generic all-to-all
high-temperature condition of \cite{bergamaschi2026fast}.  

\subsection{Quantum-to-Classical Davies Comparison}
\label{app:arithmetic-progressions}

We first explain how the spectral comparison theorem of
\cite{basso2025quantum} interacts with the present Hamiltonian family.

For a graph $\cG$, let $D_{\cG}$ denote the largest length of a proper
arithmetic progression contained in the spectrum of $H_{\cG}$.  Recall that
``proper'' means that the common difference is nonzero.

\begin{lemma}[An arithmetic progression of length proportional to the number of vertices]
\label{lem:volume-order-arithmetic-progression}
Let $\Gamma$ be a fixed nontrivial finite group, let $J>0$, and let $\cG$ be a finite
simple cubic graph with $N:=|V(\cG)|$.  Then
\begin{equation}
  \label{eq:volume-order-arithmetic-progression}
  \frac{N}{6}+1\le D_{\cG}\le N+1.
\end{equation}
In particular, $D_{\cG}=\Theta(N)$ uniformly over the graph family.
\end{lemma}

\begin{proof}
For the lower bound, we first find many edges with disjoint endpoints.
Choose a maximal matching $M$ in $\cG$.  By maximality, every vertex is either incident to
an edge of $M$ or adjacent to an endpoint of one.  Since the graph is cubic,
each endpoint has at most two neighbors outside its matched edge.  Thus the
two endpoints of one matched edge, together with all their unmatched
neighbors, account for at most $2+2+2=6$ vertices.  Charging every vertex to
one nearby matched edge therefore gives
\[
  |M|\ge \frac{N}{6}.
\]

We now assign labels so that each selected edge contributes exactly two
violated vertex constraints.  Fix a nontrivial irrep
$\sigma\in\widehat\Gamma$.  For each
$0\le k\le |M|$, choose $k$ edges of $M$, assign the Peter--Weyl label
$\sigma$ to those edges, and assign the trivial irrep to every other edge.
The tensor product of the prescribed nonzero isotypic summands is itself a
nonzero simultaneous invariant subspace for the commuting projectors
$\{A_v\}_{v\in V}$.

At a vertex not incident to a selected edge, all three incident labels are
trivial and $A_v$ has eigenvalue one.  At a vertex incident to a selected
edge, exactly one incident label is nontrivial.  Depending on the edge
orientation, the corresponding vertex representation is either
$V_\sigma$ or its dual $V_\sigma^*$; both are nontrivial irreducible
representations and hence contain no invariant vector.  Therefore $A_v$
has eigenvalue zero on the selected isotypic block.
Since the selected edges have disjoint endpoints, exactly $2k$ vertices
have eigenvalue zero and the remaining $N-2k$ have eigenvalue one.
Thus this nonzero joint eigenspace has energy
\[
  E_k=-J(N-2k),
  \qquad 0\le k\le |M|.
\]
These energies form a proper arithmetic progression of length $|M|+1$
and common difference $2J$, proving the lower bound.

For the upper bound, the commuting projectors $A_v$ have eigenvalues zero
or one, so
\[
  \operatorname{spec}(H_{\cG})\subseteq
  \{-Jm:0\le m\le N\}.
\]
A proper arithmetic progression has distinct entries and thus has length at
most $N+1$.
\end{proof}

The general temperature-independent comparison theorem of
\cite{basso2025quantum} gives, in its notation,
\begin{equation}
  \label{eq:basso-temperature-independent-comparison}
  \lambda_{\mathcal L}
  \ge \frac{1}{2D_{\cG}}\lambda_{\mathcal L,0},
\end{equation}
where $\lambda_{\mathcal L,0}$ is the gap restricted to the full
zero-Bohr-frequency space
\[
  V_0(H_{\cG})=\{X:[X,H_{\cG}]=0\}.
\]
By \cref{lem:volume-order-arithmetic-progression}, the guaranteed comparison
factor in \eqref{eq:basso-temperature-independent-comparison} is
$\Theta(N^{-1})$ on the present family.  Consequently, this
temperature-independent comparison alone does not transfer a uniform
$V_0(H_{\cG})$-gap into a volume-independent full Davies gap.
The long progression itself does not imply slow mixing; it identifies
the inverse-size factor in this particular comparison bound.

\begin{remark}[Frequencywise comparison away from infinite temperature]
\label{rem:basso-frequencywise-comparison}
For comparison, the frequencywise bound of \cite{basso2025quantum} is
\[
  \lambda_{\mathcal L,\omega}
  \ge
  \frac12\max\left\{
    \frac1{D_\omega},\;1-e^{-\beta|\omega|}
  \right\}\lambda_{\mathcal L,0}
  \qquad(\omega\ne0),
\]
where $D_\omega$ is the largest progression length with common difference
$\omega$.  Since every nonzero Bohr frequency satisfies $|\omega|\ge J$,
this yields a size-independent comparison factor for $\beta\ge\beta_0>0$.
Its application, however, requires a gap bound on the entire zero-frequency
sector $V_0(H_{\cG})$, beyond what our label-interface estimate provides.
At $\beta=0$, the $1/D_\omega$ term remains available, but its size
dependence prevents this bound alone from providing a volume-uniform
comparison throughout an interval starting at infinite temperature.
\end{remark}

\subsection{A General All-to-All High-Temperature Condition}
\label{app:all-to-all-envelope}

The all-to-all high-temperature theorem of
\cite{bergamaschi2026fast} gives a second useful calibration.  In the
notation of that result, a Hamiltonian on sites consisting of $\mathsf q$
qubits each, with local dimension $2^{\mathsf q}$, interaction locality $\mathsf k$, interaction degree
$\mathsf d$, and strength $\mathsf J$ lies in the proved regime when
\begin{equation}
  \label{eq:generic-all-to-all-envelope}
  \beta<\frac{1}{\mathsf c^{\mathsf q\mathsf k}\mathsf d\mathsf J}
\end{equation}
for a universal constant $\mathsf c\ge1$.  The cubic star Hamiltonian
has the formal locality data
\[
  \mathsf k=3,
  \qquad
  \mathsf d=4,
  \qquad
  \mathsf J=J.
\]
If, for this comparison only, the edge Hilbert space had qubit dimension
$|\Gamma|=2^{s_\Gamma}$, then
\eqref{eq:generic-all-to-all-envelope} would imply
\begin{equation}
  \label{eq:formal-star-all-to-all-envelope}
  \beta J<\frac{1}{4\mathsf c^{3s_\Gamma}}\le\frac14.
\end{equation}
By contrast, the common fixed-group theorem in this paper holds for
$\beta J<\log(5/3)$, and the $S_3$ calculation reaches
$t=2$, equivalently $\beta J=\log3>1/4$.  Thus even replacing
$\mathsf c^{3s_\Gamma}$ by its lower bound $1$ in
\eqref{eq:formal-star-all-to-all-envelope} gives only $\beta J<1/4$,
which does not reach either the common interval or the displayed
non-Abelian point.

This is a calibration of sufficient high-temperature conditions, not a
literal theorem specialization.  The result of
\cite{bergamaschi2026fast} concerns a Chen--Kastoryano--Gily\'en (CKG)
sampler \cite{chen2023exact} generated by single-site
Pauli updates, whereas the present theorem concerns the specified complete
star Davies bath.  Moreover, $|S_3|=6$ is not a power of two.  Applying
the cited qubit theorem to this local Hilbert space would
require either an arbitrary-qudit extension or a gap-preserving encoding
whose effect on locality and interaction strength is controlled.
Thus we compare sufficient temperature conditions, without establishing
a separation between the two dynamics.

\end{document}